\documentclass[12pt,oneside]{book}
\usepackage{amsmath}
\usepackage{amsthm}
\usepackage{graphics}
\usepackage{epsfig,amssymb,latexsym,verbatim}
\usepackage{graphicx}
\usepackage{multirow}
\usepackage{multicol}
\usepackage{float}
\usepackage{booktabs}
\usepackage{fancyhdr}
\usepackage{enumerate,verbatim}
\usepackage[sort&compress]{natbib}
\usepackage{rotating}
\usepackage{hyperref}
\usepackage{subfigure}
\usepackage{ulem}
\usepackage{url}
\usepackage{threeparttable}
\usepackage{xr}
\usepackage{multirow}
\usepackage{multicol}
\usepackage{wrapfig}
\usepackage{lscape}
\usepackage{rotating}
\usepackage{epstopdf}
\usepackage{subfigure}
\usepackage{amsfonts}
\usepackage[usenames,dvipsnames]{color}
\usepackage{times,color}
\usepackage{amsbsy}
\usepackage{amsthm}
\usepackage[sf,small,center, compact]{titlesec}

\usepackage[lined,algonl,boxed]{algorithm2e}

\def\bs{\boldsymbol}

\definecolor{red}{rgb}{1,0,0}
\definecolor{green}{rgb}{0,1,0}
\definecolor{blue}{rgb}{0,0,1}

\def\red{\textcolor{red}}
\definecolor{lxy}{RGB}{180,0,180}

\floatstyle{ruled}
\newfloat{algorithm}{tbp}{loa}
\floatname{algorithm}{Algorithm}
\def\bs{\boldsymbol}

\DeclareMathOperator*{\argmin}{arg\,min}

\newcommand{\thetheorem}{{\thesection. \arabic{theorem}}}
\newcommand{\thelemma}{{\thesection. \arabic{lemma}}}
\newcommand{\theproposition}{{\thesection. \arabic{proposition}}}
\newcommand{\thecorollary}{{\thesection. \arabic{corollary}}}
\newtheorem{theorem}{{\sc Theorem}}
\newtheorem{lemma}{{\sc Lemma}}
\newtheorem{corollary}{{\sc Corollary}}

\begin{document}
	\renewcommand{\baselinestretch}{1.2}
	\markboth{\hfill{\footnotesize\rm Yueying Wang, Guannan Wang, Li Wang and R. Todd Ogden}\hfill}
	{\hfill {\footnotesize\rm Simultaneous Confidence Corridors for Mean Functions in Functional Data Analysis of Imaging Data} \hfill}
	\renewcommand{\thefootnote}{}
	$\ $\par \fontsize{10.95}{14pt plus.8pt minus .6pt}\selectfont
	\vspace{0.8pc} \centerline{\large\bf Simultaneous Confidence Corridors for Mean Functions}
	\centerline{\large\bf in Functional Data Analysis of Imaging Data}
	\vspace{.4cm} \centerline{Yueying Wang$^{a}$, Guannan Wang$^{b}$, Li Wang$^{a}$ and R. Todd Ogden$^{c}$
		\footnote{\emph{Address for correspondence}: Li Wang, Department of Statistics and the Statistical Laboratory, Iowa State University, Ames, IA, USA. Email: lilywang@iastate.edu}} \vspace{.4cm} \centerline{\it $^{a}$Iowa State University, $^{b}$College of William \& Mary and $^{c}$Columbia University} \vspace{.55cm}
	\fontsize{9}{11.5pt plus.8pt minus .6pt}\selectfont
	
\begin{quotation}
	\noindent \textit{Abstract:} Motivated by recent work involving the analysis of biomedical imaging data, we present a novel procedure for constructing simultaneous confidence corridors for the mean of imaging data. We propose to use flexible bivariate splines over triangulations to handle irregular domain of the images that is common in brain imaging studies and in other biomedical imaging applications. The proposed spline estimators of the mean functions are shown to be consistent and asymptotically normal under some regularity conditions. We also provide a computationally efficient estimator of the covariance function and derive its uniform consistency. The procedure is also extended to the two-sample case in which we focus on comparing the mean functions from two populations of imaging data. Through Monte Carlo simulation studies we examine the finite-sample performance of the proposed method. Finally, the proposed method is applied to analyze brain Positron Emission Tomography (PET) data in two different studies. One dataset used in preparation of this article was obtained from the Alzheimer's Disease Neuroimaging Initiative (ADNI) database.
	
	\vspace{9pt} \noindent \textit{Key words and phrases:} Bivariate splines, Functional principal component analysis, Image analysis, Semiparametric efficiency, Triangulation.
\end{quotation}

\fontsize{10.95}{14pt plus.8pt minus .6pt}\selectfont

\thispagestyle{empty}

\setcounter{chapter}{1} \label{SEC:introduction}
\setcounter{equation}{0}
\noindent \textbf{1. Introduction} \vskip 0.1in

In recent years, as digital technology has advanced significantly, valuable imaging data of body structures and organs can be easily collected during routine clinical practice. This new paradigm presents new opportunities to innovate in both research and clinical settings. Medical imaging technology has revolutionized health care over the past three decades, allowing doctors to find or detect tumors and other abnormalities and evaluate the effectiveness of treatment. For example, radiographic imaging is one of the effective and clinically useful tools for examining various body tissues to identify various conditions. Large-scale imaging data offers an incredibly rich data resource for scientific and medical discovery.

Functional data analysis provides modern analytical tools for imaging data, which can be viewed as realizations of random functions. Let $\Omega$ be a two-dimensional bounded domain, and $\bs{z}=(z_1,z_2)$ be a point in $\Omega$. The model we consider is:
\vspace{-0.15in}\begin{equation}
	Y_{i}(\bs{z})=\mu(\bs{z})+\eta_{i}(\bs{z})+\sigma(\bs{z})\varepsilon_{i}(\bs{z}), ~ i=1,\ldots,n,~ \bs{z}\in \Omega, \vspace{-0.15in}
	\label{model:FDA}
\end{equation}
which is one instance of the general {\it function-on-scalar regression model}. In model (\ref{model:FDA}), $Y_i(\bs{z})$ denotes the imaging measurement at location $\bs{z}\in \Omega$, $\eta_i(\bs{z})$ is a stochastic process indexed by $\bs{z}$ which characterizes subject-level image variations, $\sigma(\bs{z})$ is a positive deterministic function, and $\varepsilon(\bs{z})$ is a mean zero stochastic process. We assume that $\eta_i(\bs{z})$ and $\varepsilon_{i}(\bs{z})$ are mutually independent, $\eta_{i}(\bs{z})$ are i.i.d. copies of a $L_{2}$ stochastic process $\eta(\bs{z})$ with mean zero and covariance function $G_{\eta}(\bs{z},\bs{z}^{\prime})$, $\varepsilon_{i}(\bs{z})$ are i.i.d. instances of a stochastic process $\varepsilon(\bs{z})$ with mean zero and covariance function $\text{Cov}\{\varepsilon_{i}(\bs{z}), \varepsilon_{i}(\bs{z}')\}=I(\bs{z}=\bs{z}^{\prime})$.

For biomedical imaging data, the objects (e.g., tumor tissues, brain regions, etc.) appearing in the images are typically irregularly shaped. Many smoothing methods in the literature, such as tensor product smoothing, kernel smoothing, and wavelet smoothing, suffer from the problem of ``leakage" across the complex domains, i.e., poor estimation over difficult regions that is a result of smoothing inappropriately across boundaries of features.

In this paper, we endeavor to address these challenges by applying bivariate splines over triangulations \citep{Lai:Wang:13} to preserve important features (shape, smoothness) of imaging data.  Spline functions defined this way offer more flexibility and varying amounts of smoothness allowing us to better approximate the mean functions. We study the asymptotic properties of the spline estimators of $\mu(\bs{z})$ by using  bivariate penalized splines (BPS) defined on triangulations and show that our estimator is consistent and asymptotically normal.

In addition, when analyzing biomedical imaging data, such as brain images, typical questions lie in estimating the mean function, $\mu(\bs{z})$, together with quantifying the estimation uncertainty and making comparisons between populations. However, making statistically rigorous inference for imaging data is challenging, and one of the main obstacles is the complicated spatial correlation structure. The prevailing analytic technique, termed the ``mass univariate'' approach, involves regarding each pixel/voxel as a unit, and for each unit, making a traditional univariate statistical inference, such as a simple $t$ test. The obvious multiple comparisons issue can be dealt with in many ways; popular approaches include Bonferroni correction, random field theory \citep{Worsley:etal:04,Adler:Taylor:07,Siegmund:etal:11} and cluster threshold-based approach \citep{Forman:etal:95}.

However, many of the multiple testing methods are ad hoc methods, which involve setting the threshold by eye, based on the practitioner's experience and knowledge. Our simulation study in Web Appendix A also demonstrates that those ad hoc methods heavily depend on the choice of the threshold. In this paper, we propose an alternative approach which treats the imaging data as an instance of functional data, regarded as being continuously defined but observed on a regular grid. \red{If we consider the imaging data as being functional, attention naturally turns from considering each pixel/voxel as the basic analytical unit and towards simultaneous, for instance, calculating simultaneous confidence corridors (SCCs; also called ``simultaneous confidence bands" or ``uniform confidence band/region").} As pointed out in \cite{Choi:Reimherr:18} and \cite{Degras:17}, conventional multiple comparison methods are less useful in the functional data setup because the infinite cardinality of the domain would lead to unbounded confidence regions.

In statistics, SCCs are vital and fundamental tools for inference on the global behavior of functions, see for example, \cite{Wang:Yang:2009}, \cite{Wang:Yang:10}, \cite{Krivobokova:Kneib:Claeskens:10}, \cite{Wang:Liu:Cheng:Yang:14} and \cite{Cai:Liu:Wang:Yang:19}. However, they have received relatively little attention in the literature of functional data analysis (FDA). Moreover, existing SCC work for FDA has concentrated on the one-dimensional case. For the development of SCCs for mean curves of functional data, see the simulation-based techniques \citep{Degras:11,Cao:Yang:Todem:12,Cao:2014,Zheng:Yang:Hardle:14,Degras:17,Cao:Wang:18}, the functional principal component (FPC) decomposition-based approach \citep{Goldsmith:etal:13} and the geometric approach by \cite{Choi:Reimherr:18} in Hilbert spaces. \cite{Zhu:Li:Kong:12} proposed SCCs for the regression coefficient functions for multivariate varying coefficient model for functional responses. \cite{Gu:Wang:Wolfgang:Yang:2014} and \cite{Chang:Lin:Ogden:2017} proposed the SCC for coefficient functions in the function-on-scalar regression model. However, there is scant literature on SCCs for imaging data or other more general 2D functions. Although the geometric method in \cite{Choi:Reimherr:18} can be used to construct SCCs in Hilbert spaces over rectangular domains, it doesn't work well for objects over complex domains with arbitrary shape, which are very common in biomedical imaging studies. In addition, the geometric method is conservative because it is essentially based on a modification of Scheff\'{e}'s method.

In this paper, we derive SCCs with exact coverage probability for the 2D functional mean function $\mu(\bs{z})$, $\bs{z}\in \Omega$, in (\ref{model:FDA}) via extreme value theory of Gaussian processes \citep{Adler:90} and approximating mean functions with bivariate splines. Our simulation studies indicate the proposed SCCs are computationally efficient and have the correct coverage probability for finite samples. We also show that the spline estimator and the accompanying SCC are asymptotically the same as if all the images had been observed without noise.

Motivated by the need to statistically quantify the difference between two imaging datasets that arise in medical imaging studies, we further consider two-sample inference for imaging data and extend our SCC construction procedure to a two-sample problem. Specifically, we focus on constructing SCC for the difference of the mean functions from two independent samples. The comparison of mean functions is particularly useful for the analysis of imaging data in some biomedical settings such as comparing imaging outcomes for groups randomized either to placebo or to active treatment.  Any mean differences may be localized and irregularly shaped, and so an estimation method should be flexible enough to allow for such differences. The approach developed here allows comparison of treatments simultaneously across the entire domain of interest.

We organize our paper as follows. Section 2 describes the BPS estimators, and establishes their asymptotic properties for imaging data. Section 3 proposes asymptotic pointwise confidence intervals and SCCs that are constructed based on the BPS estimators. In Section 4, we discuss how to estimate the unknown components involved in the SCC construction and other issues of implementation. Section 5 reports findings from a simulation study. In Section 6, we apply the proposed methods to two real brain imaging datasets. In Section 7, we conclude the article with some discussions. Technical proofs of the theoretical results and additional results from the simulation studies as well as the application are provided in the Appendices.

\setcounter{chapter}{2}  \renewcommand{\thetheorem}{2.\arabic{theorem}}
\renewcommand{\thelemma}{2.\arabic{lemma}}
\renewcommand{\theproposition}{2.\arabic{proposition}}
\renewcommand{\thetable}{2.\arabic{table}} \setcounter{table}{0} 
\renewcommand{\thefigure}{2.\arabic{figure}} \setcounter{figure}{0} 
\setcounter{equation}{0} \setcounter{lemma}{0} \setcounter{theorem}{0}
\setcounter{proposition}{0}\setcounter{corollary}{0}
\vskip .12in \noindent \textbf{2. Models and Estimation Method} 
\label{SEC:method}

In practice, the functional imaging response variable, $Y_i(\cdot)$, is only measured on a regular grid of pixels, $\bs{z}_j\in \Omega,~j=1,\ldots,N$.  For notation simplicity, we let $Y_{ij}=Y_i(\bs{z}_j)$ be the imaging response of subject $i$ at location $j$, and the actual data set consists of $\{(Y_{ij}, \bs{z}_j)\}$, $i=1,\ldots,n$, $j=1,\ldots,N$, which can be modeled as
\begin{equation}
	Y_{ij}=\mu(\bs{z}_{j})+\eta_{i}(\bs{z}_{j})+\sigma(\bs{z}_j)\varepsilon_{ij}.
	\label{model2}
\end{equation}

\vskip .10in \noindent \textbf{2.1. Bivariate Spline Basis Approximation over Triangulations} \vskip .10in

For model (\ref{model2}), we first consider the estimation of the mean function, $\mu(\cdot)$. Medical imaging data are typically observed on an irregular domain $\Omega$, and thus triangulation is an effective strategy to handle such type of data. We approximate the mean function in (\ref{model2}) by the bivariate splines that are piecewise polynomial functions over a 2D triangulated domain; see \cite{Lai:Wang:13}. In the following, we briefly introduce the techniques of triangulations and describe the BPS smoothing method.

Triangulation is an effective tool for handling data distributed on irregular regions with complex boundaries and/or interior holes. In the following we use $T$ to denote a triangle which is a convex hull of three points that are not collinear. A collection $\triangle=\{T_1,\ldots,T_M\}$ of $M$ triangles is called a triangulation of $\Omega=\cup_{m=1}^{M} T_{m}$ if any nonempty intersection between a pair of triangles in $\triangle$ is either a shared vertex or a shared edge. In the rest of the paper, we assume that all points $\bs{z}_j$'s lie in the interior of some triangle in $\triangle$, i.e., they are not on edges or vertices of triangles in $\triangle$.

Given a triangle $T\in \triangle$, let $|T|$ be its longest edge length, and $\varrho_{T}$ be the radius of the largest disk which can be inscribed in $T$. Define the shape parameter of $T$ as the ratio $\pi_{T}=|T|/\varrho_T$. When $\pi_T$ is small, the triangles are relatively uniform in the sense that all angles of triangles in the triangulation $\triangle$ are roughly the same. Denote the size of $\triangle$ by $|\triangle|:=\max \{|T|,T \in \triangle \}$, i.e., the length of the longest edge of all triangles in $\triangle$.

For an integer $r\geq 0$, let $\mathcal{C}^r(\Omega)$ be the collection of all $r$-th continuously differentiable functions over $\Omega$. Given a triangulation $\triangle$, let $\mathcal{S}_{d}^{r}(\triangle )=\{s\in \mathcal{C}^{r}(\Omega ):s|_{T}\in \mathbb{P}_{d}(T), T \in \triangle \}$ be a spline space of degree $d$ and smoothness $r$ over triangulation $\triangle $, where $s|_{T}$ is the polynomial piece of spline $s$ restricted on triangle $T$, and $\mathbb{P}_{d}$ is the space of all polynomials of degree less than or equal to $d$.  We use Bernstein basis polynomials to represent the bivariate splines.  For any triangle $T \in \triangle$ and any fixed point $\bs{z} \in \Omega$, let $b_1$, $b_2$ and $b_3$ be the barycentric coordinates of $\bs{z}$ relative to $T$. Then, the Bernstein basis polynomials of degree $d$ relative to triangle $T$ is defined as $B_{ijk}^{T,d}(\bs{z})=\frac{d!}{i!j!k!}b_1^ib_2^jb_3^k$, $i+j+k=d$. Let $\{B_{m}\}_{m \in \mathcal{M}}$ be the set of degree-$d$ bivariate Bernstein basis polynomials for $\mathcal{S}_{d}^{r}(\triangle)$, where $\mathcal{M}$ stands for an index set of Bernstein basis polynomials. Denote by $\mathbf{B}$ the evaluation matrix of Bernstein basis polynomials, where the $j$-th row of $\mathbf{B}$ is given by $\mathbf{B}^{\top}(\bs{z}_{j})=\{B_{m}(\bs{z}_{j}), m\in \mathcal{M}\}$, for $j=1,\ldots, N$. We can approximate the mean function $\mu(\bs{z})$ by $\mu(\bs{z})\approx \mathbf{B}^{\top}(\bs{z})\bs{\gamma}$, where $\bs{\gamma}^{\top} =(\gamma_{m},m \in \mathcal{M})$ is the spline coefficient vector. The above bivariate spline basis can be easily constructed via the R package \texttt{BPST}.

To define the penalized spline method, for any function $g(\bs{z})$ and direction $z_{h}$, $h=1,2$, let $\nabla_{z_{h}}^{v}g(\bs{z})$ denote the $v$-th order derivative in the direction $z_{h}$ at the point $\bs{z}$. We consider the following penalized least squares problem:
\begin{equation*}
	\min_{g\in S_d^r(\triangle)} \sum_{i=1}^{n}\sum_{j=1}^{N}\left\{Y_{ij}-g(\bs{z}_{j})
	\right\}^{2}+\rho_{n}\mathcal{E}(g),
\end{equation*}
where
$\mathcal{E}(s)= \sum_{T\in\triangle}\int_{T} \sum_{i+j=2}
\binom{2}{i}(\nabla_{z_{1}}^{i}\nabla_{z_{2}}^{j}s)^{2}dz_{1}dz_{2}$ is the roughness penalty, and $\rho_{n}$ is the roughness penalty parameter.

To meet the smoothness requirement of the splines, we need to impose some linear constraints on the spline coefficients $\bs{\gamma}$: $\mathbf{H}\bs{\gamma}=\mathbf{0}$ to be specific. See Section B.2 of the Supplementary Material of \cite{Yu:etal:19} for a simple example of $\mathbf{H}$. Thus, we have to minimize
\begin{equation*}
	\sum_{i=1}^{n}\sum_{j=1}^{N}\left\{Y_{ij}-\mathbf{B}^{\top}(\bs{z}_{j})
	\bs{\gamma}\right\}^{2}+\rho_{n}\bs{\gamma}^{\top}\mathbf{P}\bs{\gamma}, \mathrm{~subject~to~} \mathbf{H}\bs{\gamma}=0,
\end{equation*}
where $\mathbf{P}$ is the block diagonal penalty matrix satisfying $\bs{\gamma}^{\top}\mathbf{P}\bs{\gamma}=\mathcal{E}(\mathbf{B}\bs{\gamma})$.

We first remove the constraint via QR decomposition of $\mathbf{H}^{\top}$:
$\mathbf{H}^{\top}=\mathbf{Q}\mathbf{R}=(\mathbf{Q}_{1}~\mathbf{Q}_{2})
\binom{\mathbf{R}_{1}}{\mathbf{R}_{2}}$, where $\mathbf{Q}$ is orthogonal and $\mathbf{R}$ is upper triangular, the submatrix $\mathbf{Q}_{1}$ is the first $p$ columns of $\mathbf{Q}$, where $p$ is the rank of $\mathbf{H}$, and $\mathbf{R}_{2}$ is a matrix of zeros. Next, we reparametrize using $\bs{\gamma} = \mathbf{Q}_2\bs{\theta}$ for some $\bs{\theta}$, then it is guaranteed that $\mathbf{H}\bs{\gamma}= 0$. The minimization problem is thus converted to a conventional unrestricted penalized regression problem:
\begin{equation}
	\sum_{i=1}^{n}\sum_{j=1}^{N}\left\{Y_{ij}-\widetilde{\mathbf{B}}^{\top}(\bs{z}_{j})
	\mathbf{Q}_2\bs{\theta}\right\}^{2}+\rho_{n}\bs{\theta}^{\top}
	\mathbf{Q}_{2}^{\top}\mathbf{P}\mathbf{Q}_{2}\bs{\theta},
	\label{EQ:PLS}
\end{equation}
where $\widetilde{\mathbf{B}}(\bs{z})=\mathbf{Q}_2^{\top}\mathbf{B}(\bs{z})$.

Denote $\bar{Y}_{\cdot,j} = \frac{1}{n} \sum_{i=1}^{n} Y_{ij}$,
$\bar{\mathbf{Y}}=(\bar{Y}_{\cdot,1},\ldots, \bar{Y}_{\cdot,N})^{\top}$, $\mathbf{U} = \mathbf{B} \mathbf{Q}_2$, and $\mathbf{D}=\mathbf{Q}_2^{\top} \mathbf{P} \mathbf{Q}_2$. Then, minimizing (\ref{EQ:PLS}) is equivalent to minimizing
\begin{equation*}
	\left\|\bar{\mathbf{Y}}-\mathbf{B}\mathbf{Q}_2 \bs{\theta}\right\|^{2}+\frac{\rho_{n}}{n}\bs{\theta}^{\top}
	\mathbf{Q}_{2}^{\top}\mathbf{P}\mathbf{Q}_{2}\bs{\theta}=\left\|\bar{\mathbf{Y}}-\mathbf{U} \bs{\theta}\right\|^{2}+\frac{\rho_{n}}{n}\bs{\theta}^{\top} \mathbf{D}\bs{\theta},
\end{equation*} and the solution is given by
$
\widehat{\bs{\theta}}
=\{\mathbf{U}^{\top}\mathbf{U} + n^{-1} \rho_{n}\mathbf{D}\}^{-1} \mathbf{U}^{\top} \bar{\mathbf{Y}}.
$
Thus, the estimator of $\bs{\gamma}$ and $\mu(\cdot)$ are:
$
\widehat{\bs{\gamma}}=\mathbf{Q}_{2}\widehat{\bs{\theta}}, ~~ \widehat{\mu}(\bs{z})=\mathbf{B}(\bs{z})^{\top}\widehat{\bs{\gamma}}.
$

\vskip .12in \noindent \textbf{2.2. Functional Principal Component Analysis} \vskip .10in

For the second component, $\eta_i(\bs{z})$, in model (\ref{model2}), we consider a spectral decomposition of its covariance function $G_{\eta}(\bs{z}, \bs{z}^{\prime})$. Denote the eigenvalue and eigenfunction sequences of the covariance operator $G_{\eta}(\bs{z}, \bs{z}^{\prime})$ as $\left\{\lambda_{k}\right\}_{k=1}^{\infty}$ and $\left\{\psi_{k}(\bs{z})\right\}_{k=1}^{\infty}$, in which $\lambda_{1}\geq \lambda_{2}\geq \cdots \geq 0$, $\sum_{k=1}^{\infty}\lambda_{k}<\infty $, and $\left\{\psi_{k}\right\}_{k=1}^{\infty}$ form an orthonormal basis of $L_{2}\left(\Omega\right) $. It follows from spectral theory that $G_{\eta}(\bs{z},\bs{z}^{\prime}) =\sum_{k=1}^{\infty}\lambda_{k}\psi_{k}(\bs{z})\psi_{k}(\bs{z}^{\prime})$. The $i$th stochastic process $\left\{\eta_{i}(\bs{z}), \bs{z} \in \Omega\right\}$ allows the Karhunen-Lo\'{e}ve $L_{2}$ representation \cite{Sang:Huang:2012}:
$\eta_{i}(\bs{z})=\sum_{k=1}^{\infty}\xi_{ik}\phi_{k}(\bs{z})$, where $\phi_{k}(\bs{z})=\sqrt{\lambda_k}\psi_{k}(\bs{z})$, and the coefficients  $\xi_{ik}$'s are uncorrelated random variables with mean $0$ and $E(\xi_{ik}\xi_{ik^{\prime}})=I(k=k^{\prime})$, referred to as the $k$th functional principal component score of the $i$th subject in classical functional principal component analysis (FPCA). Thus, the response measurements in (\ref{model2}) can be represented as follows
\begin{equation}
	Y_{ij}=\mu(\bs{z}_{j})
	+\sum_{k=1}^{\infty}\xi_{ik}\phi_{k}(\bs{z}_{j})+ \sigma(\bs{z}_j)\varepsilon_{ij}.
	\label{model3}
\end{equation}

Next, we describe the method of estimating the FPCA: the variance-covariance function $G_{\eta}(\bs{z}, \bs{z}^{\prime})$ and its eigenvalues and eigenfunctions. For any $i=1,\ldots,n$, $j=1,\ldots,N$, let $\widehat{R}_{ij}=Y_{ij}- \widehat{\mu}(\bs{z}_j)$ be the residual. We  estimate $\eta_{i}(\bs{z})$ individually by employing the bivariate spline smoothing method to $\{(\widehat{R}_{ij},\bs{z}_{j})\}_{j=1}^{N}$. To be more specific, for each $i=1,\ldots, n$, we define the spline estimator of $\eta_{i}(\bs{z})$ as
\begin{equation*}
	\widehat{\eta}_{i}(\bs{z})=\argmin_{g_{i}\in \mathcal{S}_{d}^{r}(\triangle^{*})} \sum_{j=1}^{N}\left\{\widehat{R}_{ij}-g_{i}(\bs{z}_{j})\right\}^{2}+\rho^{*}_{n}\mathcal{E}(g_{i}),
\end{equation*}
where the triangulation $\triangle^{*}$ and smoothness penalty $\rho^{*}_{n}$ may be different from those introduced in Section 2 when estimating $\mu(\bs{z})$.
Next, define the estimator of $G_{\eta}(\bs{z},\bs{z}^{\prime})$ as
\begin{equation}
	\widehat{G}_{\eta}(\bs{z},\bs{z}^{\prime})=n^{-1}\sum_{i=1}^{n}
	\widehat{\eta}_i(\bs{z})\widehat{\eta}_i(\bs{z}^{\prime}),
	\label{DEF:G_eta_hat}
\end{equation}
and we estimate the eigenfunctions $\psi_{k}(\cdot)$ using the following eigenequations:
\begin{equation}
	\int_{\Omega}\widehat{G}_{\eta}(\bs{z},\bs{z}^{\prime})\widehat{\psi}_{k}(\bs{z})d\bs{z} =\widehat{\lambda}_{k}\widehat{\psi}%
	_{k} (\bs{z}^{\prime}),
	\label{EQ:eigenfunction}
\end{equation}
where $\widehat{\psi}_{k}$'s are subject to $\int_{\Omega}\widehat{\psi}_{k}^{2}\left(
\bs{z}\right) d\bs{z}=1$ and $\int_{\Omega}\widehat{\psi}_{k}(\bs{z}) \widehat{\psi}%
_{k^{\prime }}(\bs{z})d\bs{z}=0$ for $k^{\prime }<k$. If $N$ is
sufficiently large, the left hand side of (\ref{EQ:eigenfunction}) can be
approximated by $\sum_{j=1}^{N}\widehat{G} (\bs{z}_{j},\bs{z}_{j^{\prime}})\widehat{\psi}_{k}\left( \bs{z}_{j}\right) A\left( \bs{z}_{j}\right)$,
where $A\left( \bs{z}_{j}\right)$ is the area of the pixel $\bs{z}_{j}$.

\vskip .12in \noindent \textbf{2.2. Theoretical Properties of the Estimators} \vskip .10in

In this section, we investigate the asymptotic properties for the proposed bivariate spline estimators. To discuss these properties, we introduce some notation of norms. For any function $g$ over the closure of domain $\Omega$, denote $\left\| g\right\|_{L ^2(\Omega)}^{2}=\int_{\Omega} g^2(\bs{z})d\bs{z}$ the regular $L_2$ norm of $g$, and $\Vert g\Vert_{\infty ,\Omega} =\sup_{\bs{z}\in \Omega} |g(\bs{z})|$ the supremum norm of $g$. Let $|g|_{\upsilon,\infty,\Omega}=\max_{i+j=\upsilon}\Vert \nabla_{z_{1}}^{i}\nabla_{z_{2}}^{j}g\Vert_{\infty ,\Omega}$ be the maximum norms of all the $\upsilon$th order derivatives of $g$ over $\Omega $. Let $\mathcal{W}^{d,\infty}(\Omega)=\left\{g:|g|_{k,\infty, \Omega}<\infty, 0\le k\le d \right\}$ be the standard Sobolev space. Given random variables $S_{n}$ for $n\geq 1$, we write $S_{n}=O_{P}(b_{n})$ if $\lim\nolimits_{c\rightarrow \infty}\lim \sup_{n}P(|S_{n}|\geq cb_{n})=0$. Similarly, we write $S_{n}=o_{P}(b_{n})$ if $\lim_{n}P(|S_{n}|\geq cb_{n})=0$, for any constant $c>0$. We next introduce some technical conditions.
\begin{itemize}
	\item[(A1)] The bivariate function $\mu(\cdot)\in \mathcal{W}^{d+1,\infty}(\Omega )$ for an integer $d \ge 1$.
	
	\item[(A2)] For any $k\geq 1$, $\xi_{ik}$'s are i.i.d. random variables with mean $0$, variance $1$ and $E\left\vert \xi _{ik}\right\vert ^{4+\delta_{1}}<+\infty$ for some constant $\delta _{1}>0$. For any $i=1,\ldots,n$, $j=1,\ldots,N$, $\varepsilon_{ij}$'s are i.i.d with mean 0, variance 1, and $E\left\vert \varepsilon _{ij}\right\vert ^{4+\delta _{2}}<+\infty $ for some constant $\delta_{2} > 0$.
	
	\item[(A3)] The function $\sigma \in \mathcal{C}^{(1)}(\Omega)$ with $0<c_{\sigma}\leq \sigma(\bs{z}) \leq C_{\sigma}\leq \infty$ for any $\bs{z}\in \Omega$; for any $k$, $\psi_{k}\in \mathcal{C}^{(1)}(\Omega)$ and the variance function $0<c_{G}\leq G_{\eta}(\bs{z},\bs{z})\leq C_{G}\leq \infty$, for any $\bs{z}\in \Omega$.
	
	\item[(A4)] The triangulation is $\pi$-quasi-uniform, that is, there exists a positive constant $\pi$ such that $(\min_{T\in \triangle}\varrho_T)^{-1} |\triangle| \leq \pi$.
	
	\item[(A5)] As $N \rightarrow \infty$, $n \rightarrow \infty$, $N^{-1}n^{1/(d+1)}\log (n) \rightarrow 0$, the triangulation size satisfies that $N^{-1}\log (n) \ll |\triangle|^2\ll \min\{n^{(2+\delta_2)/(4+\delta_2)}N^{-1}\log^{-1}(n), n^{-1/(d+1)}\}$, and the smoothing penalty parameter $\rho_n $ satisfies $\rho_n =o \{\min(n^{1/2}N|\triangle|^{3},nN^{3/2}|\triangle|^{6},nN|\triangle|^{5})\}$.
	
	\item[(A6)] For $k \in \{1,\ldots, \kappa\}$ and a nonnegative integer $s$, $\phi_k(\bs{z}) \in \mathcal{W}^{s+1,\infty}(\Omega ) ,\ \sum_{k=1}^{\kappa} \|\phi_k\|_{\infty}<\infty$. $\frac{\rho_{n}}{nN|\triangle|^{3}}\sum_{k=1}^{\kappa_n} \|\phi_k\|_{2,\infty}=o(1)$, $\left(1+\frac{\rho_{n}}{nN|\triangle|^{5}}\right) \sum_{k=1}^{\kappa_n} |\triangle|^{s+1}\|\phi_k\|_{s+1,\infty}=o(1)$ for a sequence $\{\kappa_n\}_{n=1}^{\infty}$ of increasing integers, with $\lim_{n \rightarrow \infty} \kappa_n = \kappa$, as $n\rightarrow\infty$. Meanwhile, $\sum_{k=\kappa_n+1}^{\kappa} \|\phi_k\|_{\infty} = o(1)$. The number $\kappa $ of nonzero eigenvalues is finite or $\kappa $ is infinite.
	
	\item[(A7)] As $N \rightarrow \infty$, $n \rightarrow \infty$, for some $0<\delta_3<1$, $N^{-1}n^{1/(d+1)+\delta_3} \rightarrow 0$, $N|\triangle_{\eta}|^2\rightarrow \infty$, $n^2|\triangle_{\eta}|^4/\log n \to \infty$.
\end{itemize}
\vspace{-0.08in}
The above assumptions are mild conditions that can be satisfied in many practical situations.  Assumption (A1) is  typically assumed about the true underlying functions in the nonparametric estimation literature. Assumption (A1) can be relaxed by only requiring $\mu(\cdot)\in \mathcal{C}^{(0)}(\Omega)$ if the  imaging data has sharp edges. Assumptions (A2) and (A3) are common conditions used in the literature; see for example, \cite{Cao:Yang:Todem:12}. Assumption (A4) suggests the use of more uniform triangulations with smaller shape parameters. Assumption (A5) describes the requirement of the growth rate of the dimension of the spline spaces relative to the sample size and the image resolution. 

The following theorem provides the $L_2$ and uniform convergence rate of $\widehat{\mu}(\cdot)$. The detailed proofs of this theorem are given in Web Appendix B.2.
\begin{theorem}
	\label{THM:convergence}
	Suppose Assumptions (A1)--(A4) hold and $N^{1/2}|\triangle|\rightarrow \infty$ as $N\rightarrow \infty$.  Then the bivariate penalized spline estimator $\widehat{\mu}(\cdot)$ is consistent and satisfies
	\begin{equation*}
		\Vert \widehat{\mu}-\mu \Vert_{L_2}=O_{P}\left\{\frac{\rho_{n}}{nN|\triangle|^{3}}\|\mu\|_{2,\infty}
		+\left(1+\frac{\rho_{n}}{nN|\triangle|^{5}}\right)
		|\triangle|^{d +1}\|\mu\|_{d+1,\infty}+\frac{1}{\sqrt{n}}+\frac{1}{\sqrt{nN}|\triangle|}\right\}.
	\end{equation*}
	In addition, if Assumptions (A1)--(A5) hold, we have
	$\Vert \widehat{\mu}-\mu \Vert_{\infty}=O_{P}\{(n^{-1}\log (n))^{1/2}\}$ and $\Vert \widehat{\mu}-\mu \Vert_{L _2}=O_{P}(n^{-1/2})$.
\end{theorem}

Theorem \ref{THM:Ghat-G} below characterizes the uniform weak convergence of $\widehat{G}_{\eta}(\bs{z},\bs{z}^{\prime})$ and the convergence of $\widehat{\psi}_{k}$ and $\widehat{\lambda}_k$.

\begin{theorem}
	\label{THM:Ghat-G}
	Under Assumptions (A1)--(A7), we have the following results:
	\begin{itemize}
		\item[(i)] The spline estimator $\widehat{G}_{\eta}(\bs{z},\bs{z}^{\prime})$ in (\ref{DEF:G_eta_hat}) uniformly converges to $G_{\eta}(\bs{z},\bs{z}^{\prime})$ in probability, i.e.,
		$\sup_{(\bs{z},\bs{z}^{\prime})\in \Omega^{2}}|\widehat{G}_{\eta}(\bs{z},\bs{z}^{\prime})
		-G_{\eta}(\bs{z},\bs{z}^{\prime})|=o_{P}(1)$.
		\item[(ii)] $\|\widehat{\psi}_{k}-\psi_{k}\|=o_{P}(1)$, $|\widehat{\lambda}_k-\lambda_k|=o_{P}(1)$, for $k=1,\ldots,\kappa$.
	\end{itemize}
\end{theorem}

Although, in theory, the Karhunen-Lo\'{e}ve representation of the covariance function consists of infinite number of terms. In applications, it is typical to truncate the spectral decomposition to an integer chosen so as to account for some predetermined proportion of the variance \cite{Hall:Muller:Wang:06,Li:Wang:Carroll:13}. One can select the number of principal component using the Akaike information criterion (AIC) suggested by \cite{Yao:Muller:Wang:05} or Bayesian information criterion (BIC) proposed by \cite{Li:Wang:Carroll:13}.

\setcounter{chapter}{3} \renewcommand{\thetheorem}{3.\arabic{theorem}}
\renewcommand{\thelemma}{3.\arabic{lemma}}
\renewcommand{\theproposition}{3.\arabic{proposition}}
\renewcommand{\thetable}{3.\arabic{table}} \setcounter{table}{0} 
\renewcommand{\thefigure}{3.\arabic{figure}} \setcounter{figure}{0} 
\setcounter{equation}{0} \setcounter{lemma}{0} \setcounter{theorem}{0}
\setcounter{proposition}{0}\setcounter{corollary}{0}
\vskip .12in \noindent \textbf{3. Simultaneous Confidence Corridors (SCCs)} \vskip 0.10in
\label{SEC:SCC}

\vskip .12in \noindent \textbf{3.1. One Sample} \vskip .10in

In this section, we develop the SCCs for the mean function $\mu(\cdot)$ in (\ref{model3}).

Let $G_{\eta} (\cdot,\cdot)$ be a positive definite function defined as $G_{\eta} (\bs{z},\bs{z}^{\prime})=\sum_{k=1}^{\kappa } \lambda_{k}\psi _{k} (\bs{z})\psi _{k} (\bs{z}^{\prime}) $, $\bs{z},\bs{z}^{\prime} \in \Omega$. Denote by $\zeta (\bs{z})$, $\bs{z}\in \Omega$ a standardized Gaussian process such
that $E\zeta \left(\bs{z}\right) = 0$, $E\zeta ^{2}(\bs{z}) = 1$ with covariance function $E\zeta(\bs{z})\zeta
(\bs{z}^{\prime}) =G_{\eta} \left(\bs{z},\bs{z}^{\prime}\right) \left\{ G_{\eta} \left(\bs{z},\bs{z}\right)
G_{\eta} \left(\bs{z}^{\prime},\bs{z}^{\prime}\right) \right\} ^{-1/2}$, $\bs{z},\bs{z}^{\prime}\in \Omega$. Denote by $q_{1-\alpha }$ the $100\left( 1-\alpha \right) ^{th}$ percentile of the distribution of the absolute maximum of $\zeta \left( \bs{z}\right) $, $\bs{z}\in \Omega$, i.e.
$P\left\{ \sup_{\bs{z}\in \Omega}\left| \zeta (\bs{z}) \right| \leq
q_{1-\alpha }\right\} =1-\alpha$, $\alpha \in \left( 0,1\right)$.

Define the ``oracle" estimator $\bar{\mu}(\bs{z}) = \mu(\bs{z}) + \frac{1}{n} \sum_{i=1}^n \eta_i(\bs{z})$. Of course this is infeasible due to the finite pixel grid $\{\bs{z}_j : j = 1,\ldots, N\}$ and the measurement error. The following theorem presents the asymptotic properties of $\bar{\mu}(\bs{z})$ and shows that the difference between the BPS estimator $\widehat{\mu}(\bs{z})$ and the ``oracle" smoother $\bar{\mu}(\bs{z})$ is uniformly bounded at an $o_P(n^{1/2})$ rate.
\begin{theorem}
	\label{THM:band}
	Under Assumptions (A1)--(A6), for any $\alpha \in \left( 0,1\right)$, as $N\rightarrow \infty$, $n\rightarrow \infty$,
	\begin{equation*}
		P\left\{ \sup_{\bs{z}\in \Omega}n^{1/2}\left| \bar{\mu}(\bs{z})-\mu(\bs{z})\right| G_{\eta} \left(\bs{z},\bs{z}\right)^{-1/2}\leq q_{1-\alpha }\right\} \rightarrow 1-\alpha,
	\end{equation*}
	\begin{equation*}
		\sup_{\bs{z} \in \Omega} |\bar{\mu}(\bs{z}) - \widehat{\mu}(\bs{z})| = o_P(n^{-1/2}).
	\end{equation*}
\end{theorem}
Based on Theorems \ref{THM:convergence} and \ref{THM:band}, we obtain the following asymptotic SCCs for $\mu (\bs{z})$, $\bs{z}\in \Omega$.
\begin{corollary}
	\label{COR:Asymp muhat}
	Under Assumptions (A1)--(A6), for any $\alpha \in \left( 0,1\right)$, as $N\rightarrow \infty$, $n\rightarrow \infty $, an asymptotic $100 \left( 1-\alpha \right) \% $ exact SCC for $\mu (\bs{z})$ is
	\begin{equation*}
		P\left\{\mu (\bs{z}) \in \widehat{\mu}(\bs{z})\pm n^{-1/2}q_{1-\alpha}G_{\eta} \left(\bs{z},\bs{z}\right) ^{1/2}, ~ \bs{z}\in \Omega \right\}\rightarrow 1-\alpha .
	\end{equation*}
\end{corollary}

\vskip .12in \noindent \textbf{3.2. Extension to Two-sample Case} \vskip .10in

While one-sample confidence bands are of primary interest in many situations, in some brain imaging analysis, interest lies in comparing two groups, e.g., patients and normal control subjects. In this section, we extend our method to two-sample problems, constructing SCCs for the difference between mean functions from two independent groups, analogous to a two-sample $t$-test. With these two-sample SCCs, we can assess differences of images with quantified uncertainty.

Given two groups of imaging observations with sample sizes $n_{1}$ and $n_{2}$, respectively, defined on a common region $\Omega$. For $H=1,2$, let  $G_{H\eta}\left(\bs{z},\bs{z}^{\prime}\right)=\sum_{k=1}^{\kappa_{H}}\phi_{Hk}\left(\bs{z}\right) \phi_{Hk} \left(\bs{z}^{\prime }\right)$ be a positive definite function and $\widehat{\mu}_{H}$ be the spline estimates for the group mean function $\mu_{H}$. Let $V\left(\bs{z}, \bs{z}^{\prime} \right)=G_{1 \eta}(\bs{z}, \bs{z}^{\prime}) +\tau G _{2 \eta}(\bs{z}, \bs{z}^{\prime})$, where $\tau=\lim_{n_1\rightarrow \infty}n_{1}/n_{2}$. Denote by $W(\bs{z})$, $\bs{z}\in \Omega$, a standardized Gaussian process such that $EW(\bs{z})=0$, $EW^{2}(\bs{z})=1$ with covariance
$E[W\left(\bs{z}\right) W\left(\bs{z}^{\prime }\right)] =\{V\left(\bs{z}, \bs{z}\right)\}^{-1/2}
V\left(\bs{z}, \bs{z}^{\prime }\right) \{V\left( \bs{z}^{\prime}, \bs{z}^{\prime }\right)\}^{-1/2}$. Denote $q_{12,\alpha}$ the $(1-\alpha)$-th quantitle of the absolute maximal distribution of $W\left(\bs{z}\right)$, $\bs{z}\in \Omega$.

\begin{theorem}
	\label{THM:muA-muB}
	Under Assumptions (A1)--(A6), for any $\alpha \in \left( 0,1\right)$, as $N\rightarrow \infty$, $n_{1}\rightarrow \infty$,
	\begin{equation*}
		P\left\{ \sup_{\bs{z}\in \Omega }\frac{n_{1}^{1/2}\left\vert \left(\widehat{\mu}_{1}-\widehat{\mu}_{2}\right)(\bs{z})-\left( \mu_{1}-\mu_{2}\right)(\bs{z}) \right\vert }{\sqrt{V\left(\bs{z},\bs{z}\right)}}\leq q_{12,\alpha}\right\} \rightarrow 1-\alpha.
	\end{equation*}
\end{theorem}
Theorem \ref{THM:muA-muB} suggests that an asymptotic $100 \left( 1-\alpha \right)\%$ exact SCC for $\left(\mu_{1}-\mu_{2}\right)(\bs{z})$ can be constructed as
$\left(\widehat{\mu}_{1}-\widehat{\mu}_{2}\right)(\bs{z})\pm n_{1}^{-1/2}q_{12,\alpha}\{V\left(\bs{z},\bs{z}\right)\}^{1/2}$.

\setcounter{chapter}{4} 
\renewcommand{\thetable}{4.\arabic{table}} \setcounter{table}{0} 
\renewcommand{\thefigure}{4.\arabic{figure}} \setcounter{figure}{0} 
\setcounter{equation}{0} \setcounter{lemma}{0} \setcounter{theorem}{0}
\setcounter{proposition}{0}\setcounter{corollary}{0}
\vskip .12in \noindent \textbf{4. Implementation} \vskip 0.10in
\label{SEC:implementation}
Without loss of generality, we describe the implementation of the proposed SCCs for the one-sample case. The procedure can be similarly adopted to the two-sample mean cases.

\vskip .12in \noindent \textbf{4.1. Quantile Estimation and Smoothing Parameter Selection} \vskip .10in

The quantile $q_{1-\alpha}$ used to construct the SCCs in Corollary \ref{COR:Asymp muhat} cannot be obtained analytically, however, it can be approximated by numerical simulation as follows: first, we simulate $\zeta_b(\bs{z}) = \widehat{G}_{\eta}^{-1/2}(\bs{z},\bs{z})\sum_{k=1}^{\kappa}\widehat{\lambda}_{k}^{1/2}
Z_{k,b}\widehat{\psi}_k(\bs{z})$,
where $Z_{k,b}$ are i.i.d standard normal variables with $1\leq k \leq \kappa$ and $b=1, \ldots, B$ for a preset large integer $B$. Next, we estimate the quantile $q_{1-\alpha}$ by the corresponding empirical quantile of these maximum values by taking the maximal absolute
value for each copy of $\zeta_b(\bs{z})$.

To construct the SCC for the two-sample case, denote $\widehat{V}\left(\bs{z}, \bs{z}^{\prime} \right)=\widehat{G}_{1 \eta}(\bs{z}, \bs{z}^{\prime}) +\tau \widehat{G} _{2 \eta}(\bs{z}, \bs{z}^{\prime})$.  We simulate
\[
\widehat{W}_{b}(\bs{z})=\{\widehat{V}(\bs{z},\bs{z})\}^{-1/2}\left\{\sum_{k=1}^{\kappa_{1} }\widehat{\lambda} _{1k}^{1/2}Z_{1k,b}\widehat{\psi}_{1k}(\bs{z})- (n_1/n_2)^{1/2}\sum_{k=1}^{\kappa_{2}}\widehat{\lambda}_{2k}^{1/2}Z_{2k,b}\widehat{\psi}_{2k}(\bs{z})\right\}, ~ \bs{z}\in \Omega.
\]
Then, $q_{12,\alpha}$ can be estimated by the empirical quantile of level $1-\alpha$ of the $B$ simulated $\|\widehat{W}_{b}\|_{\infty}$'s, $b=1,\ldots,B$.

Next, for a good fit of the data, it is necessary to choose a suitable value of the smoothing parameter $\rho_{n}$. A large value of $\rho_{n}$ enforces a smoother fitted function with larger fitting errors, while a small $\rho_{n}$ may result in overfitting of the data. Since the in-sample fitting errors can not gauge the prediction accuracy of the fitted function, we select a criterion function that attempts to measure the out-of-sample performance of the fitted model. Minimizing the generalized cross-validation (GCV) criterion is one computationally efficient approach to selecting smoothing parameters that also has good theoretical properties. We choose the smoothing parameter by minimizing the following
\[
\mathrm{GCV}(\rho_n)=\frac{\left\|\bar{\mathbf{Y}}-\mathbf{S}(\rho_{n})\bar{\mathbf{Y}} \right\|^{2}}
{N\{1-\mathrm{tr}\{\mathbf{S}(\rho_n)\}/N\}^2}
\]
over a grid of values of $\rho_n$, where $\mathbf{S}(\rho_{n})=\mathbf{U}(\mathbf{U}^{\top}\mathbf{U}+n^{-1}\rho_{n}\mathbf{D})^{-1}\mathbf{U}^{\top}$.

\vskip .12in \noindent \textbf{4.2. Spline Basis and Triangulation Selection} \vskip .10in

To construct the SCC, we need to choose the spline basis functions and triangulation used in the BPS, a notoriously difficult task for constructing nonparametric pointwise confidence intervals or simultaneous confidence bands. 

When the resolution of the imaging is relatively high and the mean imaging seems to be a realization from some smooth function without sharp edges, we suggest using smooth parameter $r=1$ with degree $d\geq4$. When $d\geq5$, the proposed spline achieves full estimation power asymptotically \citep{Lai:Schumaker:07}. It is generally believed that subject-level image variation $\eta_i$'s are less smooth than the mean function. Thus, we suggest considering lower order splines, such as $d=2$, when estimating the $\eta_i$'s.

An optimal triangulation is a partition of the domain which is best according to some criterion that measures the shape, size or number of triangles. For example, a ``good" triangulation usually refers to those with well-shaped triangles, no small angles or/and no obtuse angles. Other criteria include the density control (adaptivity) and optimal size (number of triangles), etc. For a fixed number of triangles, \cite{Lai:Schumaker:07} recommend selecting the triangulation according to ``max-min" criterion which maximizes the minimum angle of all the angles of the triangles in the triangulation.

We suggest building the triangulated meshes using typical triangulation construction methods such as Delaunay Triangulation \citep{DeLoera:etal:2010}. The Matlab code \texttt{DistMesh} and R package \texttt{Triangulation} can be used to construct the triangulation. When estimating the mean function $\mu(\cdot)$, we suggest choosing the triangulation $\triangle_\mu$ based on leave-images-out $k$-fold cross-validation (CV). In the estimation of the $\eta_i(\cdot)$'s, we suggest choosing the triangulation $\triangle_\eta$ so as to minimize a bootstrap estimator of the coverage error of the SCCs. In Algorithm \ref{ALGO:tri}, we describe our selection scheme for the one-sample case.

\normalem
\begin{algorithm}
	\SetKwInOut{Input}{Input}
	\SetKwInOut{Output}{Output}
	\SetKwBlock{Begin}{Step 2.}{}
	\caption{Triangulation selection.}
	\Input{Images $\left\{ Y_{ij}\right\}_{j=1, i=1}^{N,n}$.\\}
	\Output{Triangulations $\triangle_{\mu}$ and $\triangle_{\eta}$.}
	\BlankLine
	\textbf{Step 1.} Selecting $\triangle_{\mu}$ and estimating $\mu(\bs{z})$. Based on $\left\{ Y_{ij}\right\}_{j=1, i=1}^{N,n}$, select   $\triangle_{\mu}$ via the leave-images-out $k$-fold CV, and obtain $\widehat{\mu}(\bs{z})$ using the BPS method. Define $\widehat{R}_{ij}=Y_{ij}-\widehat{\mu}(\bs{z}_j)$.
	
	\Begin(Selecting $\triangle_{\eta}$ from a set of triangulations \{$\triangle_{\eta}^{q}$, $q\in \mathcal{Q}$\}.){
		\ForEach {$q\in \mathcal{Q}$}{
			~(i) For $i=1,\ldots,n$, estimate $\widehat {\eta}_i(\bs{z})$ by smoothing $\widehat{R}_{ij}$ via the bivariate spline smoothing method based on triangulation $\triangle_{\eta}^{q}$, and let $\widehat{\varepsilon}_{ij}=\widehat{R}_{ij}-\widehat {\eta}_i(\bs{z}_j)$.\\
			\noindent (ii) Generate an independent random sample $\delta_{i}^{(b)}$ and $\delta_{ij}^{(b)}$ from $\{-1,1\}$ with probability 0.5 each, and define $Y_{ij}^{\ast(b)}=\widehat{\mu}(\bs{z}_j)+\delta_{i}^{(b)} \widehat {\eta}_i(\bs{z}_j) + \delta_{ij}^{(b)} \widehat{\varepsilon}_{ij}$. \\
			(iii) Based on $\left\{ Y_{ij}^{\ast (b)}\right\}_{j=1, i=1}^{N,n}$, obtain the estimators of the mean and covariance functions $\widehat{\mu}^{*(b)}$ and $\widehat{G}_{\eta}^{*(b)}$ using $\triangle_\mu$ and $\triangle_{\eta}^{q}$, respectively.\\
			(iv) For any fixed $\alpha\in (0,1)$,  construct $100(1-\alpha)\%$ SCCs for resampled data $\left\{ Y_{ij}^{\ast (b)}\right\}_{j=1, i=1}^{N,n}$: $\mathcal{B}^{*(b)}(\alpha),\ b=1,\ldots, B$, $\mathcal{B}^{*(b)}(\alpha) = \widehat{\mu}^{*(b)}(\bs{z})\pm n^{-1/2}q_{1-\alpha}^{*(b)}\widehat{G}_{\eta}^{*(b)} \left(\bs{z},\bs{z}\right) ^{1/2}$.}
		\vskip -.1in
		Select $\triangle_{\eta}$ by minimizing the objective function
		\begin{equation*}
			\min_{q \in \mathcal{Q}} \int_{\alpha-\delta}^{\alpha+\delta} \left\{\frac{1}{B} \sum_{b=1}^{B} I(\widehat{\mu}\in \mathcal{B}^{*(b)}(\alpha);\triangle_{\eta}^{q})-(1-\alpha)\right\}^2d\alpha,
		\end{equation*}
		\hskip -.1in for some constant $0<\delta<\alpha$, which is taken to be $0.005$ in our simulation studies.}
	\label{ALGO:tri}
\end{algorithm}
\ULforem

\vskip .12in \noindent \textbf{4.3. Variance Estimation for Measurement Errors and SCC Adjustment} \vskip .10in

For certain imaging types and modalities, our assumptions (A2) and (A3) about the measurement errors may not be completely satisfied. We propose a modification to the SCC procedure in Section 3 to deal with images with relatively large measurement errors. 

For the one-sample SCC, for any $j=1,\ldots, N$, let $\widehat{\varepsilon}_{ij} = \widehat{R}_{ij} - \widehat{\eta}_{i}(\bs{z}_{j})$, and we estimate $\sigma^{2}(\bs{z}_{j})$ by
$\widehat{\sigma}^{2}(\bs{z}_{j}) = n^{-1} \sum_{i=1}^{n} \widehat{\varepsilon}_{ij} \widehat{\varepsilon}_{ij}$. Next, denote $\widehat{\varepsilon}(\bs{z})=\widetilde{\mathbf{B}}(\bs{z})^{\top} \bs{\Gamma}_{N,\rho}^{-1}
\frac{1}{nN}\sum_{i=1}^{n}\sum_{j=1}^{N} \widetilde{\mathbf{B}}(\bs{z}_{j})\sigma(\bs{z}_j)\varepsilon_{ij}$. We estimate the variance-covariance function of $\widehat{\varepsilon}(\bs{z})$, 
$\widetilde{G}_{\varepsilon} (\bs{z}, \bs{z}^{\prime})=\mathrm{Cov} \{\widehat{\varepsilon}(\bs{z}),
\widehat{\varepsilon}(\bs{z}')\}$, by
\[
\widehat{G}_{\varepsilon} (\bs{z}, \bs{z}^{\prime})
= n^{-1}N^{-2} \widetilde{\mathbf{B}}(\bs{z})^{\top} \bs{\Gamma}_{N,\rho}^{-1} \left\{ \sum_{j=1}^{N}  \widetilde{\mathbf{B}}(\bs{z}_{j}) \widehat{\sigma}^{2}(\bs{z}_{j}) \widetilde{\mathbf{B}}(\bs{z}_{j}) ^{\top} \right\} \bs{\Gamma}_{N,\rho}^{-1} \widetilde{\mathbf{B}}(\bs{z}^{\prime}),
\] 
where $\bs{\Gamma}_{N,\rho}$ is given in (\ref{DEF:Gamma_rho}) in the Appendices.

Denote $\widehat{\Sigma}(\bs{z},\bs{z}^{\prime}) = \widehat{G}_{\eta}(\bs{z},\bs{z}^{\prime}) + n \widehat{G}_{\varepsilon} (\bs{z}, \bs{z}^{\prime})$. We adjust the approximation procedure of quantile $q_{1-\alpha}$ as follows: first, we simulate
\[
\zeta_b(\bs{z}) = \widehat{\Sigma}^{-1/2}(\bs{z},\bs{z})\left\{\sum_{k=1}^{\kappa} \widehat{\lambda}_{k}^{1/2}
\widehat{\psi}_k(\bs{z}) Z^{(b)}_{k, \xi} + \widetilde{\mathbf{B}}(\bs{z})^{\top}\bs{\Gamma}_{N,\rho}^{-1}\frac{1}{N}\sum_{j=1}^{N} \widetilde{\mathbf{B}}(\bs{z}_{j})\widehat{\sigma}(\bs{z}_j)Z^{(b)}_{j,\varepsilon} \right\},
\]
where $Z^{(b)}_{k, \xi}$ and $Z^{(b)}_{j,\varepsilon} $ are i.i.d standard normal variables with $1\leq k \leq \kappa, 1\leq j \leq N $; next, we estimate the quantile $q_{1-\alpha}$ by the corresponding empirical quantile of the $B$ simulated $\|\zeta_b\|_{\infty}$; finally, we construct the SCC as $\widehat{\mu}(\bs{z})\pm n^{-1/2}q_{1-\alpha}\widehat{\Sigma} \left(\bs{z},\bs{z}\right) ^{1/2}$, $\bs{z} \in \Omega$.

For the two-sample case, we  can similarly modify the procedure by defining $\widehat{\Sigma}_{H}(\bs{z},\bs{z}^{\prime}) = \widehat{G}_{\eta,H} + n_{H} \widehat{G}_{\varepsilon,H}$, for $H=1,2$, and $\widehat{\Xi}(\bs{z}, \bs{z}^{\prime} )= \widehat{\Sigma}_{1}(\bs{z}, \bs{z}^{\prime} )+ n_{1}/n_{2}\widehat{\Sigma}_{2}(\bs{z}, \bs{z}^{\prime} )$. Let $\widehat{\sigma}_{H}(\bs{z})$ be the estimator of $\sigma_{H}(\bs{z})$, for $H=1,2$. To estimate $q_{12,\alpha}$, we simulate
\begin{align*}
	&\widehat{W}_{b}(\bs{z})=\left\{\widehat{\Xi}(\bs{z},\bs{z})\right\}^{-1/2}\left\{\sum_{k=1}^{\kappa_{2}}\widehat{\lambda} _{1k}^{1/2}Z^{(b)}_{1k,\xi}\widehat{\psi}_{1k}(\bs{z})
	-\left(\frac{n_1}{n_2}\right)^{1/2}\sum_{k=1}^{\kappa_{2} }\widehat{\lambda}_{2k}^{1/2}Z^{(b)}_{2k,\xi}\widehat{\psi}_{2k}(\bs{z})\right. \\
	&+ \left.\widetilde{\mathbf{B}}(\bs{z})^{\top}\bs{\Gamma}_{N,\rho_1}^{-1}\frac{1}{N}\sum_{j=1}^{N} \widetilde{\mathbf{B}}(\bs{z}_{j})\widehat{\sigma}_{1}(\bs{z}_j)Z^{(b)}_{1 j,\varepsilon} 
	-\left(\frac{n_1}{n_2}\right)^{1/2} \widetilde{\mathbf{B}}(\bs{z})^{\top}\bs{\Gamma}_{N,\rho_2}^{-1}\frac{1}{N}\sum_{j=1}^{N} \widetilde{\mathbf{B}}(\bs{z}_{j})\widehat{\sigma}_{2}(\bs{z}_j)Z^{(b)}_{2 j,\varepsilon} \right\},
\end{align*}
where $Z^{(b)}_{Hk, \xi}$ and $Z^{(b)}_{Hj,\varepsilon} $ are i.i.d standard normal variables with $1\leq k \leq \kappa_{H}, 1\leq j \leq N $ for $H=1,2$. Then, $q_{12,\alpha}$ can be estimated by the empirical quantile of the $B$ simulated $\|\widehat{W}_{b}\|_{\infty}$'s, $b=1,\ldots,B$.
A modified SCC for $\mu_{1}(\bs{z}) - \mu_{2}(\bs{z})$ can thus be constructed as $\left(\widehat{\mu}_{1}-\widehat{\mu}_{2}\right)(\bs{z})\pm n_{1}^{-1/2}q_{12,\alpha}\{\widehat{\Xi}\left(\bs{z},\bs{z}\right)\}^{1/2}$.

\setcounter{chapter}{5}\setcounter{equation}{0} 
\renewcommand{\thetable}{5.\arabic{table}} \setcounter{table}{0} 
\renewcommand{\thefigure}{5.\arabic{figure}} \setcounter{figure}{0} 
\vskip .12in \noindent \textbf{5. Simulation Studies} \label{SEC:simulation} \vskip 0.10in

In this section, we describe two Monte Carlo simulations to examine the finite sample performance of the proposed method.  

\vskip .12in \noindent \textbf{5.1. One Sample SCC} \vskip .10in

In this simulation study, the measurements on the images are generated from the model:
\begin{equation*}
	Y_{ij} = \mu(\bs{z}_{j}) + \sum_{k=1}^{2} \sqrt{\lambda_k} \xi_{ij} \psi_{k}(\bs{z}_{j}) + \sigma(\bs{z}_{j}) \varepsilon_{ij}, ~ i=1,\ldots, n, ~ j=1,\ldots, N,
\end{equation*}
where $\bs{z}_j= (z_{1j}, z_{2j})\in \Omega\subset [0,1]^2$, and $\Omega$ is the same as the domain of the brain images shown in Section 6. To demonstrate the practical performance of our theoretical results, we consider the following four mean functions:
\vspace{-0.1in}
\begin{itemize}
	\item (quadratic) $\mu(\bs{z})=20\left\{(z_1-0.5)^2+(z_2-0.5)^2\right\}$,
	\item (exponential) $\mu(\bs{z})=5\exp\left[-15\left\{(z_1-0.5)^2+(z_2-0.5)^2\right\}\right]+0.5$,
	\item (cubic) $\mu(\bs{z})=3.2(-z_1^3+z_2^3)+2.4$,
	\item (sine) $\mu(\bs{z})=-10[\sin\{5\pi(z_1+0.22)\}-\sin\{5\pi(z_2-0.18)\}]+2.8$,
\end{itemize}
\vspace{-0.1in}and the corresponding mean images are shown in the first column of Figures \ref{FIG:S01} -- \ref{FIG:S06} in the Appendices.

To simulate the within-image dependence, we generate $\xi_{ik}\overset{\text{i.i.d}}{\sim}N(0,1)$, for $k=1,2$. For the eigenvalues, we set  $\lambda_1=0.5$, $\lambda_2=0.2$. For the eigenfunctions, we let $\psi_1(\bs{z})=c_1\sin(\pi z_{1})+c_2,~\psi_2(\bs{z})=c_3\cos(\pi z_{2})+c_4$, where $c_1=0.988,~c_2=0.5,~c_3=2.157,~\mathrm{and}~c_4=-0.084$ to guarantee that the eigenfunctions are orthonormal. We generate heterogenous measurement errors with $\sigma(\bs{z})=0.25\{1-(z_{1}-0.5)^2-(z_{2}-0.5)^2\}$. We consider $n=50,~100,~200$, and for each image, we consider two types of resolution: $40\times 40$ and $79\times 79$ with $N=921$ and $3682$ pixels falling inside the domain, respectively.

To apply our method, we consider three different triangulations which are also shown in the first column of Figures \ref{FIG:S01} -- \ref{FIG:S06} in the Appendices. The first triangulation ($\triangle_1$) contains 49 triangles and 38 vertices; the second triangulation ($\triangle_2$) contains 80 triangles and 54 vertices; while the third triangulation ($\triangle_3$) contains 144 triangles and 87 vertices. 
The estimated mean function based on these three triangulations are shown in the second columns of Figures \ref{FIG:S01} -- \ref{FIG:S06},  and the corresponding 99\% SCCs are given in the last two columns. From these figures, one can see that all three triangulations result in almost the same estimates and SCCs. One can also see that even when the number of images is moderately large, the estimation is very accurate regardless of the type of underling mean functions.


Table \ref{TAB:1} and Table \ref{TAB:S02} in the Appendices summarize the estimated coverage rate of the SCCs based on 1000 replications for $N=921$ and $3682$, respectively. The number in parenthesis represents the average bandwidth. These two tables also confirm that there is little difference among the three triangulations and that the coverage rate is closer to the nominal confidence level for larger values of $n$.

\begin{table}
	\caption{Empirical coverage rates of the SCCs ($N=921$).}
	\begin{center}
		\begin{tabular}{ccccccccccccccc} \hline \hline
			\multirow{2}{*}{$n$} &\multicolumn{3}{c}{$\alpha=0.10$} &\multicolumn{3}{c}{$\alpha=0.05$} &\multicolumn{3}{c}{$\alpha=0.01$}\\ \cline{2-10}
			&$\triangle_1$ &$\triangle_2$ &$\triangle_3$ &$\triangle_1$ &$\triangle_2$ &$\triangle_3$ &$\triangle_1$ &$\triangle_2$ &$\triangle_3$\\ \hline
			
			\multicolumn{10}{c}{$\mu(\bs{z})=20\left\{(z_1-0.5)^2+(z_2-0.5)^2\right\}$}\\ \hline
			\multirow{2}{*}{50} & 0.858 & 0.860 & 0.874 & 0.928 & 0.929 & 0.935 & 0.977 & 0.981 & 0.981 \\
			& (0.651) & (0.651) & (0.659) & (0.739) & (0.739) & (0.747) & (0.908) & (0.908) & (0.916) \\
			\multirow{2}{*}{100} & 0.891 & 0.893 & 0.897 & 0.944 & 0.947 & 0.949 & 0.979 & 0.979 & 0.980 \\
			& (0.473) & (0.473) & (0.474) & (0.535) & (0.535) & (0.537) & (0.657) & (0.657) & (0.659)  \\
			\multirow{2}{*}{200} & 0.896 & 0.897 & 0.897 & 0.942 & 0.949 & 0.948 & 0.987 & 0.988 & 0.988  \\
			& (0.335) & (0.336) & (0.337) & (0.379) & (0.380) & (0.381) & (0.465) & (0.466) & (0.467)  \\ \hline
			
			\multicolumn{10}{c}{$\mu(\bs{z})=5\exp\left[-15\left\{(z_1-0.5)^2+(z_2-0.5)^2\right\}\right]+0.5$}\\ \hline
			\multirow{2}{*}{50} & 0.877 & 0.879 & 0.879 & 0.939 & 0.941 & 0.937 & 0.983 & 0.983 & 0.982 \\
			&  (0.664) & (0.666) & (0.667) & (0.752) & (0.754) & (0.755) & (0.921) & (0.923) & (0.924)\\
			\multirow{2}{*}{100} & 0.888 & 0.892 & 0.892 & 0.942 & 0.944 & 0.945 & 0.979 & 0.980 & 0.980 \\
			& (0.473) & (0.474) & (0.474) & (0.535) & (0.536) & (0.537) & (0.657) & (0.658) & (0.659)  \\
			\multirow{2}{*}{200} & 0.904 & 0.890 & 0.902 & 0.947 & 0.942 & 0.949 & 0.986 & 0.986 & 0.986  \\
			& (0.341) & (0.336) & (0.342) & (0.385) & (0.381) & (0.386) & (0.470) & (0.466) & (0.472)  \\ \hline
			
			\multicolumn{10}{c}{$\mu(\bs{z})=3.2(-z_1^3+z_2^3)+2.4$}\\ \hline
			\multirow{2}{*}{50} & 0.876 & 0.879 & 0.880 & 0.934 & 0.937 & 0.938 & 0.980 & 0.981 & 0.981 \\
			& (0.639) & (0.639) & (0.639) & (0.727) & (0.728) & (0.728) & (0.896) & (0.896) & (0.897)  \\
			\multirow{2}{*}{100} & 0.870 & 0.876 & 0.884 & 0.929 & 0.935 & 0.938 & 0.979 & 0.980 & 0.980 \\
			& (0.455) & (0.455) & (0.457) & (0.517) & (0.517) & (0.519) & (0.639) & (0.640) & (0.642) \\
			\multirow{2}{*}{200} & 0.890 & 0.889 & 0.906 & 0.941 & 0.942 & 0.953 & 0.984 & 0.986 & 0.985 \\
			& (0.326) & (0.325) & (0.329) & (0.370) & (0.370) & (0.373) & (0.456) & (0.456) & (0.459)  \\ \hline
			
			\multicolumn{10}{c}{$\mu(\bs{z})=-10[\sin\{5\pi (z_{1}+0.22)\}-\sin\{5\pi (z_{2}-0.18)\}]+2.8$}\\ \hline
			\multirow{2}{*}{50} & 0.882 & 0.869 & 0.879 & 0.937 & 0.930 & 0.939 & 0.981 & 0.976 & 0.980 \\
			& (0.734) & (0.740) & (0.754) & (0.821) & (0.828) & (0.843) & (0.989) & (0.996) & (1.011) \\
			\multirow{2}{*}{100} & 0.886 & 0.901 & 0.880 & 0.938 & 0.946 & 0.935 & 0.982 & 0.983 & 0.982 \\
			& (0.522) & (0.534) & (0.536) & (0.584) & (0.596) & (0.598) & (0.705) & (0.718) & (0.721)  \\
			\multirow{2}{*}{200} & 0.877 & 0.891 & 0.887 & 0.937 & 0.951 & 0.947 & 0.985 & 0.986 & 0.984  \\
			& (0.370) & (0.378) & (0.384) & (0.414) & (0.423) & (0.429) & (0.499) & (0.508) & (0.514)  \\ \hline
			\hline
		\end{tabular}
	\end{center}
	\label{TAB:1}
\end{table}

\vskip .12in \noindent \textbf{5.2. Two Sample SCC} \vskip .10in

In this simulation study, we examine the power of detecting a difference in mean images based on the proposed two-sample SCC. Two group of images are generated from the model:
\begin{equation*}
	Y_{H,ij} = \mu_{H}(\bs{z}_{j}) + \sum_{k=1}^{\kappa} \sqrt{\lambda_k} \xi_{ij} \psi_{k}(\bs{z}_{j}) + \sigma(\bs{z}_{j}) \varepsilon_{ij}, ~ H=1,2,
\end{equation*}
where $\psi_k$'s are generated as in the simulation in Section 5.1.
We consider the following:
\begin{equation}
	H_0:~\mu_1(\bs{z})=\mu_2(\bs{z}), \hbox{~for~all~} \bs{z}\in \Omega \quad\mathrm{v.s.}\quad H_a:~\mu_1(\bs{z})\neq \mu_2(\bs{z}) \hbox{~for~some~} \bs{z} \in \Omega.
	\label{EQ:H0_2g}
\end{equation}
The mean functions for two groups considered here are $\mu_1(\bs{z}) = 20\{(z_1 -0.5)^2+(z_2 -0.5)^2\}$, and $\mu_2(\bs{z}) = \mu_1(\bs{z})+\delta (-z_1^3+z_2^3)$. The value of $\delta$ controls the difference between the two groups. The eigenvalues $\lambda_k$'s, eigenfunctions $\psi_k$'s and the measurement errors $\varepsilon_{ij}$'s are generated in the same way as in the simulation presented in Section 5.1, and we set $\sigma(\bs{z}) =0.1$.

Figure \ref{FIG:2} and Table \ref{TAB:S01} in the Appendices summarize the estimated probability of rejecting $H_0$ in (\ref{EQ:H0_2g}) with nominal level $\alpha=0.10,~0.05~\mathrm{and}~0.01$. When $\delta=0$, the probability should be close to the nominal level, and when $\delta$ is large, the estimated power should be close to $1$. From Figure \ref{FIG:2} and Table \ref{TAB:S01}, one can see even when the numbers of the images $n_1$ and $n_2$ are moderately large, the size of the test is very close to the nominal level. The estimated power increases quickly as $n_1$ and $n_2$ increase. The performance of the procedure is similar and consistent for different triangulations.


\begin{figure}
	\begin{center}
		\begin{tabular}{ccc}
			\includegraphics[height=2in,width=2.45in]{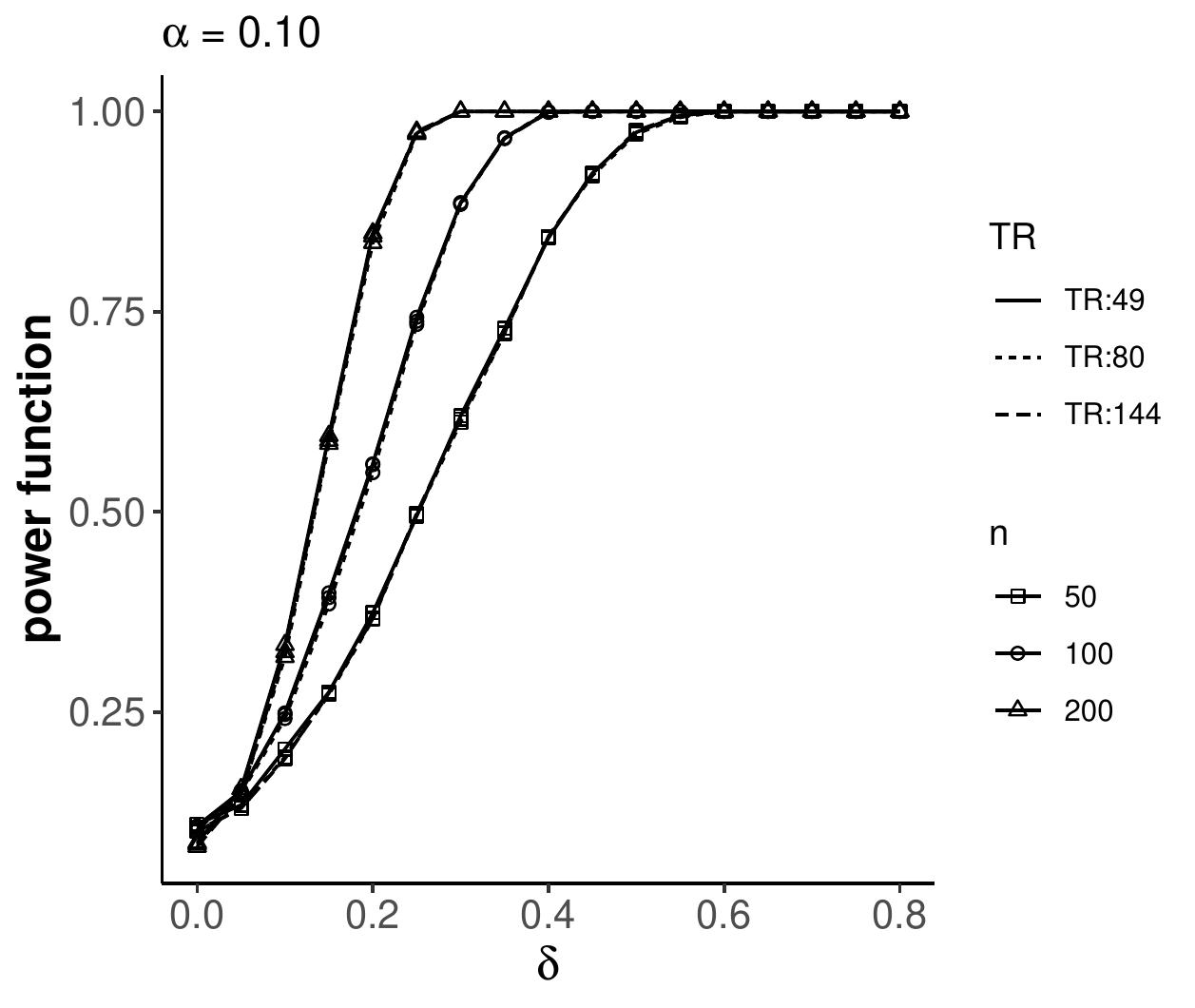} \hspace{-.7in} &
			\includegraphics[height=2in,width=2.45in]{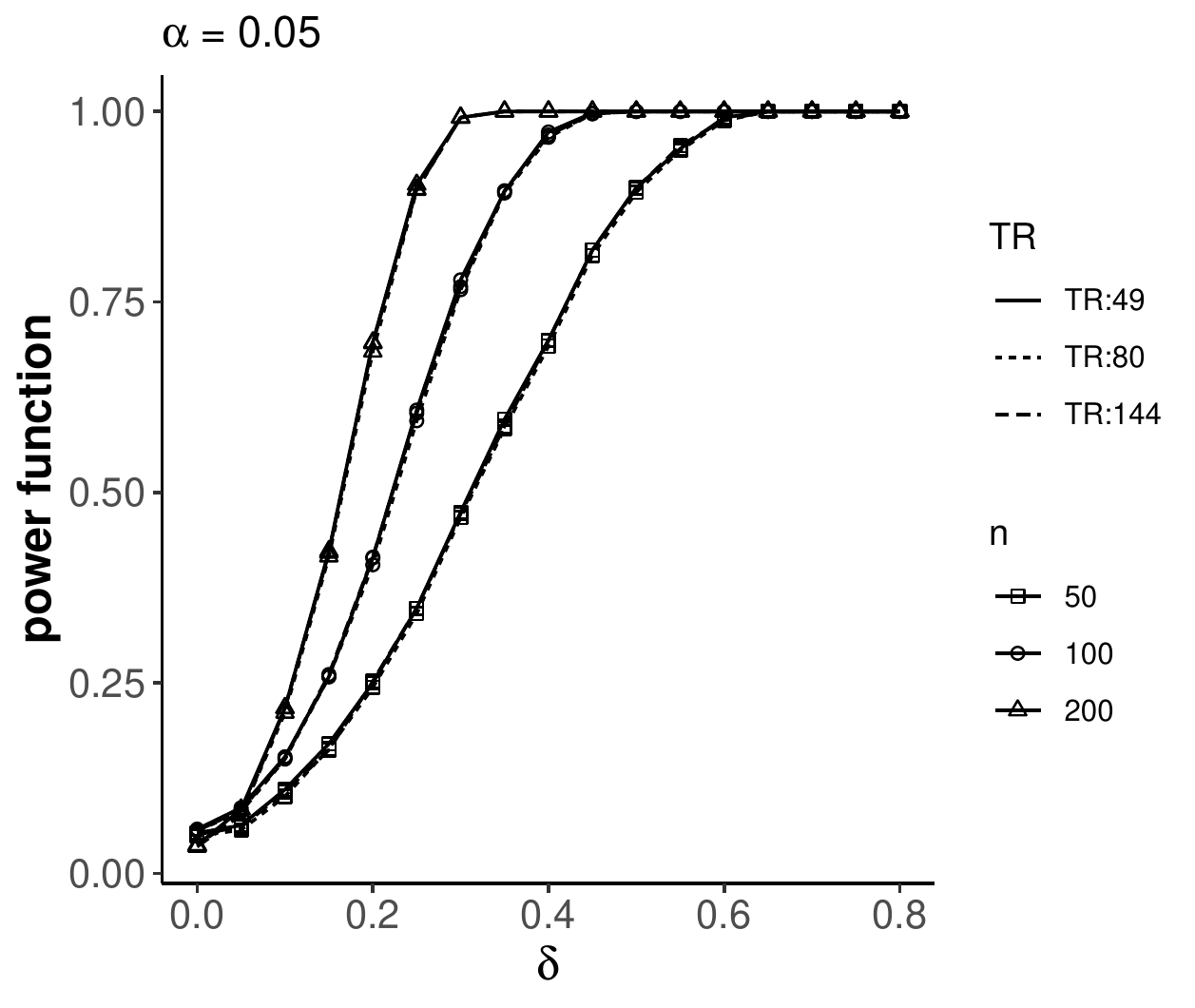} \hspace{-.7in} &
			\includegraphics[height=2in,width=2.45in]{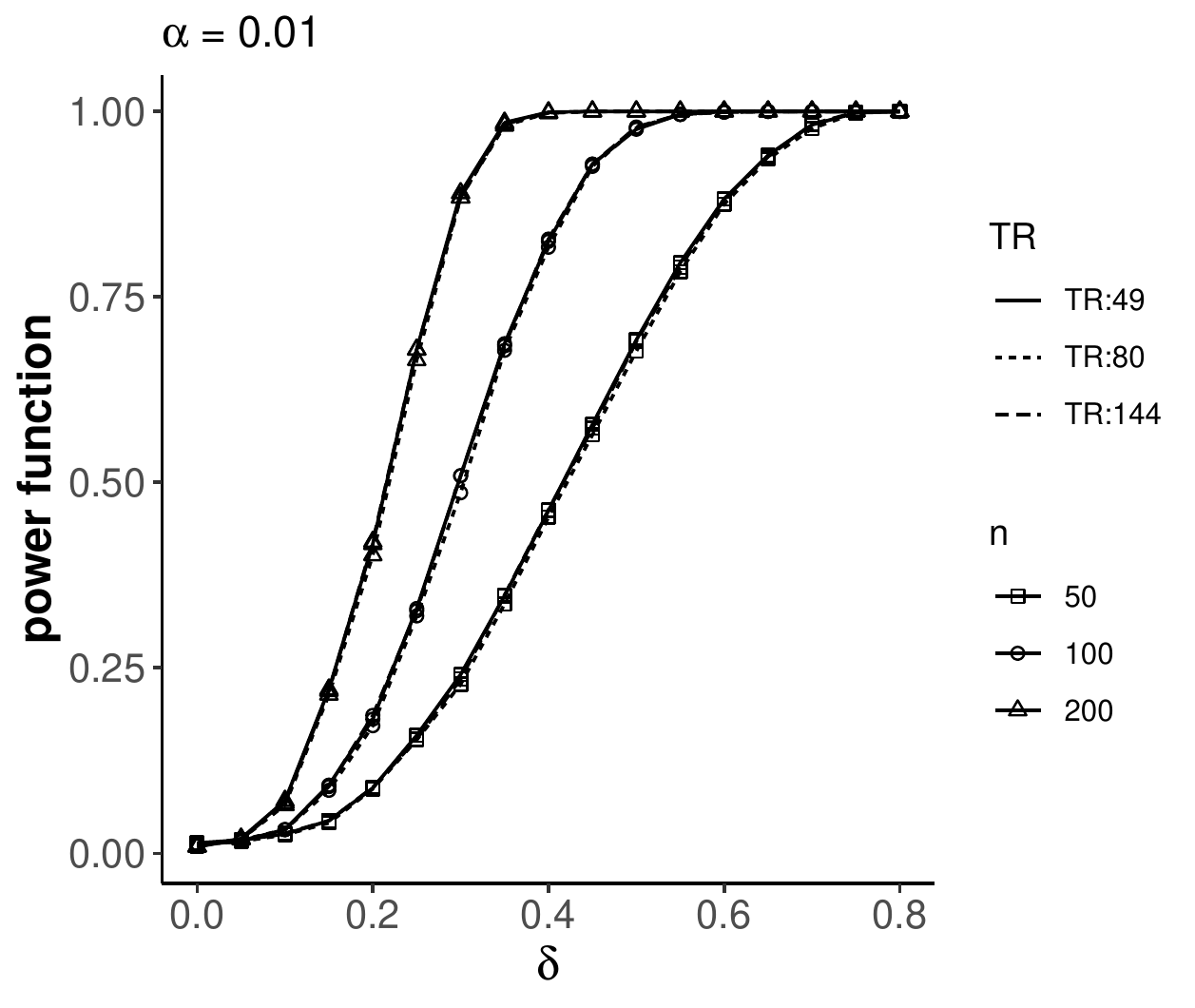}\\
			~~~~~(a) $\alpha=0.10$ & ~~~(b) $\alpha=0.05$ & ~~~(c) $\alpha=0.01$ \\[12pt]
		\end{tabular}    
	\end{center}
	\caption{Type I error and empirical power of two-sample test for different $\alpha$'s.}
	\label{FIG:2}
\end{figure}

\setcounter{chapter}{6} \renewcommand{\thetheorem}{6.\arabic{theorem}}
\renewcommand{\thelemma}{6.\arabic{lemma}}
\renewcommand{\theproposition}{6.\arabic{proposition}}
\renewcommand{\thetable}{6.\arabic{table}} \setcounter{table}{0} 
\renewcommand{\thefigure}{6.\arabic{figure}} \setcounter{figure}{0} 
\setcounter{equation}{0} \setcounter{lemma}{0} \setcounter{theorem}{0}
\setcounter{proposition}{0}\setcounter{corollary}{0}
\vskip .12in \noindent \textbf{6. Applications to Brain Imaging Data} \vskip 0.10in
\label{SEC:application}

In this section, we implement the proposed SCCs to analyze brain imaging data. In particular, we consider data taken from positron emission tomography (PET) studies with two different settings: one using the tracer [C${}^{11}$]WAY100635 that has an affinity for the serotonin 1A receptor in a study of major depressive disorder (MDD); and one using the fluorodeoxyglucose tracer [F${}^{18}$]FDG, a glucose analog, in a study of dementia. The imaging data are naturally three-dimensional in each case, but we focus here on one strategically selected slice in each setting. For the MDD study, we select the horizontal slice which passes through the midbrain and the amygdala, two regions implicated in MDD \citep{Parsey:Ogden:Tin:10}. As pointed out by \cite{Marcus:Mena:Subramaniam:14}, within the brain, the anatomical regions that are commonly affected by Alzheimer diseases are the bilateral superior medial frontal, anterior, middle cingulate and bilateral parietal cortices, while the regions such as the bilateral medial temporal lobes are usually less affected. Therefore, for the [F${}^{18}$]FDG study, we focus on the 48th horizontal slice of the brain since it passes through the frontal and parietal lobes. In each case, we consider the hypotheses in (\ref{EQ:H0_2g})  for the difference between two mean functions.

For the [C${}^{11}$]WAY100635 data, we have 40 subjects who are classified as normal controls and 26 who have been diagnosed with MDD \citep{Parsey:Oquendo:Ogden:06}. Figure \ref{FIG:app_04} displays the results of the application of the proposed procedure to these data.  The portions of the SCCs not containing zero can be seen in (a); the estimation of the mean difference between the two groups is shown in (b), and the lower and upper SCCs are shown in (c) and (d).

\begin{figure}
	\begin{center}
		\begin{tabular}{cccc}
			\includegraphics[scale=0.15]{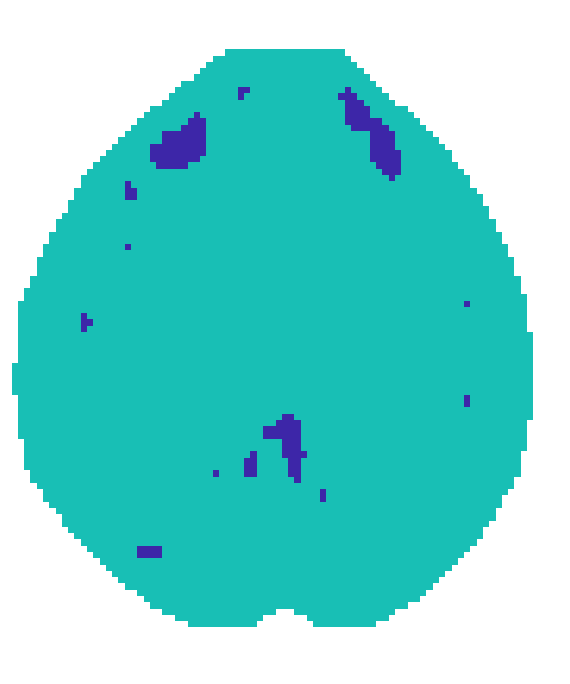} &\includegraphics[scale=0.15]{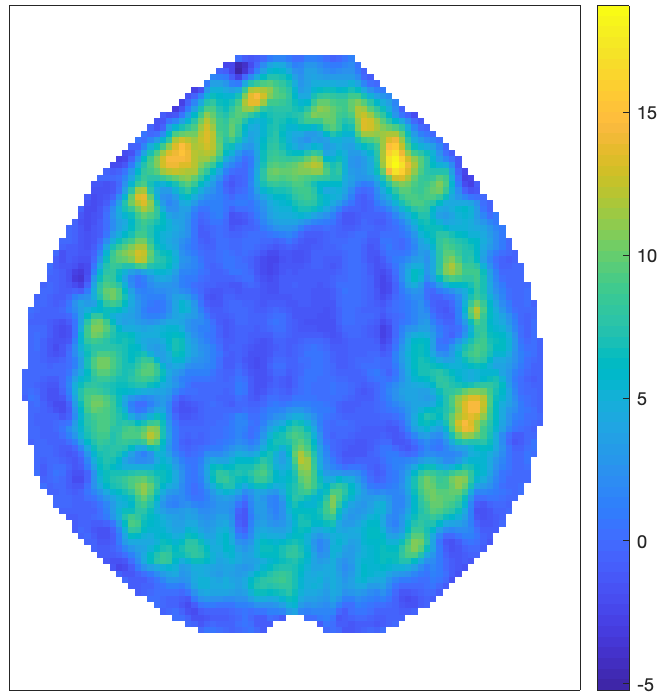} &\includegraphics[scale=0.15]{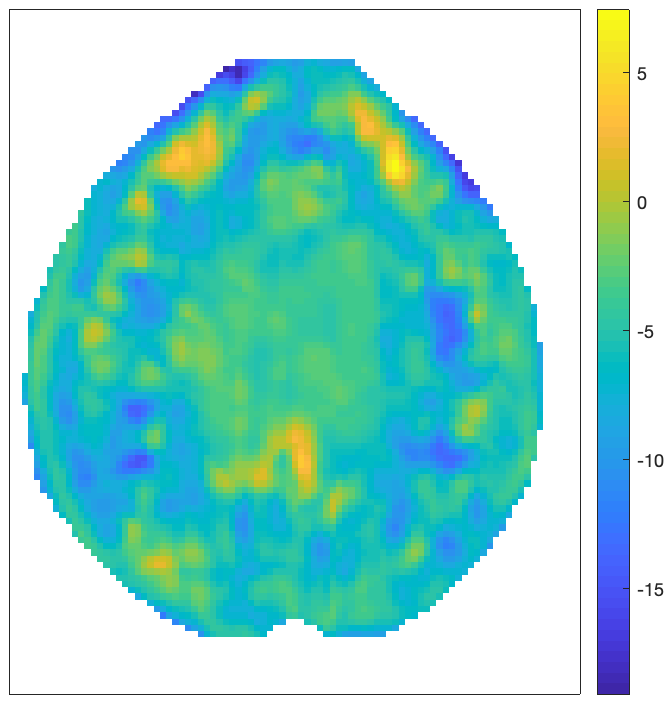} &\includegraphics[scale=0.15]{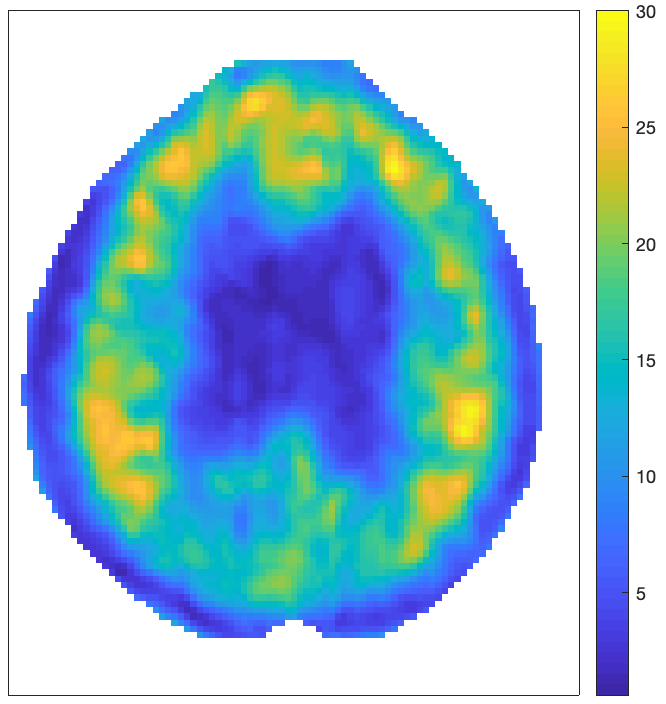}\\
			(a) Coverage of zero  &(b) $\widehat{\mu}_{\mathrm{MDD}}-\widehat{\mu}_{\mathrm{CON}}$  &(c) Lower SCC &(d) Upper SCC\\
			& & &\\
		\end{tabular}
	\end{center}
	\caption{SCC for comparison between CON and MDD. (In (a), yellow color indicates zero falls above the upper band and blue color indicates zero falls beneath the lower band.)}
	\label{FIG:app_04}
\end{figure}

Next, we illustrate these procedures by applying them to PET data from the Alzheimer's Disease Neuroimaging Initiative (ADNI; \url{adni.loni.usc.edu}). One of the primary goals of the ADNI study is to test whether PET and some other biological markers can be combined to measure the progression of mild cognitive impairment (MCI) and early Alzheimer's disease (AD). This dataset consists of 112 subjects with normal cognitive functions (control group; CON), 213 subjects with mild cognitive impairment (MCI), and 122 subjects who have been diagnosed  with Alzheimer's Disease (AD).

We use the proposed method in Section 4.1 to choose the triangulation. Among the three triangulation candidates ($\triangle_1$--$\triangle_3$) considered in simulation studies, we choose $\triangle_3$ when estimating the mean functions, and $\triangle_1$ when estimating the covariance functions. As we suggested in Section 4.2, we use smooth parameter $r=1$ with degree $d=5$ for the estimation of mean function and $d=2$ for the estimation of $\eta_{i}$'s. The results of this application are displayed in Figure \ref{FIG:app_02}. The first row of Figure \ref{FIG:app_02} displays the areas in which zero is not contained within the 95\% SCC comparing each pair of diagnostic groups.  This suggests that the AD group has widespread mean differences from each of the other two groups.  Since this dataset is relatively large, we also stratify the data according to sex and age (greater or less than 75 years) and within each stratum we examine the SCC for the difference between all pairs of diagnostic groups.  The breakdowns of these data in terms of these variables are given in Table \ref{TAB:S03}.

\begin{figure}
	\begin{center}
		\begin{tabular}{cccc}
			Group  &\hspace{0.0in} CON vs MCI & CON vs AD & MCI vs AD\\
			
			Entire Group & \includegraphics[scale=0.15]{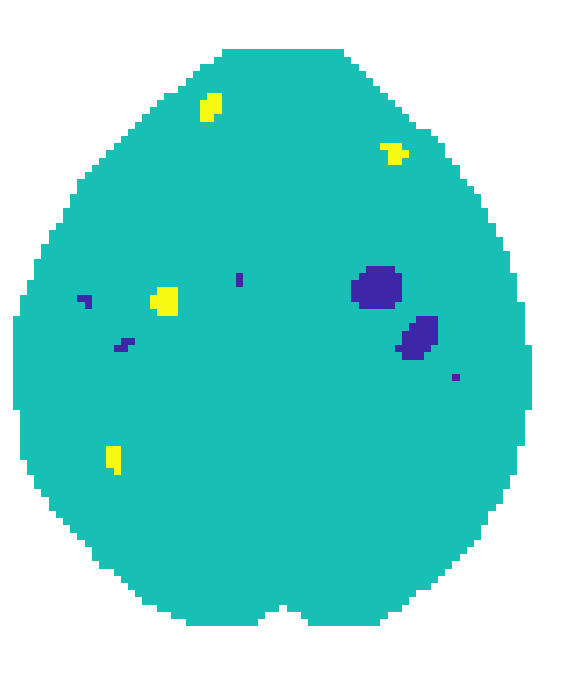} &   \includegraphics[scale=0.15]{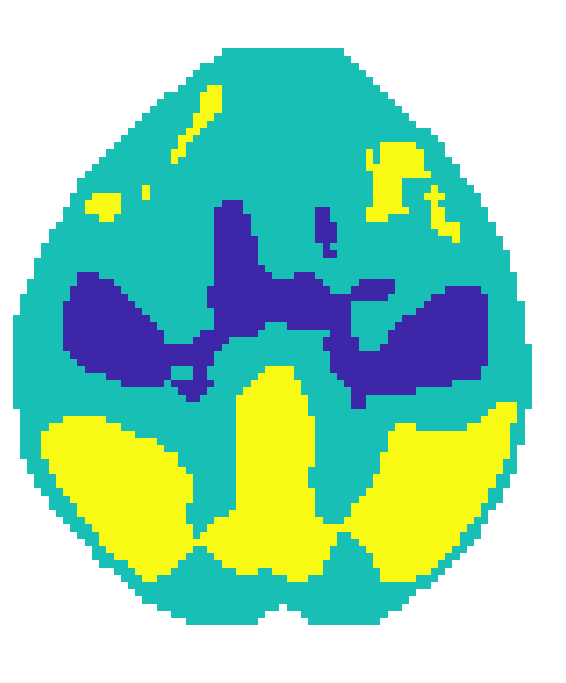} &   \includegraphics[scale=0.15]{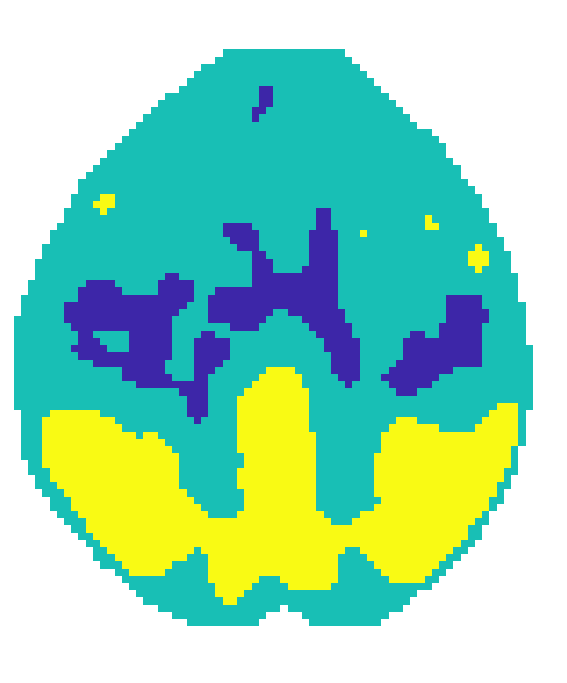}\\
			
			Female &  \includegraphics[scale=0.15]{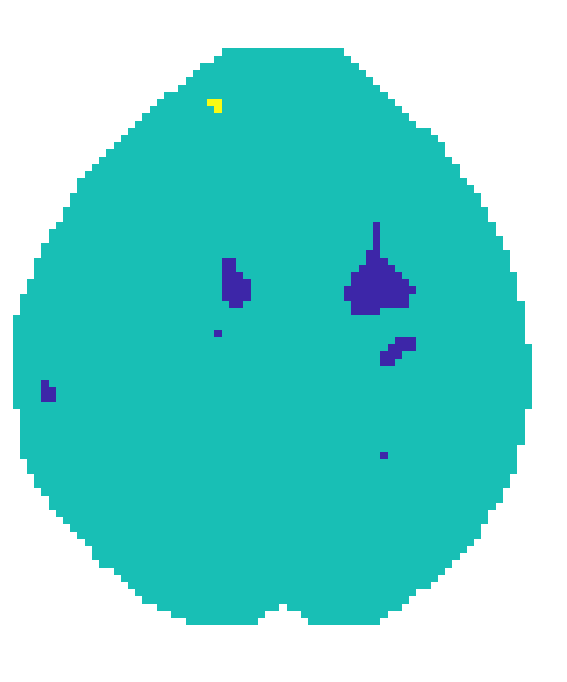} & \includegraphics[scale=0.15]{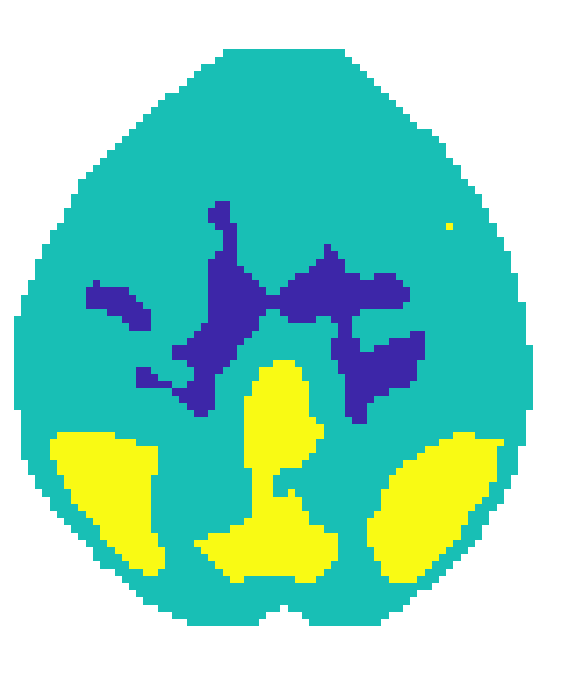} & \includegraphics[scale=0.15]{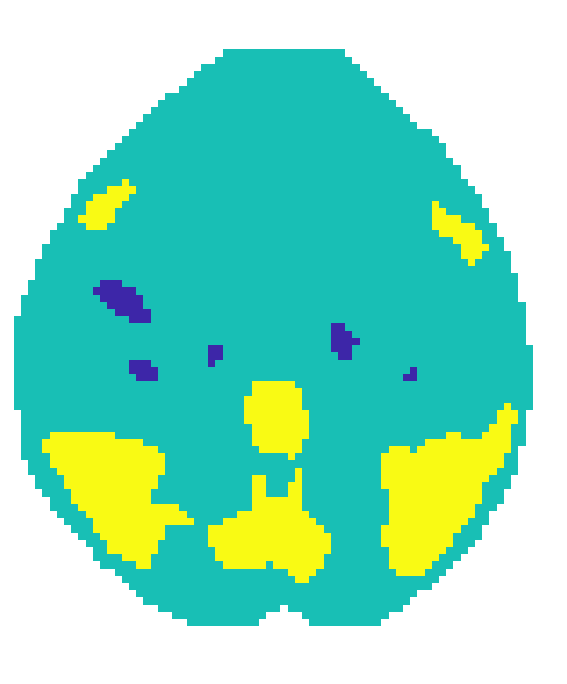}\\
			
			Male & \includegraphics[scale=0.15]{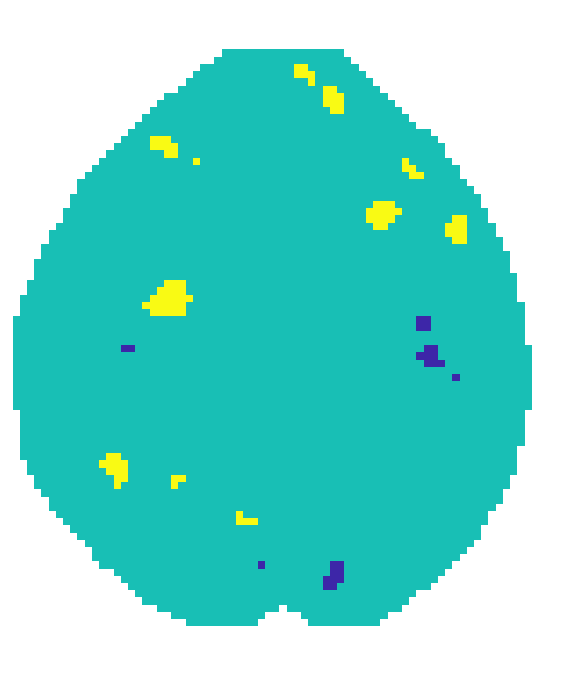} &\includegraphics[scale=0.15]{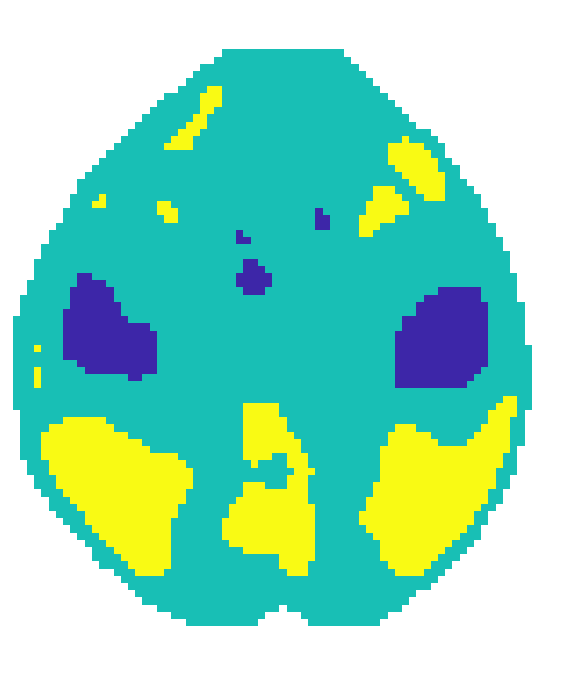} & \includegraphics[scale=0.15]{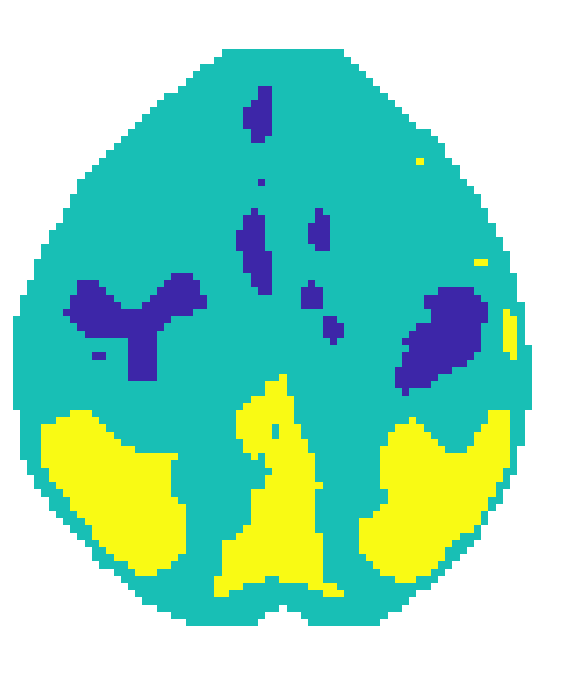}\\
			
			Age $\leq 75$ & \includegraphics[scale=0.15]{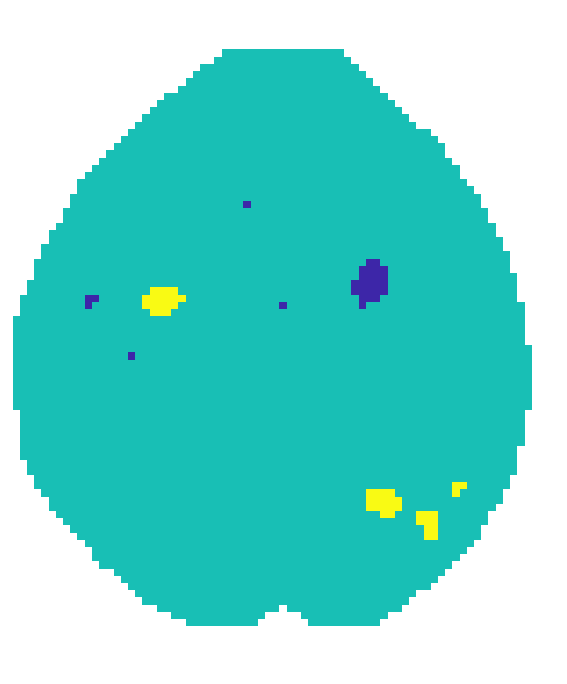} & \includegraphics[scale=0.15]{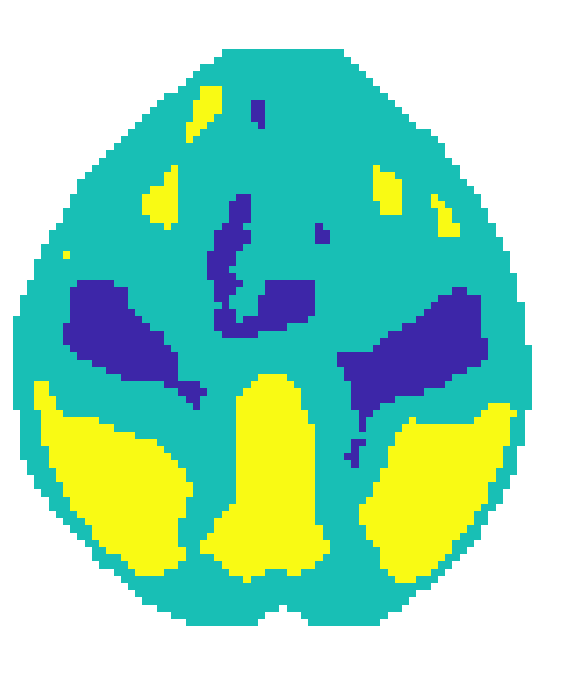} & \includegraphics[scale=0.15]{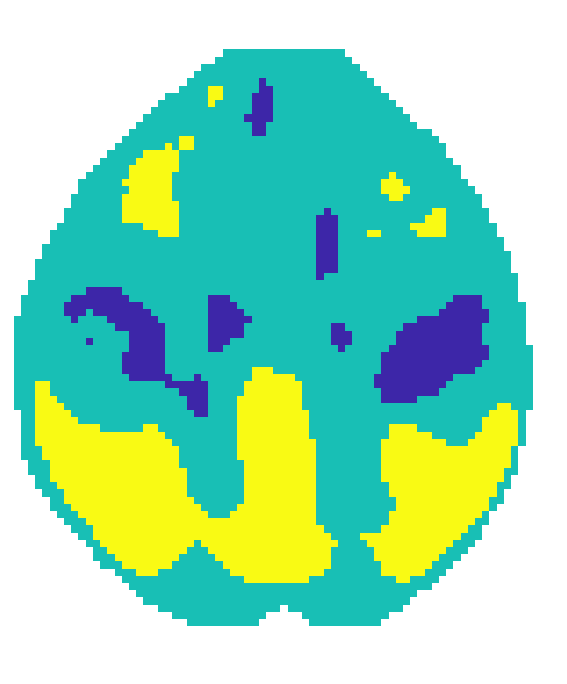}\\
			
			Age $>75$ & \includegraphics[scale=0.15]{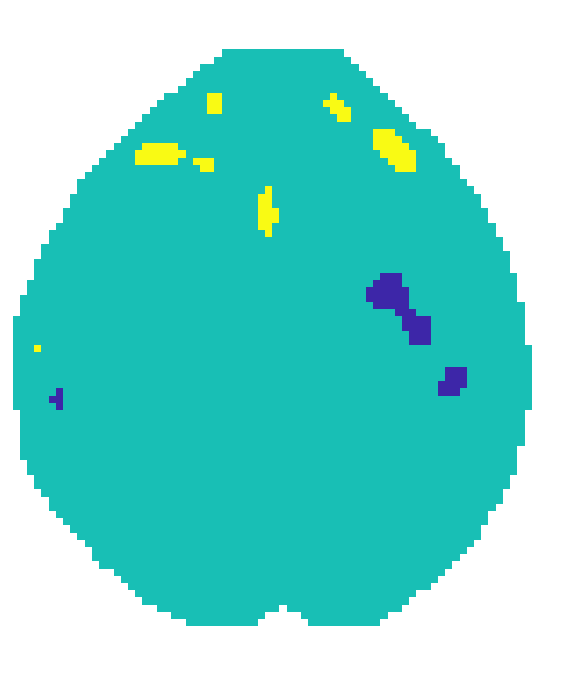} & \includegraphics[scale=0.15]{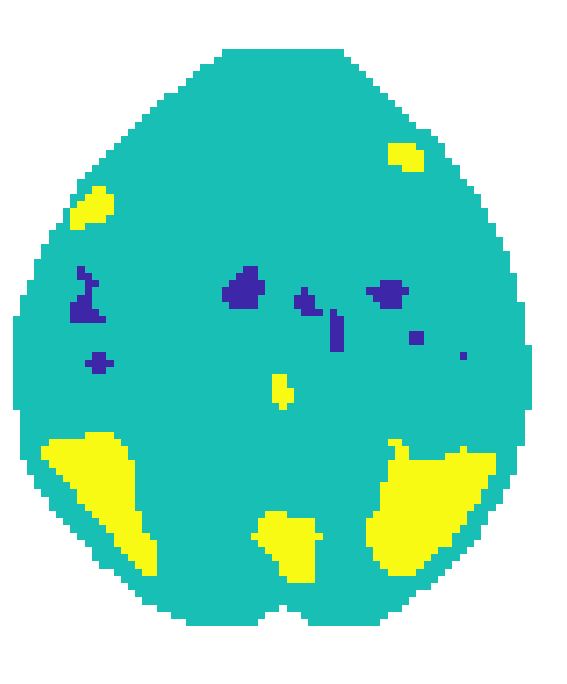} & \includegraphics[scale=0.15]{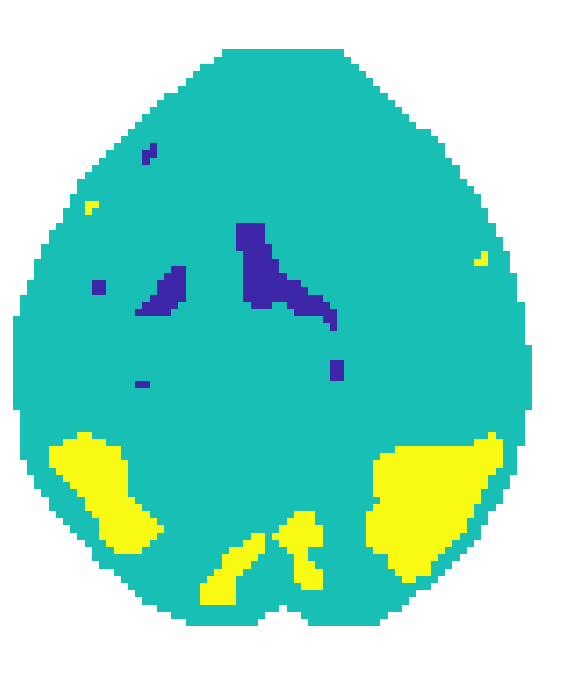}\\
			& & &\\
		\end{tabular}
	\end{center}
	\caption{Coverage of zero of SCC for pairwise comparisons among CON, MCI and AD. (Yellow color indicates zero falls above the upper band and blue color indicates zero falls beneath the lower band.)}
	\label{FIG:app_02}
\end{figure}

\begin{table}
	\caption{Two way table of diagnosis vs. gender and age group.}
	\begin{center}
		\begin{tabular}{llcccc} \hline\hline
			& &\multicolumn{3}{c}{Diagnosis} & \multirow{2}{*}{Total}\\ \cline{3-5}
			& &CON &MCI &AD & \\ \hline
			\multirow{2}{*}{Gender} &Female & 42 & 77 & 50 & 169\\
			&Male & 70 & 136 & 72 & 278\\ \hline
			\multirow{2}{*}{Age} &Age $\leq 75$ & 54 & 107 & 60 & 221\\
			&Age $> 75$ & 58 & 106 & 62 & 226\\ \hline
			Total & & 112 & 213 & 122 & 447\\ \hline\hline
		\end{tabular}
	\end{center}
	\label{TAB:S03}
\end{table}
The large apparent differences in the full group analysis can be seen (but to a lesser extent) in the comparisons among the males and among the relatively younger population, but are less pronounced in the other sub-group analyses.

\setcounter{chapter}{7} \setcounter{equation}{0} \vskip .10in
\noindent \textbf{7. Discussion} \label{SEC:discussion} \vskip 0.05in

We develop SCCs for mean functions of imaging data in the functional data framework. We show that the proposed procedure has desirable statistical properties: the estimators are semiparametrical efficient, asymptotically efficient as if all images were observed with no error. One main advantage of our method is its computational efficiency and feasibility for large-scale imaging data. It greatly enhances the application of SCCs to imaging data in biomedical studies.

In this paper, we approximate the bivariate function of the spatial effect using the bivariate splines over triangulations. We prefer the bivariate penalized splines (BPS) due to their (i) convenient representations with flexible degrees and various smoothness, (ii) computational efficiency, and (iii) great ability of handling the sparse designs. 

A few more issues still merit further research. For instance, the triangulation selection using the cross-validation and wild bootstrap works well in practice, but a stronger theoretical justification for their use is still needed in the FDA context. In recent years, there has been a great deal of work on functional regression. It is interesting to extend the proposed methodology to functional regression models. The construction of SCCs in such models is a significant challenge and requires more in-depth investigation. Last but not least, it is also interesting to develop SCCs for large-scale longitudinal imaging data, in which accounting for the dependence within the subject as well as for the longitudinal design is crucial for making inference.

\setcounter{chapter}{8} \setcounter{equation}{0} \vskip .10in
\noindent \textbf{Acknowledgment}

The  authors  are  truly  grateful to the editor, the associate editor and two reviewers for their constructive suggestions that led to significant improvement of the article. Li Wang's research was supported in part by National Science Foundation award DMS-1542332. Todd Ogden's work was partially supported by NIH grants 5 R01 EB024526 and 2 P50 MH090964. Data used in preparation of this article were obtained from the Alzheimer's Disease Neuroimaging Initiative (ADNI) database (\url{adni.loni.usc.edu}). As such, the investigators within the ADNI contributed to the design and implementation of ADNI and/or provided data but did not participate in analysis or writing of this report. A complete listing of ADNI investigators can be found at: \url{http://adni.loni.usc.edu/wp-content/uploads/how_to_apply/ADNI_Acknowledgement_List.pdf}.
\vspace*{-8pt}

\newpage
\vskip 0.10in \noindent \textbf{Appendices}

\setcounter{chapter}{9} \renewcommand{\thetheorem}{A.\arabic{theorem}}
\renewcommand{\theproposition}{A.\arabic{proposition}}
\renewcommand{\thelemma}{A.\arabic{lemma}}
\renewcommand{\thecorollary}{A.\arabic{corollary}}
\renewcommand{\theequation}{A.\arabic{equation}} \renewcommand{\thesubsection}{A.\arabic{subsection}}
\renewcommand{\thetable}{{\arabic{table}}} \setcounter{table}{0}
\renewcommand{\thefigure}{\arabic{figure}} \setcounter{figure}{0}
\setcounter{equation}{0} \setcounter{lemma}{0} \setcounter{proposition}{0}
\setcounter{theorem}{0} \setcounter{subsection}{0}\setcounter{corollary}{0}

\vskip .05in \noindent \textbf{A. More Results from Simulation Studies}
\label{sec:web-A}
In this section, we present more simulation results from Sections 5.1 and 5.2 in the main paper.
For the simulation example presented in Section 5.1, Figures \ref{FIG:S01} - \ref{FIG:S03} present the 99\% SCCs for the quadratic mean function based on sample size $n = 50$, $100$ and $200$. Figures \ref{FIG:S04} - \ref{FIG:S06} present the 99\% SCCs for the exponential, cubic and sine mean functions with $n=50$, respectively. Table \ref{TAB:S02} summarizes the estimated coverage rate of the SCCs based on 1000 replications for $3682$. Table \ref{TAB:S01} provides the type I error and the empirical power of the two-sample test presented in Section 5.2.

To illustrate the benefits of using our method, we conduct the following simulation study to compare the proposed SCC with the traditional multiple testing with Bonferroni correction and the cluster threshold-based method \citep{Poldrack:Mumford:Nichols:11}. Similar as in Sections 5.1 in the main paper, we generate the images from the following model:
\begin{equation*}
	Y_{ij} = \mu(\bs{z}_{j}) + \sum_{k=1}^{\kappa} \sqrt{\lambda_k} \xi_{ij} \psi_{k}(\bs{z}_{j}) + \sigma(\bs{z}_{j}) \varepsilon_{ij},~\bs{z}_j\in \Omega\subset [0,1]^2.
\end{equation*}
For comparison, we consider the following mean function, which is similar as the exponential function in Example 1 in Section 5.1:
\begin{equation*}
	\mu(\bs{z})=\left\{\begin{array}{ll}
		\exp\left[-30\left\{(z_1-0.5)^2+(z_2-0.5)^2\right\}\right], &(z_1-0.5)^2+(z_2-0.5)^2\leq 0.10\\
		0, &(z_1-0.5)^2+(z_2-0.5)^2>0.10\\ \end{array}\right.,
\end{equation*}
and the corresponding images are shown in Figure \ref{Fig:S07}. To simulate the within-image dependence, we generate $\xi_{ik}\overset{\text{i.i.d}}{\sim}N(0,1)$ for $i=1,\ldots,n,~k=1,2$, and orthonormal basis functions $\psi_1(\bs{z})=0.988\sin(\pi z_1)+0.5,~\psi_2(\bs{z})=2.157\cos(\pi z_2)-0.084$. For the eigenvalues, we set $\lambda_1=0.2$, $\lambda_2=0.05$. We consider $n=100,~200$ and for each image, the number of pixels is set to be the same as in typical brain imaging which is $N=79\times 95=7,505$.

Based on these images, we are interested in testing $H_0:~\mu(\bs{z}_{j})=0,~\bs{z}_{j}\in\Omega$, $j=1,\ldots,N$, at significance level $\alpha=0.05$. For the cluster approach, the threshold is usually set by the practitioner's experience and prior knowledge. In this example, we consider three thresholds: 0.1, 0.05 and 0.01, as suggested in \cite{Poldrack:Mumford:Nichols:11}. For comparison, we consider the following criteria:
\begin{itemize}
	\item False Positive Rate (FPR): the proportion of pixels within the domain which are discovered incorrectly as positive (significantly different from zero);
	\item False Negative Rate (FNR): the proportion of pixels within the domain which are discovered incorrectly as negative (not significantly different from zero);
	\item False Discovery Rate (FDR): the proportion of detected pixels that are false positives.
\end{itemize} 

Table \ref{TAB:S00} summarizes all results based on 100 replications. Figure \ref{FIG:S08} shows the discovery of the true signal via different methods for a typical replication with $n=200$. Based on Table \ref{TAB:S00} and Figure \ref{FIG:S08}, it is obvious that the pixel-wise inference with Bonferroni correction is very conservative. The FPRs and FDRs of the Bonferroni correction are very close to zero, while the FNRs are very high, even greater than 30\%. Although the FPRs and FDRs for the proposed SCC are above zero, they are still very small, usually less than 1\%. Meanwhile, the FNRs for the proposed SCC are much smaller than the Bonferroni correction. In addition, one sees that the cluster threshold-based method heavily depends on the choice of threshold. When using 0.01 as the threshold instead of 0.1, the FPR dramatically decreases while the FNR considerably increases. For $n=200$, the FPR and FDR of the SCC are both smaller than those of the Cluster-threshold method. From Figure \ref{FIG:S08}, we can also see that our method aims at detecting contiguous groups of active pixels because it is able to account for the spatial dependence within data. 

\setcounter{chapter}{10} \renewcommand{\thetheorem}{B.\arabic{theorem}}
\renewcommand{\theproposition}{B.\arabic{proposition}}
\renewcommand{\thelemma}{B.\arabic{lemma}}
\renewcommand{\thecorollary}{B.\arabic{corollary}}
\renewcommand{\theequation}{B.\arabic{equation}} \renewcommand{\thesubsection}{B.\arabic{subsection}}
\renewcommand{\thetable}{{B.\arabic{table}}} \setcounter{table}{0}
\renewcommand{\thefigure}{B.\arabic{figure}} \setcounter{figure}{0}
\setcounter{equation}{0} \setcounter{lemma}{0} \setcounter{proposition}{0}
\setcounter{theorem}{0} \setcounter{subsection}{0} \setcounter{corollary}{0}
\vskip .05in \noindent \textbf{B. Technical Proofs}
\label{sec:web-B}

In the following, we use $c$, $C$, $c_1$, $c_2$, $C_1$, $C_2$, etc. as generic constants, which may be different even in the same line. For any sequence $a_n$ and $b_n$, we write $a_n \asymp b_n$ if there exist two positive constants $c_1, c_2$ such that $c_1 |a_n| \le |b_n| \le c_2 |a_n|$, for all $n\ge 1$. For a real valued vector $\bs{a}$, denote $\|\bs{a}\|$ its Euclidean norm. For a matrix $\mathbf{A}=(a_{ij})$, denote $\|\mathbf{A}\|_{\infty}=\max_{i,j} |a_{ij}|$. For any positive definite matrix $\mathbf{A}$, let $\lambda_{\min}(\mathbf{A})$ and $\lambda_{\max}(\mathbf{A})$ be the
smallest and largest eigenvalues of $\mathbf{A}$.

For $g_{1}(\bs{z})$, $g_{1}(\bs{z})$, define the theoretical and empirical inner products as
\begin{equation}
	\langle g_{1}, g_{2}\rangle = \int_{\Omega}g_{1}(\bs{z})g_{2}(\bs{z})d\bs{z},~\langle g_{1}, g_{2}\rangle_{N} = \frac{1}{N}\sum_{j=1}^N g_{1}(\bs{z}_j) g_{2}(\bs{z}_j),
	\label{DEF:inner_product}
\end{equation}
and denote the corresponding theoretical and empirical norms $\|\cdot\|$ and $\|\cdot\|_{N}$.
Furthermore, let $\| \cdot\|_{\mathcal{E}}$ be the norm introduced by the inner product $\langle \cdot, \cdot\rangle_{\mathcal{E}}$, where, for $g_{1}(\bs{z})$ and $g_{2}(\bs{z})$,
\begin{equation*}
	\left\langle g_{1},g_{2}\right\rangle_{\mathcal{E}}=\int_{\Omega}\left\{\sum_{i+j=2}
	\binom{2}{i}
	(\nabla_{z_{1}}^{i}\nabla_{z_{2}}^{j}g_{1}(\bs{z}))\right\}\!\left\{\sum_{i+j=2}
	\binom{2}{i}
	(\nabla_{z_{1}}^{i}\nabla_{z_{2}}^{j}g_{2}(\bs{z}))\right\}dz_{1}dz_{2}.
	\label{DEF:energy_product}\end{equation*}
Let $A(\Omega)$ be the area of the domain $\Omega$, and without loss of generality, we assume $A(\Omega)=1$ in the rest of the article.

\vskip .05in \noindent \textbf{B.1. Properties of Bivariate Splines}

We cite two important results from \cite{Lai:Schumaker:07}.

\begin{lemma}[Theorem 2.7, \cite{Lai:Schumaker:07}]
	\label{LEM:normequity}
	Let $\{B_m\}_{m\in \mathcal{M}}$ be the Bernstein polynomial basis for spline space $\mathcal{S}_{d}^{r}(\triangle)$ defined over a $\pi$-quasi-uniform triangulation $\triangle$.  Then there exist positive constants $c$, $C$ depending on the smoothness $r$, $d$, and the shape parameter $\pi$ such that
	\[
	c|\triangle|^{2}\sum_{m\in \mathcal{M}}\gamma_{m}^{2}\leq
	\left\Vert \sum_{m\in \mathcal{M}} \gamma_{m}B_{m}\right\Vert^{2}\leq C|\triangle|^{2}\sum_{m\in \mathcal{M}}\gamma_{m}^{2}.
	\label{EQ:normequity}
	\]
\end{lemma}

\begin{lemma}[Theorems 10.2 and 10.10, \cite{Lai:Schumaker:07}]
	\label{LEM:appord}
	Suppose that $\triangle$ is a $\pi$-quasi-uniform triangulation of a polygonal domian $\Omega$, and $g(\cdot) \in  \mathcal{W}^{d+1,\infty}(\Omega)$.
	\begin{itemize}
		\item[(i)] For bi-integer $(a_{1},a_{2})$ with $0\leq {a_{1}}+{a_{2}} \leq d $, there exists a spline $g^{\ast}(\cdot)\in \mathcal{S}_{d}^{0}(\triangle)$ such that $\Vert \nabla_{z_{1}}^{a_{1}}\nabla_{z_{2}}^{a_{2}}\left(g-g^{\ast}\right) \Vert_{\infty}\leq C|\triangle|^{d+1-a_{1}-a_{2}}|g|_{d+1,\infty}$, where $C$ is a constant depending on $d$, and the shape parameter $\pi$.
		\item[(ii)] For bi-integer $(a_{1},a_{2})$ with $0\leq {a_{1}}+{a_{2}} \leq d $, there exists a spline $g^{\ast\ast}(\cdot)\in \mathcal{S}_{d}^{r}(\triangle)$ ($d\geq 3r+2$) such that $\Vert \nabla_{z_{1}}^{a_{1}}\nabla_{z_{2}}^{a_{2}}\left(g-g^{\ast\ast}\right) \Vert_{\infty}\leq C|\triangle|^{d+1-a_{1}-a_{2}}|g|_{d+1,\infty}$, where $C$ is a constant depending on $d$, $r$, and the shape parameter $\pi$.
	\end{itemize}
\end{lemma}
Lemma \ref{LEM:appord} shows that $\mathcal{S}_{d}^{0}(\triangle)$ has full approximation power, and $\mathcal{S}_{d}^{r}(\triangle)$ also has full approximation power if $d\geq 3r+2$.

\begin{lemma}[Lemma B.4 in Supplemental Materials, \cite{Yu:etal:18}]
	\label{LEM:integration}
	Under Assumptions (A3) and (A4), for any Bernstein basis polynomials $B_{m}(\bs{z}),~m \in \mathcal{M}$, of degree $d\geq0$, one has
	\begin{equation*}
		\max_{m\in\mathcal{M}}\left|\frac{1}{N}\sum_{j=1}^{N}B_{m}^{k}(\bs{z}_{j})
		-\int_{\Omega}B_{m}^{k}(\bs{z})d\bs{z}\right|
		= O(N^{-1/2}|\triangle|),~~ 1 \leq k < \infty,
	\end{equation*}
	\begin{equation*}
		\max_{m, m'\in\mathcal{M}}\left|\frac{1}{N}\sum_{j=1}^{N}B_{m}(\bs{z}_{j})B_{m'}(\bs{z}_{j})
		-\int_{\Omega}B_{m}(\bs{z})B_{m'}(\bs{z})d\bs{z}\right|\\
		= O(N^{-1/2}|\triangle|),~~ 1 \leq k < \infty,
	\end{equation*}
	\begin{equation*}
		\max_{m, m^{\prime} \in \mathcal{M}}\left|\frac{1}{N^2}\!\!\sum_{j,j^{\prime}=1}^{N}G_{\eta}(\bs{z}_{j},\bs{z}_{j^{\prime}})B_{m}(\bs{z}_{j})B_{m^{\prime}}(\bs{z}_{j^{\prime}})
		-\!\!\!\int_{\Omega^{2} }G_{\eta}(\bs{z},\bs{z}^{\prime})B_m(\bs{z})B_{m^{\prime}}(\bs{z}^{\prime})
		d\bs{z}d\bs{z}^{\prime}\right|
		=O(N^{-1/2}|\triangle|^{3}),
	\end{equation*}
	\begin{equation*}
		\max_{m\in \mathcal{M}}\left|\|\sigma B_m\|_{N}^2 - \|\sigma B_m\|^2\right|=\max_{m\in \mathcal{M}}\left| \frac{1}{N}
		\sum_{j=1}^{N}B_{m}^{2}(\bs{z}_{j})\sigma^{2}(\bs{z}_{j})
		-\int_{\Omega}\sigma^{2}(\bs{z})B_{m}^{2}(\bs{z})d\bs{z}\right|=O(N^{-1/2}|\triangle|).
	\end{equation*}
\end{lemma}

The following lemma provides the uniform convergence rate at which the empirical
inner product approximates the theoretical inner product defined in (\ref{DEF:inner_product}).
\begin{lemma}
	\label{LEM:Rnorder-vec}
	Let $g_{1}(\bs{z})=\sum_{m\in \mathcal{M}}\gamma_{1,m}B_{m}(\bs{z})$, $g_{2}(\bs{z})=\sum_{m\in \mathcal{M}}\gamma_{2, m}B_{m}(\bs{z})$ be any spline functions in $\mathcal{S}_{d}^{r}(\triangle)$. Suppose Assumptions (A1), (A2) and (A4) hold, and  $N^{1/2}|\triangle|\rightarrow \infty$ as $N\rightarrow \infty$, then
	\[
	\omega_{N}=\sup\limits_{g_{1},g_{2}\in \mathcal{S}_{d}^{r}(\triangle)}\left|
	\frac{\left\langle g_{1},g_{2}\right\rangle_{N}-\left\langle g_{1},g_{2}\right\rangle}{\left\|
		g_{1}\right\|\left\|g_{2}\right\|}\right| =O_{P}\left(N^{-1/2}|\triangle|^{-1}\right)=o_P(1).
	\]
\end{lemma}

\noindent\begin{proof}
	It is easy to see
	\begin{align*}
		\langle g_{1},g_{2}\rangle_{N}&=\frac{1}{N}\sum_{j=1}^{N}\left\{\sum_{m\in \mathcal{M}}\gamma_{1,m}B_{m}(\bs{z}_{j})\right\} \left\{\sum_{m^{\prime}\in \mathcal{M}}\gamma_{2,m^\prime}B_{m^\prime}(\bs{z}_{j})\right\}\\
		&=\sum_{m}\sum_{m^\prime}\gamma_{1,m}\gamma_{2,m^\prime}\frac{1}{N}\sum_{j=1}^{N} B_{m}(\bs{z}_{j})B_{m^\prime}(\bs{z}_{j}).
	\end{align*}
	Note that $\langle g_{1}, g_{2}\rangle =\sum_{m}\sum_{m^\prime}\gamma_{1,m}\gamma_{2,m^{\prime}}
	\int_{\Omega}B_{m}(\bs{z})B_{m^\prime}(\bs{z})d\bs{z}$. It follows from Assumptions (A1), (A2) and Lemma \ref{LEM:normequity} that, for any $l=1,2$, $\widetilde{c}_{l}|\triangle|^{2}\sum_{m}\gamma_{l,m}^2\leq \Vert g_{l}\Vert^{2} \leq \widetilde{C}_{l}|\triangle|^{2}\sum_{m}\gamma_{l,m}^2$, and
	\[ C_{1}|\triangle|^{2}\left(\sum_{m}\gamma_{1,m}^2\sum_{m^\prime}\gamma^2_{2,m^{\prime}}\right)^{1/2}\leq \| g_{1}\| \| g_{2}\| \leq C_{2}|\triangle|^{2}\left(\sum_{m}\gamma_{1,m}^2\sum_{m^\prime}\gamma_{2,m^{\prime}}^2\right)^{1/2}.
	\]
	Therefore, one has
	\begin{align*}
		\nonumber
		\omega_{N}&\leq 
		\frac{\sum_{|m^{\prime}-m| \leq (d+2)(d+1)/2}|\gamma_{1,m}\gamma_{2,m^{\prime}}|}{C_{1}|\triangle|^{2}\left[\sum_{m}\gamma_{1,m}^2\sum_{m^\prime}\gamma_{2,m^{\prime}}^2\right]^{1/2}}
		\max_{m,m^{\prime}\in\mathcal{M}}\left| \frac{1}{N}\sum_{j=1}^{N}B_{m}(\bs{z}_{j})B_{m^{\prime}}(\bs{z}_{j})-\int_{\Omega}B_{m}(\bs{z})B_{m^{\prime}}(\bs{z})d\bs{z}\right| \\
		&\leq C|\triangle|^{-2} \max_{m,m^{\prime}\in\mathcal{M}}\left| \frac{1}{N}\sum_{j=1}^{N}B_{m}(\bs{z}_{j})B_{m^{\prime}}(\bs{z}_{j})-\int_{\Omega}B_{m}(\bs{z})B_{m^{\prime}}(\bs{z})d\bs{z}\right|.
	\end{align*}
	The desired result follows from Lemma \ref{LEM:integration}.	
\end{proof}

As a direct result of Lemma \ref{LEM:Rnorder-vec}, we can see that
\begin{eqnarray}
	\sup_{g\in \mathcal{S}_{d}^{r}(\triangle)}\left| \left. \left\| g\right\|_{N}^{2}\right/ \Vert g\Vert^{2}-1\right| =O_{P}\left(N^{-1/2}|\triangle|^{-1}\right)=o_{P}(1).  \label{EQ:normratio}
\end{eqnarray}

\begin{lemma}
	\label{LEM:Normratios}
	Suppose Assumption (A4) hold, and  $N^{1/2}|\triangle|\rightarrow \infty$ as $N\rightarrow \infty$, then
	\begin{equation}
		S_{N}=\sup_{g \in \mathcal{S}_{d}^{r}(\triangle)}\left\{ \frac{\|g\|
			_{\infty}}{\|g\|_{N}},\|g\|_{N}\neq 0\right\}=O(|\triangle|^{-1}),  \label{EQ:Sn}
	\end{equation}
	\begin{equation}
		\overline{S}_{N}=\sup_{g \in \mathcal{S}_{d}^{r}(\triangle)}\left\{\frac{\|g\|_{\mathcal{E}}}{\|g\|_{N}},\|g\|_{N}\neq 0\right\} = O(|\triangle|^{-2}). \label{EQ:Snbar}
	\end{equation}
\end{lemma}

\noindent\begin{proof}
	By Markov's inequality, for any $g\in \mathcal{S}_{d}^{r}(\triangle)$, $\|g\|_{\infty}\leq C|\triangle| ^{-1}\|g\|$, $\|g\|_{\mathcal{E}}\leq C|\triangle|^{-2}\|g\|$.
	Equation (\ref{EQ:normratio}) implies that
	$
	\|g\|_{N}/\|g\| \geq \left[1-O_{P}\left\{N^{-1/2}|\triangle|^{-1}\right\}\right]^{1/2}.
	$
	Thus, one has
	\begin{align*}
		S_{N} &\leq C|\triangle| ^{-1}\left[1-O_{P}\left\{N^{-1/2}|\triangle|^{-1}\right\} \right]
		^{-1/2}=O_{P}\left( |\triangle| ^{-1}\right) , \\
		\overline{S}_{N} &\leq C|\triangle|
		^{-2}\left[1-O_{P}\left\{N^{-1/2}|\triangle|^{-1}\right\}
		\right]^{-1/2}=O_{P}\left( |\triangle| ^{-2}\right).
	\end{align*}
	Lemma \ref{LEM:Normratios} is established.
\end{proof}

\vskip .05in \noindent \textbf{B.2. Convergence of Penalized Spline Estimators}

Let $\left\{ \widetilde{B}_{m}(\bs{z}), m \in \widetilde{\mathcal{M}}\right\}$ be a set of transformed Bernstein basis polynomials and $\widetilde{\mathbf{B}}(\bs{z})=\mathbf{Q}_2^{\top}\mathbf{B}(\bs{z})$,  then, for $\mathbf{U}=\mathbf{B}\mathbf{Q}_{2}$ defined in Section 2.1, $
\mathbf{U}^{\top}\mathbf{U}=\sum_{j=1}^{N} \widetilde{\mathbf{B}}(\bs{z}_{j})\widetilde{\mathbf{B}}^{\top}(\bs{z}_{j})$ and $
\mathbf{U}^{\top}\mathbf{Y}=\sum_{j=1}^{N} \widetilde{\mathbf{B}}(\bs{z}_{j})Y_{ij}$.

Denote by
\vspace{-0.1in}\begin{equation}
	\bs{\Gamma}_{N,\rho}=\frac{1}{N}\sum_{j=1}^{N}
	\{\widetilde{\mathbf{B}}(\bs{z}_{j})\widetilde{\mathbf{B}}^{\top}(\bs{z}_{j})\}
	+\frac{\rho_{n}}{nN}  \mathbf{Q}_2^{\top} [\langle B_{m},B_{m^{\prime}}\rangle_{\mathcal{E}}]_{m,m^{\prime}\in \mathcal{M}}\mathbf{Q}_2
	\label{DEF:Gamma_rho}
\end{equation}
a symmetric positive definite matrix.

The following lemma shows that the maximum and minimum eigenvalue of $\bs{\Gamma}_{N,\rho}$ are bounded by certain orders.
\begin{lemma}
	\label{LEM:Gamma_rho}
	Under Assumption (A4), if $N^{1/2}|\triangle|\rightarrow \infty$ as $N\rightarrow \infty$, then there exist constants $0 < c_{\rho} < C_{\rho} < \infty$, such that with probability approaching 1 as $N\rightarrow \infty$, $n\rightarrow \infty$,
	\[
	c_{\rho}|\triangle|^{2} \leq \lambda_{\min}(\bs{\Gamma}_{N,\rho}) \leq
	\lambda_{\max}(\bs{\Gamma}_{N,\rho}) \leq  C_{\rho}\left(|\triangle|^{2}+\frac{\rho_{n}}{nN|\triangle|^{2}}\right).
	\]
	Specifically, when $\rho_{n}=0$, one has
	$
	c_{0}|\triangle|^{2} \leq \lambda_{\min}(\bs{\Gamma}_{N,0}) \leq
	\lambda_{\max}(\bs{\Gamma}_{N,0}) \leq  C_{0}|\triangle|^{2}.
	$
\end{lemma}

\noindent\begin{proof}
	For any vector $\bs{\theta}$ with the same dimension as that of $\widetilde{\mathbf{B}}(\bs{z})$, there exists $h \in \mathcal{S}_{d}^{r}(\triangle)$ such that $h(\bs{z}) = \widetilde{\mathbf{B}}^{\top}(\bs{z}) \bs{\theta}=\mathbf{B}^{\top}(\bs{z}) \bs{\gamma}$, where $\bs{\gamma}=\mathbf{Q}_2\bs{\theta} $ and
	\begin{equation*}
		\bs{\theta}^{\top}\bs{\Gamma}_{N,\rho}\bs{\theta} =\frac{1}{N}\bs{\gamma}^{\top}\sum_{j=1}^{N} \{\mathbf{B}(\bs{z}_{j})\mathbf{B}^{\top}(\bs{z}_{j})\}\bs{\gamma}+ \frac{\rho_n}{nN} \bs{\gamma}^{\top}[\langle B_{m},B_{m^{\prime}}\rangle_{\mathcal{E}}]_{m,m^{\prime}\in \mathcal{M}} \bs{\gamma}=\|h\|_{N}^2
		+\frac{\rho_n}{nN} \| h\|_{\mathcal{E}}^2.
	\end{equation*}
	By (\ref{EQ:normratio}) and Lemma \ref{LEM:normequity}, one has $\Big| \|h\|^2_{N}/\|h\|^2 -1\Big| \leq \omega_{N}$ and
	\[
	c (1-\omega_{N}) |\triangle|^{2}\Vert \bs{\gamma}\Vert^{2}   \leq (1-\omega_{N}) \Vert h\Vert^{2}
	\leq \Vert h \Vert_{N}^{2} = (1+\omega_{N}) \Vert h \Vert^{2}\leq C (1+\omega_{N})
	|\triangle|^{2}\Vert \bs{\gamma}\Vert^{2}.
	\]
	Thus, $\lambda_{\min}(\mathbf{\Gamma}_{N,\rho}) \geq c_{\rho} |\triangle|^2$ for some positive constant $c_{\rho}$.
	
	On the other hand, similar as in the supplement of \cite{Lai:Wang:13}, using the Markov's inequality and Lemma \ref{LEM:normequity}, one has
	$
	\left\|h\right\|_{\mathcal{E}}^2\leq C|\triangle|^{-4}\left\Vert h\right\Vert^{2}\leq C|\triangle|^{-2}\Vert\bs{\gamma}\Vert ^{2}.
	$
	Thus, the largest eigenvalue of the matrix $\bs{\Gamma}_{N,\rho}$ in (\ref{DEF:Gamma_rho}) satisfies that
	\[
	\lambda_{\max}(\bs{\Gamma}_{N,\rho}) \leq C\left\{(1+\omega_{N})
	|\triangle|^{2}+\frac{\rho_{n}}{nN}\frac{1}{|\triangle|^{2}}\right\} \leq C_{\rho} \left(|\triangle|^{2}+\frac{\rho_{n}}{nN|\triangle|^{2}}\right),
	\]
	for some positive constant $C_{\rho}$.
\end{proof}

Using $\bs{\Gamma}_{N,\rho}$ defined in (\ref{DEF:Gamma_rho}), the solution of the penalized regression problem (\ref{EQ:PLS}) is given by
\[
\widehat{\bs{\theta}} = \bs{\Gamma}_{N,\rho}^{-1} \frac{1}{nN}\sum_{i=1}^{n}\sum_{j=1}^{N} \widetilde{\mathbf{B}}(\bs{z}_{j}) Y_{ij}.
\]
Next we define
\begin{align}
	\nonumber
	\widehat{\bs{\theta}}_{\mu}&=\bs{\Gamma}_{N,\rho}^{-1}\frac{1}{N}\sum_{j=1}^{N} \widetilde{\mathbf{B}}(\bs{z}_{j}) \mu(\bs{z}_j), ~~
	\widehat{\bs{\theta}}_{\eta}=\bs{\Gamma}_{N,\rho}^{-1}
	\frac{1}{nN}\sum_{i=1}^{n}\sum_{j=1}^{N}  \widetilde{\mathbf{B}}(\bs{z}_{j})\sum_{k=1}^{\infty}\xi_{ik}\phi_{k}(\bs{z}_j), \\
	\widehat{\bs{\theta}}_{\varepsilon}&=\bs{\Gamma}_{N,\rho}^{-1}
	\frac{1}{nN}\sum_{i=1}^{n}\sum_{j=1}^{N} \widetilde{\mathbf{B}}(\bs{z}_{j})\sigma(\bs{z}_j)\varepsilon_{ij}.
	\label{DEF:theta-eta-eps}
\end{align}
Note that, the BPS estimator $\widehat{\mu}$ in Section 2.1 can be written as $\widehat{\mu}(\bs{z})=\widehat{\mu}^{o}(\bs{z})+\widehat{\eta}(\bs{z})
+\widehat{\varepsilon}(\bs{z})$, where
\begin{equation}
	\widehat{\mu}^{o}(\bs{z})=\widetilde{\mathbf{B}}(\bs{z})^{\top}\widehat{\bs{\theta}}_{\mu},~~
	\widehat{\eta}(\bs{z})=\widetilde{\mathbf{B}}(\bs{z})^{\top}\widehat{\bs{\theta}}_{\eta},~~
	\widehat{\varepsilon}(\bs{z})=\widetilde{\mathbf{B}}(\bs{z})^{\top}\widehat{\bs{\theta}}_{\varepsilon},
	\label{DEF:pls_estimator}
\end{equation}
Therefore,
\begin{equation}
	\widehat{\mu}(\bs{z})-\mu(\bs{z})=\widehat{\mu}^{o}(\bs{z})-\mu(\bs{z})
	+\widehat{\eta}(\bs{z})+\widehat{\varepsilon}(\bs{z}).
	\label{EQ:pbeta_decompose}
\end{equation}

\begin{lemma}
	\label{LEM:thetatilde-eta}
	Suppose Assumptions (A2)--(A4) hold and  $N^{1/2}|\triangle|\rightarrow \infty$ as $N\rightarrow \infty$, then $\|\widehat{\bs{\theta}}_{\eta}\|^2=O_P(n^{-1}|\triangle|^{-2})$ and $\|\widehat{\bs{\theta}}_{\varepsilon}\|^2=O_P(n^{-1}N^{-1}|\triangle|^{-4})$.
\end{lemma}

\noindent\begin{proof}
	Note that
	\vspace{-0.1in}\[
	\widehat{\bs{\theta}}_{\eta}=\bs{\Gamma}_{n,\rho}^{-1}
	\frac{1}{nN}\sum_{i=1}^{n}\sum_{j=1}^{N} \widetilde{\mathbf{B}}(\bs{z}_{j})\sum_{k=1}^{\infty}\xi_{ik}\phi_{k}(\bs{z}_j).
	\]
	By Lemma \ref{LEM:Gamma_rho}, one has
	\[
	\|\widehat{\bs{\theta}}_{\eta}\|^2
	\asymp \frac{1}{n^2N^2|\triangle|^4}  \sum_{i,i^{\prime}=1}^{n}\sum_{j,j^{\prime}=1}^{N} \widetilde{\mathbf{B}}(\bs{z}_j)^{\top}
	\widetilde{\mathbf{B}}(\bs{z}_{j^{\prime}}) \sum_{k, k^{\prime}=1}^{\infty}\xi_{ik}\phi_{k}(\bs{z}_j)
	\xi_{i^{\prime}k^{\prime}}\phi_{k^{\prime}}(\bs{z}_{j^{\prime}}).
	\]
	Note that by Assumption (A2), for any $i\neq i^{\prime}$, $j$, $j^{\prime}$, one has
	\[
	E\Big\{\widetilde{\mathbf{B}}(\bs{z}_j)^{\top}
	\widetilde{\mathbf{B}}(\bs{z}_{j^{\prime}})\!\! \sum_{k, k^{\prime}=1}^{\infty}\xi_{ik}\phi_{k}(\bs{z}_j)
	\xi_{i^{\prime}k^{\prime}}\phi_{k^{\prime}}(\bs{z}_{j^{\prime}})\Big\}
	=\!\!\!\sum_{m\in \widetilde{\mathcal{M}}}\!\widetilde{B}_{m}(\bs{z}_j)\widetilde{B}_{m}(\bs{z}_{j^{\prime}}) \!\!\sum_{k,k^{\prime}} E\xi_{ik}\xi_{i^{\prime}k^{\prime}}\phi_{k}(\bs{z}_{j})\phi_{k^{\prime}}(\bs{z}_{j^{\prime}})=0.
	\]
	Next, for any $i$, because $ \widetilde{\mathbf{B}}(\bs{z}_j)^{\top} \widetilde{\mathbf{B}}(\bs{z}_{j^{\prime}}) = \mathbf{B}(\bs{z}_j)^{\top}\mathbf{Q}_2 \mathbf{Q}_2^{\top}\mathbf{B}(\bs{z}_{j^{\prime}}) $ and the eigenvalues of $\mathbf{Q}_2\mathbf{Q}_2^{\top}$ are either 0 or 1,
	\begin{align*}
		\frac{1}{N^2}&\sum_{j=1}^N\sum_{j^{\prime}=1}^{N}
		E\left\{\widetilde{\mathbf{B}}(\bs{z}_j)^{\top}\widetilde{\mathbf{B}}(\bs{z}_{j^{\prime}})\sum_{k, k^{\prime}=1}^{\infty}\xi_{ik}\phi_{k}(\bs{z}_j)
		\xi_{ik^{\prime}}\phi_{k^{\prime}}(\bs{z}_{j^{\prime}})\right\}\\
		&\leq \frac{1}{N^2}\sum_{j=1}^N\sum_{j^{\prime}=1}^{N} E\left\{ \mathbf{B}(\bs{z}_j)^{\top}\mathbf{B}(\bs{z}_{j^{\prime}})\sum_{k=1}^{\infty}\xi_{ik}^{2}
		\phi_{k}(\bs{z}_j)\phi_{k}(\bs{z}_{j^{\prime}})\right\}\\
		&= \sum_{m\in \mathcal{M}}
		\frac{1}{N^2}\sum_{j=1}^N\sum_{j^{\prime}=1}^{N} B_{m}(\bs{z}_j)B_{m}(\bs{z}_{j^{\prime}})
		G_{\eta}(\bs{z}_j,\bs{z}_{j^{\prime}}).
	\end{align*}
	Assumption (A4) and Lemma \ref{LEM:integration} imply that
	\begin{align*}
		\frac{1}{N^2}\sum_{j\neq j^{\prime}} B_{m}(\bs{z}_j)B_{m}(\bs{z}_{j^{\prime}})
		G_{\eta}(\bs{z}_j,\bs{z}_{j^{\prime}})=&\int_{\Omega^{2}}G_{\eta}(\bs{z},\bs{z}^{\prime})B_m(\bs{z})B_m(\bs{z}^{\prime})
		d\bs{z}d\bs{z}^{\prime}\\
		& \times \{1+O(N^{-1/2}|\triangle|^{3})\}=O(|\triangle|^{4}).
	\end{align*}
	Thus,
	\[
	\frac{1}{N^2}\sum_{j=1}^N\sum_{j^{\prime}=1}^{N} E\left\{\widetilde{\mathbf{B}}(\bs{z}_j)^{\top}\widetilde{\mathbf{B}}(\bs{z}_{j^{\prime}})\sum_{k, k^{\prime}=1}^{\infty}\xi_{ik}\phi_{k}(\bs{z}_j)
	\xi_{ik^{\prime}}\phi_{k^{\prime}}(\bs{z}_{j^{\prime}})\right\} \leq C|\triangle|^{2}.
	\]
	Therefore, $E\|\widehat{\bs{\theta}}_{\eta}\|^2\leq C  (n^{-1}|\triangle|^{-2})$.
	
	Similarly, by the definition of $\widehat{\bs{\theta}}_{\varepsilon}$ in (\ref{DEF:theta-eta-eps}) and Lemma \ref{LEM:Gamma_rho}, one has
	\[
	\|\widehat{\bs{\theta}}_{\varepsilon}\|^2  \asymp \frac{1}{n^2N^2|\triangle|^4}  \sum_{i,i^{\prime}=1}^{n}\sum_{j,j^{\prime}=1}^{N} \widetilde{\mathbf{B}}(\bs{z}_j)^{\top}
	\widetilde{\mathbf{B}}(\bs{z}_{j^{\prime}}) \sigma(\bs{z}_j) \sigma(\bs{z}_{j^{\prime}})\varepsilon_{ij} \varepsilon_{i^{\prime}j}.
	\]
	Note that for any $i\neq i^{\prime}$, $j$, $j^{\prime}$, $E(\varepsilon_{ij}\varepsilon_{i^{\prime}j^{\prime}})=0$ and for any $i,j\neq j^{\prime}$, $E(\varepsilon_{ij} \varepsilon_{ij^{\prime}})=0$. Because the eigenvalues of $\mathbf{Q}_2\mathbf{Q}_2^{\top}$ are either 0 or 1, by Assumption (A2) and Lemma \ref{LEM:integration}, for any $i$,
	\begin{align*}
		&E\Big\{\frac{1}{N}\sum_{j,j^{\prime}=1}^{N} \widetilde{\mathbf{B}}(\bs{z}_j)^{\top}
		\widetilde{\mathbf{B}}(\bs{z}_{j^{\prime}}) \sigma(\bs{z}_j) \sigma(\bs{z}_{j^{\prime}})\varepsilon_{ij} \varepsilon_{i^{\prime}j}\Big\}
		=\frac{1}{N}\sum_{j=1}^{N}\mathbf{B}(\bs{z}_j)^{\top}\mathbf{Q}_2\mathbf{Q}_2^{\top} \mathbf{B}(\bs{z}_{j})\sigma^{2}(\bs{z}_j)\\
		&\leq C \sum_{m\in \mathcal{M}}\frac{1}{N}\sum_{j=1}^{N}B_{m}^{2}(\bs{z}_j) \sigma^{2}(\bs{z}_j)
		\leq C \sum_{m\in \mathcal{M}} \int_{\Omega}\sigma^{2}(\bs{z})B_{m}^{2}(\bs{z})d\bs{z}\{1+O(N^{-1/2}|\triangle|^{-1})\}= O(1).
	\end{align*}
	Therefore,
	\[
	E\|\widehat{\bs{\theta}}_{\varepsilon}\|^2 \asymp \frac{1}{nN|\triangle|^4}\frac{1}{N}\sum_{j=1}^{N} \widetilde{\mathbf{B}}(\bs{z}_j)^{\top}
	\widetilde{\mathbf{B}}(\bs{z}_{j}) \sigma^{2}(\bs{z}_j)
	\leq C (nN)^{-1} |\triangle|^{-4}.
	\]
	The conclusion of the lemma follows.
\end{proof}

Next, the following lemmas give the uniform convergence rate of $\widehat{\mu}(\bs{z})$ to $\mu(\bs{z})$. We start by introducing some notations for the specific situation when there is no penalty in the regression problem, i.e., $\rho_n=0$.
Denote
$
\bs{\Gamma}_{N,0}=\frac{1}{N}\sum_{j=1}^{N}
\widetilde{\mathbf{B}}(\bs{z}_{j})\widetilde{\mathbf{B}}^{\top}(\bs{z}_{j}).
$
Let $\bar{\xi}_{\cdot k} = \frac{1}{n} \sum_{i=1}^{n} \xi_{ik}$, for any $k\geq 1$, and $\bar{\varepsilon}_{\cdot j}= \frac{1}{n} \sum_{i=1}^{n} \varepsilon_{ij}$ for any $j=1,\ldots, N$, and denote
\begin{align*}
	\widetilde{\bs{\theta}}_{\mu}&=\bs{\Gamma}_{N,0}^{-1}
	\frac{1}{N}\sum_{j=1}^{N} \widetilde{\mathbf{B}}(\bs{z}_{j})\mu(\bs{z}_j),\notag\\
	\widetilde{\bs{\theta}}_{\eta}&=\bs{\Gamma}_{N,0}^{-1}\frac{1}{nN}\sum_{i=1}^{n}\sum_{j=1}^{N} \widetilde{\mathbf{B}}(\bs{z}_{j})\eta_{i}(\bs{z}_j)=\bs{\Gamma}_{N,0}^{-1}\frac{1}{N}\sum_{j=1}^{N}\sum_{k=1}^{\kappa}  \widetilde{\mathbf{B}}(\bs{z}_{j})\bar{\xi}_{\cdot k} \phi_k(\bs{z_{j}}),\notag\\
	\widetilde{\bs{\theta}}_{\varepsilon}&=\bs{\Gamma}_{N,0}^{-1}
	\frac{1}{nN}\sum_{i=1}^{n}\sum_{j=1}^{N} \widetilde{\mathbf{B}}(\bs{z}_{j})\sigma(\bs{z}_j)\varepsilon_{ij}=\bs{\Gamma}_{N,0}^{-1}
	\frac{1}{N}\sum_{j=1}^{N} \widetilde{\mathbf{B}}(\bs{z}_{j})\sigma(\bs{z}_j)\bar{\varepsilon}_{\cdot j},
\end{align*}
Then we can have the following decomposition $\widetilde{\mu}(\bs{z})=\widetilde{\mu}^{o}(\bs{z})
+\widetilde{\eta}(\bs{z})+\widetilde{\varepsilon}(\bs{z})$, where
\begin{equation}
	\widetilde{\mu}^{o}(\bs{z})=\widetilde{\mathbf{B}}(\bs{z})^{\top}\widetilde{\bs{\theta}}_{\mu},~
	\widetilde{\eta}(\bs{z})=\widetilde{\mathbf{B}}(\bs{z})^{\top}\widetilde{\bs{\theta}}_{\eta},~
	\widetilde{\varepsilon}(\bs{z})=\widetilde{\mathbf{B}}(\bs{z})^{\top}\widetilde{\bs{\theta}}_{\varepsilon}.
	\label{DEF:beta_decomp}
\end{equation}

\begin{lemma}
	\label{LEM:uniformbiasrate}
	Under Assumptions (A1) and (A4), if $N^{1/2}|\triangle|\rightarrow \infty$ as $N\rightarrow \infty$, the functions $\widehat{\mu}^{o}(\bs{z})$  satisfy $\|\widehat{\mu}^{o}-\mu\|_{\infty}=O_{P}\left\{
	\frac{\rho_{n}}{nN|\triangle|^{3}}|\mu|_{2,\infty}
	+\left(1+\frac{\rho_{n}}{nN|\triangle|^{5}}\right)
	|\triangle|^{d +1}|\mu|_{d+1,\infty}
	\right\}$.
\end{lemma}

\noindent\begin{proof}
	Note that
	$\|\mu-\widehat{\mu}^{o}\|_{\infty}\leq \|\mu-\widetilde{\mu}^{o}\|_{\infty}+\| \widetilde{\mu}^{o}-\widehat{\mu}^{o}\|_{\infty}$,
	where $\widetilde{\mu}^{o}$ is given in (\ref{DEF:beta_decomp}).\\
	According to Proposition 1 in \cite{Lai:Wang:13}, $\|\widetilde{\mu}^{o}-\mu\|_{\infty}\leq C|\triangle|^{d+1}|\mu|_{d+1,\infty}$. Thus we only need to show the order of $\|\widetilde{\mu}^{o}-\widehat{\mu}^{o}\|_{\infty}$.
	
	By the definition of $S_{N}$ in  (\ref{EQ:Sn}), one has
	\vspace{-0.1in}\begin{equation}
		\|\widetilde{\mu}^{o}-\widehat{\mu}^{o}\|_{\infty}\leq
		S_{N}\|\widetilde{\mu}^{o}-\widehat{\mu}^{o}\|_{N}.
		\label{EQ:pls-ps-sup}
	\end{equation}
	Note that the penalized spline $\widehat{\mu}^{o}$ of $\mu$ is
	characterized by the orthogonality relation: 
	$nN\langle\mu-\widehat{\mu}^{o},g\rangle _{N}=\rho_{n} \langle \widehat{\mu}^{o},g\rangle _{\mathcal{E}},\quad \textrm{for all } g\in \mathcal{S}_{d}^{r}(\triangle)$, while $\widetilde{\mu}^{o}$ is characterized by
	$
	\langle\mu -\widetilde{\mu}^{o},g\rangle _{N}=0$, for all $g\in \mathcal{S}_{d}^{r}(\triangle)$. Combining the two orthogonality relations, one has $nN\langle \widetilde{\mu}^{o}-\widehat{\mu}^{o},g\rangle _{N}=\rho_{n} \langle
	\widehat{\mu}^{o},g\rangle _{\mathcal{E}}$, for all
	$g\in \mathcal{S}_{d}^{r}(\triangle)$. Inserting $g=\widetilde{\mu}^{o}-\widehat{\mu}^{o}$ yields
	that
	\begin{eqnarray*}
		nN\|\widetilde{\mu}^{o}-\widehat{\mu}^{o}\|_{N}^{2}=\rho_{n} \langle \widehat{\mu}^{o},\widetilde{\mu}^{o}-\widehat{\mu}^{o}\rangle _{\mathcal{E}} = \rho_{n} \left\{\langle \widehat{\mu}^{o}, \widetilde{\mu}^{o} \rangle_{\mathcal{E}} - \|\widehat{\mu}^{o}\|^{2}_{\mathcal{E}} \right\} \geq 0.
	\end{eqnarray*}
	Thus, by Cauchy-Schwarz inequality, $\|\widehat{\mu}^{o}\|_{\mathcal{E}}^{2}\leq \langle \widehat{\mu}^{o},\widetilde{\mu}^{o}\rangle _{\mathcal{E}}\leq \|\widehat{\mu}^{o}\|_{\mathcal{E}}\|\widetilde{\mu}^{o}\|_{\mathcal{E}}$,
	which implies that $\|\widehat{\mu}^{o}\|_{\mathcal{E}}\leq \|\widetilde{\mu}^{o}\|_{\mathcal{E}}$. Meanwhile, by the definition of $\overline{S}_{N}$,
	\vspace{-0.1in}\[
	nN\|\widetilde{\mu}^{o}-\widehat{\mu}^{o}\|_{N}^{2}\leq
	\rho_n \|\widehat{\mu}^{o}\|_{\mathcal{E}}\|\widetilde{\mu}^{o}
	-\widehat{\mu}^{o}\|_{\mathcal{E}}\leq \rho_{n} \overline{S}_{N} \|\widehat{\mu}^{o}\|_{\mathcal{E}}\| \widetilde{\mu}^{o}-\widehat{\mu}^{o}\|_{N} \leq  \rho_{n} \overline{S}_{N} \|\widetilde{\mu}^{o}\|_{\mathcal{E}}\| \widetilde{\mu}^{o}-\widehat{\mu}^{o}\|_{N}.
	\vspace{-0.1in}\]
	Therefore,
	\vspace{-0.1in}\begin{eqnarray}
		\|\widetilde{\mu}^{o}-\widehat{\mu}^{o}\|_{N}\leq
		\rho_n (nN)^{-1}\overline{S}_{N} \|\widetilde{\mu}^{o}\|_{\mathcal{E}}. \label{EQ:pls-ps2}
	\end{eqnarray}
	Combining (\ref{EQ:pls-ps-sup}) and (\ref{EQ:pls-ps2}) yields that
	$\|\widetilde{\mu}^{o}-\widehat{\mu}^{o}\|_{\infty}\leq S_{N}
	\|\widetilde{\mu}^{o}-\widehat{\mu}^{o}\|_{N}\leq
	\rho_n (nN)^{-1}S_{N} \overline{S}_{N} \|\widetilde{\mu}^{o}\|_{\mathcal{E}}$. 
	By Lemma \ref{LEM:appord}, one has
	\[
	\|\widetilde{\mu}^{o}\|_{\mathcal{E}} =C_{1}
	\{|\mu|_{2,\infty}+\sum_{a_{1}+a_{2}=2 }\|\nabla_{z_{1}}^{a_{1}}\nabla_{z_{2}}^{a_{2}}(
	\mu- \widetilde{\mu}^{o})\|_{\infty}\}
	\leq C_{2}(|\mu|_{2,\infty}+
	|\triangle|^{d -1}|\mu|_{d +1,\infty}).
	\]
	It follows
	$
	\|\widetilde{\mu}^{o}-\widehat{\mu}^{o}\|_{\infty}=\rho_n (nN)^{-1}S_{N}\overline{S}_{N}
	C_{2}(|\mu|_{2,\infty}+|\triangle|^{d -1}|\mu|_{d+1,\infty})
	$.
	By Lemma \ref{LEM:Normratios}, one has $S_{N} =O_{P}\left( |\triangle| ^{-1}\right)$ and
	$\overline{S}_{N}=O_{P}\left( |\triangle| ^{-2}\right)$.
	Thus,
	\[
	\|\widetilde{\mu}^{o}-\widehat{\mu}^{o}\|_{\infty}
	=O_{P}\left\{ \frac{\rho_n}{nN|\triangle| ^{3}}(|\mu|_{2,\infty}
	+|\triangle|^{d -1}|\mu|_{d +1,\infty}) \right\}.
	\]
	Hence,
	$
	\|\widehat{\mu}^{o}-\mu\|_{\infty}\leq C_{1}|\triangle|^{d +1}|\mu|_{d+1,\infty}+O_{P}\left\{ \frac{\rho_n}{nN|\triangle|^{3}}
	\left(|\mu|_{2,\infty}+|\triangle|^{d -1}|\mu|_{d +1,\infty}\right) \right\}
	$.
	Lemma \ref{LEM:uniformbiasrate} is established.
\end{proof}

\begin{lemma}
	\label{LEM:epshatorder}
	Under Assumptions (A2)--(A4), if $N^{1/2}|\triangle|\rightarrow \infty$ as $N\rightarrow \infty$ and $n^{1/(4+\delta_{2})} \ll n^{1/2}N^{-1/2}|\triangle|^{-1}$, then $\left\|\widetilde{\varepsilon} \right\|_{\infty}=O_{P}\{(nN)^{-1/2}(\log n)^{1/2} |\triangle|^{-1}\}$. In addition, if Assumption (A6) holds, then $\|\widetilde{\eta}\|_{\infty}=O_{P}\{n^{-1/2}(\log n)^{1/2}\}$.
\end{lemma}

\noindent\begin{proof}
	Note that $\widetilde{\varepsilon}(\bs{z})=\sum_{m \in \widetilde{\mathcal{M}}}\widetilde{\theta}_{\varepsilon,m}\widetilde{B}_{m}(\bs{z})$ for some coefficients $\widetilde{\theta}_{\varepsilon,m}$, so the order of $\widehat{\varepsilon}(\bs{z})$ is related to that of $\widetilde{\theta}_{\varepsilon,m}$. In fact
	\begin{align*}
		\left\|\widetilde{\varepsilon}\right\|_{\infty} &= \left\| \mathbf{\widetilde{B}}(\bs{z})^{\top}\mathbf{\Gamma}_{N,0}^{-1} \left[\frac{1}{nN}\sum_{i=1}^n\sum_{j=1}^{N} \widetilde{B}_{m}(\bs{z}_{j})\sigma(\bs{z}_j)\varepsilon_{ij}\right] _{m \in \widetilde{\mathcal{M}}}\right\| _{\infty} \\
		& \leq C |\triangle|^{-2} \max_{m \in \widetilde{\mathcal{M}}}\ \left| \frac{1}{nN}\sum_{i=1}^n\sum_{j=1}^{N} \widetilde{B}_{m}(\bs{z}_{j})\sigma(\bs{z}_j)\varepsilon_{ij}\right|,
	\end{align*}
	almost surely, where $\widetilde{\bs{\theta}}_{\varepsilon}=(\widetilde{\theta}_{\varepsilon,m})_{m \in \widetilde{\mathcal{M}}}$ with $\widetilde{\mathcal{M}}$ being an index set of the transformed Bernstein basis polynomials $\widetilde{B}_{m}(\bs{z})$.
	Next we show that with probability $1$
	\begin{equation}
		\max_{m \in \widetilde{\mathcal{M}}}\left| \frac{1}{nN}\sum_{i=1}^n\sum_{j=1}^{N} \widetilde{B}_{m}(\bs{z}_{j})\sigma(\bs{z}_j)\varepsilon_{ij}\right| =O\left\{(\log n)^{1/2}|\triangle|/(nN)^{1/2}\right\}.  \label{EQ:BEsupnorm}
	\end{equation}
	To prove (\ref{EQ:BEsupnorm}), let $\tau_{i}=\tau_{i,m}=\frac{1}{nN}\sum_{j=1}^{N}\widetilde{B}_{m}(\bs{z}_j)\sigma(\bs{z}_j)\varepsilon_{ij}$. We decompose the random variable $\tau_{i}$ into a truncated part and a tail part,
	\begin{align*}
		\tau_{i,1}^{L_{n}}&=\frac{1}{nN}\sum_{j=1}^{N}\widetilde{B}_{m}(\bs{z}_j)\sigma(\bs{z}_j)\varepsilon_{ij}I\left\{ \left|\varepsilon_{ij}\right|
		>L_{n}\right\},~
		\tau_{i,2}^{L_{n}}=\frac{1}{nN}\sum_{j=1}^{N}\widetilde{B}_{m}(\bs{z}_j)\sigma(\bs{z}_j)\varepsilon_{ij}I\left\{ \left| \varepsilon_{ij}\right| \leq L_{n}\right\} -\mu_{i}^{L_{n}}, \\
		\mu_{i}^{L_{n}}&=
		\frac{1}{nN}\sum_{j=1}^{N}\widetilde{B}_{m}(\bs{z}_j)\sigma(\bs{z}_j)E\left[\varepsilon_{ij}I\left\{ \left| \varepsilon_{ij}\right| \leq L_{n}\right\} \right],
	\end{align*}
	where $%
	L_{n}=n^{\alpha }$, and $n^{1/(4+\delta_2)} \ll n^{\alpha} \ll \sqrt{\frac{n}{N\log n}}|\triangle|^{-1} $.
	
	It is straightforward to verify that $\mu_{i}^{L_{n}} =O(n^{-1}L_{n}^{-2}|\triangle|^2)$.
	Next we show that tail part vanishes almost surely. Note that
	\begin{equation}
		\sum_{n=1}^{\infty }P\left\{ \left| \varepsilon_{nj} \right| >L_{n}\right\} \leq
		\sum_{n=1}^{\infty }\frac{E\left| \varepsilon_{nj} \right| ^{4+\delta_2}}{%
			L_{n}^{4+\delta_2}}\leq \upsilon _{\delta }\sum_{n=1}^{\infty }L_{n}^{-(4+\delta_2)}<\infty .  \label{EQ:tail_vanish}
	\end{equation}
	By Borel Cantelli lemma,  one has
	$
	\left| \sum_{i=1}^{n}\tau_{i,1}^{L_{n}} \right| =O_{a.s.}\left( n^{-k}\right),\textrm{\ for\ any\ } k>0.
	$
	Next, note that $E\left(\tau_{i,2}^{L_{n}}\right) =0$, one has
	\begin{align*}
		\mathrm{Var}(\tau_{i,2}^{L_{n}})  &=\frac{1}{n^2N^2}\sum_{j=1}^{N} \widetilde{B}^2_{m}(\bs{z}_j) \sigma^2(\bs{z}_j)
		\left\{E\left( \varepsilon_{ij}^{2}\right) -E[ \varepsilon_{ij}^{2}I\{| \varepsilon_{ij}| >L_{n}\}] -\left(  E\left[\varepsilon_{ij}I\left\{ \left| \varepsilon_{ij}\right| \leq L_{n}\right\} \right] \right) ^{2} \right\} \\
		&\asymp n^{-2}N^{-1}|\triangle|^2.
	\end{align*}
	Using the independence of $\tau_{i,2}^{L_{n}}$, $i=1,\ldots ,n$, one has $\mathrm{Var}\left( \sum_{i=1}^{n}\tau_{i,2}^{L_{n}}\right) \asymp \left( nN\right)^{-1} |\triangle|^2$.
	
	Now Minkowski's inequality implies that
	\begin{align*}
		E&\left|\tau_{i,2}^{L_{n}}\right| ^{k} = E\Big|\frac{1}{nN}\sum_{j=1}^{N}\widetilde{B}_{m}(\bs{z}_j)\sigma(\bs{z}_j)\varepsilon_{ij}I\left\{ \left| \varepsilon_{ij}\right| \leq L_{n}\right\} -\mu _{i}^{L_{n}}\Big|^{k} \\
		&\leq 2^{k-1} \Bigg[ \left\{\frac{1}{nN}\sum_{j=1}^{N}\widetilde{B}_{m}(\bs{z}_j)\sigma(\bs{z}_j) L_{n} \right\}^{k-2}
		\!\!E \left|\frac{1}{nN}\sum_{j=1}^{N}\widetilde{B}_{m}(\bs{z}_j)\sigma(\bs{z}_j)\varepsilon_{ij}I\left\{ \left| \varepsilon_{ij}\right| \leq L_{n}\right\}\right|^2\!\!
		+\!\! \left( \mu _{i}^{L_{n}}\right)^{k}\Bigg]  \\
		&\leq 2\left(\frac{C|\triangle|^2L_{n}}{n}\right)^{k-2}E\left|\tau_{i,2}^{L_{n}}\right| ^{2},\quad
		k\geq 3.
	\end{align*}
	Thus, $E\left| \tau_{i,2}^{L_{n}}\right| ^{k}\leq \left( \frac{C L_{n}|\triangle|^2}{n}\right)^{k-2}k!E|\tau_{i,2}^{L_{n}}|^2<\infty $ with the Cramer constant $c^{*}=Cn^{-1}L_{n}|\triangle|^2$.\\
	By the Bernstein inequality, for any large enough $\delta >0$,
	\begin{align*}
		&\left. P\left( \left|\frac{1}{nN}\sum_{i=1}^n\sum_{j=1}^{N} \widetilde{B}_{m}(\bs{z}_{j})\sigma(\bs{z}_j)\varepsilon_{ij} \right| \ge \delta |\triangle|\sqrt{\frac{\log n}{nN}}\right)
		=P\left( \left| \sum_{i=1}^{n}\tau_{i,2}^{L_{n}}\right| \ge \delta |\triangle| \sqrt{\frac{\log n}{nN}%
		}\right) \right.  \\
		&\leq 2\exp \left\{ \frac{-\delta ^{2}|\triangle|^2\frac{\log n}{nN}}{4\mathrm{Var}\left( \sum_{i=1}^{n}\tau_{i,2}^{L_{n}}\right) +2c^{*}\delta |\triangle| \sqrt{\frac{\log n}{nN}}}\right\}
		=2\exp \left\{ \frac{-\delta ^{2}|\triangle|^2\frac{\log n}{nN}}{\frac{4c}{nN}|\triangle|^2+2CL_{n}{n}^{-1}\delta |\triangle| ^3\sqrt{\frac{\log n}{nN}}}\right\} \\
		&=2\exp \left\{ \frac{-\delta ^{2} \log n}{4c+2CL_{n}\delta |\triangle| \sqrt{\frac{N\log n}{n}}}\right\} \leq
		2n^{-3},
	\end{align*}
	given that $L_{n} = n^{\alpha} = o\left( \sqrt{\frac{n}{N \log n}} |\triangle|^{-1} \right)$.
	Hence
	\[
	\sum_{n=1}^{\infty }P\left( \max_{m \in \mathcal{M}}\left|\frac{1}{nN}\sum_{i=1}^n\sum_{j=1}^{N} \widetilde{B}_{m}(\bs{z}_{j})\sigma(\bs{z}_j)\varepsilon_{ij} \right|
	\geq \delta |\triangle| \sqrt{\frac{\log n}{nN}}\right) \leq c |\triangle|^{-2} \sum_{n=1}^{\infty }n^{-3}<\infty,
	\]
	for such $\delta >0$. Thus, Borel-Cantelli's lemma implies (\ref{EQ:BEsupnorm}).
	
	Similarly, for $\widetilde{\eta}(\bs{z})=\sum_{m \in \widetilde{\mathcal{M}}}\widetilde{\theta}_{\eta,m}\widetilde{B}_{m}(\bs{z})$ one has
	\[
	\left\|\widetilde{\eta}\right\|_{\infty} \leq C|\triangle|^{-2}\max_{m \in \widetilde{\mathcal{M}}}\left| \frac{1}{nN} \sum_{i=1}^{n}\sum_{j=1}^{N}
	\eta_{i}(\bs{z}_j)\widetilde{B}_{m}(\bs{z}_j)\right|,
	\]
	almost surely. Then we can show that with probability $1$
	\begin{equation*}
		\max_{m \in \widetilde{\mathcal{M}}}\left| \frac{1}{nN} \sum_{i=1}^{n}\sum_{j=1}^{N} \eta_{i}(\bs{z}_j)\widetilde{B}_{m}(\bs{z}_j)\right| =O\left\{n^{-1/2}|\triangle|^{2}(\log n)^{1/2}\right\},
	\end{equation*}
	by decomposing mean $0$ random variable $u_{i}=u_{i,m}=\frac{1}{N}\sum_{j=1}^{N}\eta_{i}(\bs{z}_j)\widetilde{B}_{m}(\bs{z}_j)$ into
	\begin{align*}
		u_{i,1}^{T_{n}}&=\sum_{k=1}^{\infty}\left\{\frac{1}{nN}\sum_{j=1}^{N}
		\widetilde{B}_{m}(\bs{z}_j)\phi_{k}(\bs{z}_{j})\right\}\xi_{ik}I\left\{ \left|\xi_{ik}\right|
		>T_{n}\right\},\\
		u_{i,2}^{T_{n}}&=\sum_{k=1}^{\infty}\left\{\frac{1}{nN}\sum_{j=1}^{N}
		\widetilde{B}_{m}(\bs{z}_j)\phi_{k}(\bs{z}_{j})\right\} \xi_{ik}I\left\{ \left|\xi_{ik}\right| \leq T_{n}\right\} -\mu_{i}^{T_{n}}, \\
		\mu_{i}^{T_{n}}&=
		\sum_{k=1}^{\infty}\left\{\frac{1}{nN}\sum_{j=1}^{N}\widetilde{B}_{m}(\bs{z}_j)
		\phi_{k}(\bs{z}_{j})\right\} E\left[\xi_{ik}I\left\{ \left| \xi_{ik}\right| \leq T_{n}\right\} \right],
	\end{align*}
	where
	$T_{n}=n^{\alpha }$ and  $n^{1/(4+\delta_1)} \ll n^{\alpha} \ll (n/\log n)^{1/2}$.
	
	Using Borel Cantelli lemma and similar method in (\ref{EQ:tail_vanish}), we can show that tail part vanishes almost surely, i.e., $\left| \sum_{i=1}^{n}u_{i,1}^{T_{n}} \right| =O_{a.s.}\left( n^{-r}\right),\textrm{\ for\ any\ } r>0.$ As $E u_{i} =0$, then it is straightforward to verify that $\mu_{i}^{T_{n}}=-Eu_{i,1}^{T_{n}}=O(n^{-1}T_{n}^{-2}|\triangle|^2)$.
	
	Next, notice that $Eu_{i,2}^{T_{n}}=0$. Then, one has
	\begin{align*}
		\mathrm{Var}\left(u_{i,2}^{T_{n}}\right)=&\frac{1}{n^{2}N^{2}}\sum_{k=1}^{\infty} \left\{\sum_{j=1}^{N} \sum_{j^{\prime}=1}^{N}\widetilde{B}_{m}(\bs{z}_{j})\widetilde{B}_{m}(\bs{z}_{j^{\prime}}) \phi_{k}(\bs{z}_l) \phi_{k}(\bs{z}_{j^{\prime}})\right\} \\
		& \times\left\{ E(\xi_{ik}^{2}) - E \Big[ \xi_{ik}^2 I\{|\xi_{ik}| > T_{n}\}\Big] - \left(  E\left[\xi_{ik}I\left\{ \left| \xi_{ik}\right| \leq T_{n}\right\} \right] \right) ^{2}  \right\} = O(n^{-2}|\triangle|^4),
	\end{align*}
	which indicates $\mathrm{Var}\left(\sum_{i=1}^{n}u_{i,2}^{T_{n}}\right) = n^{-1}|\triangle|^{4}$.
	
	Similarly, we can show that there exists some constant $C$, such that for any $r\geq 3$, we have $E|u_{i,2}^{T_{n}}|^{r} \leq  \left(C|\triangle|^2 T_{n}/n\right)^{r-2}r! E|u_{i,2}^{T_{n}}|^2$. Using Bernstein inequality, one has
	\[
	P\left\{\left| \sum_{i=1}^{n}u_{i}\right| \geq \delta n^{-1/2}|\triangle|^{2}(\log n)^{1/2}\right\}
	\leq 2\exp \left\{\frac{-\delta^{2}\log n}{4c+2\delta CT_{n} (\log n)^{1/2} n^{-1/2}} \right\} \leq
	2n^{-3}.
	\]
	Hence,
	\[
	\sum_{n=1}^{\infty}P\left\{ \max_{m \in \widetilde{\mathcal{M}}}\left|\sum_{i=1}^{n}u_{i} \right| \geq \delta n^{-1/2}|\triangle|^{2}(\log n)^{1/2}\right\} \leq C |\triangle|^{-2} \sum_{n=1}^{\infty}n^{-3}<\infty
	\]
	for such $\delta >0$. Thus, Borel-Cantelli's lemma implies that $\left\|\widetilde{\eta} \right\|_{\infty}=O_{P}\{n^{-1/2}(\log n)^{1/2}\}$.
\end{proof}

\begin{lemma}
	\label{LEM:maxnorm}
	Under Assumptions (A2)--(A4),  one has
	\[
	\|\widehat{\varepsilon}\|_{\infty} = O_{P}\left\{\frac{(\log n)^{1/2}}{
		\sqrt{nN}|\triangle|}+ \frac{\rho_n}{n^{3/2}N^{3/2} |\triangle|^6} \right\}, ~
	\|\widehat{\eta}\|_{\infty} = O_{P}\left\{\frac{(\log n)^{1/2}}{
		\sqrt{n}}+ \frac{\rho_n}{n^{3/2}N|\triangle|^5} \right\}.
	\]
\end{lemma}
\noindent\begin{proof}
	We only show the infinity norm of $\widehat{\varepsilon}$. The conclusion of
	$\|\widehat{\eta}\|_{\infty}$ follows similarly. Note that the penalized spline $\widehat{\varepsilon}$ of $\varepsilon$ is characterized by the orthogonality relations: $nN\left\langle \varepsilon -\widehat{\varepsilon},g\right\rangle _{N}=\rho_n
	\left\langle \widehat{\varepsilon},g\right\rangle _{\mathcal{E}}$, for all $g\in \mathcal{S}_{d}^{r}(\triangle)$. In particular, $\widetilde{\varepsilon}$ is characterized by $\left\langle \varepsilon -\widetilde{\varepsilon},g\right\rangle _{N}=0$, for all $g\in \mathcal{S}_{d}^{r}(\triangle)$. Inserting $g=\widetilde{\varepsilon}-\widehat{\varepsilon}$ yield that
	$nN\left\| \widetilde{\varepsilon}-\widehat{\varepsilon}\right\| _{N}^{2}=\rho_n
	\left\langle \widehat{\varepsilon},\widetilde{\varepsilon}-\widehat{\varepsilon}\right\rangle _{\mathcal{E}} =
	\rho_n (\langle \widehat{\varepsilon}, \widetilde{\varepsilon}\rangle_{\mathcal{E}} -\|\widehat{\varepsilon}\|_{\mathcal{E}})$.
	
	It follows, by Cauchy-Schwarz inequality, that
	$
	\left\| \widehat{\varepsilon}\right\|
	_{\mathcal{E}}^{2}\leq \left\langle \widehat{\varepsilon},\widetilde{\varepsilon}\right\rangle _{\mathcal{E}}
	\leq \left\|\widehat{\varepsilon}\right\| _{\mathcal{E}}\left\| \widetilde{\varepsilon}
	\right\|_{\mathcal{E}},
	$
	which implies that $\left\| \widehat{\varepsilon}\right\|
	_{\mathcal{E}}\leq \left\| \widetilde{\varepsilon}\right\|
	_{\mathcal{E}}$.
	Thus, by Cauchy-Schwarz inequality and the definition of $\overline{S}_{N}$ in (\ref{EQ:Snbar}), one has
	\[
	nN\left\| \widetilde{\varepsilon}-\widehat{\varepsilon}\right\| _{N}^{2}\leq
	\rho_n \left\| \widehat{\varepsilon}\right\| _{\mathcal{E}}\left\| \widetilde{\varepsilon}-\widehat{\varepsilon}\right\| _{\mathcal{E}}\leq \overline{S}_{N}\rho_n \left\| \widehat{\varepsilon}\right\|
	_{\mathcal{E}}\left\| \widetilde{\varepsilon}-\widehat{\varepsilon}\right\|
	_{N}.
	\]
	Hence, $\left\| \widetilde{\varepsilon}-\widehat{\varepsilon}\right\| _{N}\leq
	(nN)^{-1}\overline{S}_{N} \rho_n \left\| \widetilde{\varepsilon}\right\|
	_{\mathcal{E}}$.
	Using (\ref{EQ:Sn}), we obtain
	$$
	\left\| \widetilde{\varepsilon}-\widehat{\varepsilon}\right\| _{\infty}\leq S_{N}
	\left\| \widetilde{\varepsilon}-\widehat{\varepsilon}\right\| _{N}\leq
	(nN)^{-1}S_{N} \overline{S}_{N} \rho_n \left\| \widetilde{\varepsilon}\right\|
	_{\mathcal{E}}.
	$$
	Finally, we use Markov's inequality to get
	$\left\| \widetilde{\varepsilon}\right\|_{\mathcal{E}} \le  C_1 |\triangle|^{-2} \left\|\widetilde{\varepsilon}
	\right\|$. It therefore follows that
	$$
	\|\widehat{\varepsilon}\|_{\infty} \le \|\widetilde{\varepsilon}\|_{\infty} + \|\widetilde{\varepsilon}- \widehat{\varepsilon}
	\|_{\infty} \le\|\widetilde{\varepsilon}\|_{\infty}+ \frac{\rho_n}{nN} S_{N} \overline{S}_{N} \frac{C_1}{|\triangle|^2}
	\|\widetilde{\varepsilon}\|_{L_2}.
	$$
	According to Lemmas \ref{LEM:Normratios},  \ref{LEM:thetatilde-eta} and  \ref{LEM:epshatorder}, one has $\|\widetilde{\varepsilon}\|_{\infty}=O_{P}\left\{ (nN)^{-1/2}(\log n)^{1/2}|\triangle|^{-1}\right\}$ and $\Vert \widehat{\varepsilon}\Vert_{L_2}^2 \asymp  |\triangle|^2 \Vert \widehat{\bs{\theta}}_{\varepsilon}\Vert^2 =O_P(n^{-1}N^{-1}|\triangle|^{-2})$. Hence, $\|\widehat{\varepsilon}\|_{\infty} = O_{P}\left\{\frac{(\log n)^{1/2}}{
		\sqrt{nN}|\triangle|}+ \frac{\rho_n}{n^{3/2}N^{3/2} |\triangle|^6} \right\}$.
\end{proof}

\noindent\begin{proof}[Proof of Theorem \ref{THM:convergence}]
	Note that
	$\widehat{\mu} - \mu = \widehat{\mu}^{o} - \mu + \widehat{\eta} + \widehat{\varepsilon}$. Therefore, $
	\Vert \widehat{\mu}-\mu\Vert_{L_2} \leq \Vert \widehat{\mu}^{o}-\mu\Vert_{L_2} + \Vert \widehat{\eta}\Vert_{L_2}+ \Vert \widehat{\varepsilon}\Vert_{L_2}.$
	By Lemmas \ref{LEM:normequity} and \ref{LEM:thetatilde-eta}, one has
	\[
	\Vert \widehat{\eta}\Vert_{L_2}^2 \asymp  |\triangle|^2 \Vert \widehat{\bs{\theta}}_{\eta}\Vert^2 =O_P(n^{-1}), ~
	\Vert \widehat{\varepsilon}\Vert_{L_2}^2 \asymp  |\triangle|^2 \Vert \widehat{\bs{\theta}}_{\varepsilon}\Vert^2 =O_P(n^{-1}N^{-1}|\triangle|^{-2}),
	\]
	and the asymptotic order of $\Vert \widehat{\mu}^{o}-\mu\Vert_{L_2}$ is the same as $\Vert \widehat{\mu}^{o}-\mu\Vert_{\infty}$.
	By Lemmas \ref{LEM:uniformbiasrate} and \ref{LEM:maxnorm},
	\begin{align*}
		&\| \widehat{\mu}^{o} - \mu\|_{\infty}
		=O_{P}\left\{\frac{\rho_{n}}{nN|\triangle|^{3}}|\mu|_{2,\infty}
		+\left(1+\frac{\rho_{n}}{nN|\triangle|^{5}}\right)
		|\triangle|^{d +1}|\mu|_{d+1,\infty}
		\right\}, \\
		&\|\widehat{\eta}\|_{\infty} = O_{P}\left\{\frac{(\log n)^{1/2}}{
			\sqrt{n}}+ \frac{\rho_n}{n^{3/2}N|\triangle|^5} \right\}, ~
		\|\widehat{\varepsilon}\|_{\infty} =  O_P\left\{\frac{(\log n)^{1/2}}{
			\sqrt{nN}|\triangle|}+ \frac{\rho_n}{n^{3/2}N^{3/2} |\triangle|^6}\right\}.
	\end{align*}
	Thus, by Assumption (A5),
	$\| \widehat{\mu} - \mu\|_{\infty}= O_{P}\{(n^{-1}\log (n))^{1/2}\}$ and $\Vert \widehat{\mu}-\mu \Vert_{L _2}=O_{P}(n^{-1/2})$.
\end{proof}

\vskip .05in \noindent \textbf{B.3. Simultaneous Confidence Bands}
\vskip .05in \noindent \textbf{B.3.1. Proof of Theorem \ref{THM:band}}
\begin{lemma}[Lemma A.5, \cite{Cao:Yang:Todem:12}]
	\label{LEM: Cao A.5}
	Let $\bar{\xi}_{\cdot k} = \frac{1}{n} \sum_{i=1}^{n} \xi_{ik}$ and $\bar{\varepsilon}_{\cdot j}= \frac{1}{n} \sum_{i=1}^{n} \varepsilon_{ij}$. If Assumption (A2) holds, then there exists some constant $C_{\beta}>0$ such that
	$\max_{1\leq k \leq \kappa} E|\bar{\xi}_{\cdot k} - \bar{Z}_{\cdot k, \xi}| \leq C_{\beta} n^{\beta -1},~\mathrm{and}~\max_{1 \leq j \leq N} |\bar{\varepsilon}_{\cdot j} - \bar{Z}_{\cdot j, \varepsilon}| = O_{a.s.}(n^{\beta -1}), \mathrm{~for~some~\beta\in (0,1/2)}$,
	where  $\{Z_{ik,\xi}\}_{i=1,k=1}^{n,\kappa}$ and $\{Z_{ij,\varepsilon}\}_{i=1,j=1}^{n,N}$ are iid $N(0,1)$ variables and $\bar{Z}_{\cdot k,\xi} = \frac{1}{n} \sum_{i=1}^{n} Z_{ik,\xi},\ \bar{Z}_{\cdot j,\varepsilon} = \frac{1}{n} \sum_{i=1}^{n} Z_{ij,\varepsilon},\ 1 \leq j \leq N,\ 1 \leq k \leq \kappa$.
\end{lemma}

\begin{lemma}
	\label{LEM:etatilde-etabar-penalized}
	Let
	$\bar{\eta}(\bs{z}) = \frac{1}{n} \sum_{i=1}^{n} \eta_i(\bs{z}) = \sum_{k=1}^{\kappa} \bar{\xi}_{\cdot k}\phi_{k}(\bs{z})$, then under Assumptions (A2)--(A6),  for $\widehat{\eta}(\bs{z})$ defined in (\ref{DEF:pls_estimator}), one has $n^{1/2}  \|\widehat{\eta}- \bar{\eta}\|_{\infty}  = o_{P}(1)$. In addition, as $N\rightarrow \infty$, $n\rightarrow \infty$,
	\begin{align*}
		&P\left\{\sup_{\bs{z} \in \Omega} n^{1/2} G_{\eta}(\bs{z}, \bs{z})^{-1/2} |\bar{\eta}(\bs{z})| \leq q_{1-\alpha}\right\} \rightarrow 1-\alpha,  \\
		&P\left\{ n^{1/2} G_{\eta}(\bs{z}, \bs{z})^{-1/2} |\bar{\eta}(\bs{z})| \leq z_{1-\alpha/2}\right\} \rightarrow 1-\alpha,\  \ \mathrm{~for~any~} \bs{z} \in \Omega.
	\end{align*}
\end{lemma}

\noindent\begin{proof}
	Denote $\widetilde{\zeta}_k(\bs{z}) = \bar{Z}_{\cdot k, \xi} \phi_k(\bs{z}),\ k=1, \ldots, \kappa$, and define
	\[
	\widetilde{\zeta}(\bs{z}) = n^{1/2} \left\{ \sum_{k=1}^{\kappa}\phi_k^2(\bs{z}) \right\}^{-1/2} \sum_{k=1}^{\kappa} \widetilde{\zeta}_k(\bs{z}) = n^{1/2} G_{\eta}(\bs{z}, \bs{z})^{-1/2} \sum_{k=1}^{\kappa} \widetilde{\zeta}_k(\bs{z}),
	\]
	then $\{\widetilde{\zeta}(\bs{z}), \bs{z} \in \Omega\}$ is a Gaussian random field with mean 0, variance 1 and covariance function $\text{Cov}\left\{\widetilde{\zeta}(\bs{z}), \widetilde{\zeta}(\bs{z}^{\prime})\right\}
	=  G_{\eta}(\bs{z}, \bs{z})^{-1/2} G_{\eta}(\bs{z}, \bs{z}^{\prime}) G_{\eta}(\bs{z}^{\prime}, \bs{z}^{\prime})^{-1/2}$. 
	Therefore, $\widetilde{\zeta}(\bs{z})$ has the same distribution as $\zeta(\bs{z}),\bs{z} \in \Omega$.
	
	Next, let
	$
	\widehat{\phi}_k(\bs{z}) = \widetilde{\mathbf{B}}(\bs{z})^{\top} \bs{\Gamma}_{N,\rho}^{-1}\frac{1}{N}\sum_{j=1}^N \widetilde{\mathbf{B}}(\bs{z}_j) \phi_k(\bs{z}_j).
	$
	Similar to the proof for Lemma \ref{LEM:epshatorder}, by Lemma \ref{LEM:Gamma_rho}, 
	$\| \widehat{\phi}_{k}\|_{\infty}\leq C |\triangle|^{-2} \Big\| \frac{1}{N} \sum_{j=1}^{N}\widetilde{\mathbf{B}}(\bs{z}_j)\phi_k(\bs{z}_j)\Big\|_{\infty}
	\leq C_1  \|\phi_k\|_{\infty}$. According to Lemma \ref{LEM:uniformbiasrate},
	\[
	\|\widehat{\phi}_k - \phi_k  \|_{\infty} = O_{P}
	\left\{\frac{\rho_{n}}{nN|\triangle|^{3}}|\phi_k|_{2,\infty}
	+\left(1+\frac{\rho_{n}}{nN|\triangle|^{5}}\right)|\triangle|^{s+1}|\phi_k|_{s+1,\infty}\right\}.
	\]
	Therefore, by Assumptions (A4)--(A6), one has
	\begin{align*}
		E&\Big\{n^{1/2} \sup_{\bs{z} \in \Omega} G_{\eta}(\bs{z}, \bs{z})^{-1/2} |\widehat{\eta}(\bs{z}) - \bar{\eta}(\bs{z})| \Big\} =  E\left[n^{1/2} \sup_{\bs{z} \in \Omega} G_{\eta}(\bs{z}, \bs{z})^{-1/2} \left| \sum_{k=1}^{\kappa} \bar{\xi}_{\cdot k} \{ \phi_k(\bs{z}) - \widehat{\phi}_k(\bs{z})\}\right| \right]  \\
		\leq& n^{1/2} G_{\eta}(\bs{z}, \bs{z})^{-1/2} \left( \sum_{k=1}^{\kappa_n} \| \widehat{\phi}_k - \phi_k  \|_{\infty} E|\bar{\xi}_{\cdot k}| +C \sum_{k=\kappa_n+1}^{\kappa} \| \phi_k  \|_{\infty} E|\bar{\xi}_{\cdot k}|\right).
	\end{align*}
	Thus,
	\begin{align*}
		E&\Big\{n^{1/2} \sup_{\bs{z} \in \Omega} G_{\eta}(\bs{z}, \bs{z})^{-1/2} |\widehat{\eta}(\bs{z}) - \bar{\eta}(\bs{z})| \Big\}\\
		\leq& n^{1/2} G_{\eta}(\bs{z}, \bs{z})^{-1/2}  C_1\Bigg[ \sum_{k=1}^{\kappa_n}
		\left\{\frac{\rho_{n}}{nN|\triangle|^{3}}|\phi_k|_{2,\infty}
		+\left(1+\frac{\rho_{n}}{nN|\triangle|^{5}}\right)|\triangle|^{s+1}|\phi_k|_{s+1,\infty}\right\} \\
		&\times E|\bar{\xi}_{\cdot k}| +\sum_{k=\kappa_n+1}^{\kappa} E|\bar{\xi}_{\cdot k}|  \| \phi_k \|_{\infty} \Bigg]\\
		\leq& C_2 \left\{\frac{\rho_{n}}{nN|\triangle|^{3}}\sum_{k=1}^{\kappa_n} |\phi_k|_{2,\infty}
		+\left(1+\frac{\rho_{n}}{nN|\triangle|^{5}}\right)
		\sum_{k=1}^{\kappa_n} |\triangle|^{s+1}|\phi_k|_{s+1,\infty}
		+ \sum_{k=\kappa_n+1}^{\kappa} \|\phi_k\|_{\infty} \right\} = o(1).
	\end{align*}
	
	Hence, $n^{1/2} \sup_{\bs{z} \in \Omega} G_{\eta}(\bs{z}, \bs{z})^{-1/2} |\widehat{\eta}(\bs{z}) - \bar{\eta}(\bs{z})| =o_P(1)$.
	Under Assumption (A3), it follows that $ \|\widehat{\eta} - \bar{\eta}\|_{\infty}  = o_{P}(n^{-1/2} )$.
	
	By Lemma \ref{LEM: Cao A.5}, for some $\beta\in (0,1/2)$,
	\begin{align*}
		E\Big\{ \sup_{\bs{z} \in \Omega} \left| \widetilde{\zeta}(\bs{z}) - n^{1/2} G_{\eta}(\bs{z}, \bs{z})^{-1/2} \bar{\eta}(\bs{z})\right| \Big\}
		&=  E\left\{ n^{1/2}\sup_{\bs{z} \in \Omega}G_{\eta}(\bs{z}, \bs{z})^{-1/2}  \left| \sum_{k=1}^{\kappa}
		(\bar{Z}_{\cdot k,\xi} - \bar{\xi}_{\cdot k}) \phi_k(\bs{z}) \right| \right\} \\
		&\leq  C n^{1/2} \sum_{k=1}^{\kappa} \|\phi_k\|_{\infty} E|\bar{Z}_{\cdot k,\xi} - \bar{\xi}_{\cdot k}|
		\leq  \widetilde{C} n^{\beta - 1/2} \sum_{k=1}^{\kappa} \|\phi_k\|_{\infty}.
	\end{align*}
	
	Thus, by Assumption (A6),
	$\sup_{\bs{z} \in \Omega} \left| \widetilde{\zeta}(\bs{z}) - n^{1/2}G_{\eta}(\bs{z}, \bs{z})^{-1/2} \bar{\eta}(\bs{z})\right| = o_{P}(1)$. Finally, note that
	$P\left\{ \sup_{\bs{z} \in \Omega} \left| \widetilde{\zeta}(\bs{z})\right| \leq q_{1-\alpha}\right\} = 1-\alpha.$
	The lemma is proved.
\end{proof}

\noindent\begin{proof}[Proof of Theorem \ref{THM:band}]
	Note that ``oracle" estimator $\bar{\mu}(\bs{z}) = \mu(\bs{z}) +\bar{\eta}(\bs{z})$ implies that
	$\widehat{\mu} - \bar{\mu} = \widehat{\mu}^{o} - \mu + \widehat{\eta} - \bar{\eta} + \widehat{\varepsilon}.$
	By Lemmas \ref{LEM:uniformbiasrate}, \ref{LEM:maxnorm}, Assumptions (A5) and (A6),
	\begin{align*}
		&\| \widehat{\mu}^{o} - \mu\|_{\infty}
		=O_{P}\left\{\frac{\rho_{n}}{nN|\triangle|^{3}}|\mu|_{2,\infty}
		+\left(1+\frac{\rho_{n}}{nN|\triangle|^{5}}\right)
		|\triangle|^{d +1}|\mu|_{d+1,\infty}
		\right\}= o_{P}(n^{-1/2}), \\
		&\|\widehat{\varepsilon}\|_{\infty} =  O_P\left\{\frac{(\log n)^{1/2}}{
			\sqrt{nN}|\triangle|}+ \frac{\rho_n}{n^{3/2}N^{3/2} |\triangle|^6}\right\}= o_{P}(n^{-1/2}).
	\end{align*}
	Thus, according to Lemma \ref{LEM:etatilde-etabar-penalized}, the theorem is established.
\end{proof}

\vskip .05in \noindent \textbf{B.3.2. Proof of Theorem \ref{THM:muA-muB}}
According to (\ref{EQ:pbeta_decompose}), for $H=1,2$, we can decompose the unpenalized spline estimator $\widehat{\mu}_{H}(\cdot)$ as $\widehat{\mu}_{H}(\bs{z}) = \widehat{\mu}^{o}_{H}(\bs{z}) + \widehat{\eta}_{H}(\bs{z}) + \widehat{\varepsilon}_{H}(\bs{z})$. Therefore, asymptotic error $(\widehat{\mu}_{1} -\widehat{\mu}_{2}) -( \mu_{1}  - \mu_{2} ) $ can be
decomposed into three components:
$(\widehat{\mu}^{o}_{1} - \widehat{\mu}^{o}_{2} - \mu_{1} +\mu_{2}) + (\widehat{\eta}_{1} -\widehat{\eta}_{2})
+(\widehat{\varepsilon}_{1} -\widehat{\varepsilon}_{2})$. Similar as the proof of Theorem 2, the first and third components of the decomposition can be proved to have $\sqrt{n}$ asymptotic efficiency. Here we focus on the second component.

By Lemma \ref{LEM: Cao A.5}, one can find i.i.d
$Z_{Hik,\xi}\thicksim N\left( 0,1\right)$, $i=1,\ldots,n_{H}$ such that
$\max_{1\leq k \leq \kappa_{H}} E|\bar{\xi}_{H \cdot k} - \bar{Z}_{H\cdot k, \xi}| \leq C_0 n^{\beta -1}$ and $\bar{Z}_{H\cdot k ,\xi}=n_{H}^{-1}\sum_{i=1}^{n_{H}}Z_{Hik,\xi}$. Likewise,
for the white noise sequence $\{\varepsilon_{Hij}, i\geq 1\}$, one can also find iid
$Z_{Hik,\varepsilon}\thicksim N\left( 0,1\right)$, $i=1,\ldots,n_{H}$ such that $%
\max_{1 \leq j \leq N} |\bar{\varepsilon}_{H\cdot j} - \bar{Z}_{H\cdot j, \varepsilon}| = O_{a.s.}(n^{\beta -1})$, where $\beta \in (0, 1/2)$.
Let  $V\left(\bs{z}, \bs{z}^{\prime} \right)=
G_{\eta, 1}(\bs{z}, \bs{z}^{\prime}) +\tau G_{\eta, 2}(\bs{z}, \bs{z}^{\prime})$, where $\tau=\lim_{n_1\rightarrow \infty}n_{1}/n_{2}$, and define
\[
\widetilde{W}(\bs{z}) =n_{1}^{1/2} V\left(\bs{z}, \bs{z}\right)^{-1/2}\left\{  \sum_{k=1}^{\kappa_{1} }\bar{Z}
_{1\cdot k, \xi }\phi_{1k}(\bs{z}) - \sum_{k=1}^{\kappa_{2} }\bar{Z}_{2\cdot k, \xi}\phi_{2k}(\bs{z}) \right\}.
\]
Then, for any $\bs{z}\in \Omega$, $\widetilde{W}(\bs{z})$ is Gaussian with mean 0 and variance 1, and the covariance
\[
E\left\{\widetilde{W}(\bs{z}) \widetilde{W}(\bs{z}^{\prime})\right\}
=V\left(\bs{z}, \bs{z}\right)^{-1/2} V\left(\bs{z}, \bs{z}^{\prime }\right)
V\left( \bs{z}^{\prime}, \bs{z}^{\prime }\right)^{-1/2}.
\]
That is, the distribution of $\widetilde{W}(\bs{z})$, $\bs{z}\in \Omega$ and the distribution of $W (\bs{z})$, $\bs{z}\in \Omega$ are identical.
Similarly, for $H=1,2$, let $\widehat{\phi}_{Hk}(\bs{z}) = \widetilde{\mathbf{B}}_{H}(\bs{z})^{\top} \bs{\Gamma}_{H,N,\rho}^{-1}\frac{1}{N}\sum_{j=1}^N \widetilde{\mathbf{B}}_{H}(\bs{z}_j) \phi_{Hk}(\bs{z}_j)$.
Note that
\[
\bar{\eta}_{H}(\bs{z}) = \frac{1}{n} \sum_{i=1}^{n} \eta_{H i}(\bs{z}) = \sum_{k=1}^{\kappa_{H}} \bar{\xi}_{H \cdot k}\phi_{Hk}(\bs{z}), ~\widehat{\eta}_{H}(\bs{z}) =  \sum_{k=1}^{\kappa_{H}} \bar{\xi}_{H \cdot k} \widehat{\phi}_{Hk}(\bs{z}).
\]
And we have shown in Lemma \ref{LEM:etatilde-etabar-penalized} that $n_{H}^{1/2}  \|\widehat{\eta}_{H}- \bar{\eta}_{H}\|_{\infty}  = o_{P}(1)$.
\begin{lemma}
	\label{LEM:two_group_penalized}
	If Assumptions (A1)--(A6), are modified for each group accordingly,
	then one has as $N \rightarrow \infty$, $n_{1} \rightarrow \infty$,
	\[
	P\left\{  \sup_{\bs{z} \in \Omega} n_{1}^{1/2} V(\bs{z},\bs{z})^{-1/2} | \bar{\eta}_{1}(\bs{z}) - \bar{\eta}_{2}(\bs{z}) | \leq q_{12,\alpha} \right\} \rightarrow 1-\alpha.
	\]
\end{lemma}

\noindent\begin{proof}
	Note that, by similar discussion in Lemma \ref{LEM:etatilde-etabar-penalized}, 
	\begin{align*}
		E\Bigg[&\sup_{\bs{z}\in \Omega}\Big\vert \widetilde{W}(\bs{z})
		-n_{1}^{1/2}V\left(\bs{z}, \bs{z}\right)^{-1/2}\left\{ \bar{\eta}_{1}(\bs{z}) -
		\bar{\eta}_{2}(\bs{z}) \right\} \Big\vert \Bigg] \\
		= &n_{1}^{1/2} E\left[ \sup_{\bs{z}\in \Omega}V\left(\bs{z}, \bs{z}\right)^{-1/2} \left|
		\sum_{k=1}^{\kappa_{1}}
		(\overline{Z}_{1\cdot k,\xi} - \overline{\xi}_{1\cdot k}) \phi_{1k}(\bs{z}) - \sum_{k=1}^{\kappa_{2}}
		(\overline{Z}_{2\cdot k,\xi} - \overline{\xi}_{2\cdot k}) \phi_{2k}(\bs{z}) \right|\right]
		= o(1).
	\end{align*}
	Therefore, one has $\sup_{\bs{z}\in \Omega}\left\vert \widetilde{W}(\bs{z})
	-n_{1}^{1/2}V\left(\bs{z}, \bs{z}\right)^{-1/2}\left\{ \bar{\eta}_{1}(\bs{z}) -
	\bar{\eta}_{2}(\bs{z}) \right\} \right\vert = o_P(1)$.
	Observe that $P\left\{\sup_{\bs{z} \in \Omega} | \widetilde{W}(\bs{z})| \leq q_{12,\alpha} \right\} = 1-\alpha$, for any $\alpha \in (0,1)$, as $N\rightarrow \infty$, $n_{1}\rightarrow \infty$
	\begin{align*}
		&P\left\{\sup_{\bs{z} \in \Omega}n_{1}^{1/2} V(\bs{z}, \bs{z})^{-1/2} \left| \bar{\eta}_{1}(\bs{z}) - \bar{\eta}_{2}(\bs{z})\right| \leq q_{12,\alpha} \right\} \rightarrow 1-\alpha, \\
		&P\left\{ n_{1}^{1/2} V(\bs{z}, \bs{z})^{-1/2} \left| \bar{\eta}_{1}(\bs{z}) - \bar{\eta}_{2}(\bs{z})\right| \leq z_{1-\alpha/2} \right\} \rightarrow 1-\alpha, \ \ \text{for~all~} \bs{z} \in \Omega.
	\end{align*}
	The conclusion of the lemma is proved.
\end{proof}

\vskip .05in \noindent \textbf{B.4. Convergence of the Covariance Estimator}
Without loss of generality, we prove Theorem \ref{THM:Ghat-G} based on the unpenalized bivariate spline estimator. Using similar arguments in Section B.2, we can easily extend this proof to the penalized case.
Based on the estimated residuals $\widehat{R}_{ij}=Y_{ij}-\widehat{\mu}(\bs{z}_j)$, $i=1,\ldots,n$, $j=1,\ldots, N$, denote
$\widehat{\bs{\beta}}_{i}=\argmin_{\bs{\beta}}\sum_{j=1}^{N}
\left\{\widehat{R}_{ij}-\mathbf{B}^{*}(\bs{z}_{j})^{\top}\mathbf{Q}^{*}_{2}\bs{\beta} \right\}^{2}$, where $\mathbf{B}^{*}(\bs{z})$ is the set of bivariate spline basis functions used to estimate $\eta_i(\bs{z})$, and the transpose of $\mathbf{H}^{*}$ admits the following QR decomposition:
$
\mathbf{H}^{*\top}=\mathbf{Q}^{*}\mathbf{R}^{*}=(\mathbf{Q}^{*}_{1}~\mathbf{Q}^{*}_{2})
\binom{\mathbf{R}^{*}_{1}}{\mathbf{R}^{*}_{2}}.
$
Then, the bivariate spline estimator of $\eta_{i}(\bs{z})$  can be written as
\begin{equation}
	\widehat{\eta}_{i}(\bs{z})=\mathbf{B}^{*}(\bs{z})^{\top}\mathbf{Q}^{*}_{2}
	\widehat{\bs{\beta}}_{i}=\widetilde{\mathbf{B}}^{*}(\bs{z})^{\top}
	\widehat{\bs{\beta}}_{i}.
	\label{DEF:eta_i_hat}
\end{equation}
Let
$
\bs{\Gamma^{*}}_{N}=\frac{1}{N}\sum_{j=1}^{N}\widetilde{\mathbf{B}}^{*}(\bs{z}_{j})
\widetilde{\mathbf{B}}^{*}(\bs{z}_{j})^{\top},
$
then one has
\begin{align*}
	\widehat{\bs{\beta}}_{i}&=\bs{\Gamma}^{*-1}_{N}\frac{1}{N}\sum_{j=1}^{N}\widetilde{\mathbf{B}}^{*}(\bs{z}_{j}) \widehat{R}_{ij}
	=\bs{\Gamma}^{*-1}_{N}\frac{1}{N}\sum_{j=1}^{N}\widetilde{\mathbf{B}}^{*}(\bs{z}_{j}) \{Y_{ij}-\widehat{\mu}(\bs{z}_j)\}\notag \\
	&=\bs{\Gamma}^{*-1}_{N}\frac{1}{N}\sum_{j=1}^{N}\widetilde{\mathbf{B}}^{*}(\bs{z}_{j}) \left[
	\{\mu(\bs{z}_j)-\widehat{\mu}(\bs{z}_j)\}
	+\eta_{i}(\bs{z}_j)+\sigma(\bs{z}_j)\varepsilon_{ij}\right].
\end{align*}
Next we define $\widetilde{r}(\bs{z})=\widetilde{\mathbf{B}}^{*}(\bs{z})^{\top}\bs{\Gamma}^{*-1}_{N}\frac{1}{N}\sum_{j=1}^{N} \widetilde{\mathbf{B}}^{*}(\bs{z}_{j})\{\mu(\bs{z}_j)-\widehat{\mu}(\bs{z}_j)\}$, and 
\[
\widetilde{\eta}_{i}(\bs{z})=\widetilde{\mathbf{B}}^{*}(\bs{z})^{\top}
\bs{\Gamma}^{*-1}_{N}\frac{1}{nN}\sum_{i=1}^n\sum_{j=1}^{N} \widetilde{\mathbf{B}}^{*}(\bs{z}_{j})\eta_{i}(\bs{z}_j), ~
\widetilde{\varepsilon}_{i}(\bs{z})=\widetilde{\mathbf{B}}^{*}(\bs{z})^{\top}
\bs{\Gamma}^{*-1}_{N}\frac{1}{nN}\sum_{i=1}^n\sum_{j=1}^{N} \widetilde{\mathbf{B}}^{*}(\bs{z}_{j})\sigma(\bs{z}_j)\varepsilon_{ij}.
\]
Then, the estimation error $d_i(\bs{z})=\widehat{\eta}_{i}(\bs{z})-\eta_{i}(\bs{z})$ in (\ref{DEF:eta_i_hat})
can be decomposed as the following: $d_i(\bs{z})=\widetilde{r}(\bs{z})+\widetilde{\eta}_{i}(\bs{z}) - \eta_{i}(\bs{z}) +\widetilde{\varepsilon}_{i}(\bs{z})$.

For any $\bs{z}$, $\bs{z}^{\prime}\in \Omega$, denote
$
\widetilde{G}_{\eta}(\bs{z},\bs{z}^{\prime})=n^{-1}\sum_{i=1}^{n}\eta_i(\bs{z})\eta_i(\bs{z}^{\prime}).
$
The following lemma shows the uniform convergence of $\widetilde{G}_{\eta}(\bs{z},\bs{z}^{\prime})$ to $G_{\eta}(\bs{z},\bs{z}^{\prime})$ in probability over all $(\bs{z},\bs{z}^{\prime})\in \Omega^2$.

\begin{lemma}[Lemma B.18, \cite{Yu:etal:18}]
	\label{LEM:Gtilde-G}
	Under Assumptions (A1)--(A7), $\sup_{(\bs{z},\bs{z}^{\prime})\in \Omega^2}|\widetilde{G}_{\eta}(\bs{z},\bs{z}^{\prime})
	-G_{\eta}(\bs{z},\bs{z}^{\prime})|=O_{P}\{n^{-1/2}(\log n)^{1/2}\}$.
\end{lemma}

\noindent\begin{proof}[Proof of Theorem \ref{THM:Ghat-G} (i)]
	Note that
	\[
	\sup_{(\bs{z},\bs{z}^{\prime})\in \Omega^2}|\widehat{G}_{\eta}(\bs{z},\bs{z}^{\prime})
	-G_{\eta}(\bs{z},\bs{z}^{\prime})|
	\leq \sup_{(\bs{z},\bs{z}^{\prime})\in \Omega^2}\{|\widehat{G}_{\eta}(\bs{z},\bs{z}^{\prime})
	-\widetilde{G}_{\eta}(\bs{z},\bs{z}^{\prime})|+|\widetilde{G}_{\eta}(\bs{z},\bs{z}^{\prime})
	-G_{\eta}(\bs{z},\bs{z}^{\prime})|\}
	\]
	where $\sup_{(\bs{z},\bs{z}^{\prime})\in \Omega^2}|\widetilde{G}_{\eta}(\bs{z},\bs{z}^{\prime})
	-G_{\eta}(\bs{z},\bs{z}^{\prime})|=o_{P}(1)$ according to Lemma \ref{LEM:Gtilde-G}, and
	\begin{align*}
		\sup_{(\bs{z},\bs{z}^{\prime})\in \Omega^2}|\widehat{G}_{\eta}(\bs{z},\bs{z}^{\prime})
		-\widetilde{G}_{\eta}(\bs{z},\bs{z}^{\prime})|
		\leq & \sup_{(\bs{z},\bs{z}^{\prime})\in \Omega^2}\left\vert n^{-1}\sum_{i=1}^{n}\eta_i(\bs{z})d_i(\bs{z}^{\prime})\right\vert+\sup_{(\bs{z},\bs{z}^{\prime})\in \Omega^2}\left\vert n^{-1}\sum_{i=1}^{n}\eta_i(\bs{z}^{\prime})d_i(\bs{z})\right\vert \notag\\
		&+\sup_{(\bs{z},\bs{z}^{\prime})\in \Omega^2}
		\left\vert n^{-1}\sum_{i=1}^{n}d_i(\bs{z})d_i(\bs{z}^{\prime})\right\vert.
	\end{align*}
	Similar to the proof of \cite{Yu:etal:18}, one can show that
	\[
	\sup_{(\bs{z},\bs{z}^{\prime})\in \Omega^2}
	\left\vert \frac{1}{n}\sum_{i=1}^{n}d_i(\bs{z})d_i(\bs{z}^{\prime})\right\vert=o_{P}(1),
	\sup_{(\bs{z},\bs{z}^{\prime})\in \Omega^2}\left\vert \frac{1}{n}\sum_{i=1}^{n}\eta_i(\bs{z})d_i(\bs{z}^{\prime})
	+n^{-1}\sum_{i=1}^{n}\eta_i(\bs{z}^{\prime})d_i(\bs{z})\right\vert
	=o_{P}(1).
	\]
	The desired result is established.
\end{proof}

\noindent\begin{proof}[Proof of Theorem \ref{THM:Ghat-G} (ii)]
	Denote $\Delta \psi_k(\bs{z}) = \int (\widehat{G} - G)(\bs{z}, \bs{z}^{\prime}) \psi_k(\bs{z}^{\prime}) d\bs{z}^{\prime}$.
	By Theorem \ref{THM:Ghat-G} (i), $\| \widehat{G} - G\|_{\infty} = o_P(1)$. Thus, for any $ k \geq 1$, $\|\Delta \psi_k\|_{\infty} = o_P(1)$.
	According to \cite{Hall:Hosseini-Nasab:2006},
	let $\| \mathit{\Delta}\|_2 = [\iint  (\widehat{G}(\bs{z}, \bs{z}^{\prime})-G(\bs{z}, \bs{z}^{\prime}))^2 d\bs{z}d\bs{z}^{\prime} ]^{1/2}$, then $\widehat{\psi}_k - \psi_k = \sum_{j \neq k} (\lambda_k - \lambda_j)^{-1} \langle \Delta \psi_k, \psi_j \rangle \psi_j + O(\| \mathit{\Delta}\|_2^2)$.
	It follows from Bessel's inequality that $\|\widehat{\psi}_k - \psi_k\|_2 \leq C(\|\Delta\psi_k\|^2_{\infty}+O(\| \mathit{\Delta}\|_2^2))=o_P(1)$.
	By (2.9) in \cite{Hall:Hosseini-Nasab:2006},
	$$\widehat{\lambda}_k - \lambda_k = \iint (\widehat{G} - G)(\bs{z}, \bs{z}^{\prime}) \psi_k(\bs{z}) \psi_k(\bs{z}^{\prime}) d\bs{z}d\bs{z}^{\prime} + O(\|\Delta \psi_k\|_2^2). $$
	Thus, using Theorem \ref{THM:Ghat-G} (i), one has $|\widehat{\lambda}_k - \lambda_k| = o_P(1), \forall k \geq 1$.\\
	Next, note that
	\begin{align*}
		\widehat{\lambda}_k &\widehat{\psi}_k(\bs{z}) - \lambda_k \psi_k(\bs{z})
		= \int \widehat{G}(\bs{z}, \bs{z}^{\prime}) \widehat{\psi}_k(\bs{z}^{\prime})d\bs{z}^{\prime} - \int G(\bs{z}, \bs{z}^{\prime}) \psi_k(\bs{z}^{\prime}) d\bs{z}^{\prime} \\
		=& \int (\widehat{G} - G)(\bs{z}, \bs{z}^{\prime}) (\widehat{\psi}(\bs{z}^{\prime}) - \psi_k(\bs{z}^{\prime})) d\bs{z}^{\prime} + \int(\widehat{G}-G)(\bs{z}, \bs{z}^{\prime}) \psi_k(\bs{z}^{\prime}) d\bs{z}^{\prime}\\
		&+ \int G(\bs{z}, \bs{z}^{\prime})\{\widehat{\psi}_k(\bs{z}^{\prime}) - \psi_k(\bs{z}^{\prime})\}d\bs{z}^{\prime}.
	\end{align*}
	By Cauchy-Schwarz inequality and Theorem \ref{THM:Ghat-G} (i), for all $\bs{z} \in \Omega$,
	\begin{align*}
		&\int G(\bs{z}, \bs{z}^{\prime})\{\widehat{\psi}_k(\bs{z}^{\prime}) - \psi_k(\bs{z}^{\prime})\}d\bs{z}^{\prime}
		\leq \left(\int G^2(\bs{z}, \bs{z}^{\prime}) d\bs{z}^{\prime} \right) ^{1/2} \|\widehat{\psi}_k - \psi_k\|_2 = o_P(1), \\
		&\int (\widehat{G} - G)(\bs{z}, \bs{z}^{\prime}) (\widehat{\psi}(\bs{z}^{\prime}) - \psi_k(\bs{z}^{\prime})) d\bs{z}^{\prime}
		\leq \|\widehat{G} -G\|_{\infty} \|\widehat{\psi}_k - \psi_k\|_2 = o_P(1), \\
		&\int(\widehat{G}-G)(\bs{z}, \bs{z}^{\prime}) \psi_k(\bs{z}^{\prime}) d\bs{z}^{\prime}
		\leq \|\widehat{G} -G\|_{\infty} \|\psi_k\|_2 = o_P(1).
	\end{align*}
	Therefore, $\|\widehat{\lambda}_k \widehat{\psi}_k- \lambda_k \psi_k\|_{\infty}=o_P(1)$, and 
	$\lambda_k \|\widehat{\psi}_k - \psi_k\|_{\infty} \leq \|\widehat{\lambda}_k \widehat{\psi}_k - \lambda_k \psi_k\|_{\infty}
	+ |\widehat{\lambda}_k - \lambda_k| \|\widehat{\psi}_k\|_{\infty} = o_P(1)$. 
	It follows that $\|\widehat{\psi}_k - \psi_k\|_{\infty} = o_P(1)$.
\end{proof}

\medskip

\renewcommand{\thetable}{A\arabic{table}}
\renewcommand{\thefigure}{A\arabic{figure}}
\begin{figure}
	\begin{center}
		\begin{tabular}{cccc}
			\hspace{-0.5in}Triangulation &\hspace{-1.3in}$\widehat{\mu}$ &\hspace{-1.3in}Lower SCC &\hspace{-1.3in}Upper SCC\\
			\includegraphics[scale=0.14]{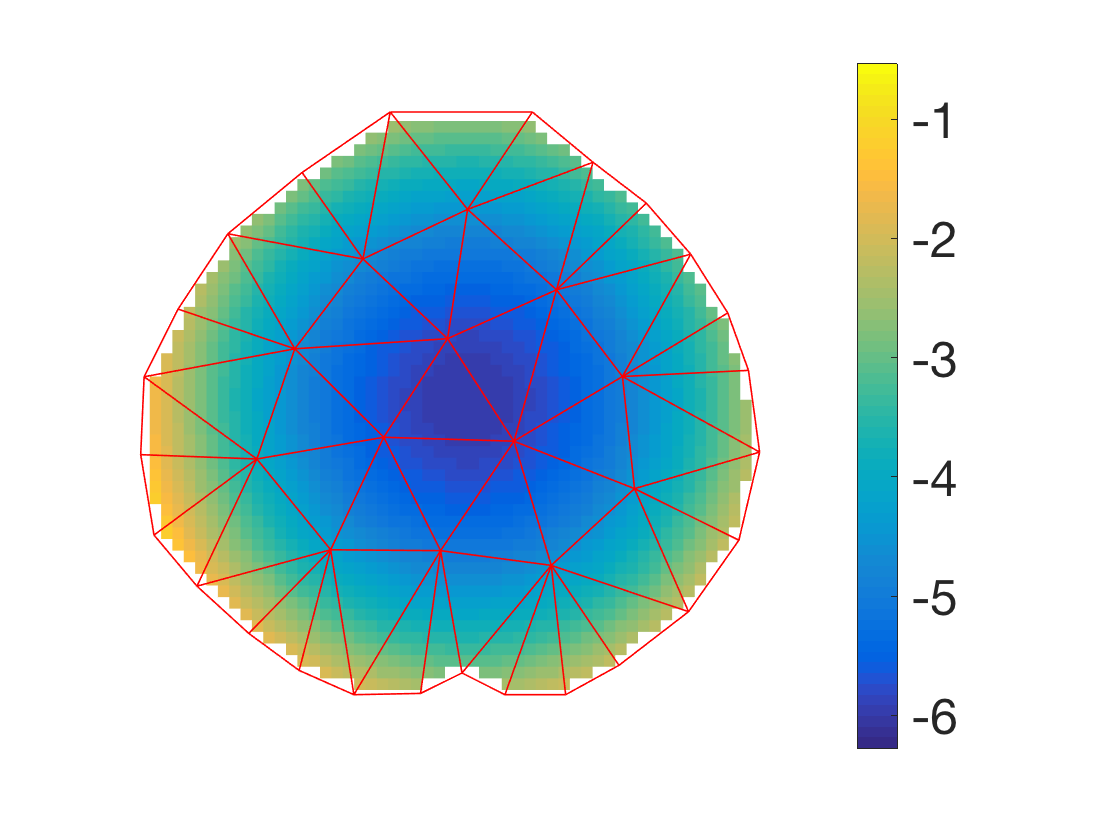} &\hspace{-0.8in}\includegraphics[scale=0.14]{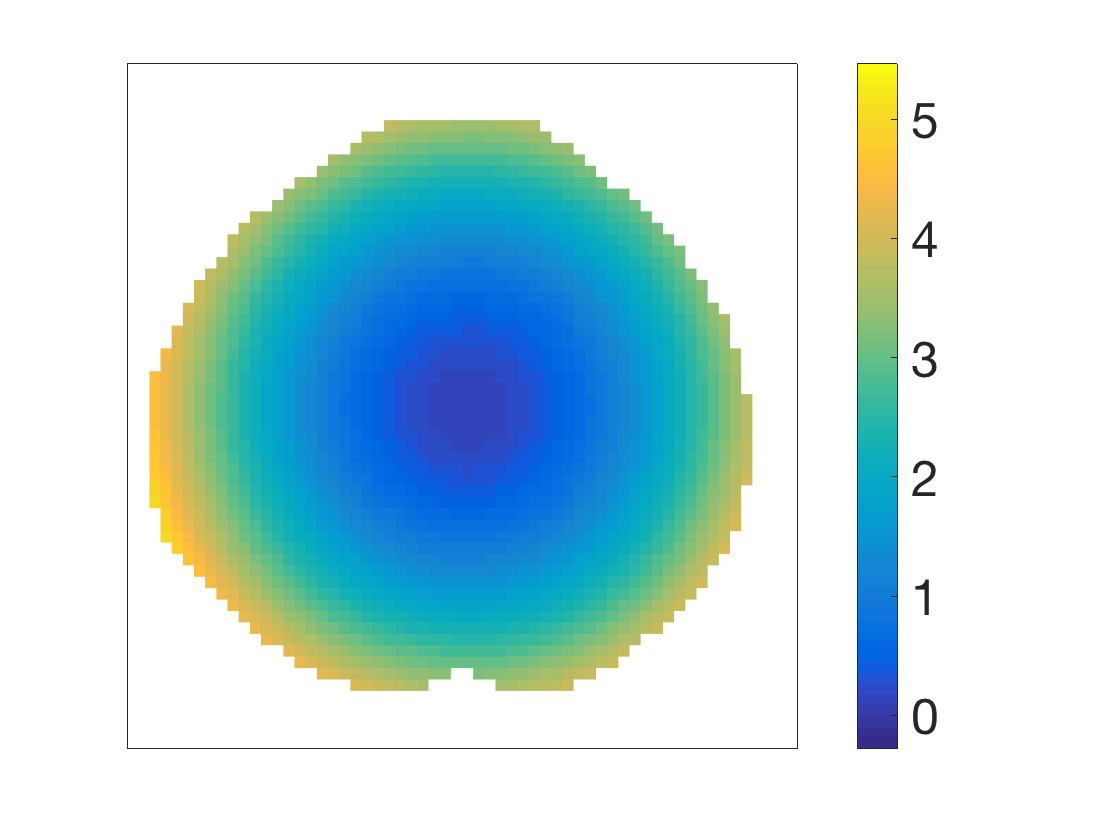} &\hspace{-0.7in}\includegraphics[scale=0.14]{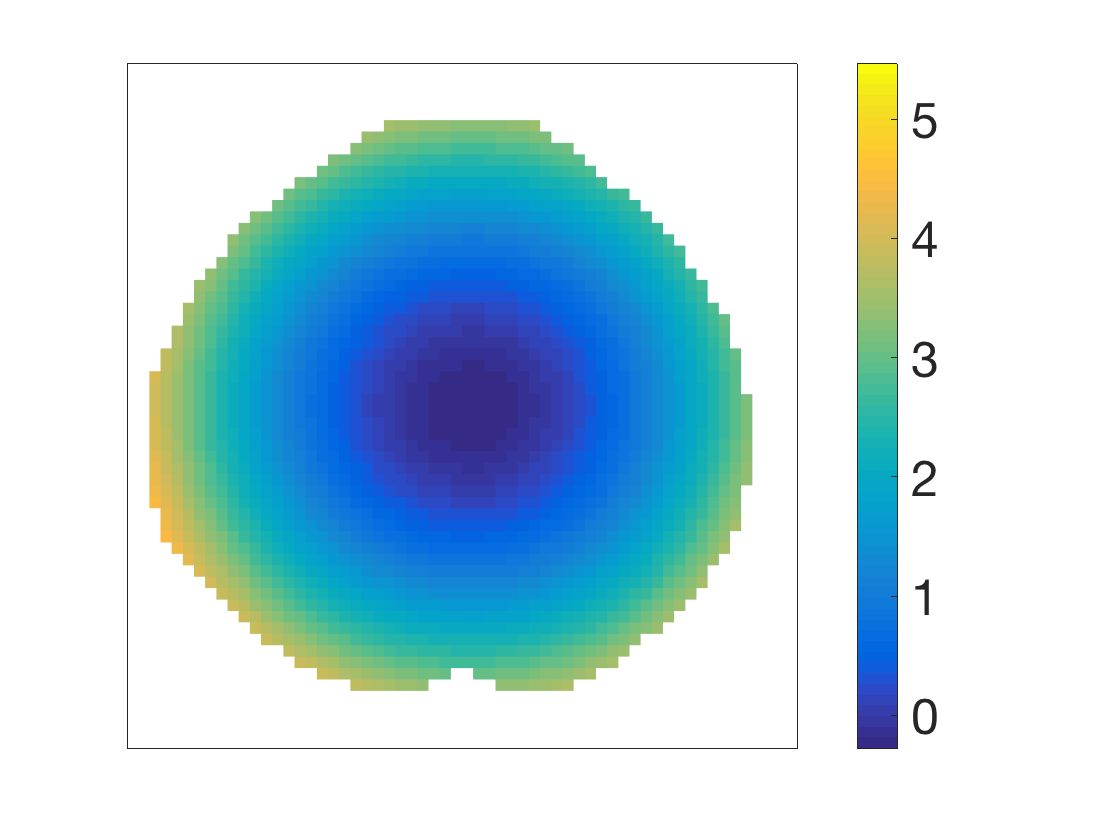} &\hspace{-0.7in}\includegraphics[scale=0.14]{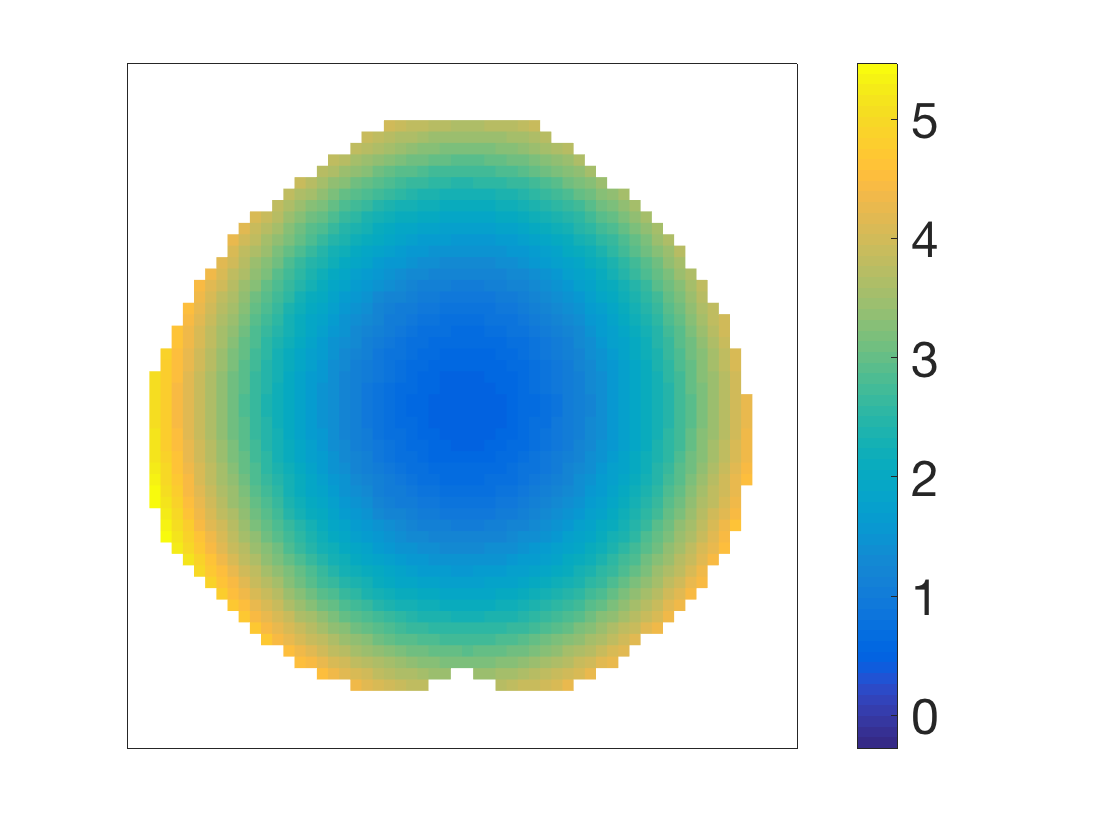}\\ [-10pt]
			\hspace{-0.3in}$\triangle_1$ & & &\\
			
			\includegraphics[scale=0.14]{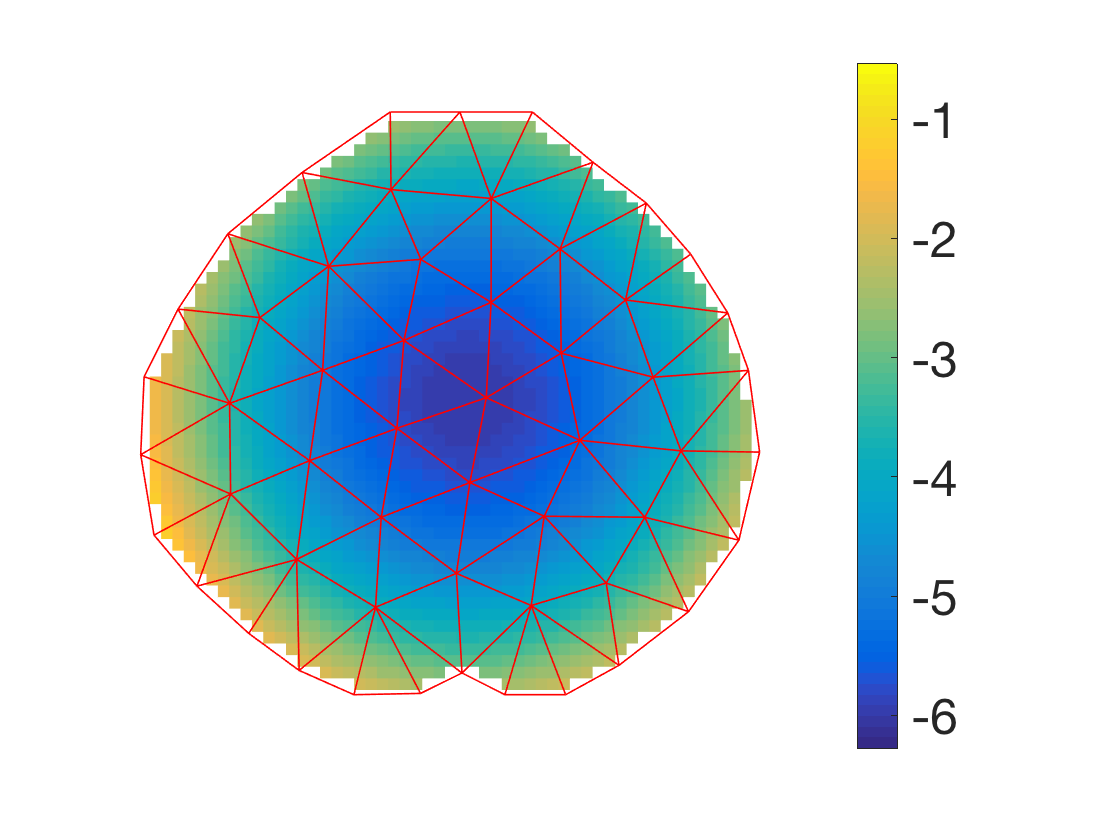} &\hspace{-0.8in}\includegraphics[scale=0.14]{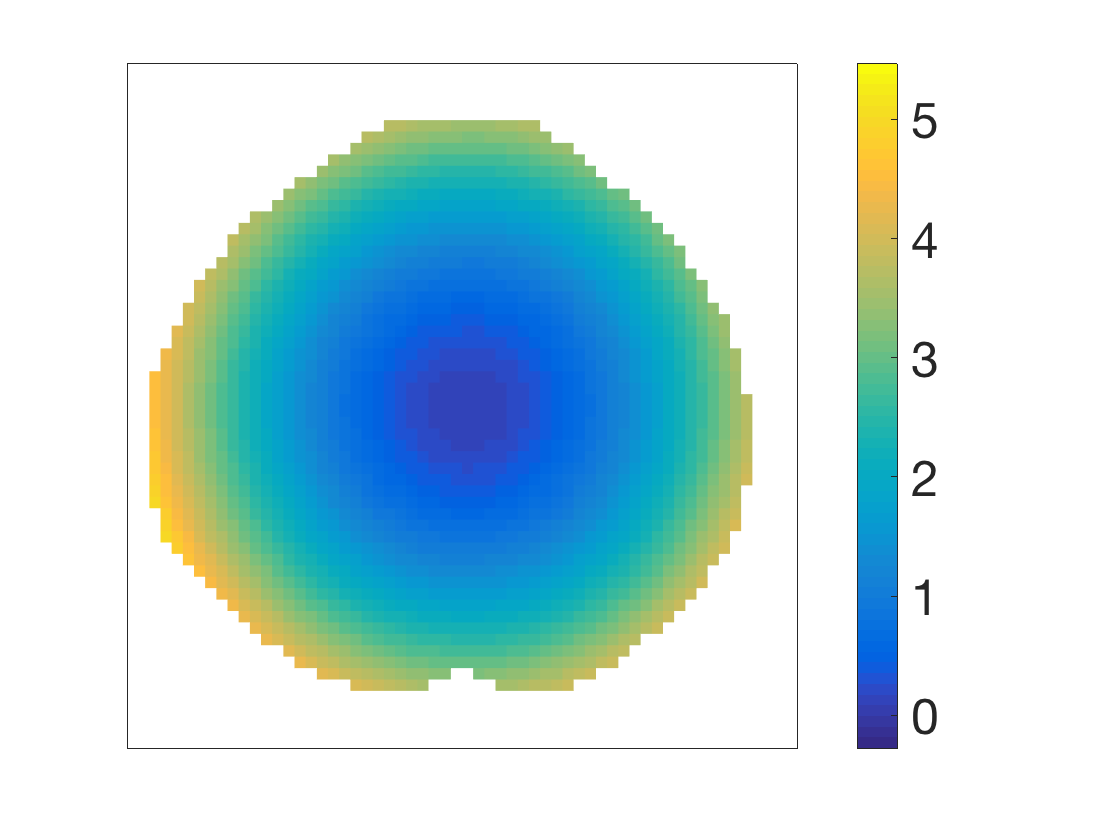} &\hspace{-0.7in}\includegraphics[scale=0.14]{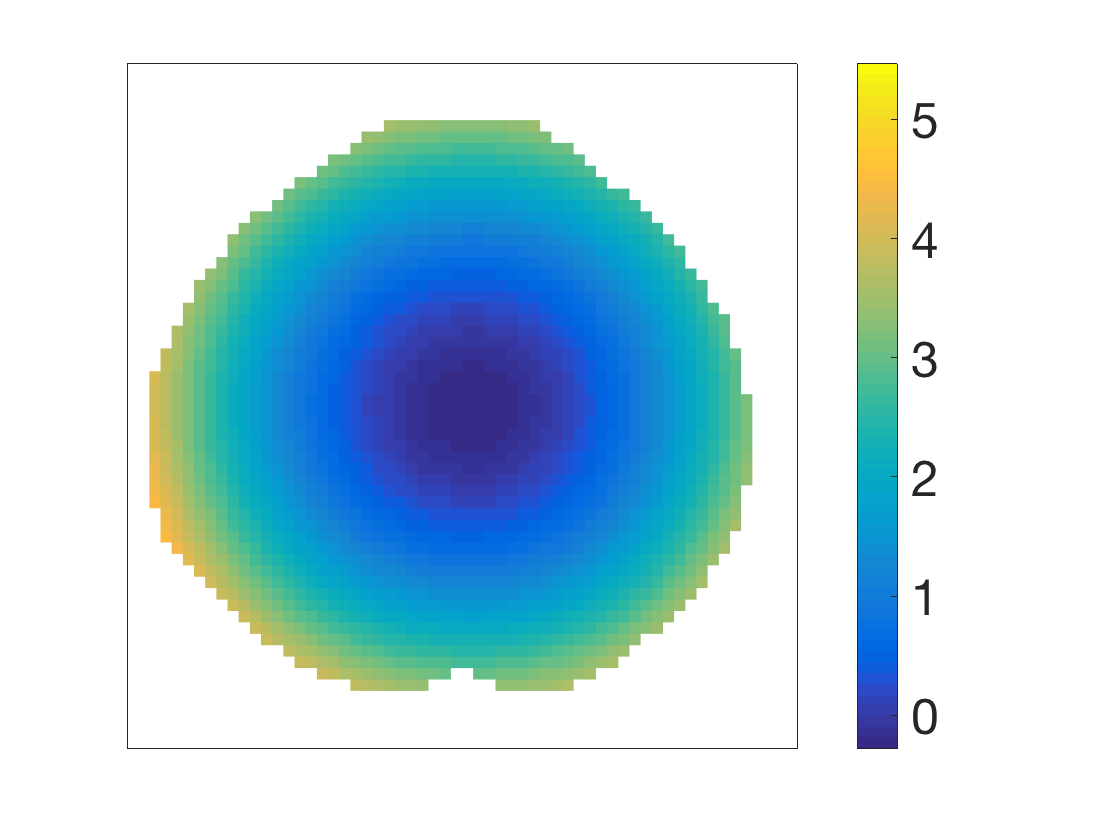} &\hspace{-0.7in}\includegraphics[scale=0.14]{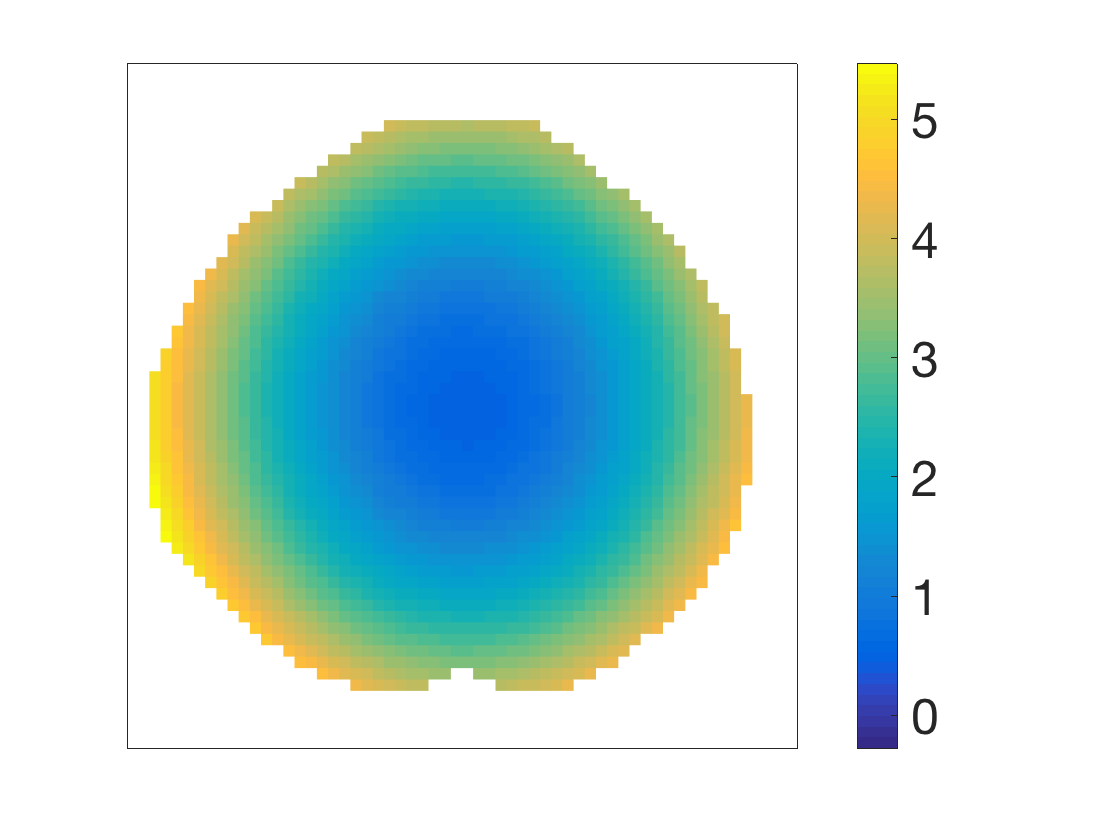}\\ [-10pt]
			\hspace{-0.3in}$\triangle_2$ & & &\\
			
			\includegraphics[scale=0.14]{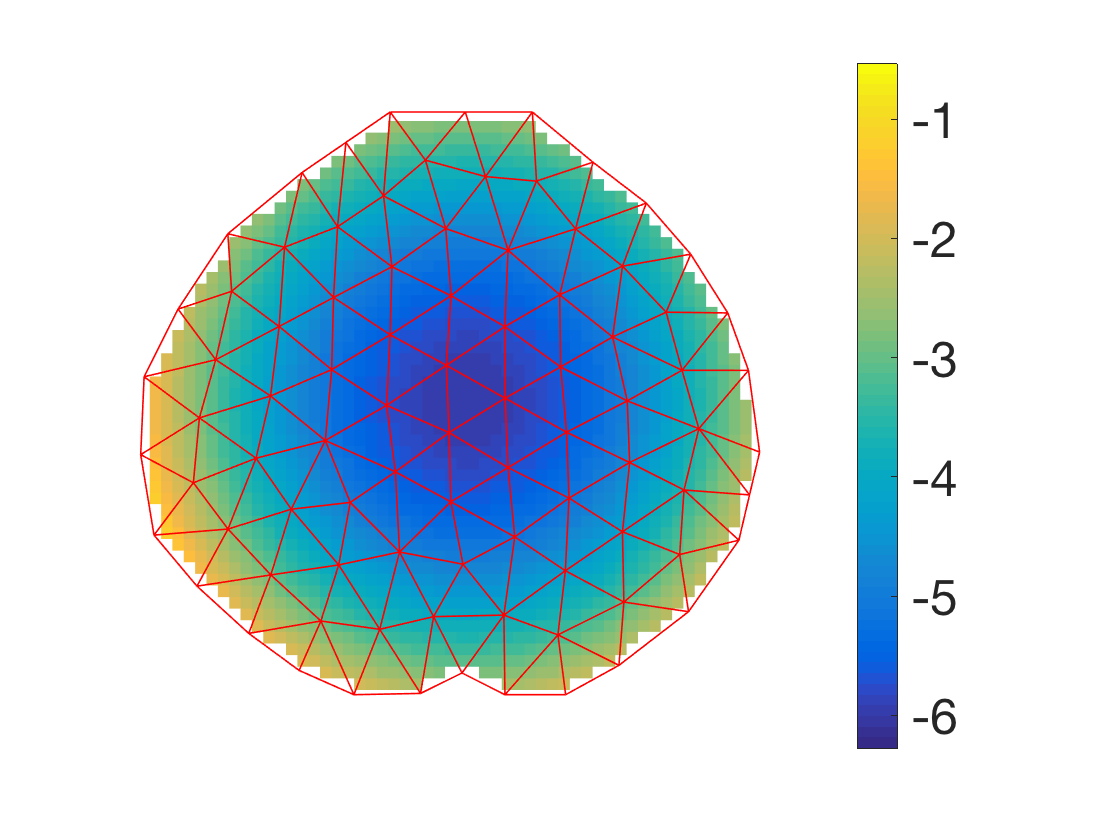} &\hspace{-0.8in}\includegraphics[scale=0.14]{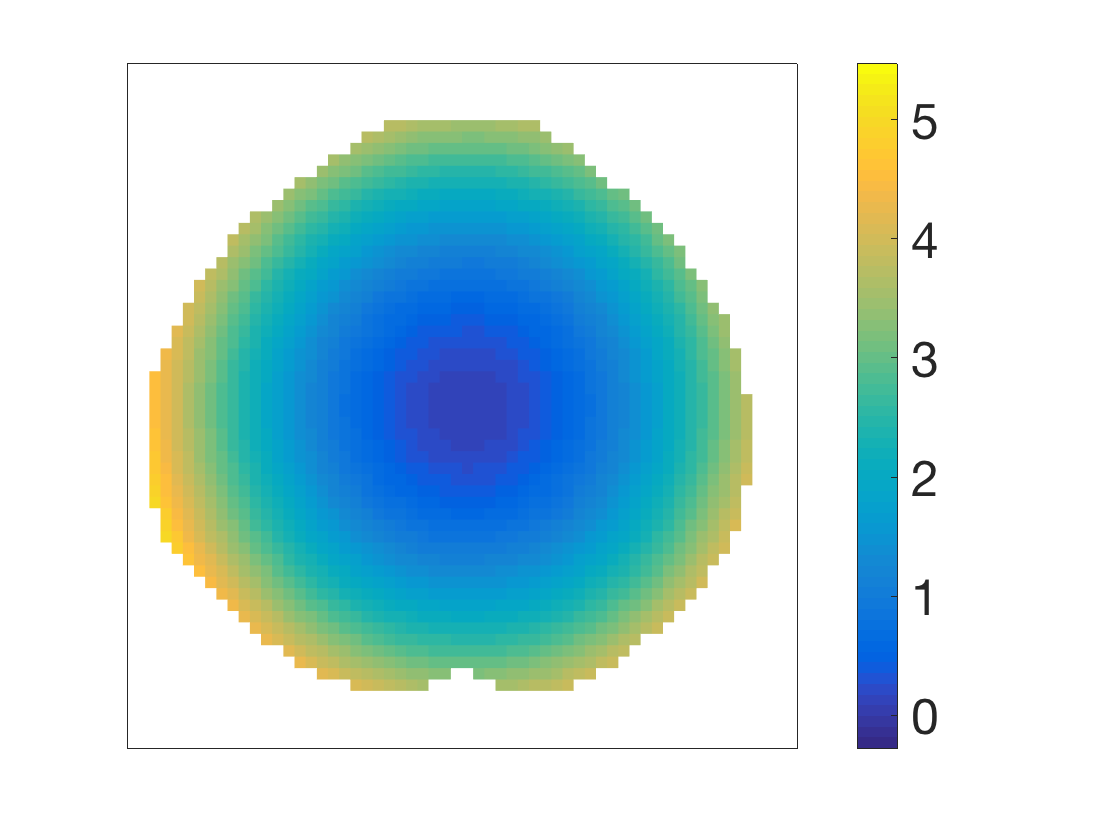} &\hspace{-0.7in}\includegraphics[scale=0.14]{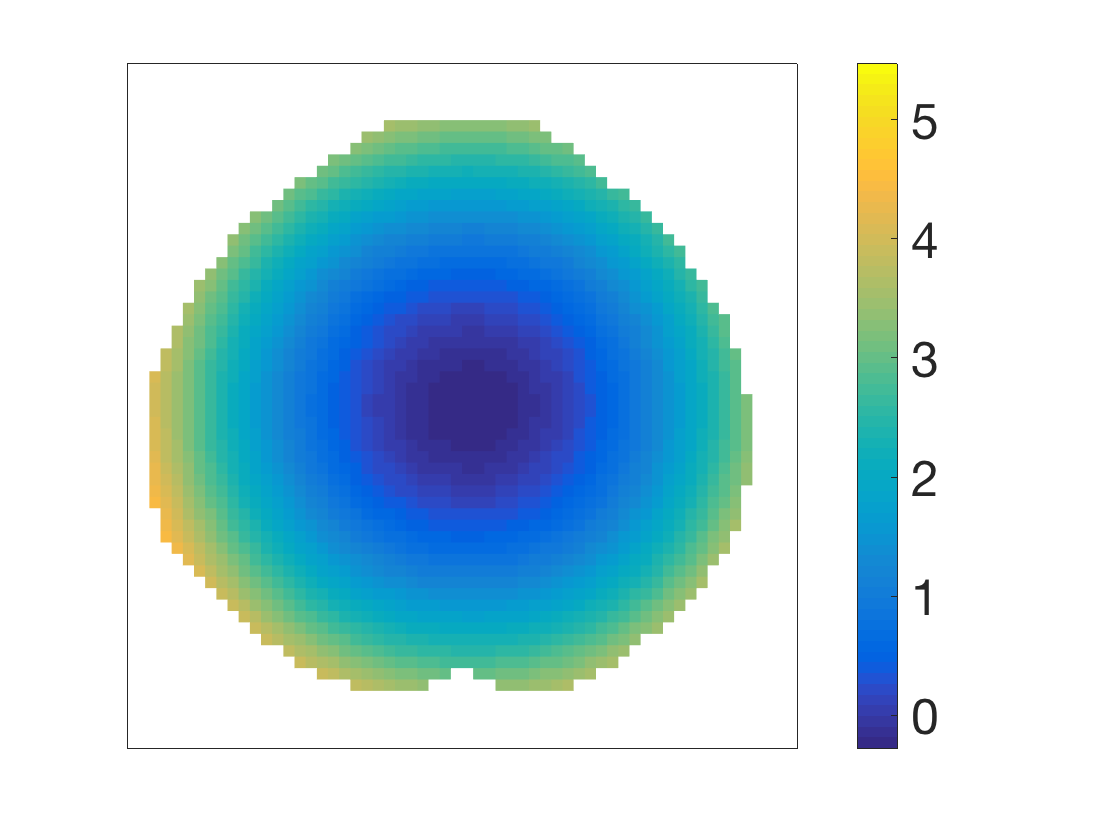} &\hspace{-0.7in}\includegraphics[scale=0.14]{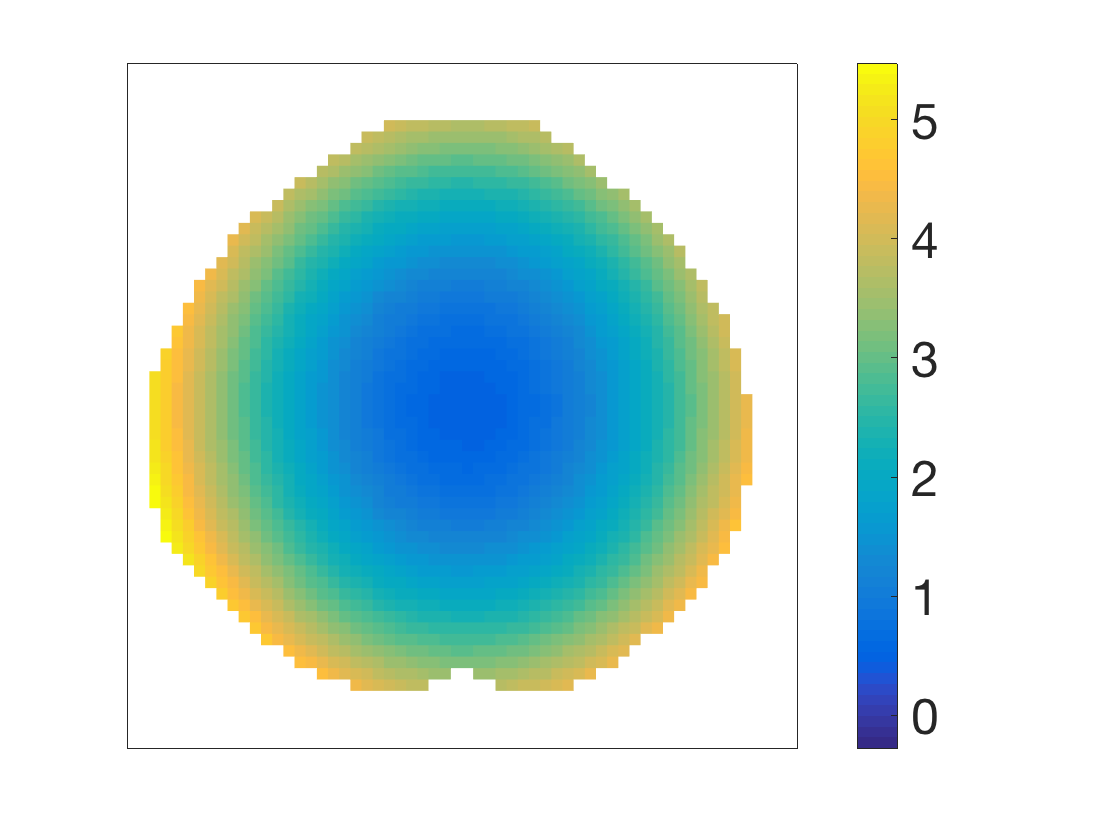}\\ [-10pt]
			\hspace{-0.3in}$\triangle_3$ & & &\\
			& & &
		\end{tabular}
	\end{center}
	\caption{SCCs for quadratic function with $n=50$ and $\alpha=0.01$.}
	\label{FIG:S01}
\end{figure}

\begin{figure}
	\begin{center}
		\begin{tabular}{cccc}
			\hspace{-0.5in}Triangulation &\hspace{-1.3in}$\widehat{\mu}$ &\hspace{-1.3in}Lower SCC &\hspace{-1.3in}Upper SCC\\
			\includegraphics[scale=0.14]{mu_fn1_tri1} &\hspace{-0.8in}\includegraphics[scale=0.14]{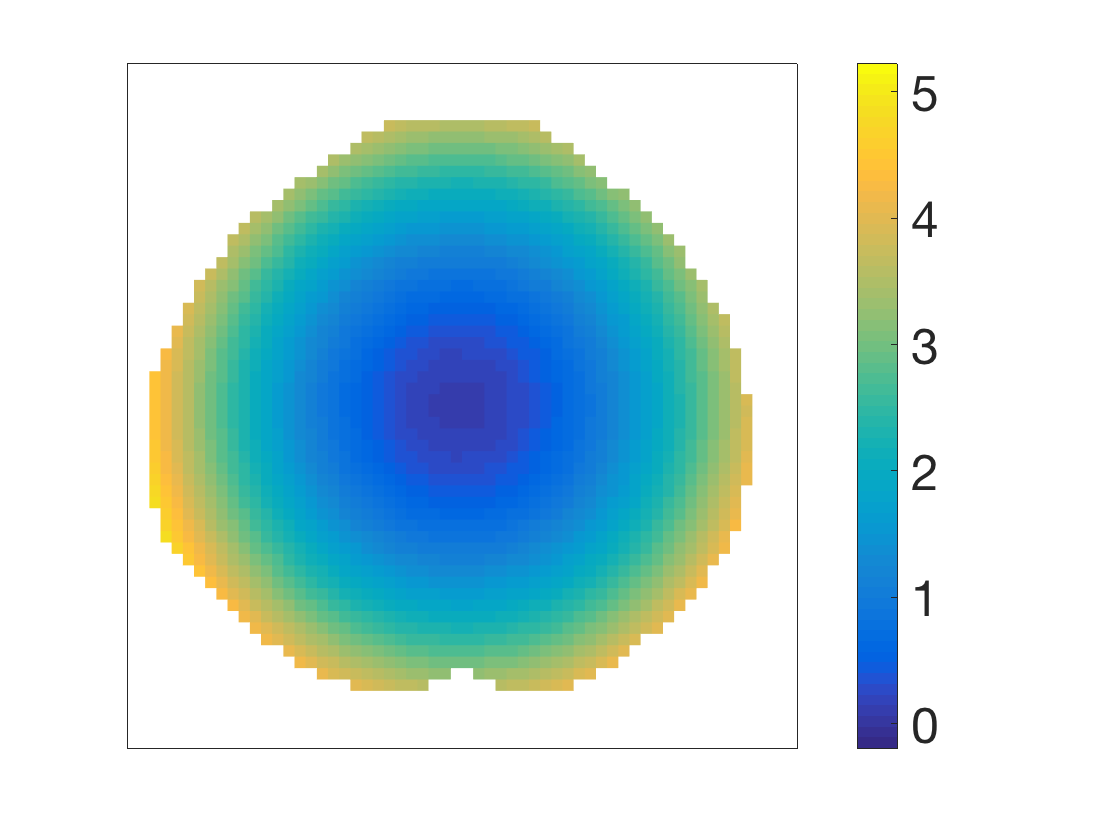} &\hspace{-0.7in}\includegraphics[scale=0.14]{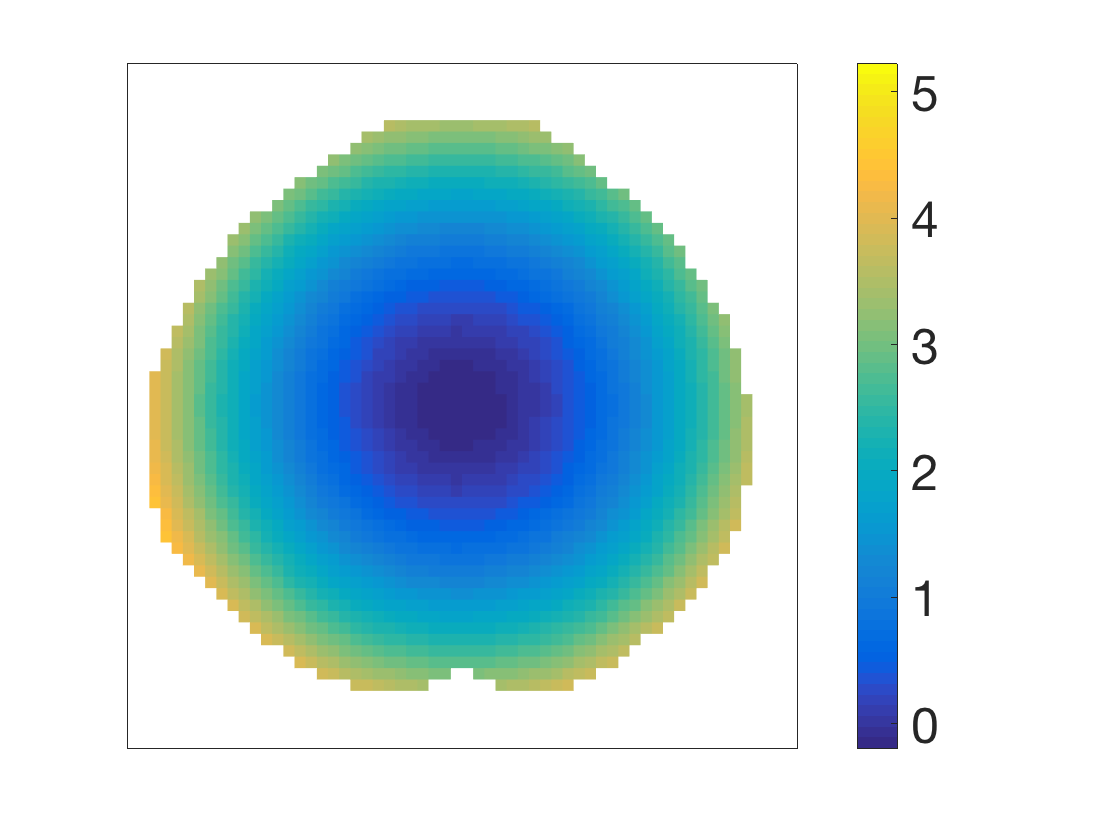} &\hspace{-0.7in}\includegraphics[scale=0.14]{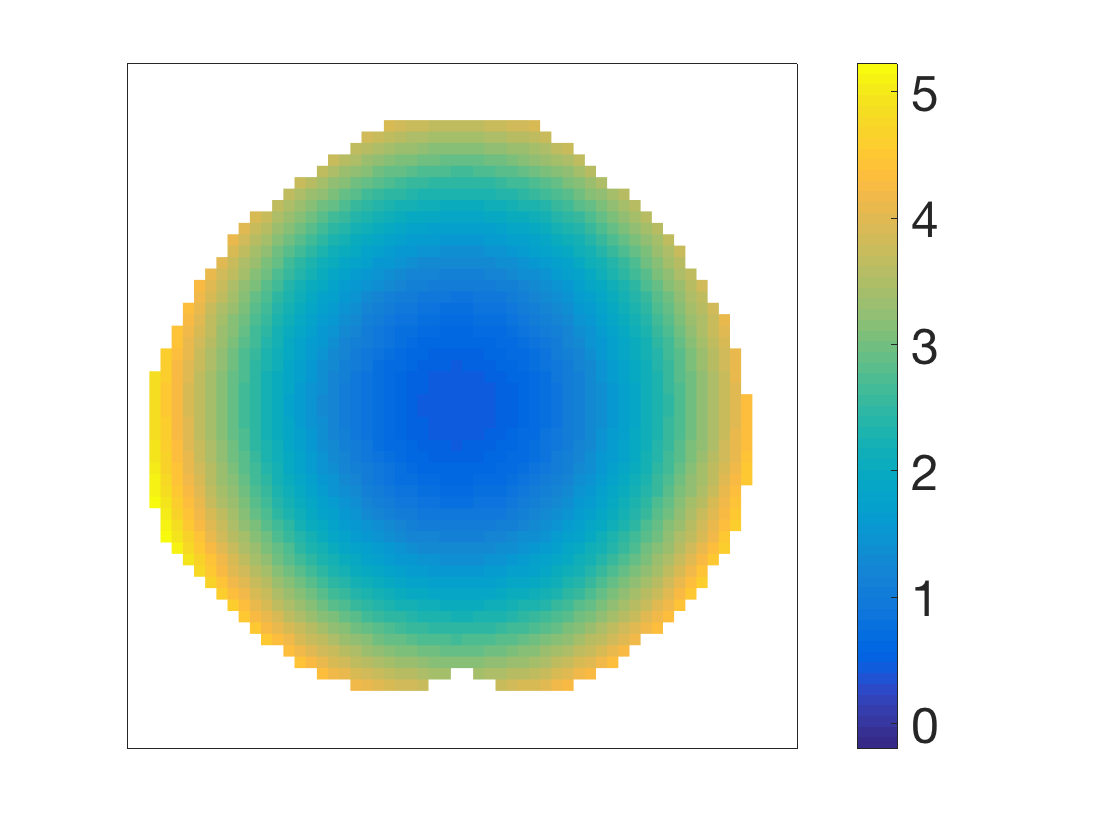}\\ [-10pt]
			\hspace{-0.3in}$\triangle_1$ & & &\\
			
			\includegraphics[scale=0.14]{mu_fn1_tri2} &\hspace{-0.8in}\includegraphics[scale=0.14]{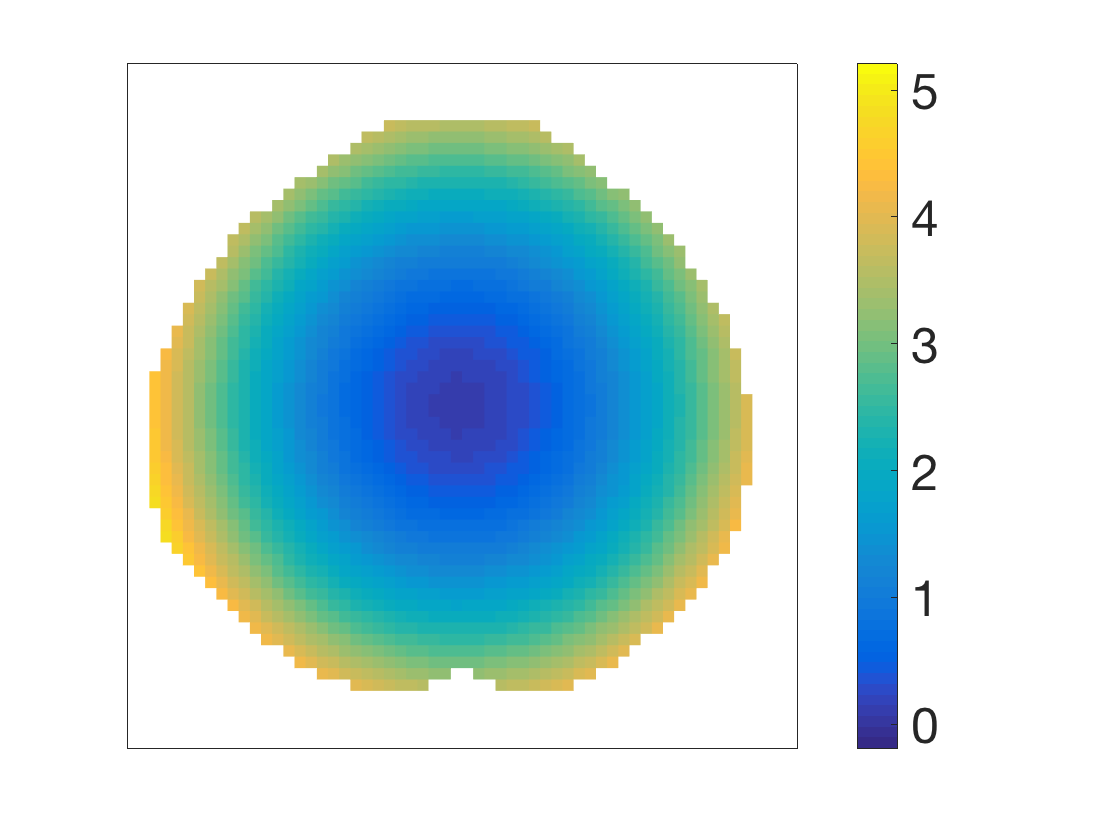} &\hspace{-0.7in}\includegraphics[scale=0.14]{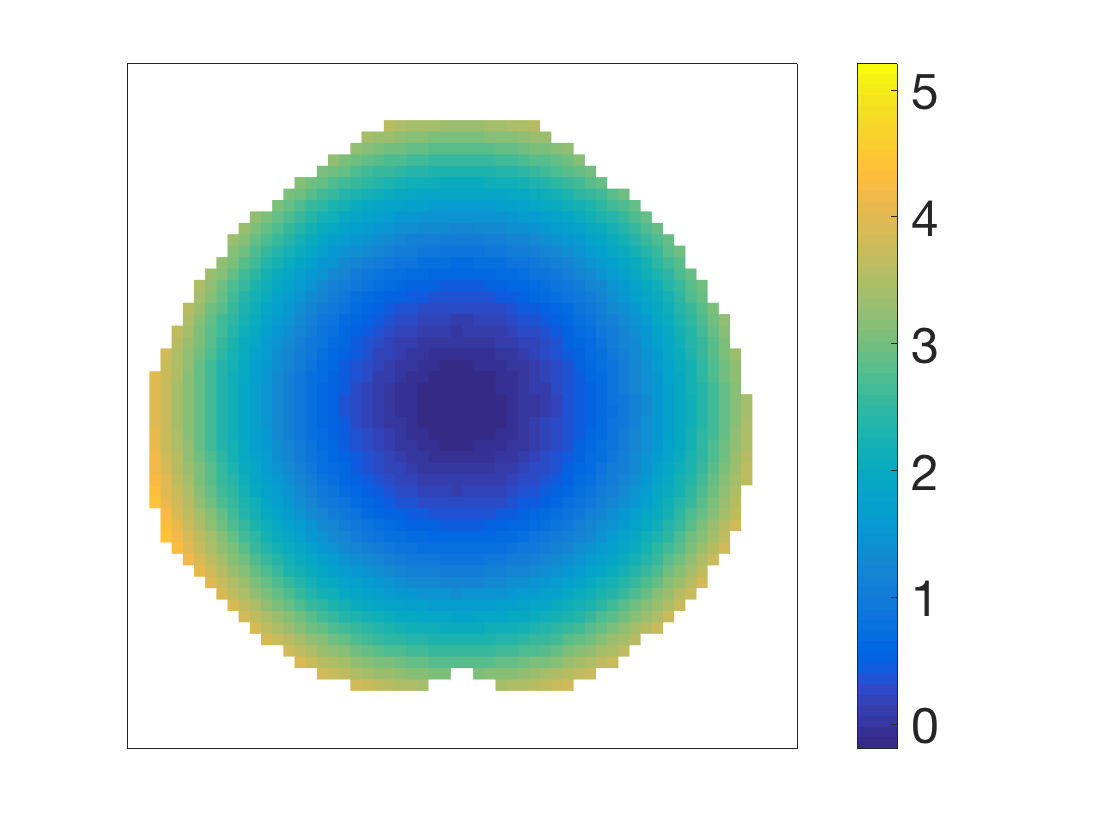} &\hspace{-0.7in}\includegraphics[scale=0.14]{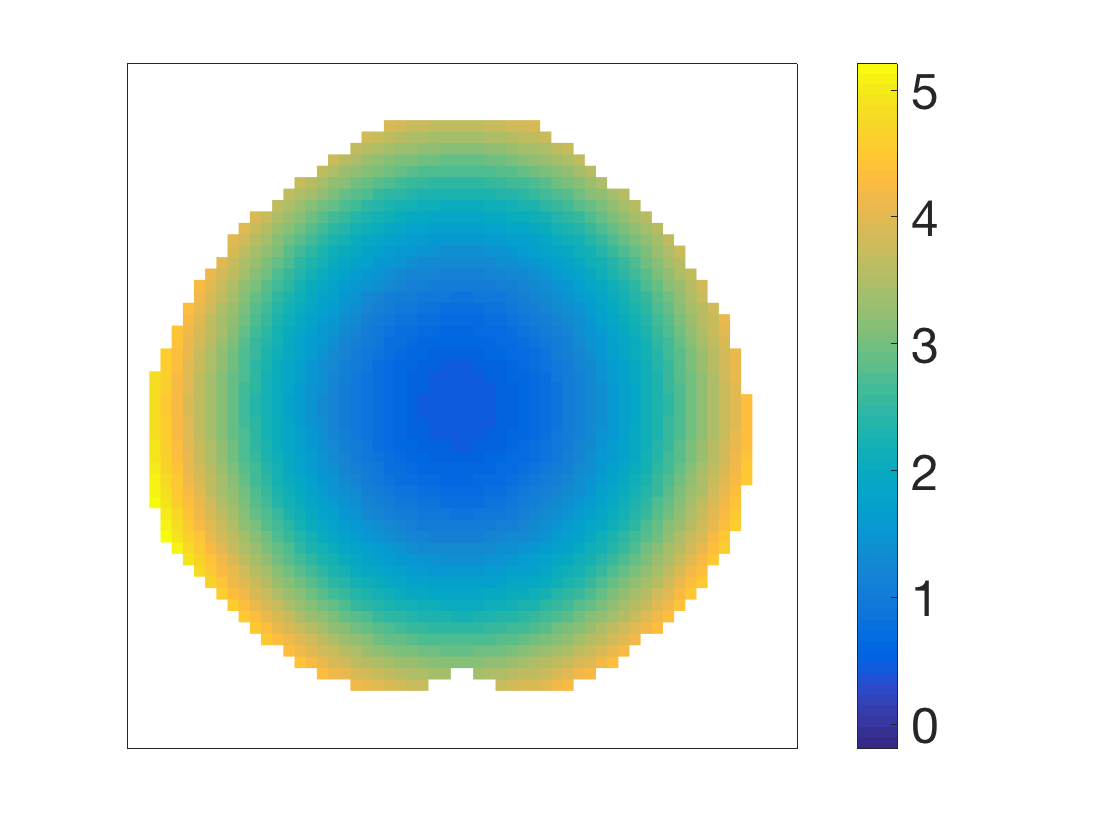}\\ [-10pt]
			\hspace{-0.3in}$\triangle_2$ & & &\\
			
			\includegraphics[scale=0.14]{mu_fn1_tri3} &\hspace{-0.8in}\includegraphics[scale=0.14]{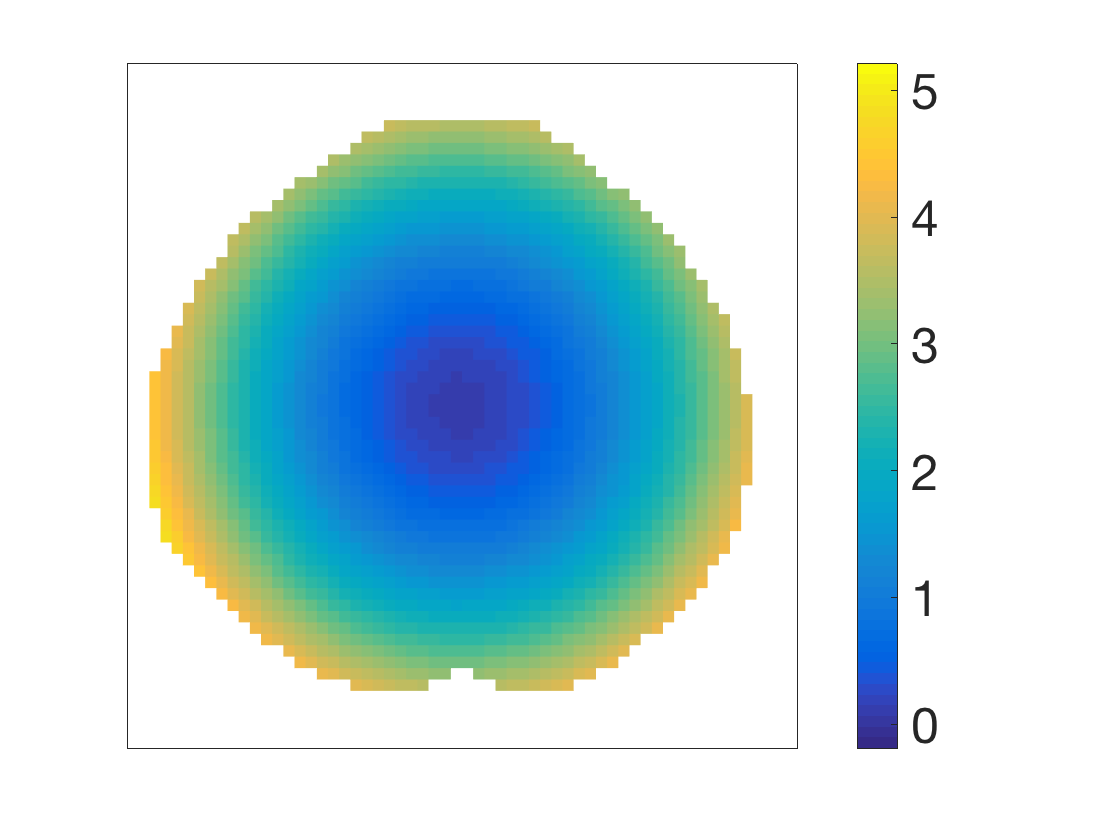} &\hspace{-0.7in}\includegraphics[scale=0.14]{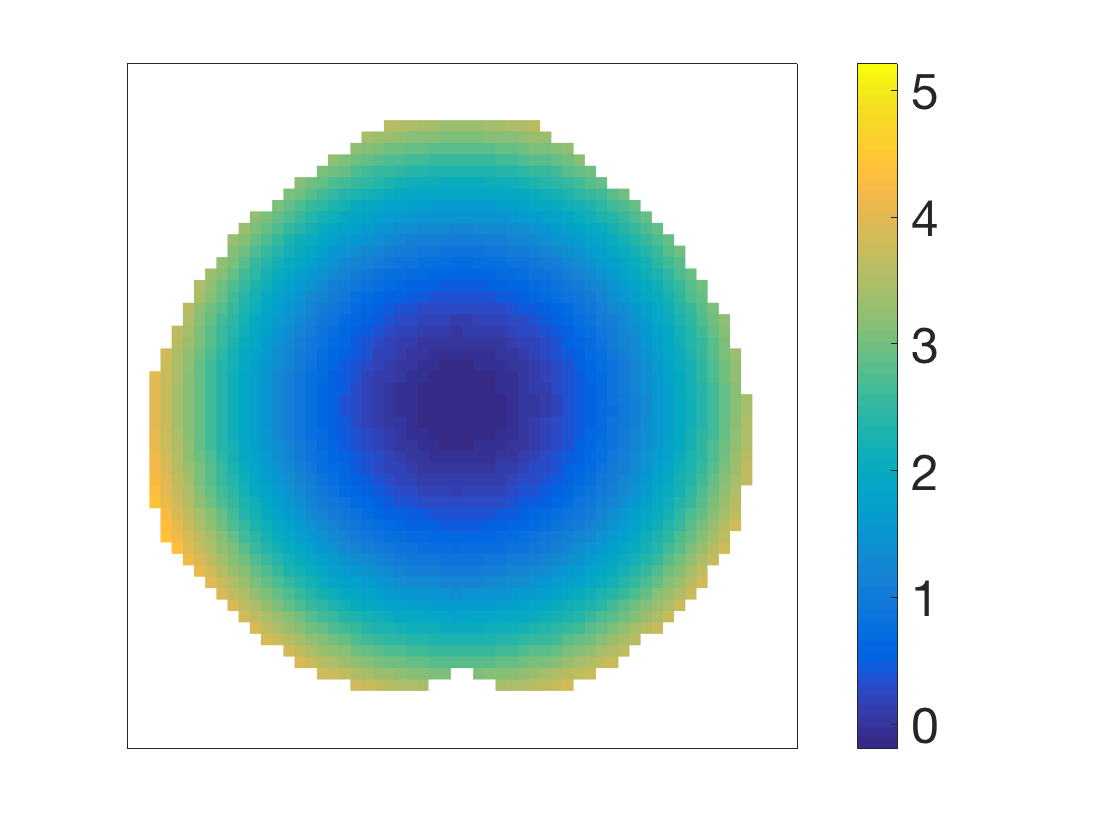} &\hspace{-0.7in}\includegraphics[scale=0.14]{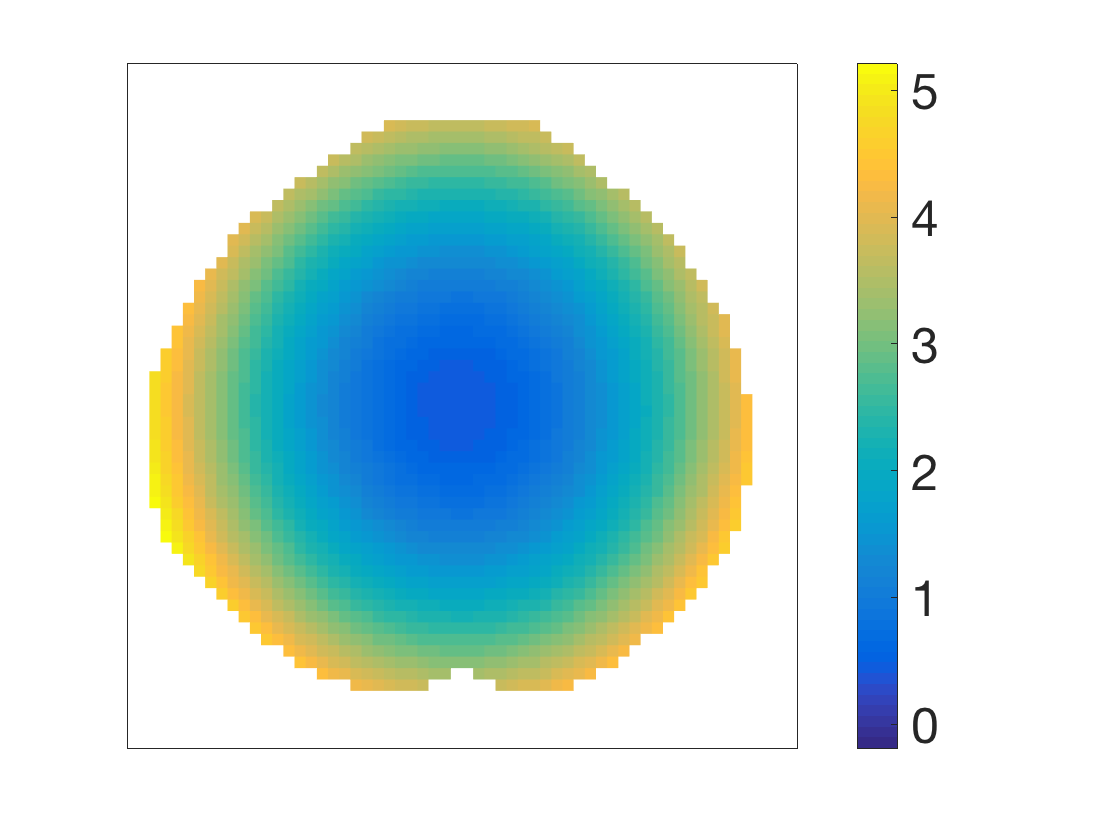}\\ [-10pt]
			\hspace{-0.3in}$\triangle_3$ & & &\\
			& & &
		\end{tabular}
	\end{center}
	\caption{SCCs for quadratic function with $n=100$ and $\alpha=0.01$.}
	\label{FIG:S02}
\end{figure}

\begin{figure}
	\begin{center}
		\begin{tabular}{cccc}
			\hspace{-0.5in}Triangulation &\hspace{-1.3in}$\widehat{\mu}$ &\hspace{-1.3in}Lower SCC &\hspace{-1.3in}Upper SCC\\
			\includegraphics[scale=0.14]{mu_fn1_tri1} &\hspace{-0.8in}\includegraphics[scale=0.14]{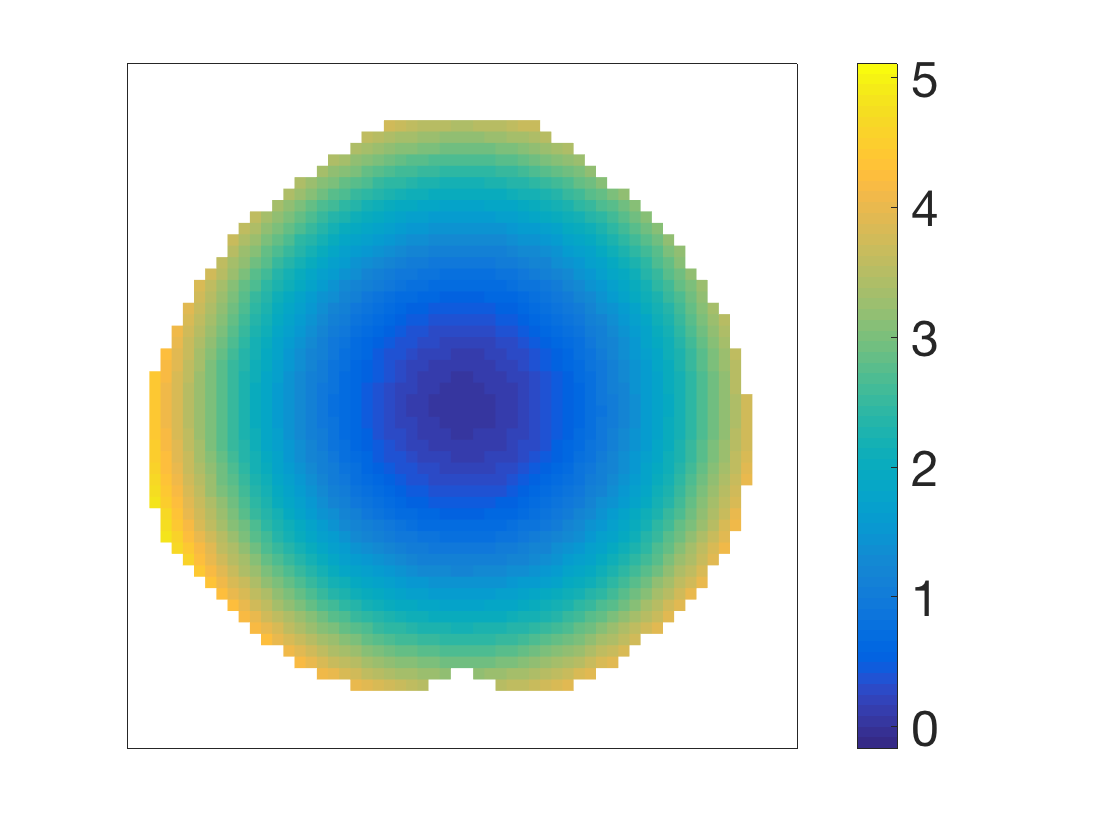} &\hspace{-0.7in}\includegraphics[scale=0.14]{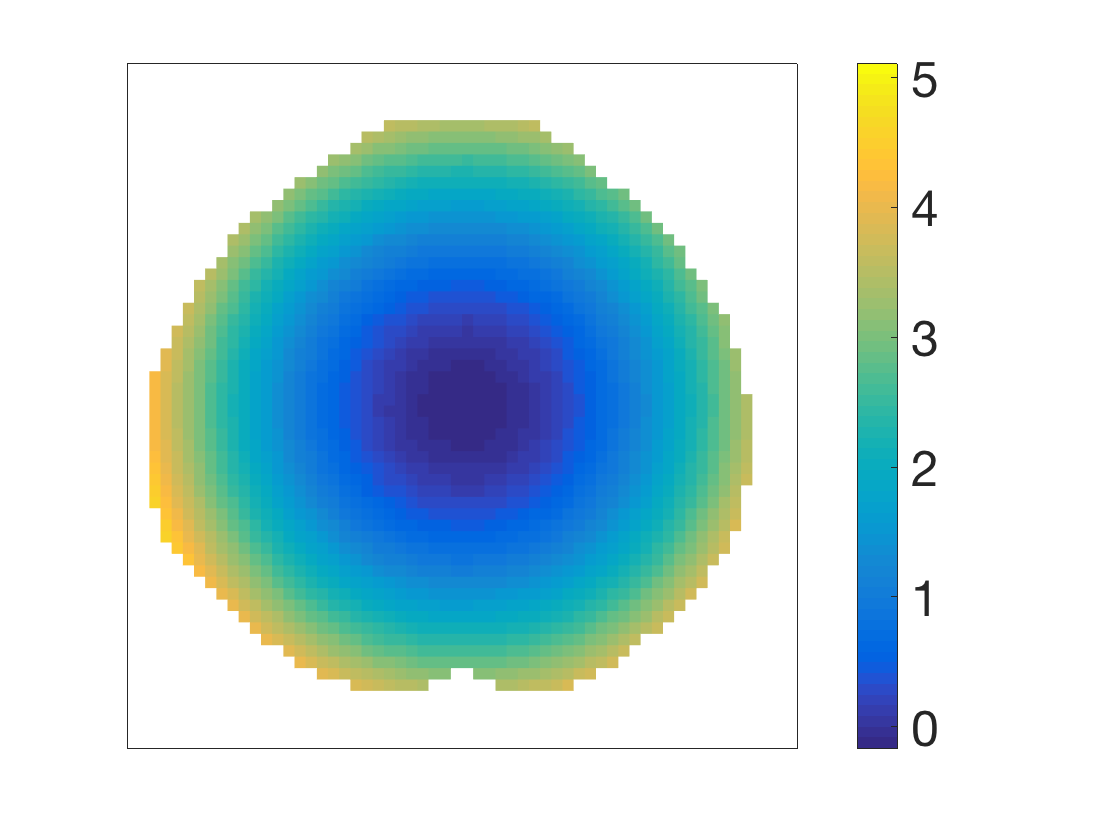} &\hspace{-0.7in}\includegraphics[scale=0.14]{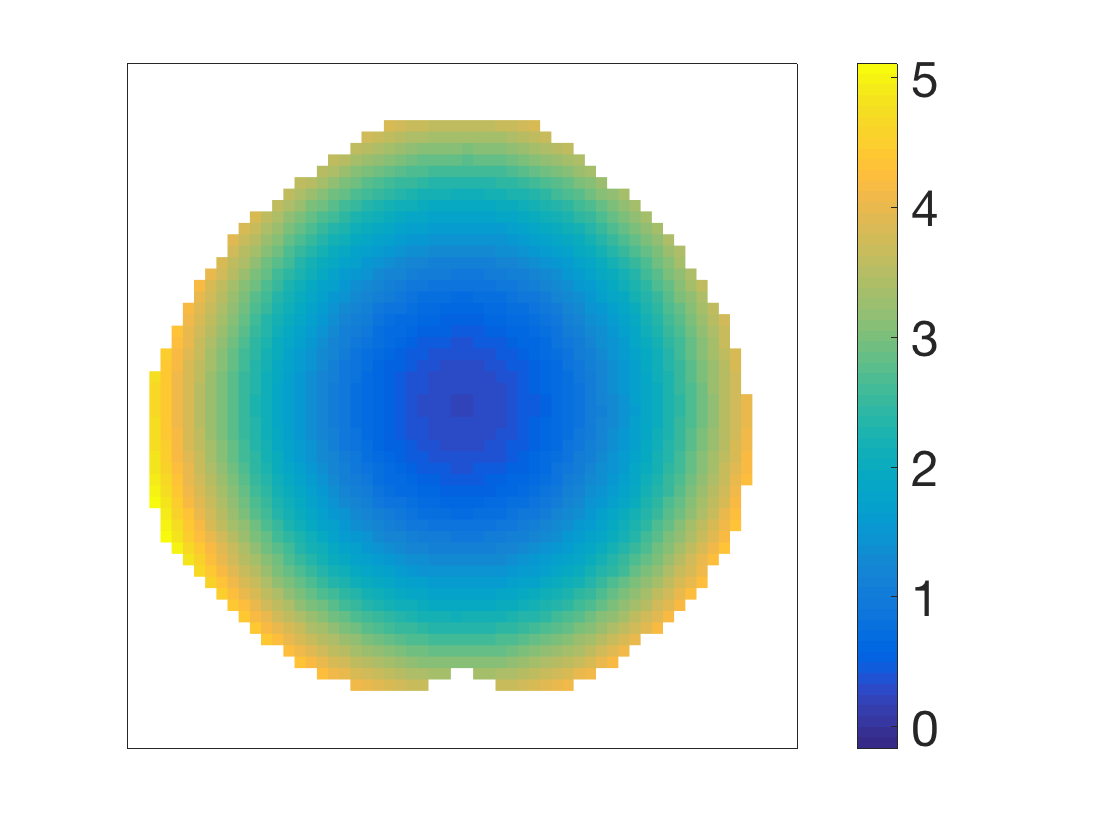}\\ [-10pt]
			\hspace{-0.3in}$\triangle_1$ & & &\\
			
			\includegraphics[scale=0.14]{mu_fn1_tri2} &\hspace{-0.8in}\includegraphics[scale=0.14]{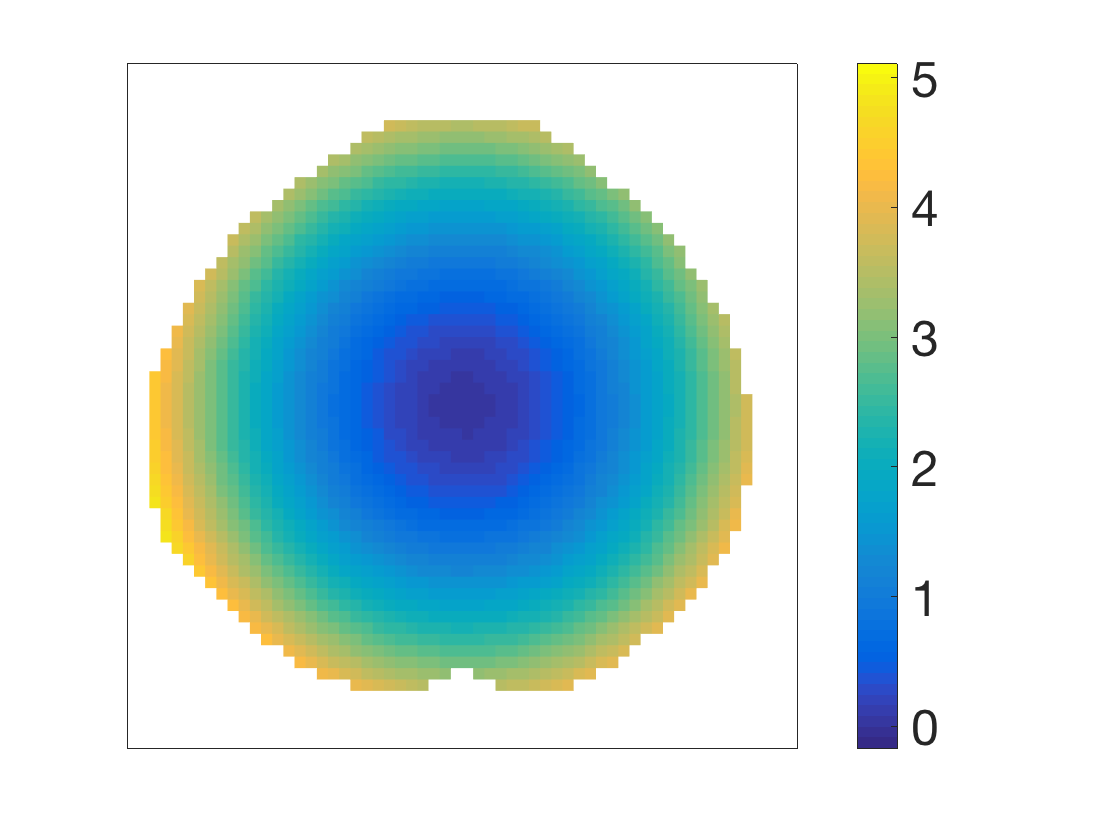} &\hspace{-0.7in}\includegraphics[scale=0.14]{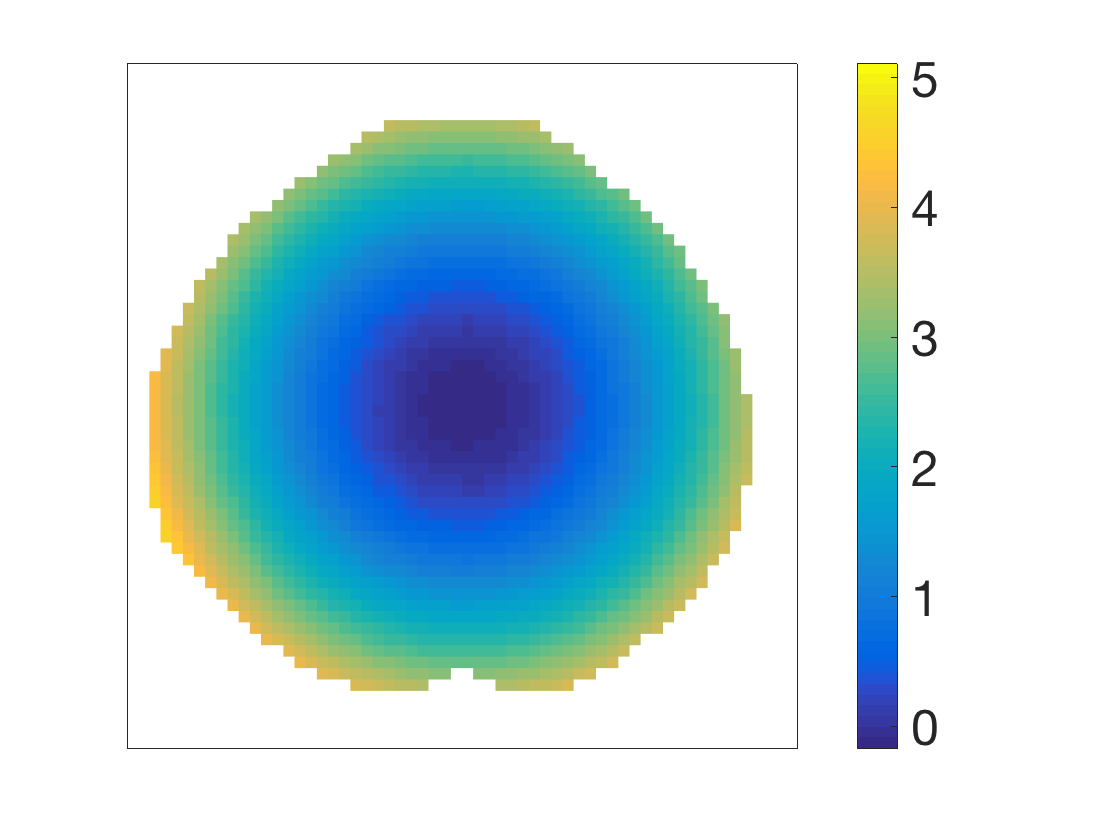} &\hspace{-0.7in}\includegraphics[scale=0.14]{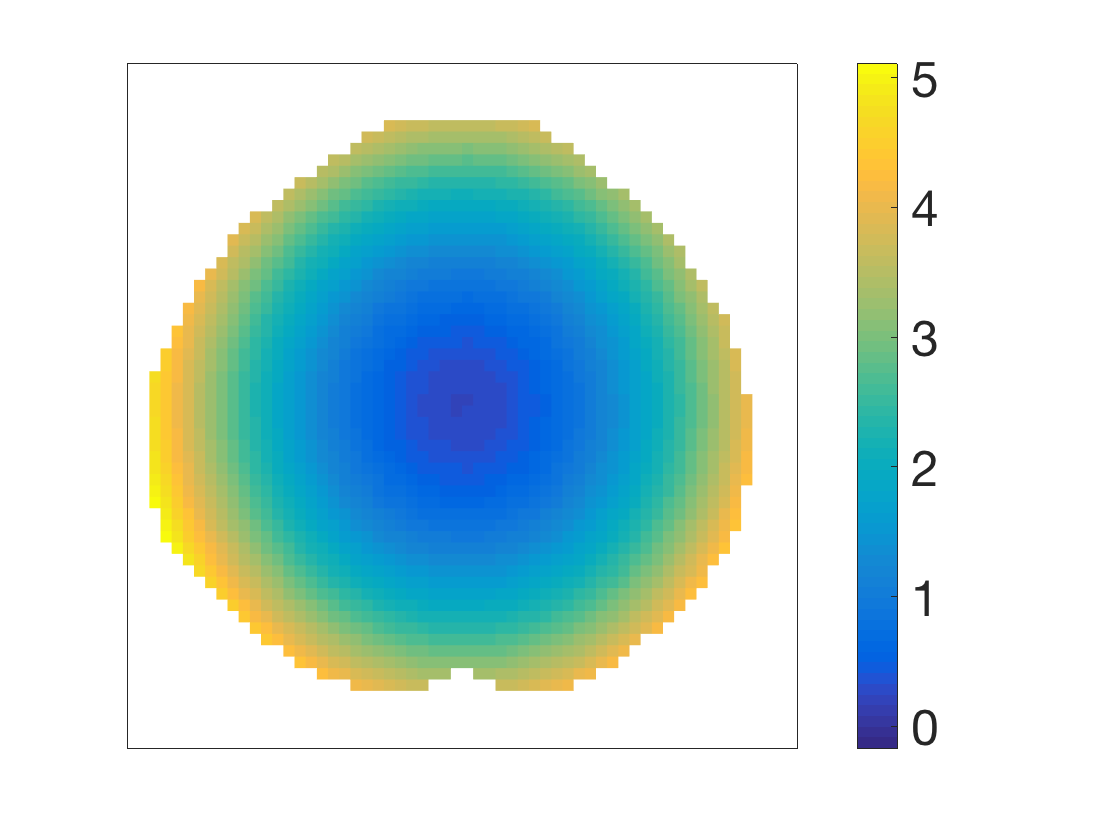}\\ [-10pt]
			\hspace{-0.3in}$\triangle_2$ & & &\\
			
			\includegraphics[scale=0.14]{mu_fn1_tri3} &\hspace{-0.8in}\includegraphics[scale=0.14]{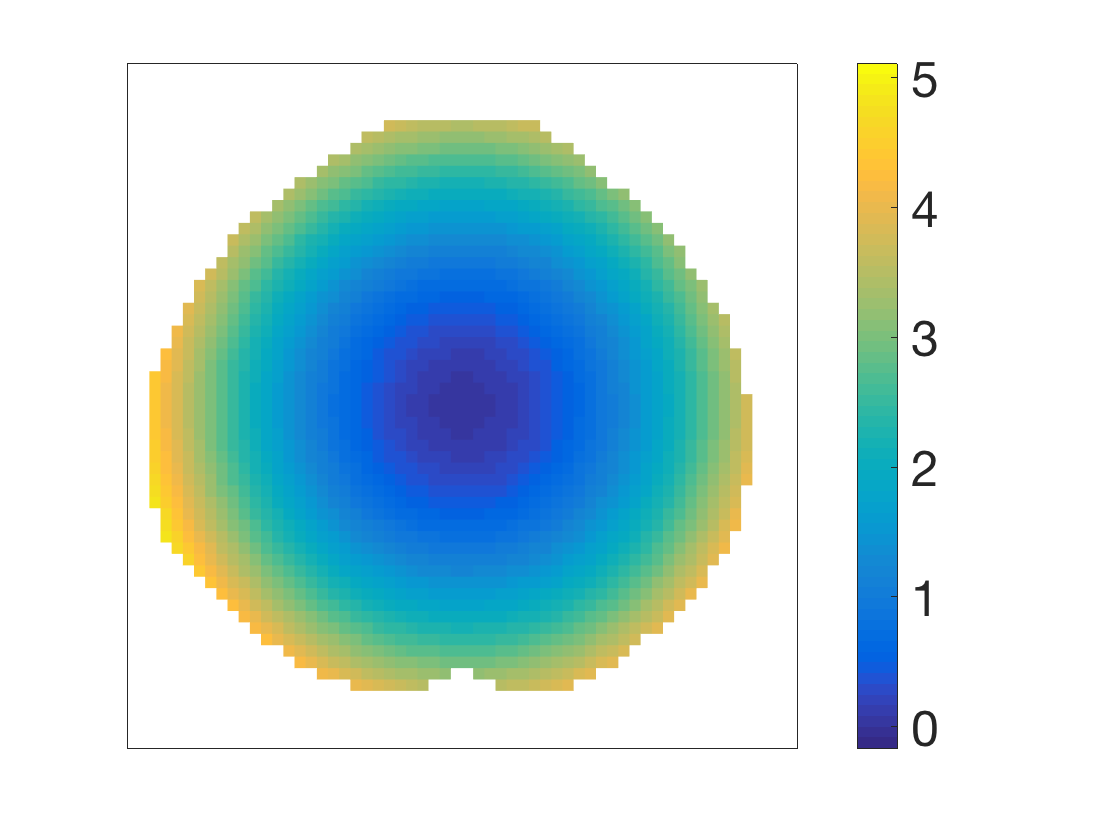} &\hspace{-0.7in}\includegraphics[scale=0.14]{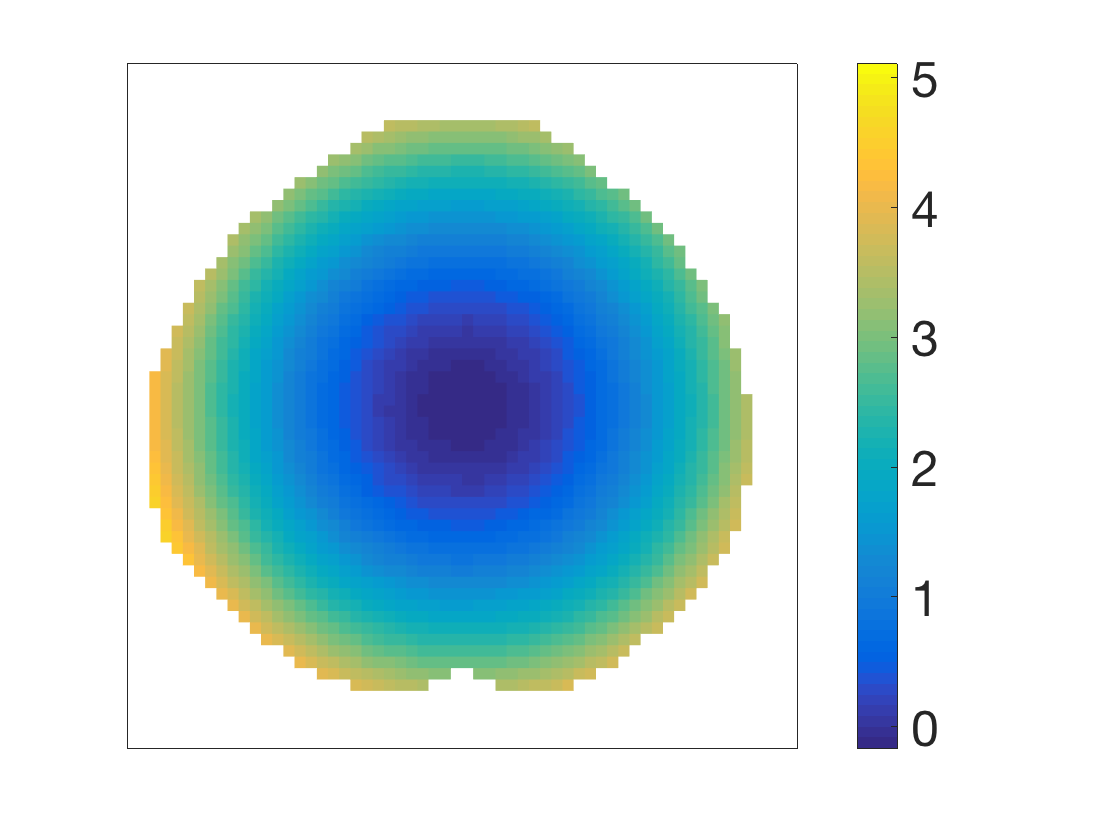} &\hspace{-0.7in}\includegraphics[scale=0.14]{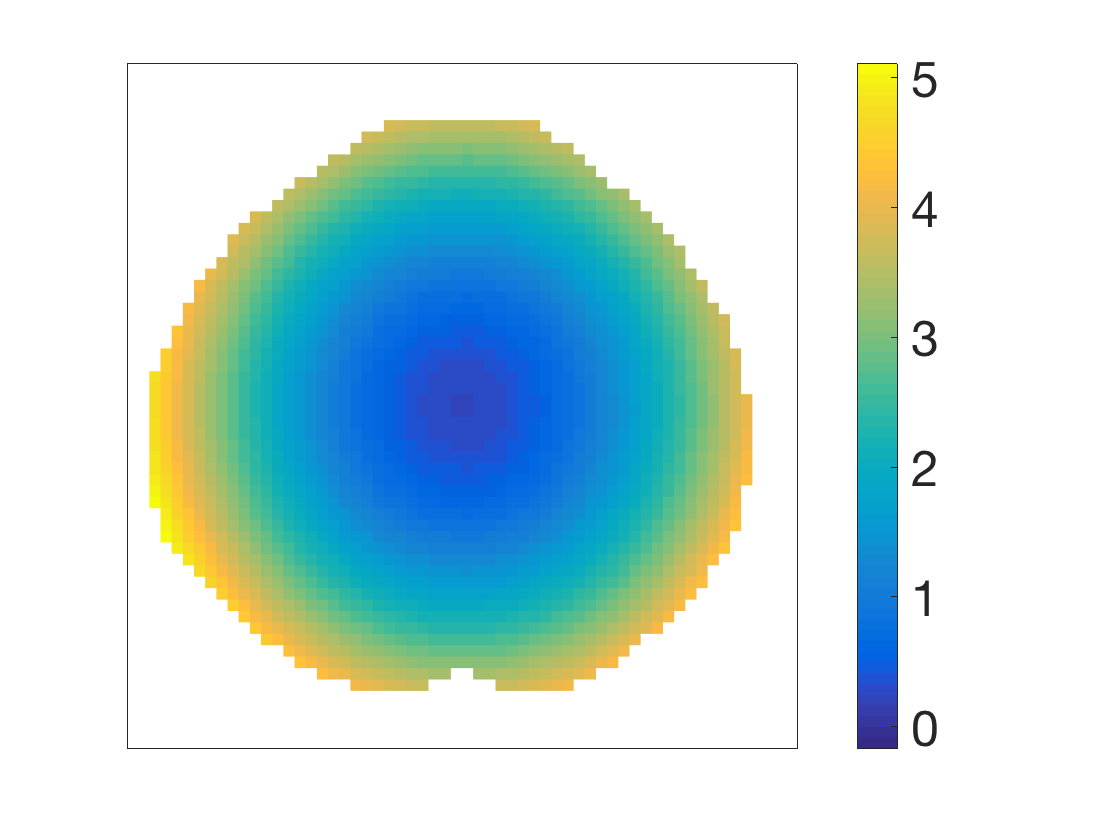}\\ [-10pt]
			\hspace{-0.3in}$\triangle_3$ & & &\\
			& & &
		\end{tabular}
	\end{center}
	\caption{SCCs for quadratic function with $n=200$ and $\alpha=0.01$.}
	\label{FIG:S03}
\end{figure}

\begin{figure}
	\begin{center}
		\begin{tabular}{cccc}
			\hspace{-0.5in}Triangulation &\hspace{-1.3in}$\widehat{\mu}$ &\hspace{-1.3in}Lower SCC &\hspace{-1.3in}Upper SCC\\
			\includegraphics[scale=0.14]{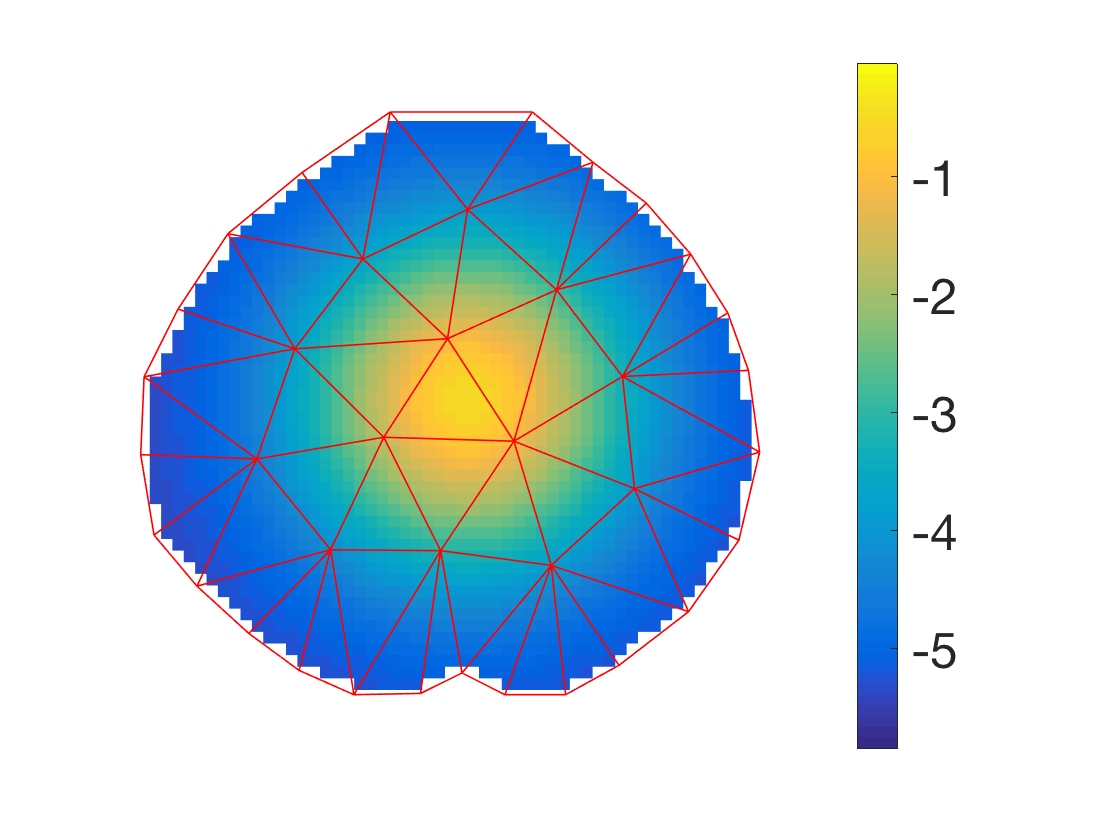} &\hspace{-0.8in}\includegraphics[scale=0.14]{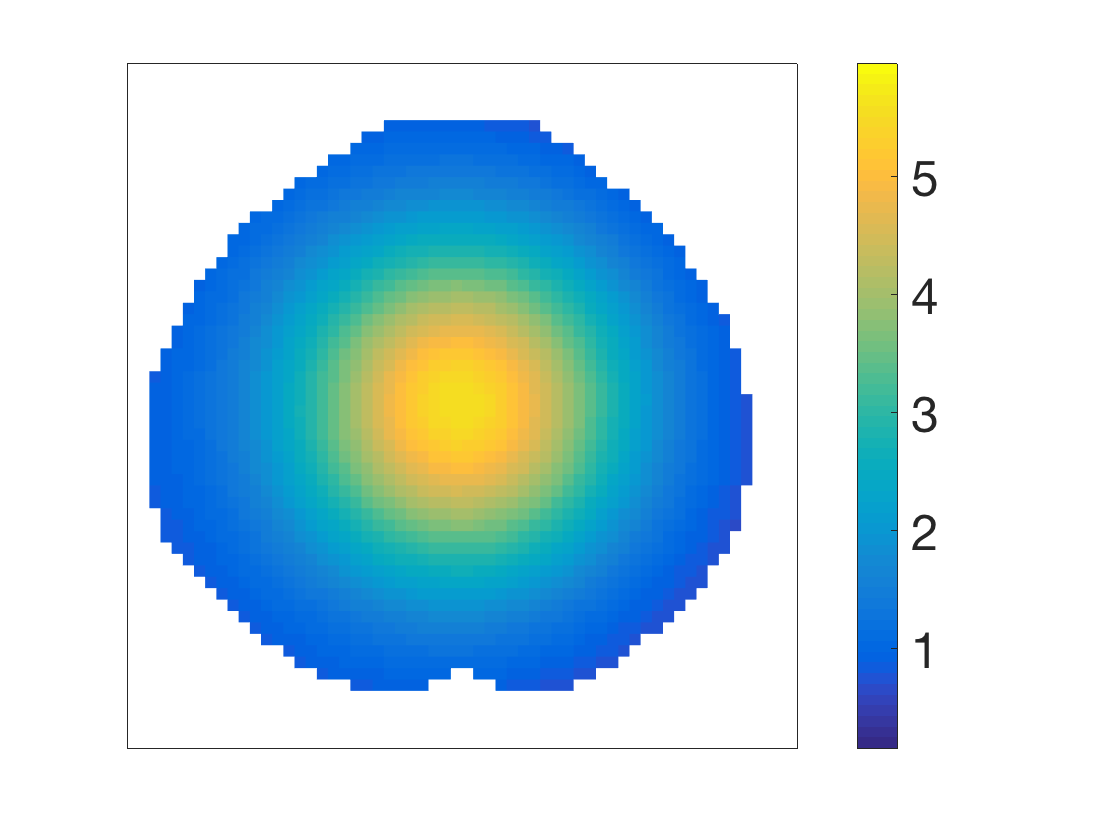} &\hspace{-0.7in}\includegraphics[scale=0.14]{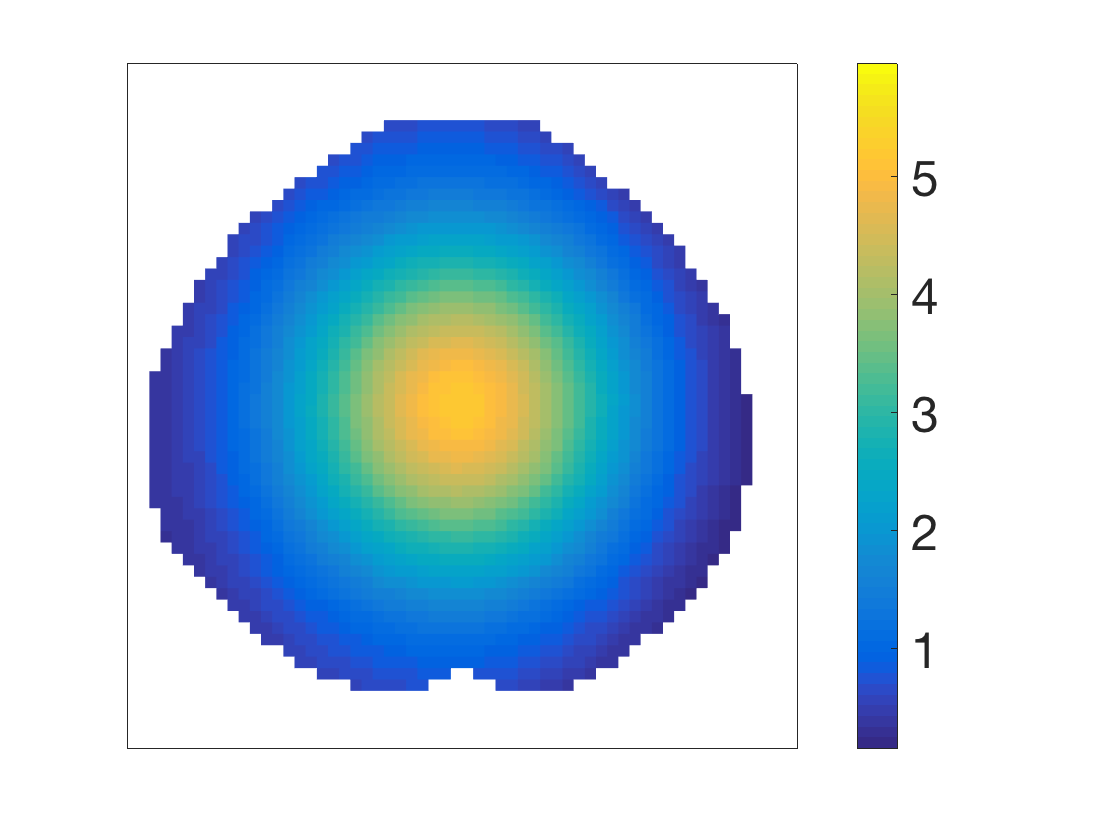} &\hspace{-0.7in}\includegraphics[scale=0.14]{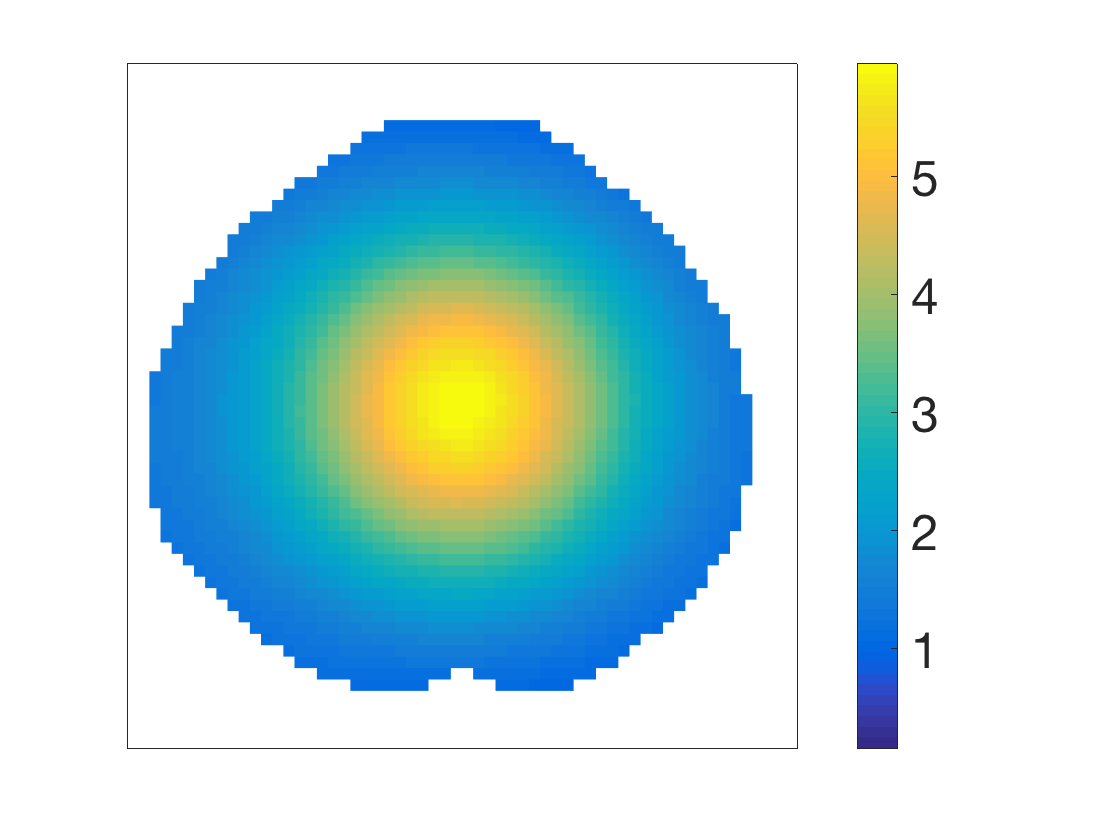}\\ [-10pt]
			\hspace{-0.3in}$\triangle_1$ & & &\\
			
			\includegraphics[scale=0.14]{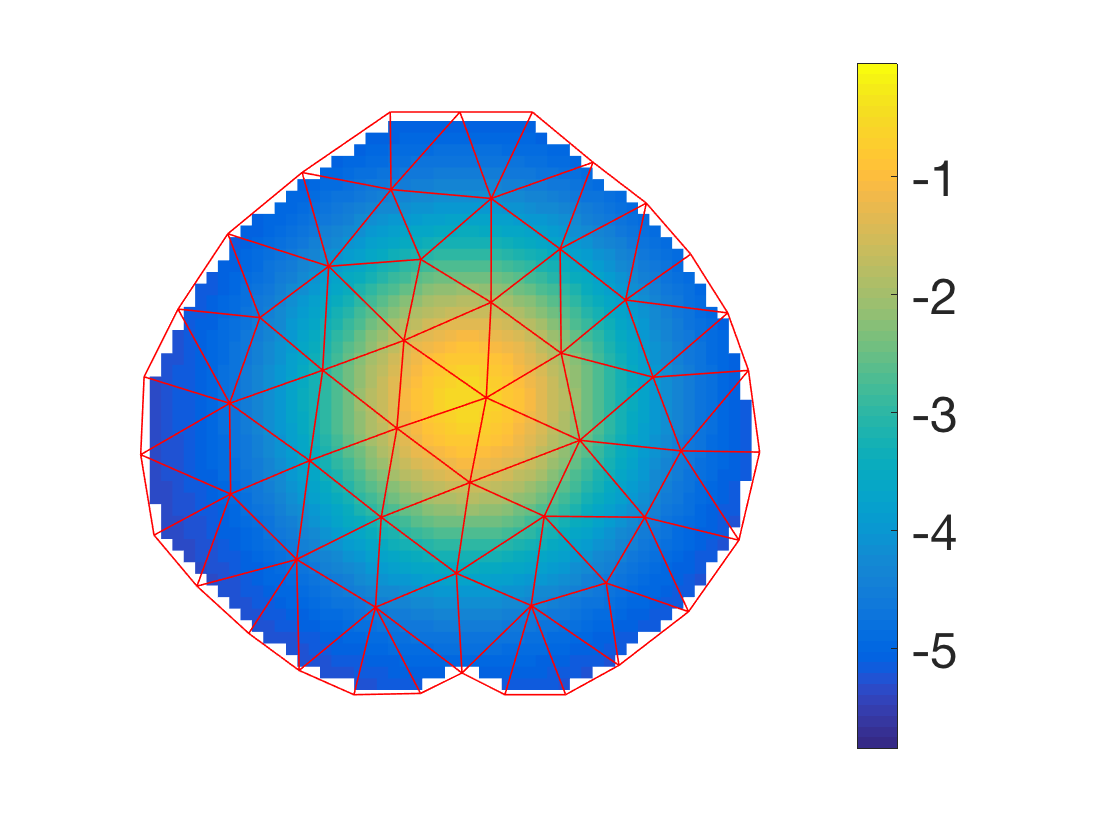} &\hspace{-0.8in}\includegraphics[scale=0.14]{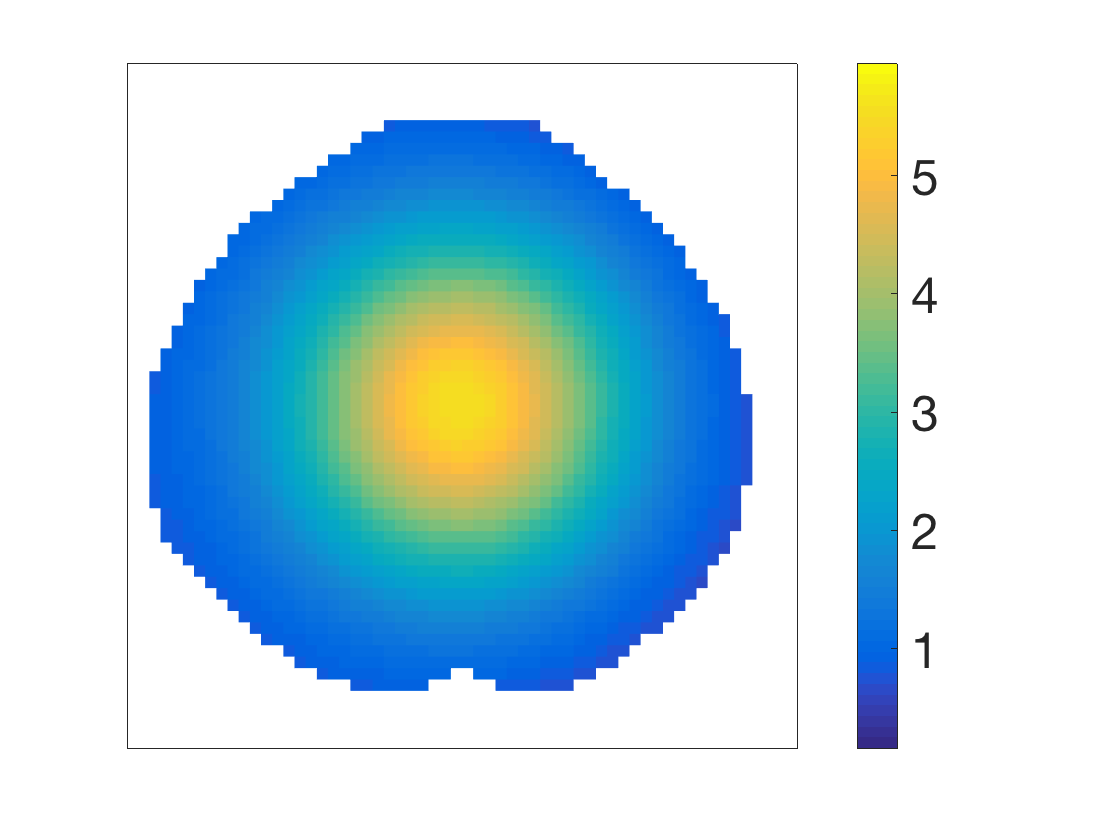} &\hspace{-0.7in}\includegraphics[scale=0.14]{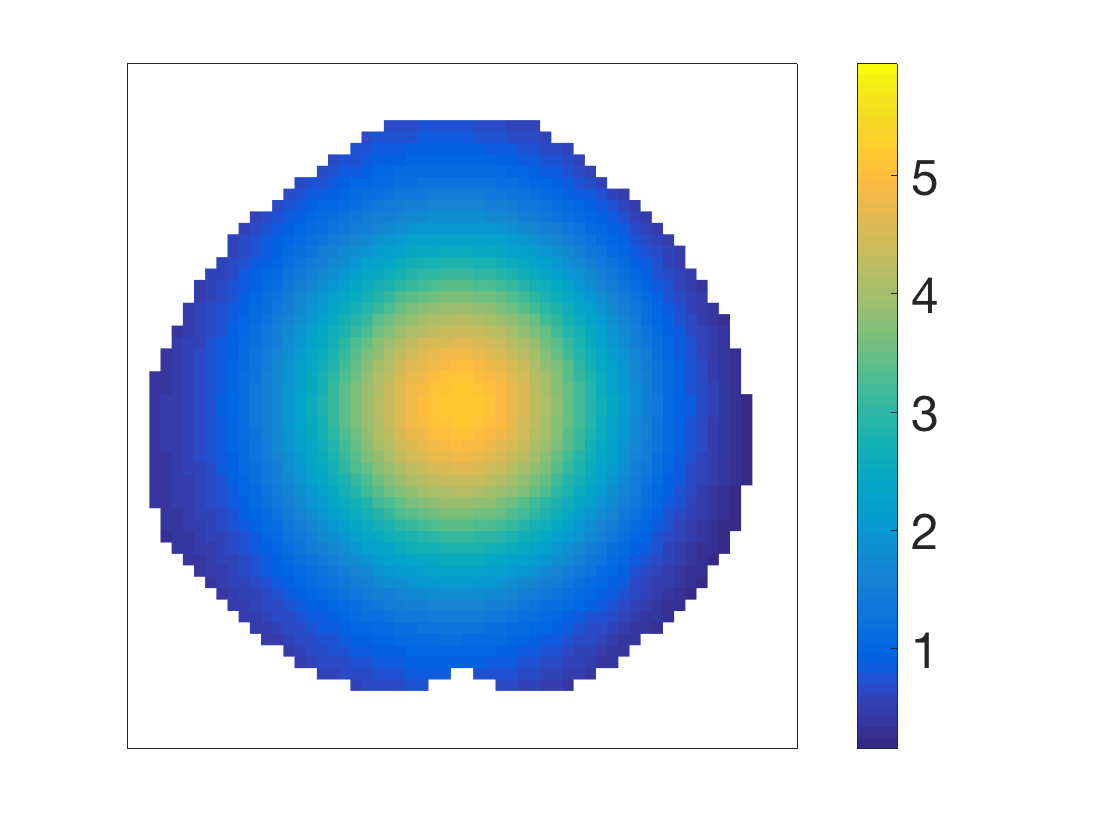} &\hspace{-0.7in}\includegraphics[scale=0.14]{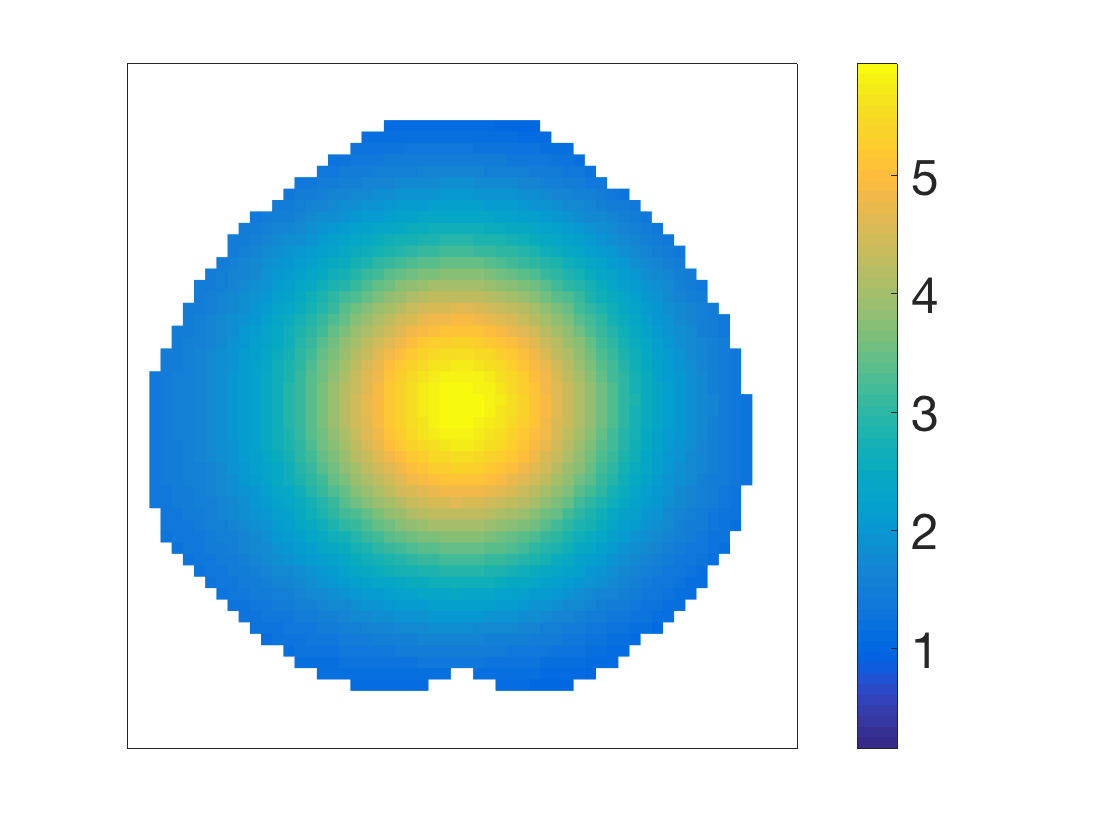}\\ [-10pt]
			\hspace{-0.3in}$\triangle_2$ & & &\\
			
			\includegraphics[scale=0.14]{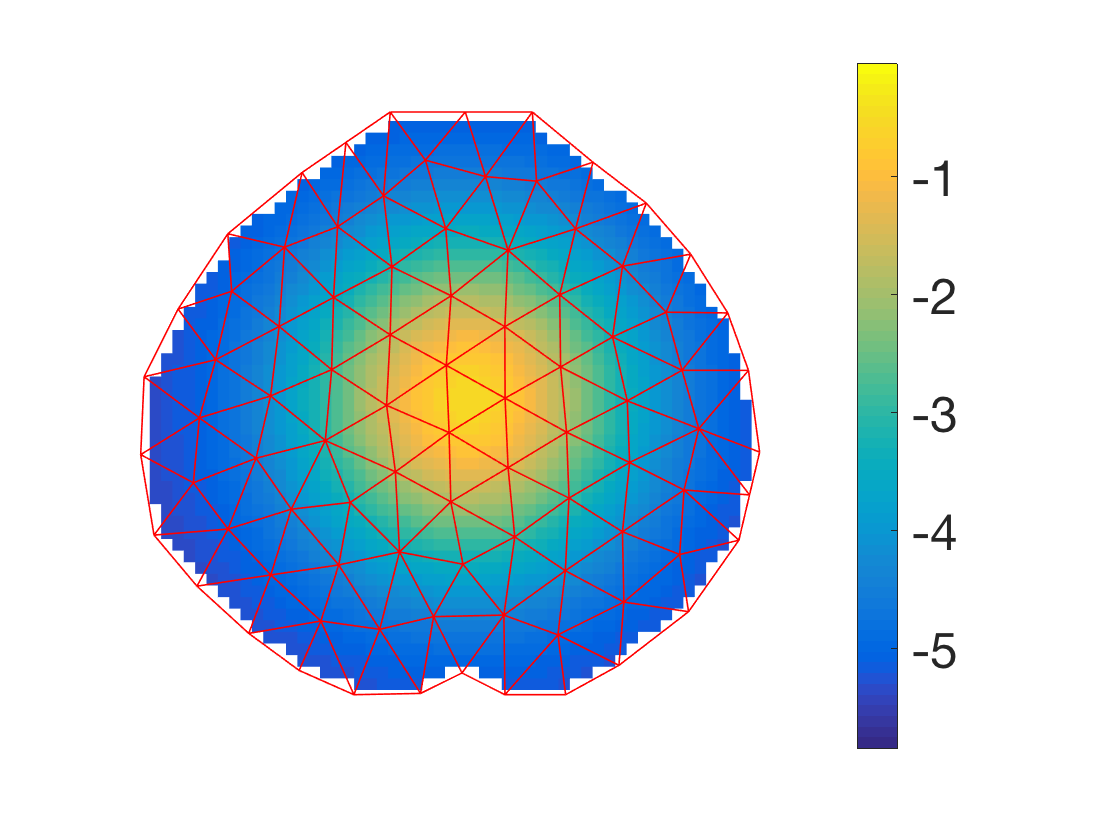} &\hspace{-0.8in}\includegraphics[scale=0.14]{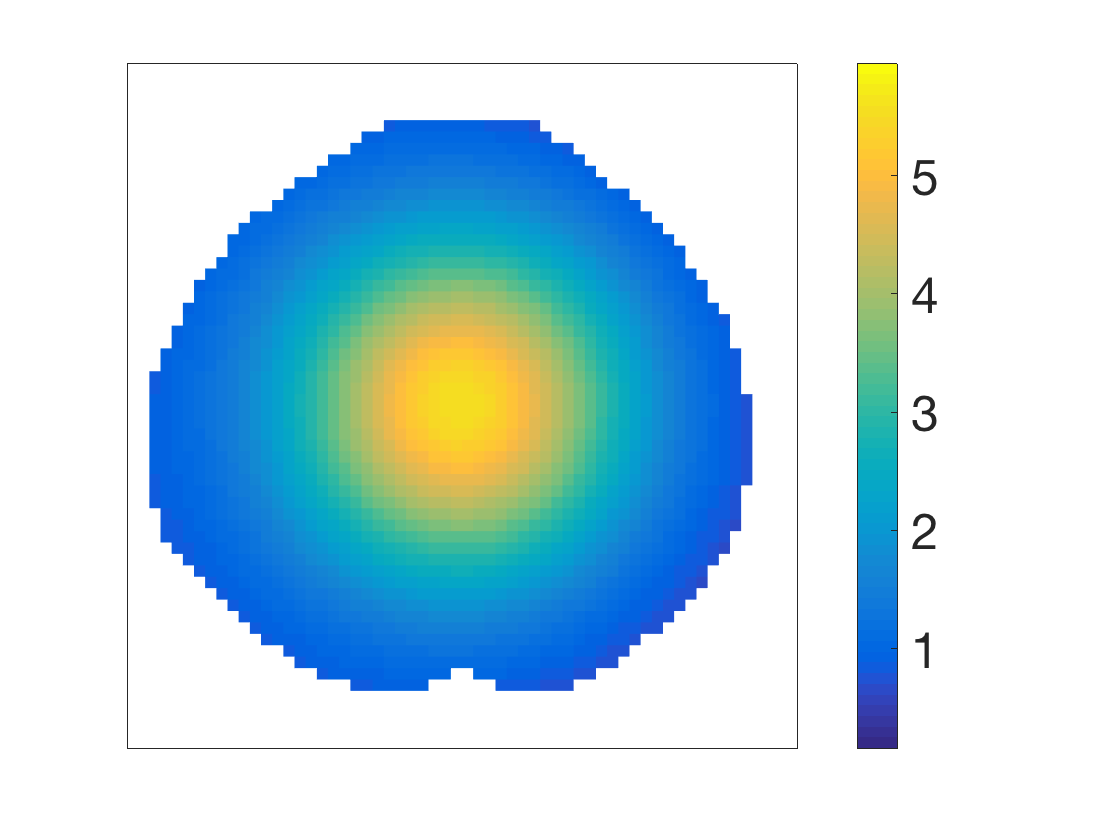} &\hspace{-0.7in}\includegraphics[scale=0.14]{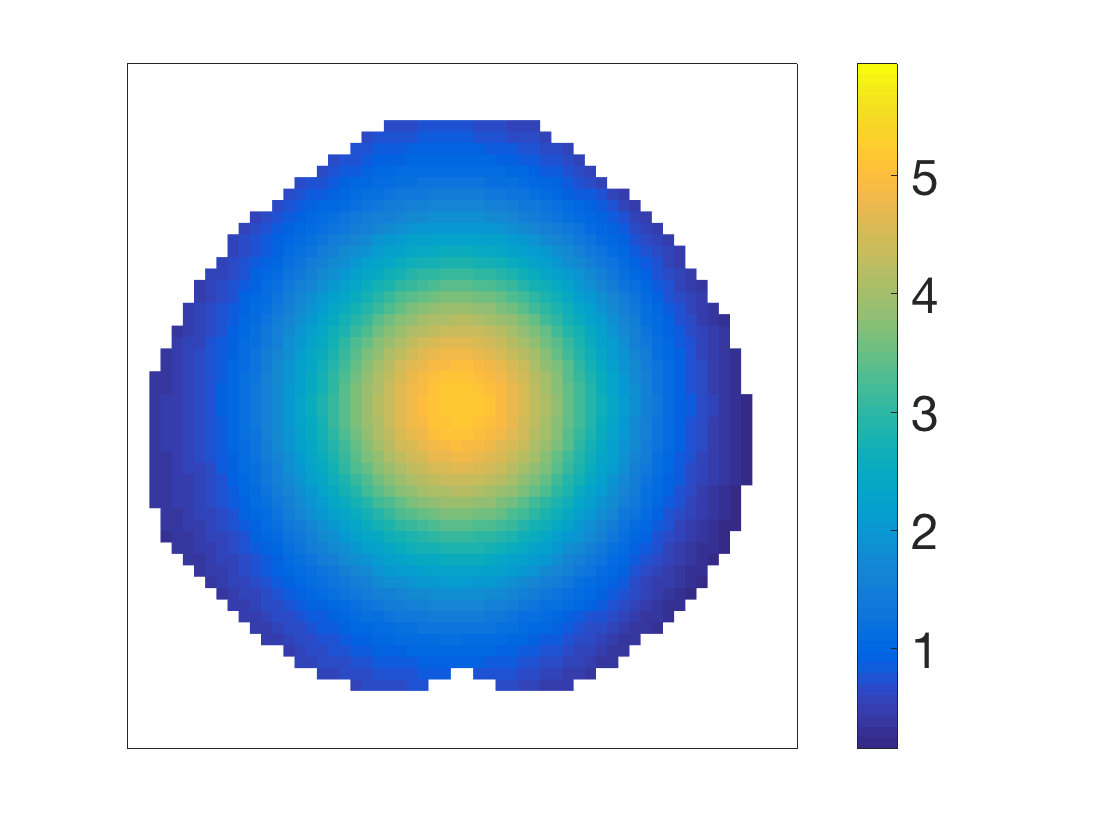} &\hspace{-0.7in}\includegraphics[scale=0.14]{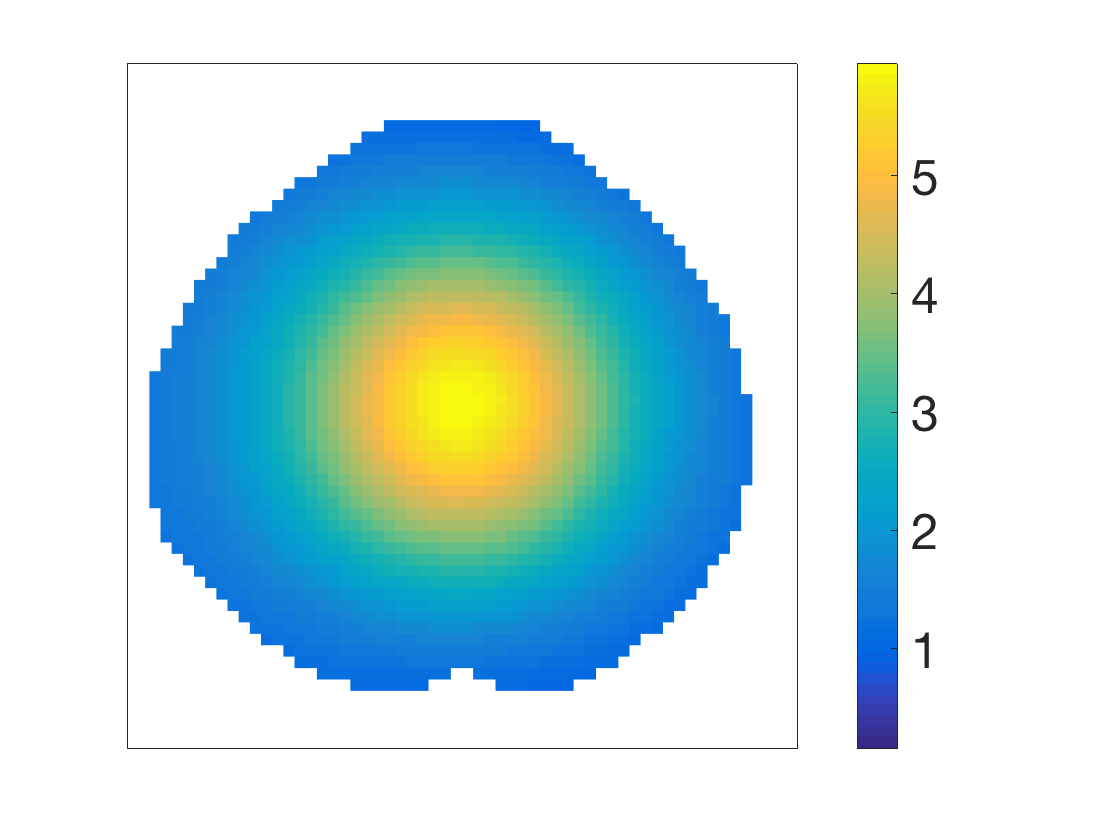}\\ [-10pt]
			\hspace{-0.3in}$\triangle_3$ & & &\\
			& & &
		\end{tabular}
	\end{center}
	\caption{SCCs for bump function with $n=50$ and $\alpha=0.01$.}
	\label{FIG:S04}
\end{figure}

\begin{figure}
	\begin{center}
		\begin{tabular}{cccc}
			\hspace{-0.5in}Triangulation &\hspace{-1.3in}$\widehat{\mu}$ &\hspace{-1.3in}Lower SCC &\hspace{-1.3in}Upper SCC\\
			\includegraphics[scale=0.14]{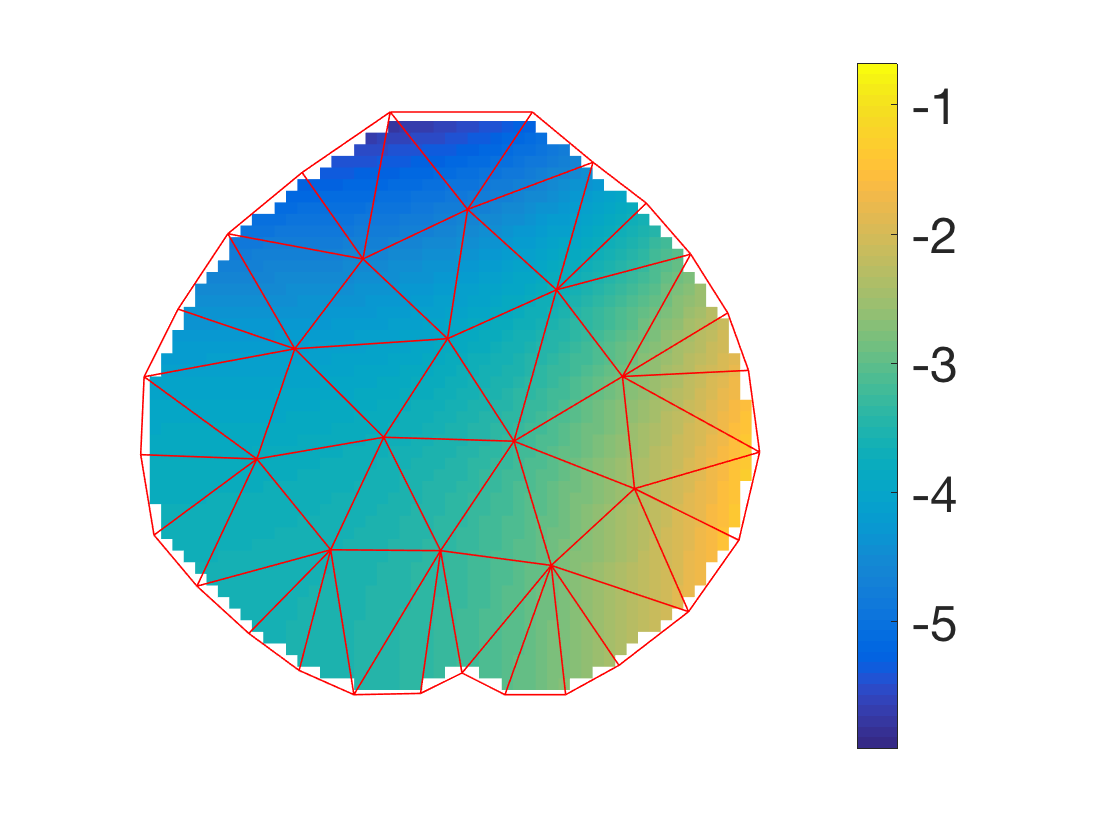} &\hspace{-0.8in}\includegraphics[scale=0.14]{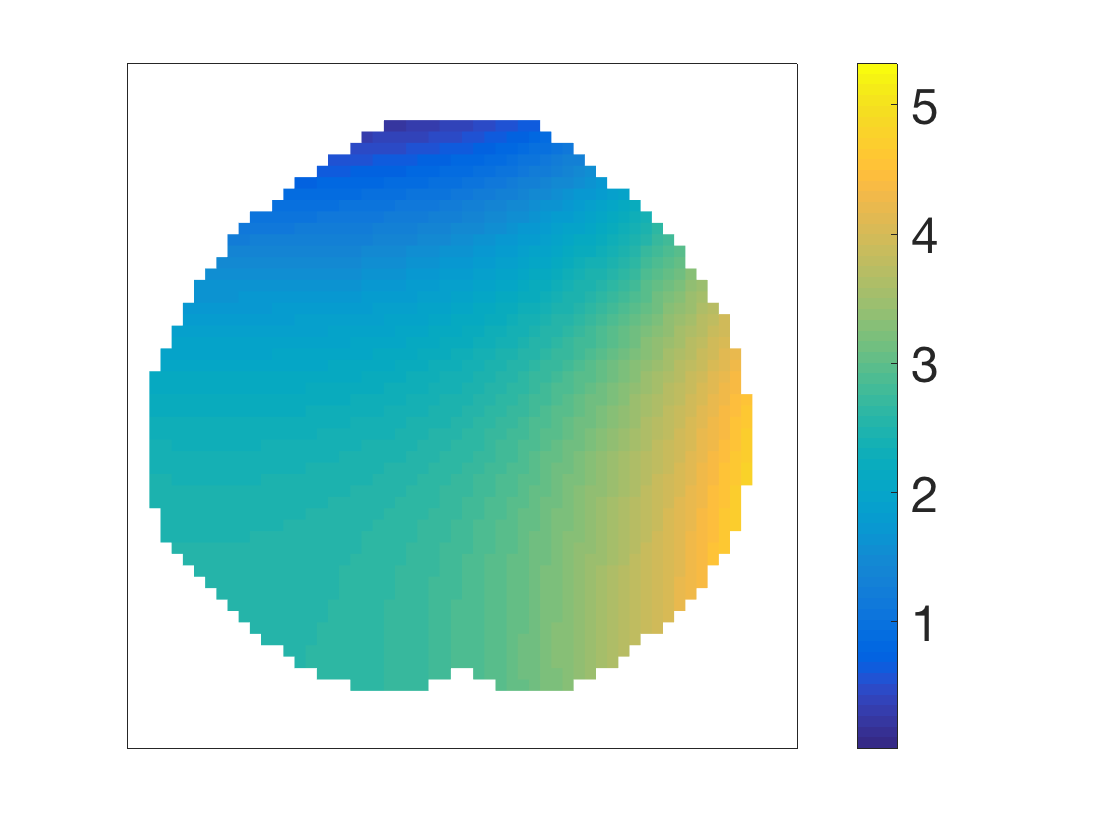} &\hspace{-0.7in}\includegraphics[scale=0.14]{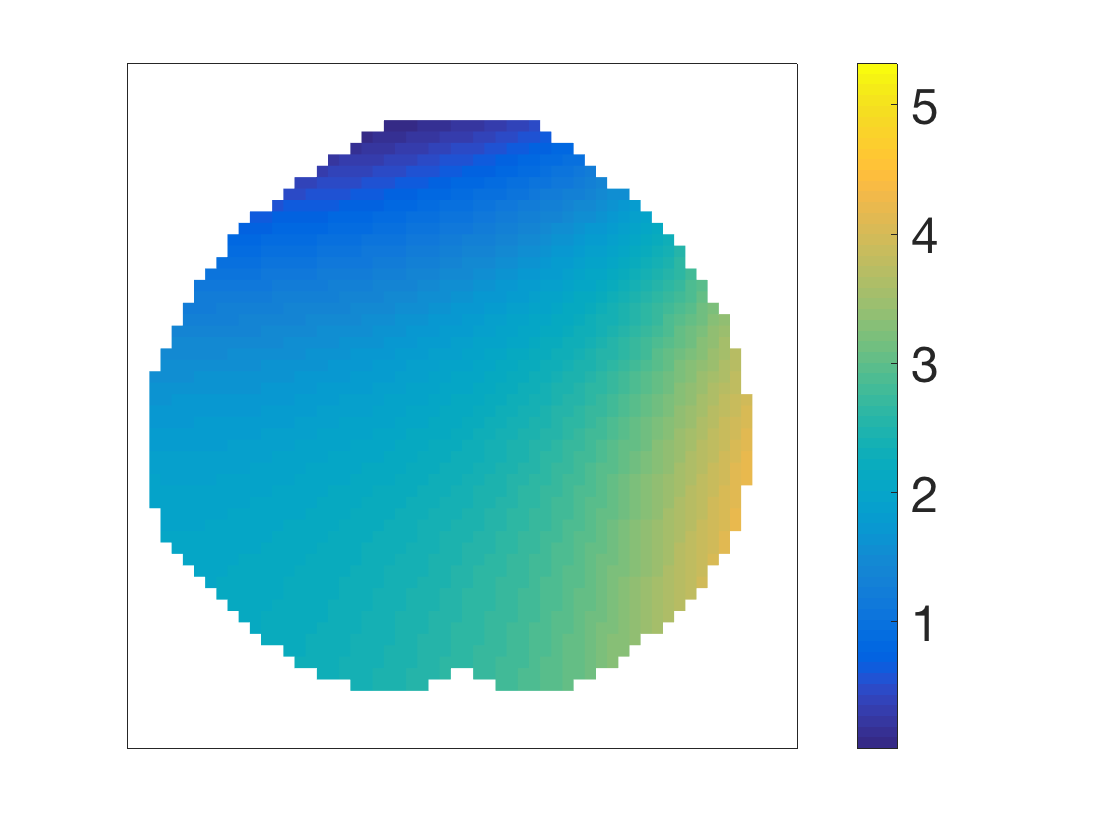} &\hspace{-0.7in}\includegraphics[scale=0.14]{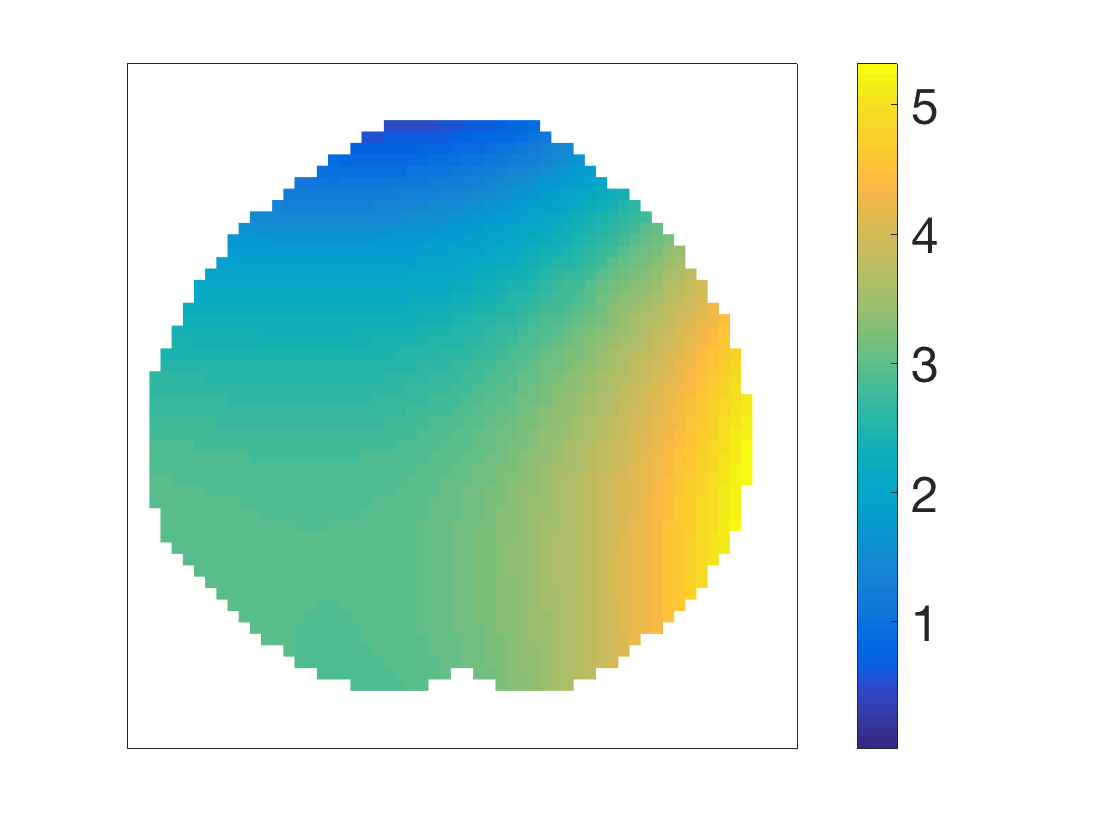}\\ [-10pt]
			\hspace{-0.3in}$\triangle_1$ & & &\\
			
			\includegraphics[scale=0.14]{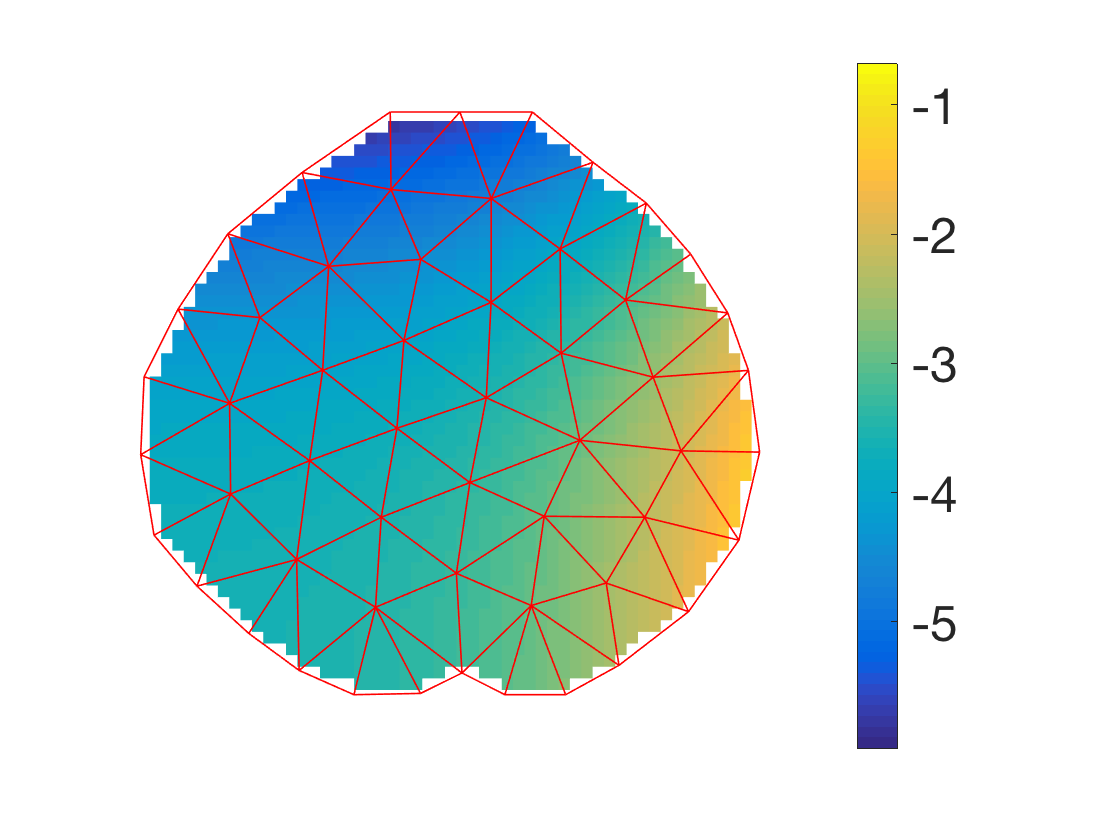} &\hspace{-0.8in}\includegraphics[scale=0.14]{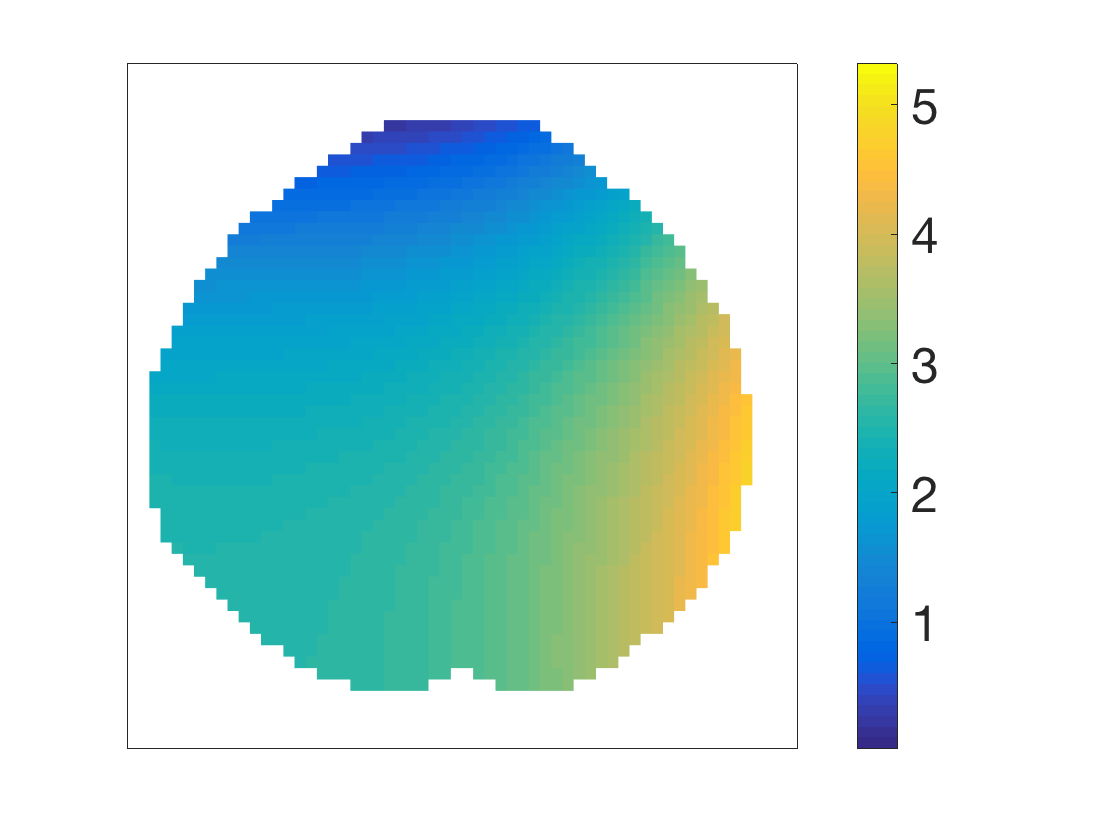} &\hspace{-0.7in}\includegraphics[scale=0.14]{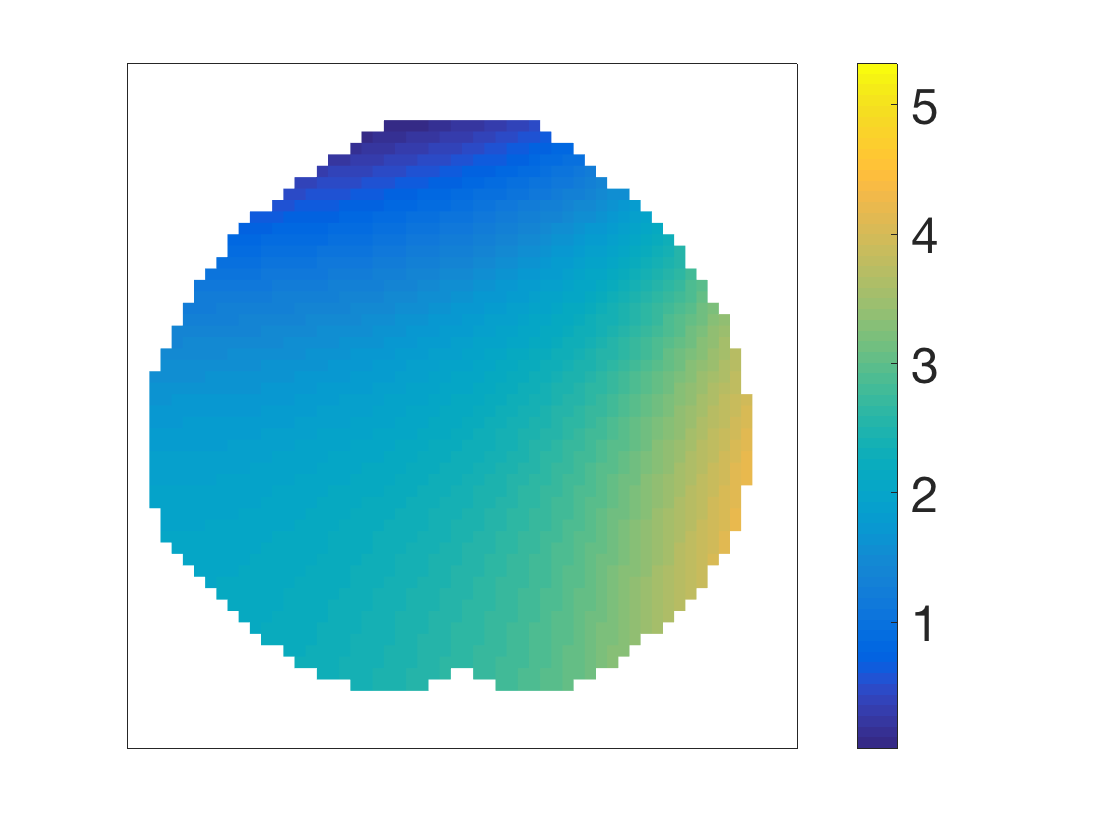} &\hspace{-0.7in}\includegraphics[scale=0.14]{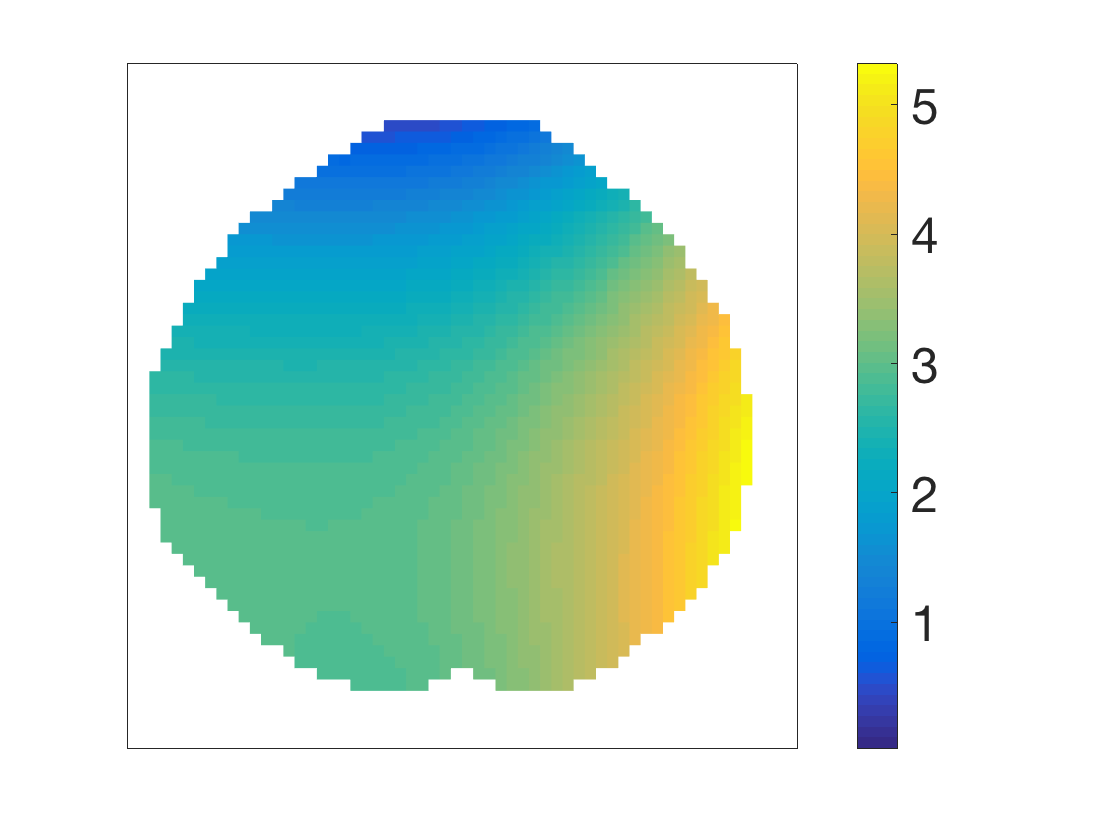}\\ [-10pt]
			\hspace{-0.3in}$\triangle_2$ & & &\\
			
			\includegraphics[scale=0.14]{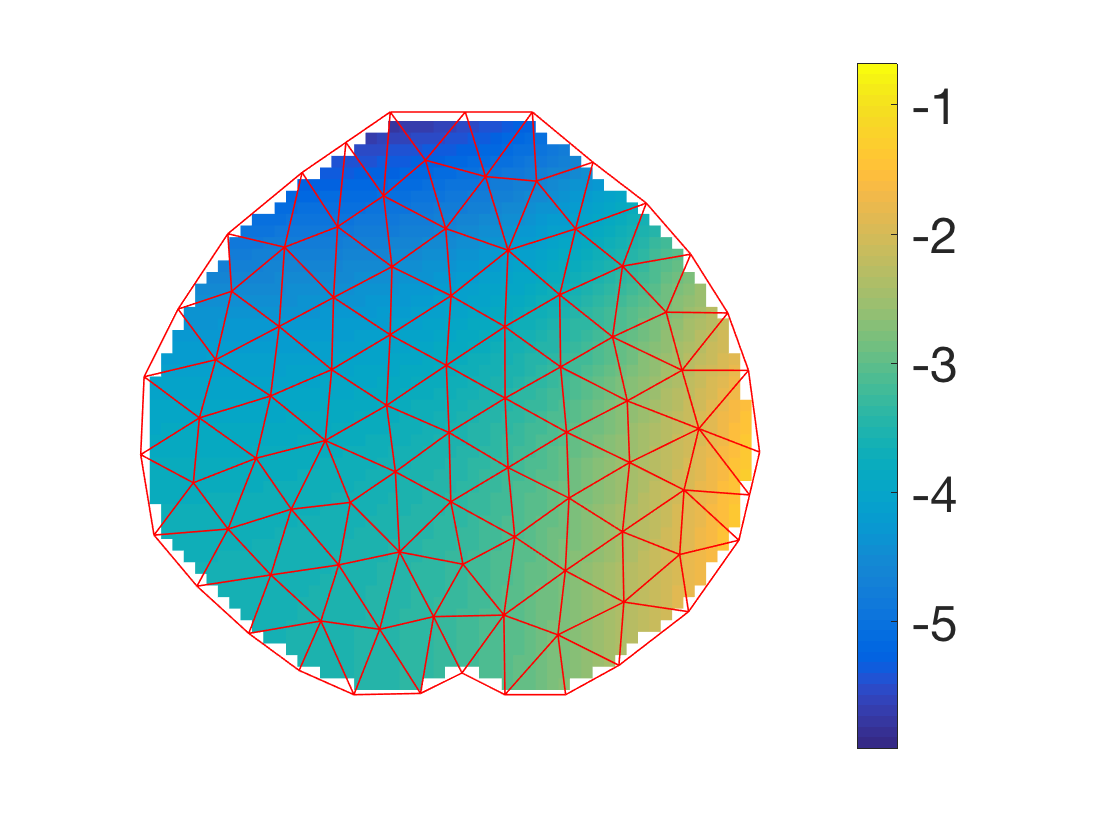} &\hspace{-0.8in}\includegraphics[scale=0.14]{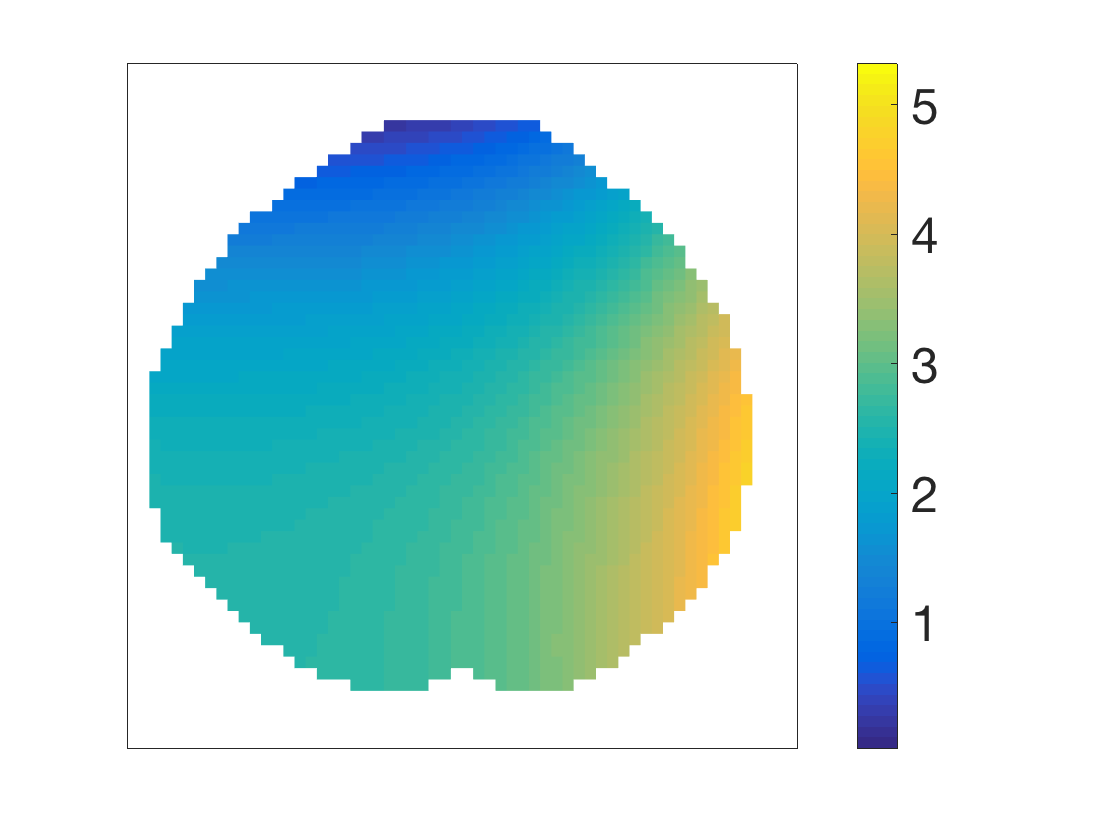} &\hspace{-0.7in}\includegraphics[scale=0.14]{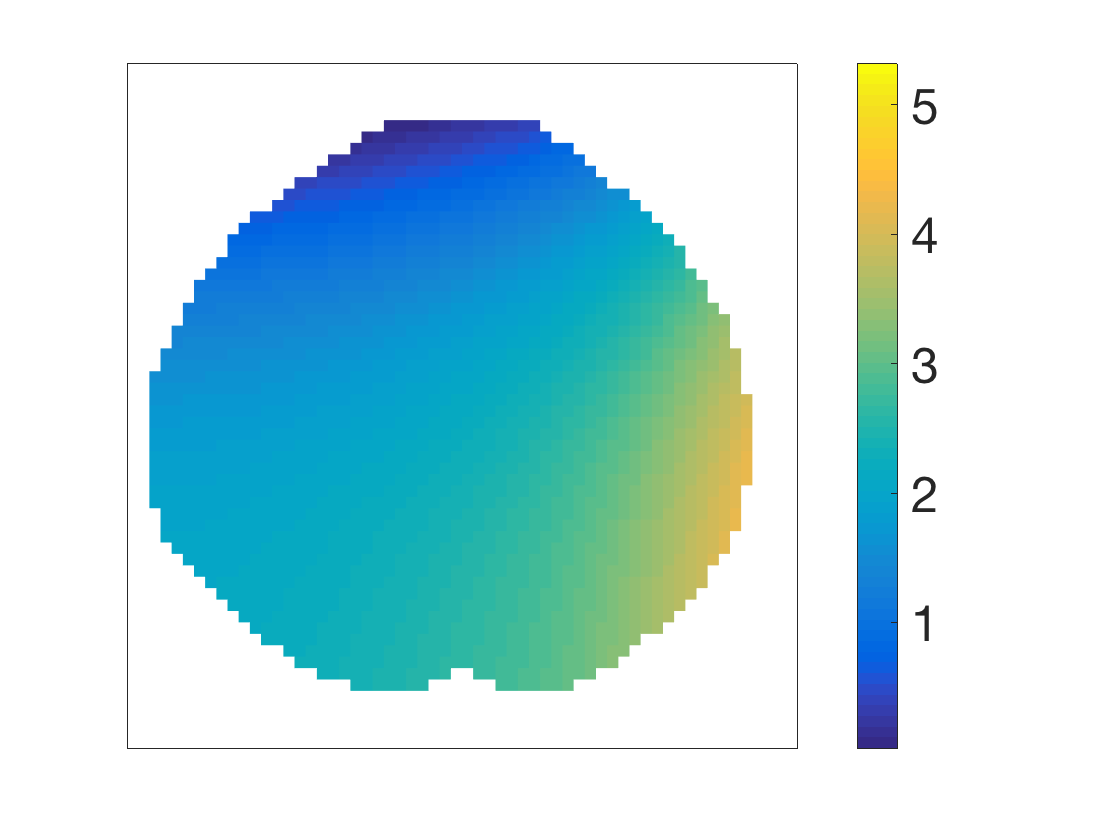} &\hspace{-0.7in}\includegraphics[scale=0.14]{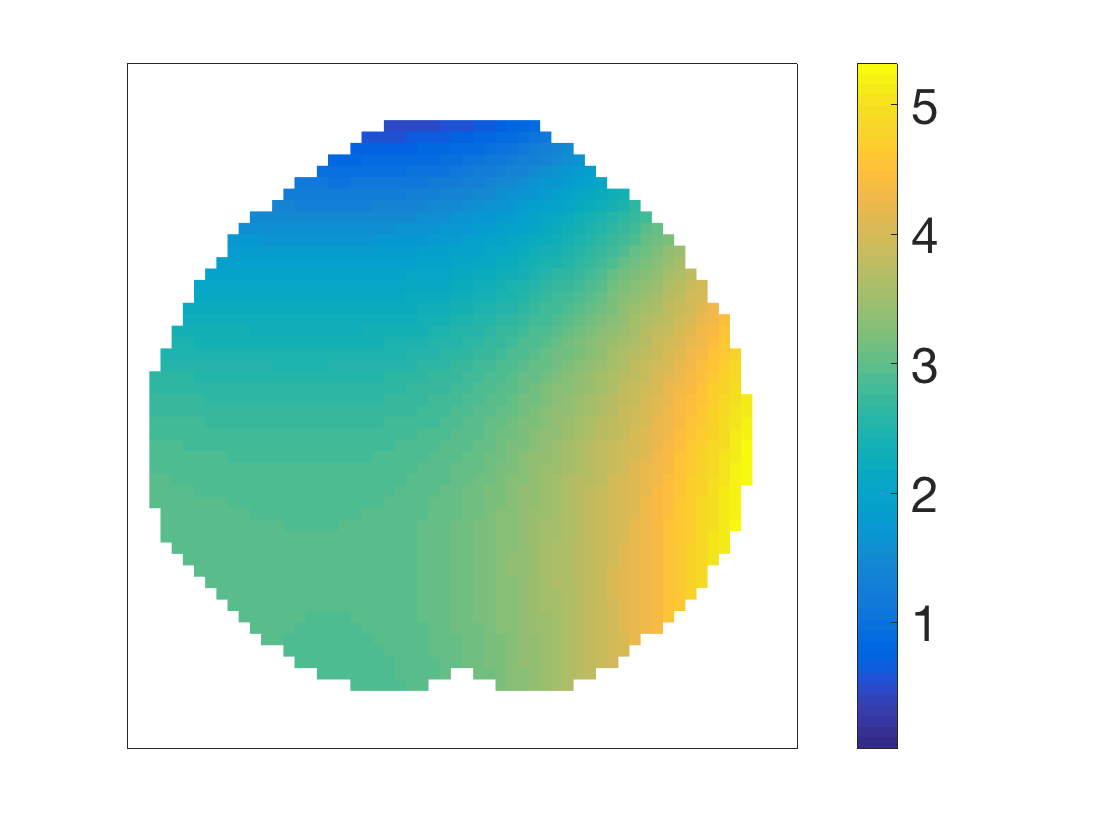}\\ [-10pt]
			\hspace{-0.3in}$\triangle_3$ & & &\\
			& & &
		\end{tabular}
	\end{center}
	\caption{SCCs for cubic function with $n=50$ and $\alpha=0.01$.}
	\label{FIG:S05}
\end{figure}

\begin{figure}
	\begin{center}
		\begin{tabular}{cccc}
			\hspace{-0.5in}Triangulation &\hspace{-1.3in}$\widehat{\mu}$ &\hspace{-1.3in}Lower SCC &\hspace{-1.3in}Upper SCC\\
			\includegraphics[scale=0.14]{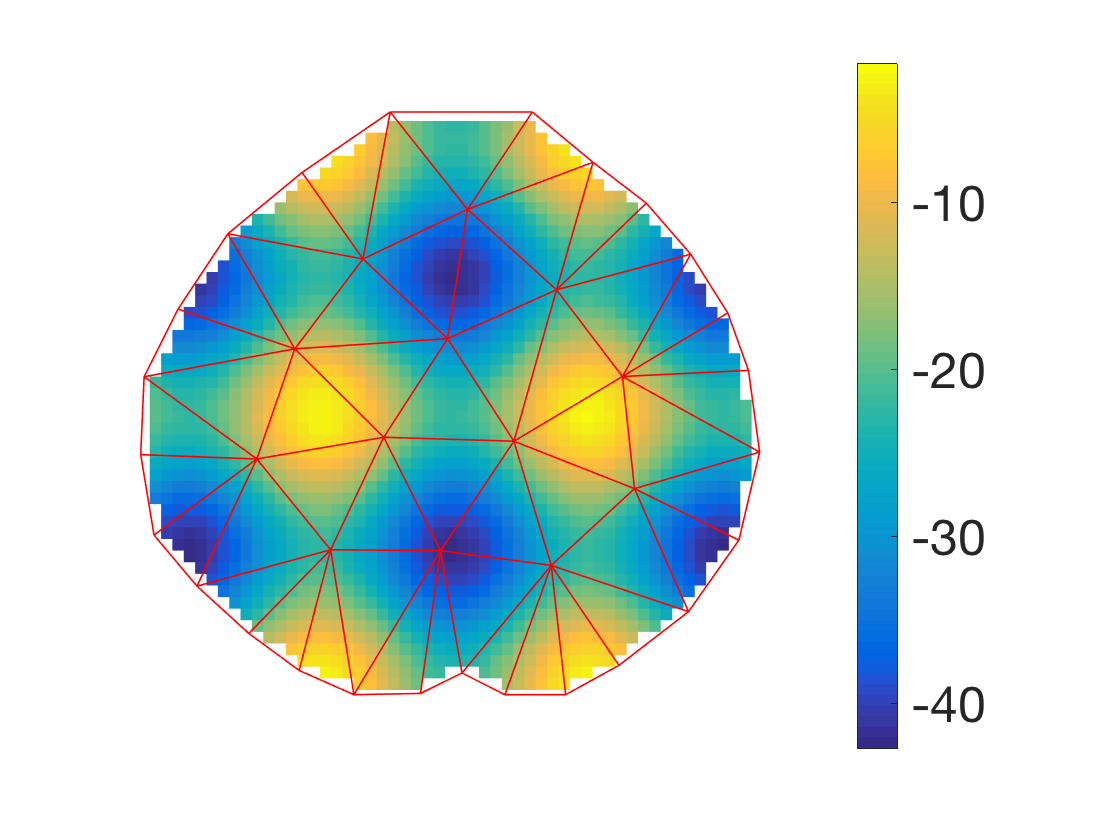} &\hspace{-0.8in}\includegraphics[scale=0.14]{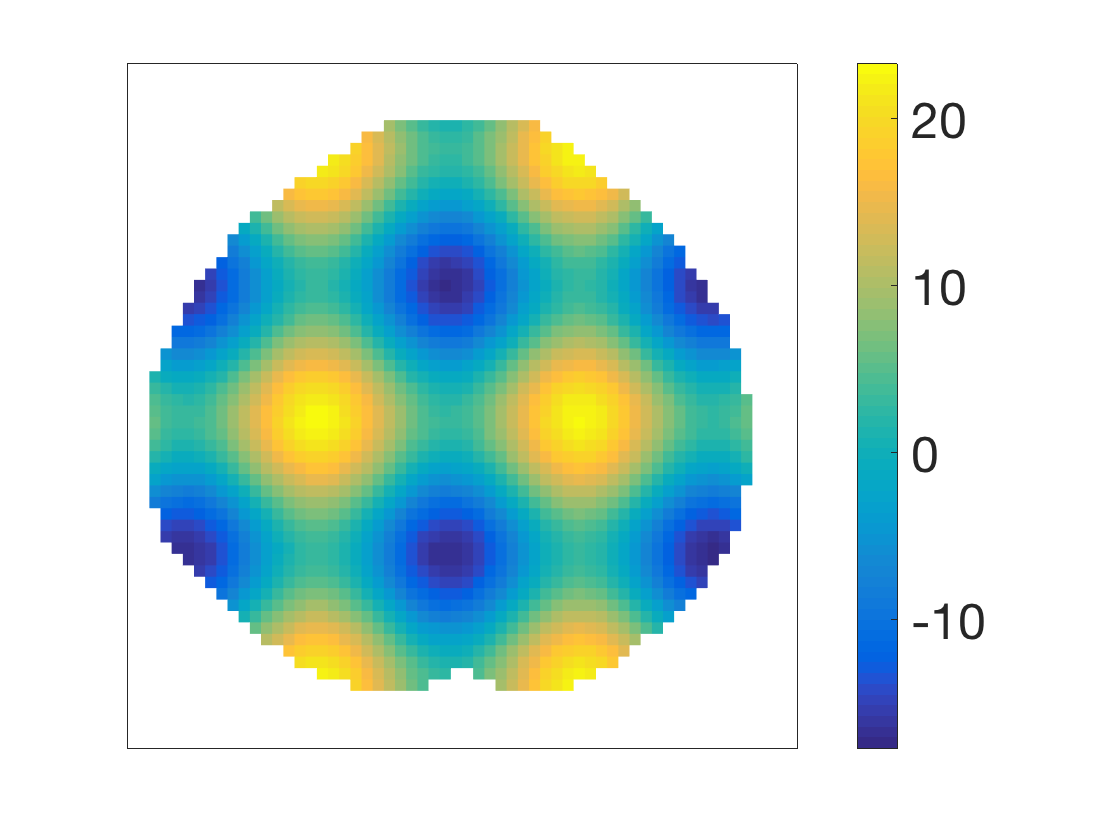} &\hspace{-0.7in}\includegraphics[scale=0.14]{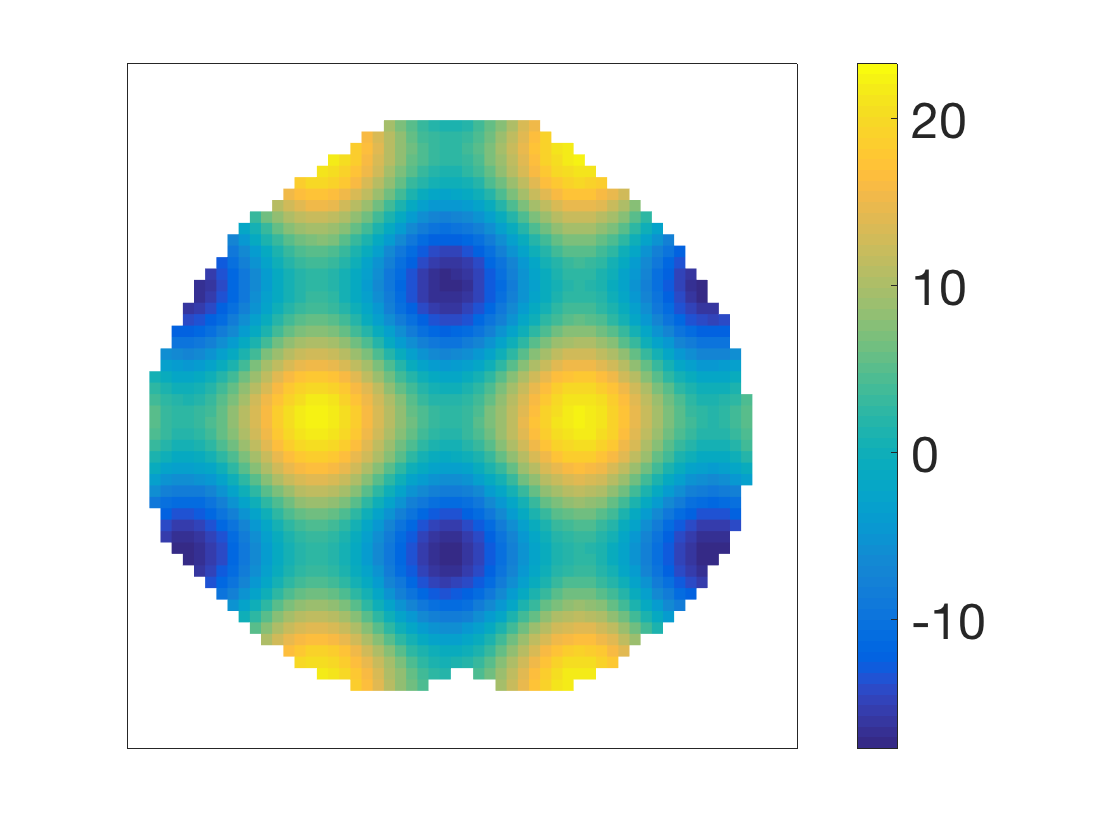} &\hspace{-0.7in}\includegraphics[scale=0.14]{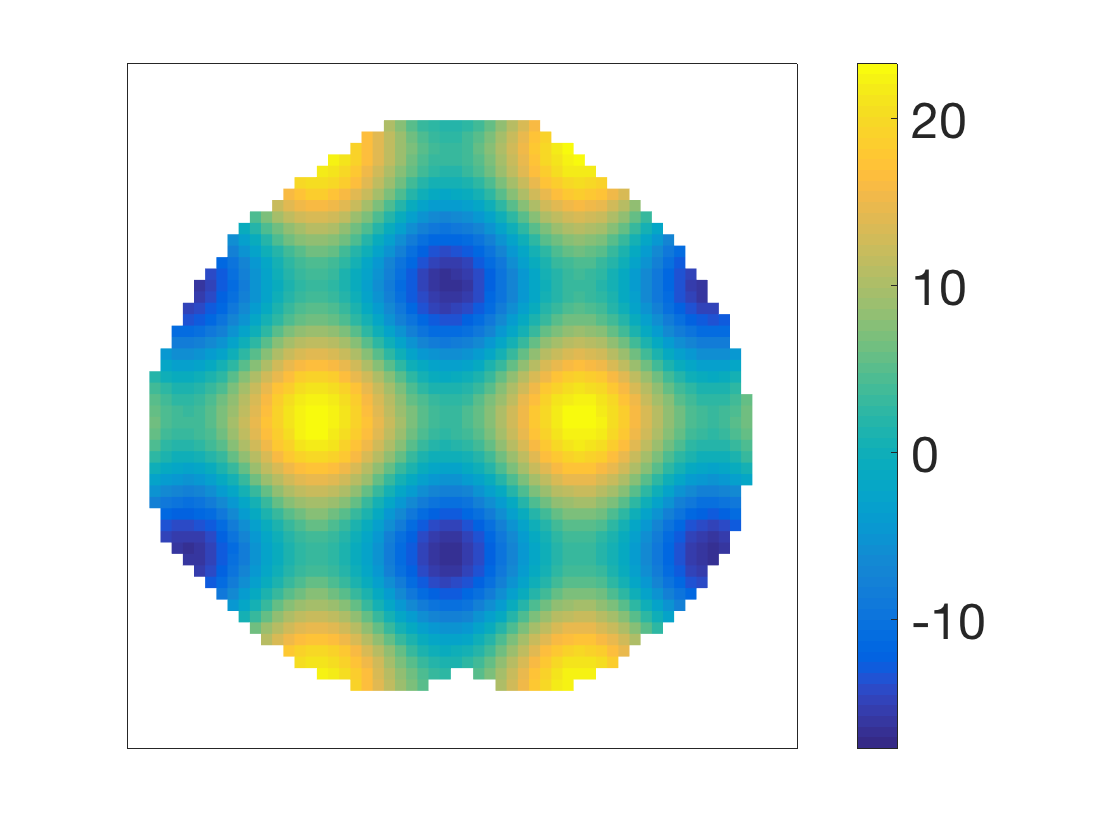}\\ [-10pt]
			\hspace{-0.3in}$\triangle_1$ & & &\\
			
			\includegraphics[scale=0.14]{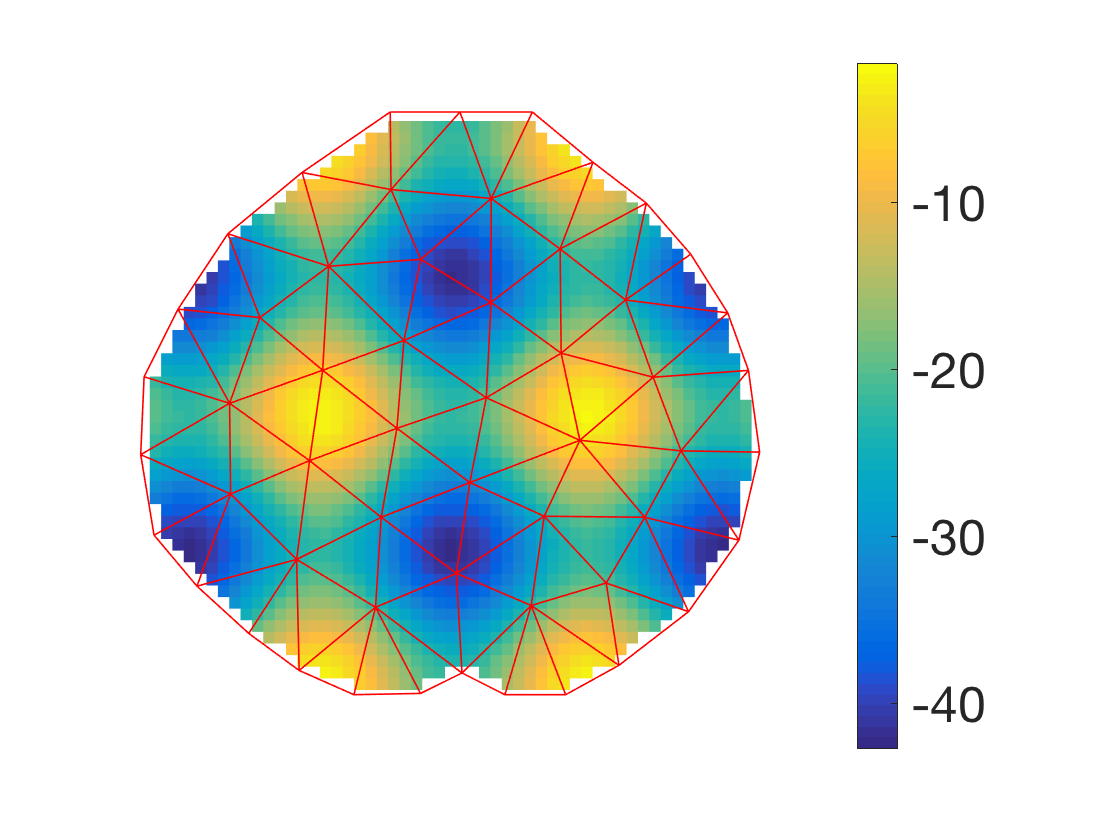} &\hspace{-0.8in}\includegraphics[scale=0.14]{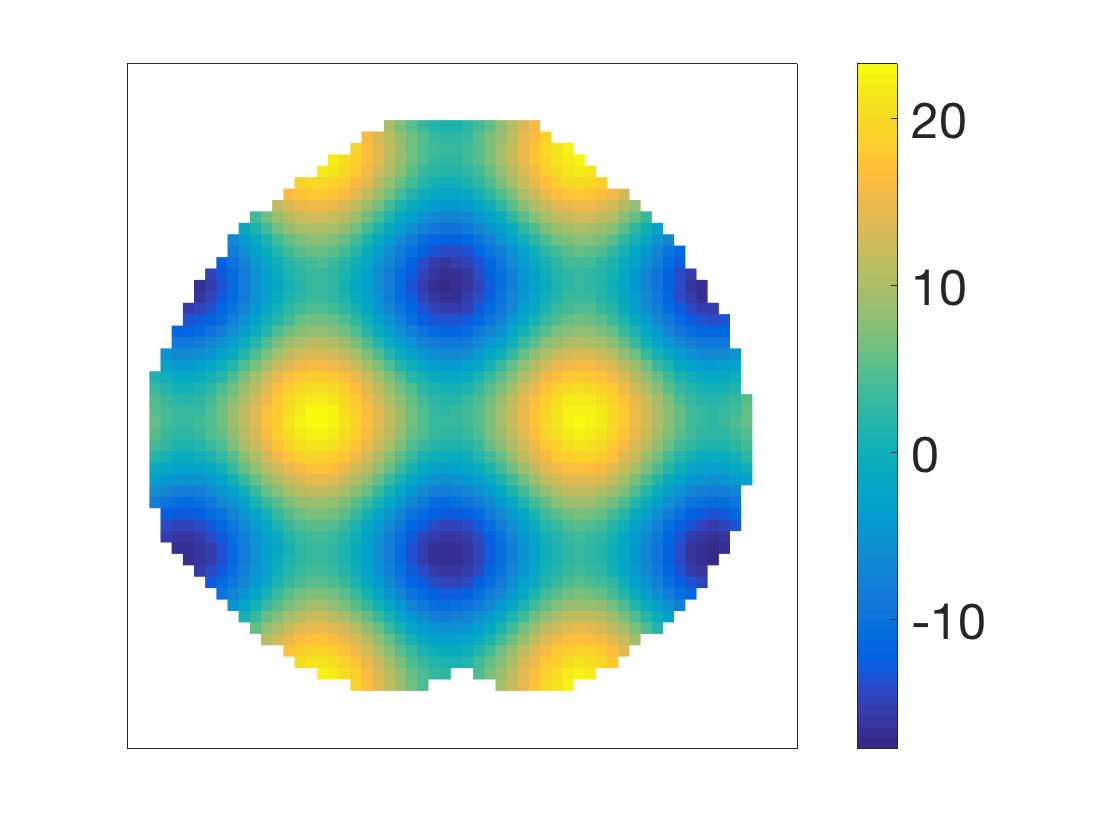} &\hspace{-0.7in}\includegraphics[scale=0.14]{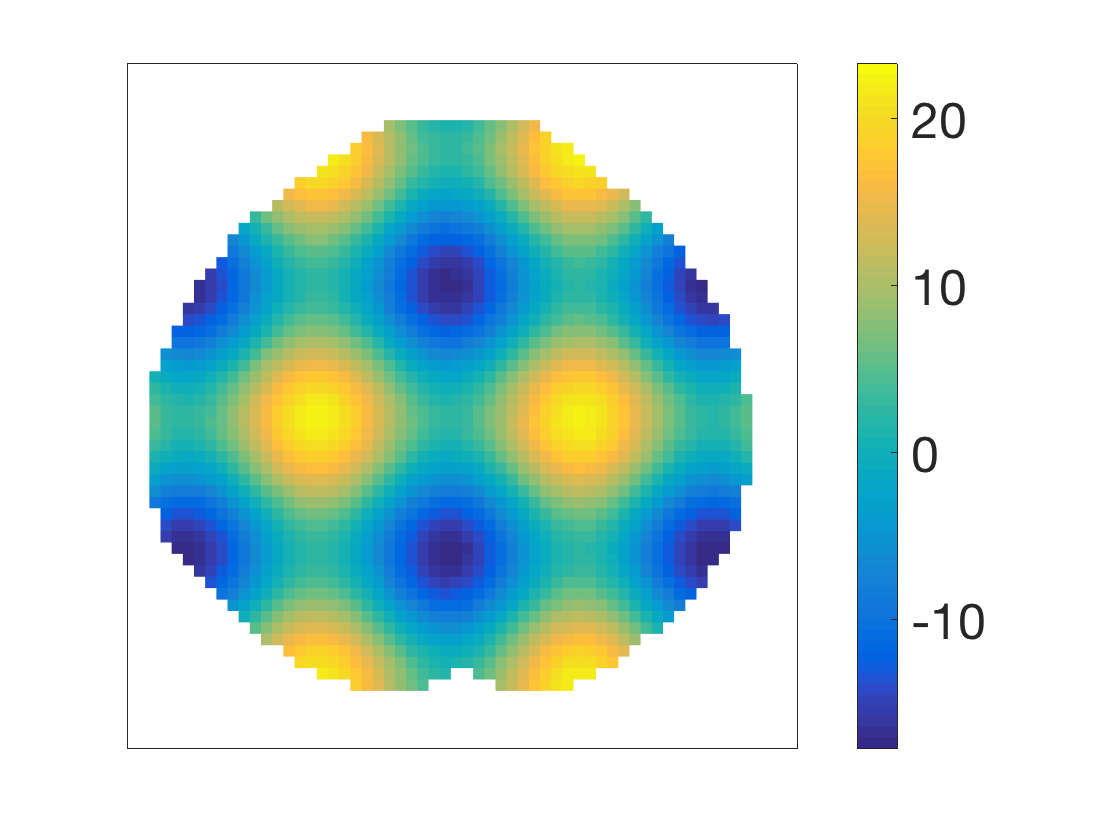} &\hspace{-0.7in}\includegraphics[scale=0.14]{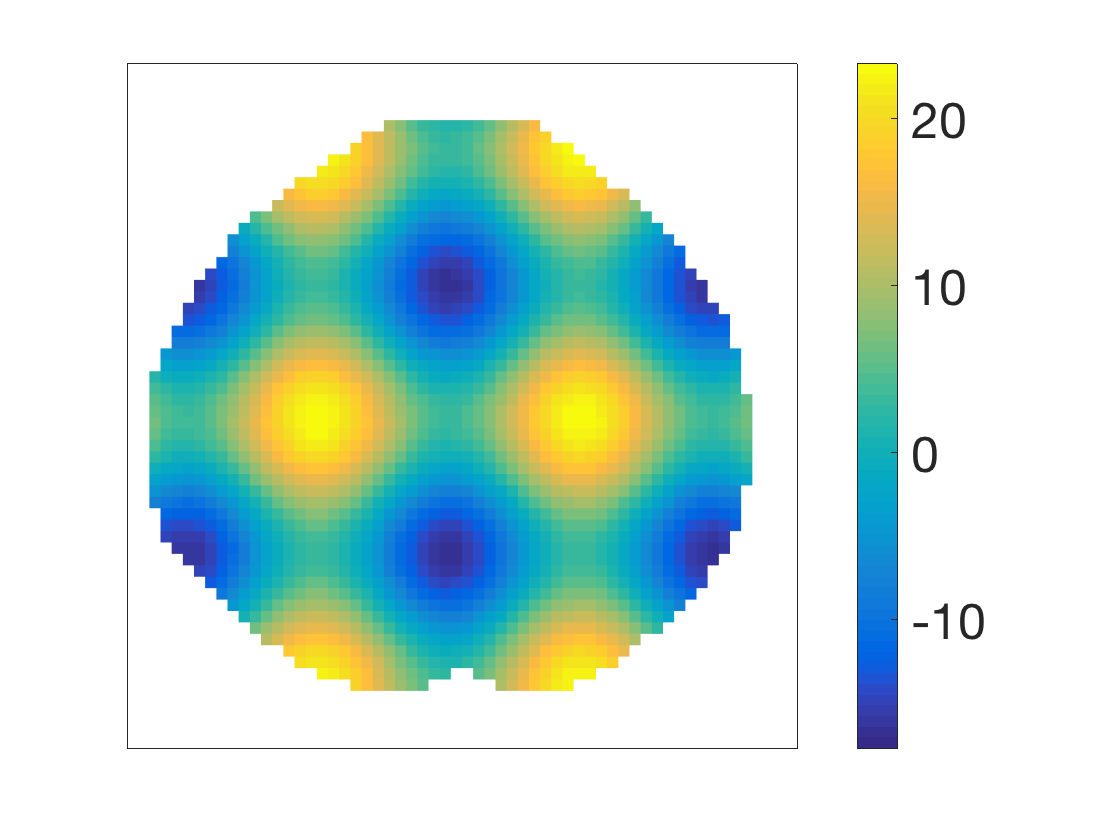}\\ [-10pt]
			\hspace{-0.3in}$\triangle_2$ & & &\\
			
			\includegraphics[scale=0.14]{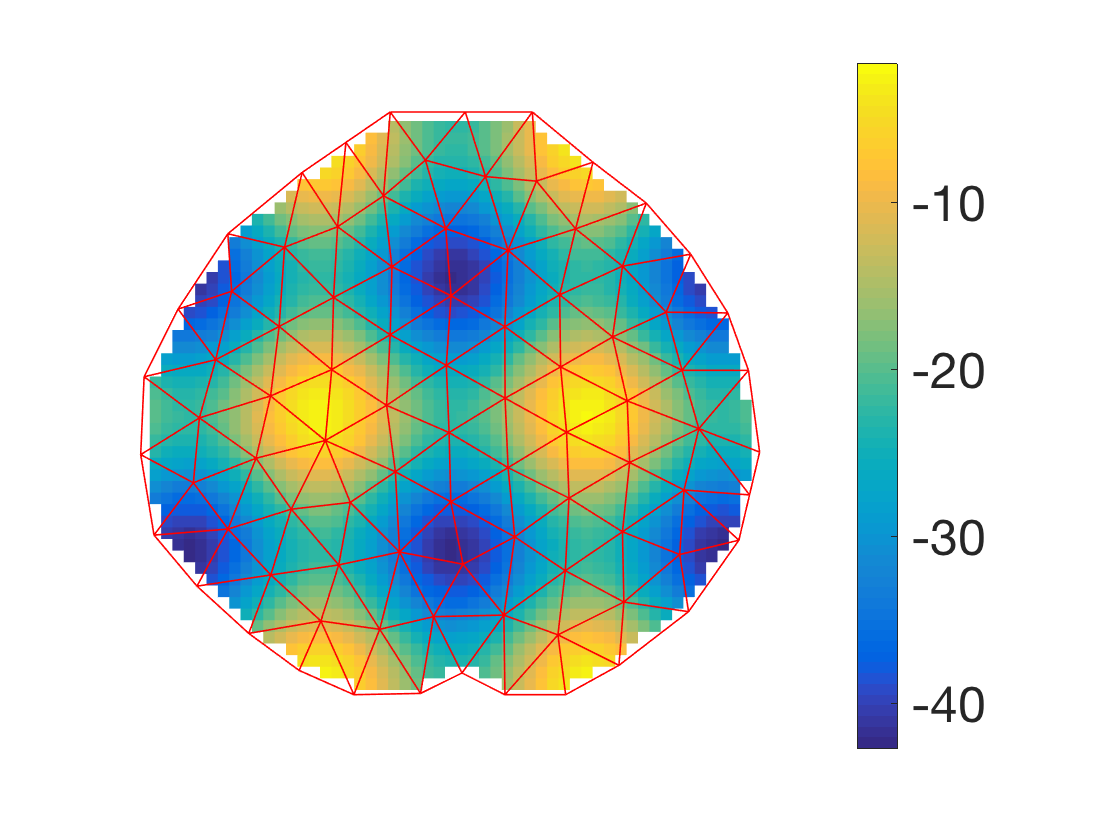} &\hspace{-0.8in}\includegraphics[scale=0.14]{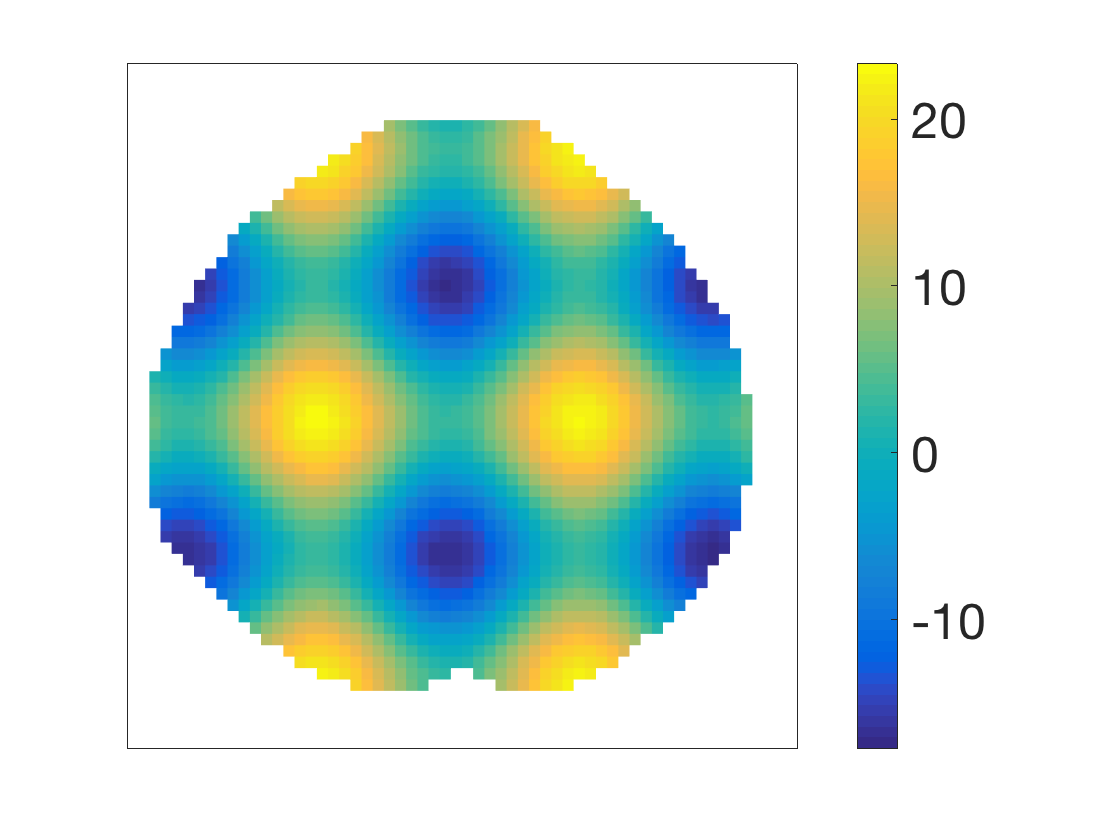} &\hspace{-0.7in}\includegraphics[scale=0.14]{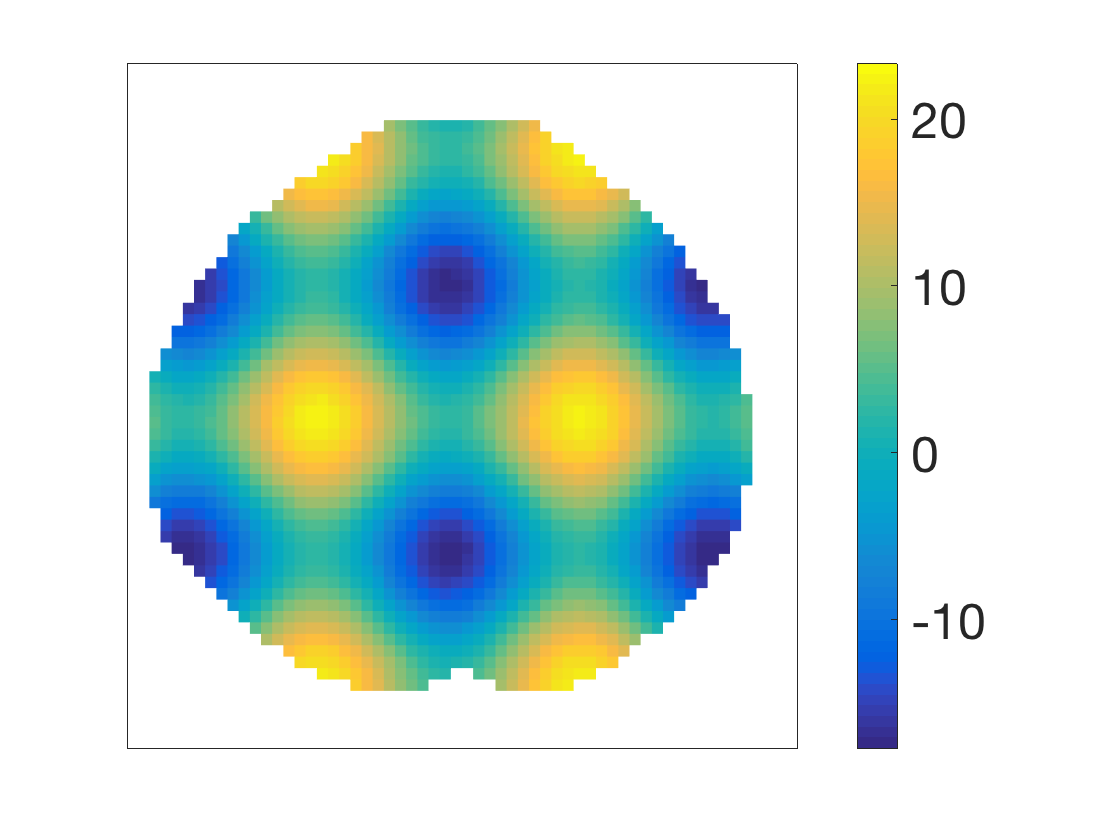} &\hspace{-0.7in}\includegraphics[scale=0.14]{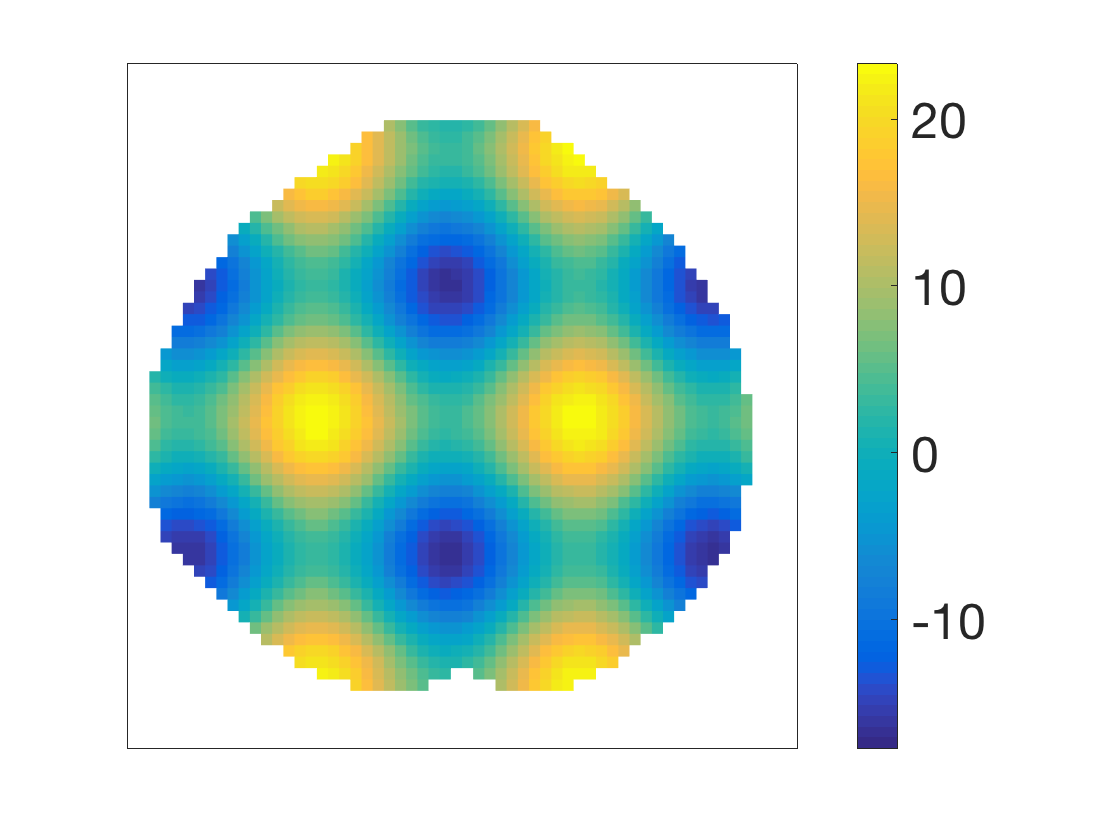}\\ [-10pt]
			\hspace{-0.3in}$\triangle_3$ & & &\\
			& & &
		\end{tabular}
	\end{center}
	\caption{SCCs for sine function with $n=50$ and $\alpha=0.01$.}
	\label{FIG:S06}
\end{figure}

\begin{figure}
	\begin{center}
		\begin{tabular}{cc}
			\includegraphics[scale=0.14]{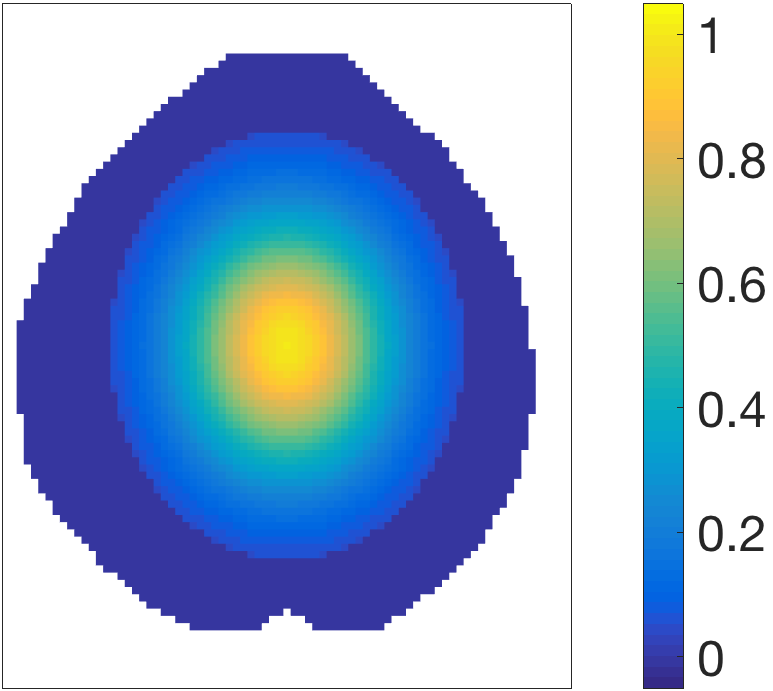} & \includegraphics[scale=0.14]{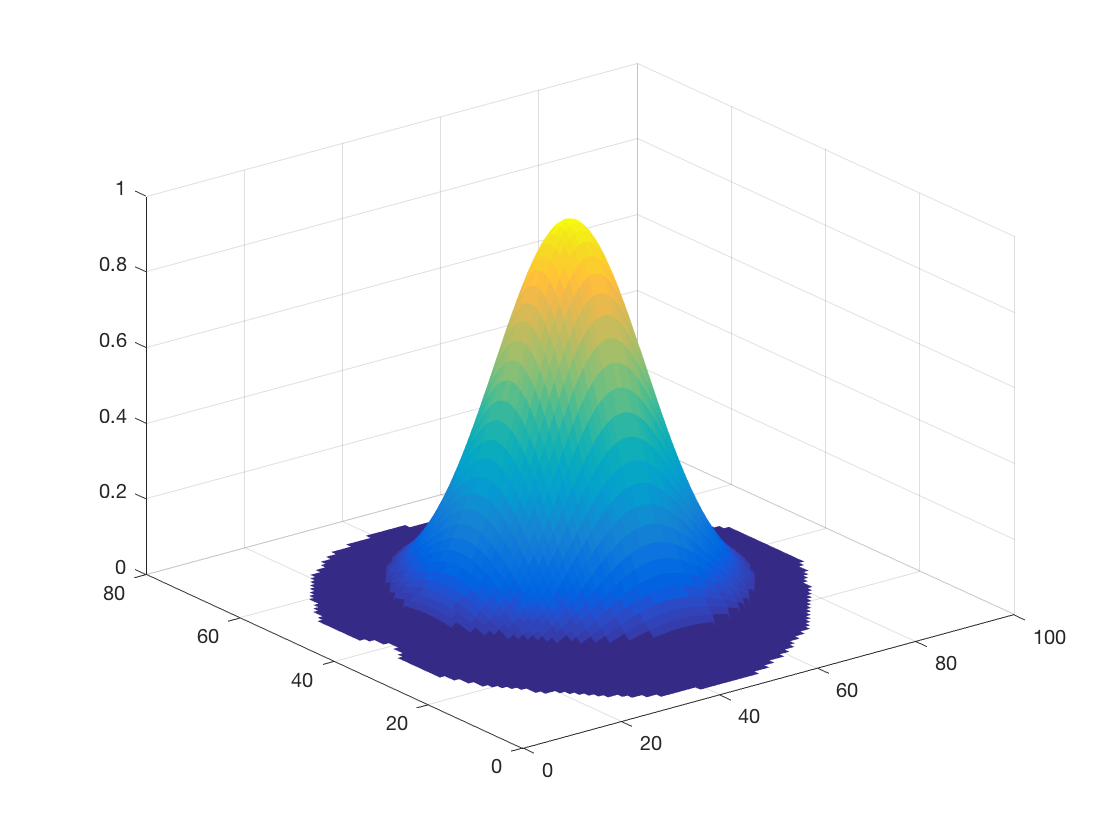}
		\end{tabular}
		\caption{True mean function: (a) image map and (b) surface plot.}
		\label{Fig:S07}
	\end{center}
\end{figure}

{\renewcommand{\baselinestretch}{1.2}
	\begin{figure}{h}
		\begin{center}
			\begin{tabular}{ccccc}
				\includegraphics[scale=0.14]{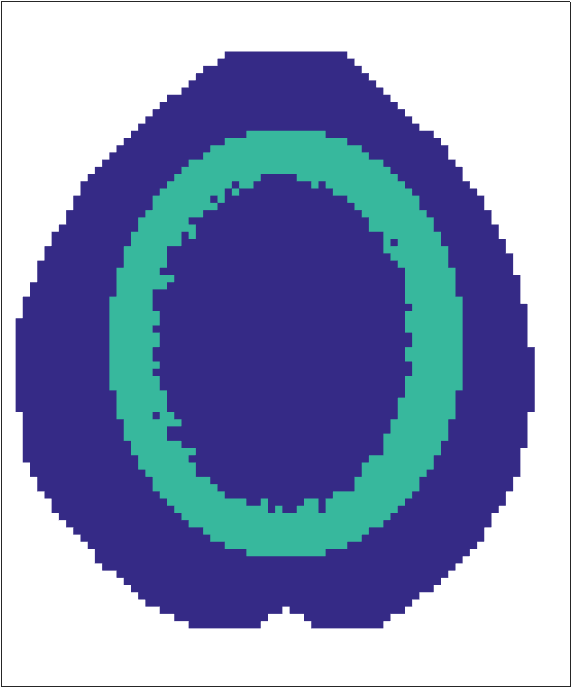} &\includegraphics[scale=0.14]{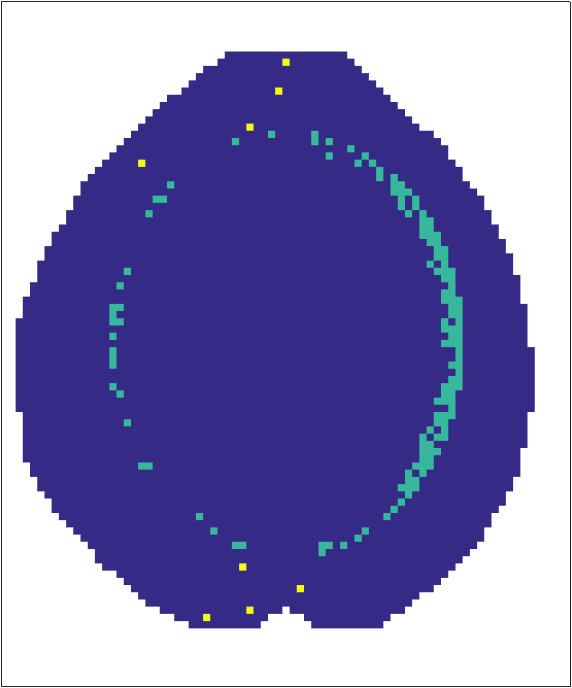} &\includegraphics[scale=0.14]{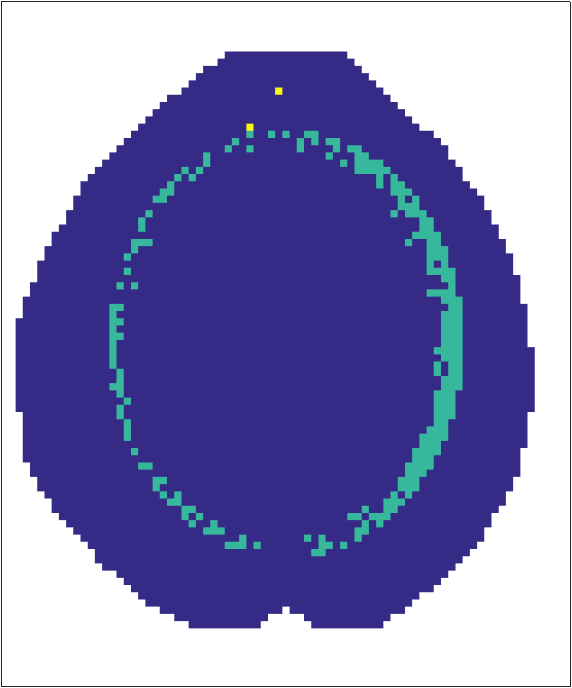} &\includegraphics[scale=0.14]{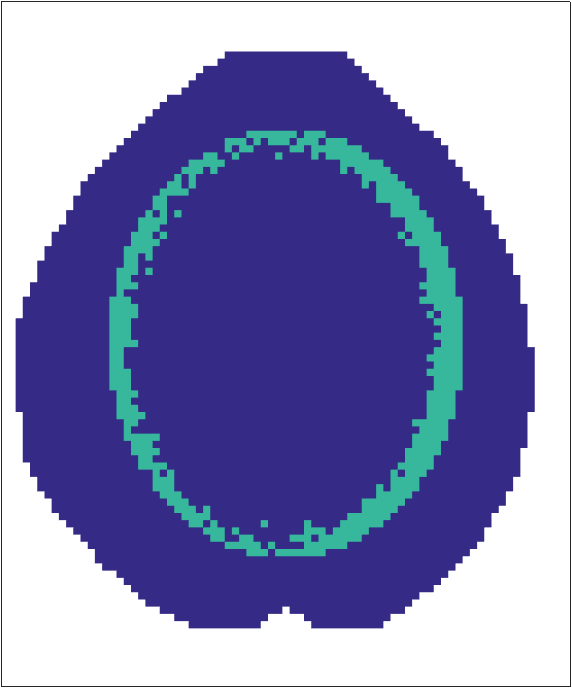} &\includegraphics[scale=0.14]{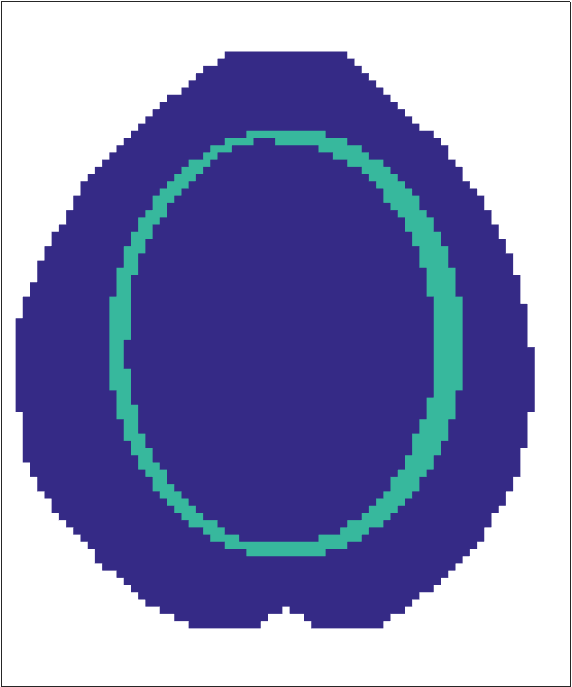}\\[6pt]
				(a) Bonferroni &(b) Cluster (0.1) &(b) Cluster (0.05) &(c) Cluster (0.01) &(d) SCC
			\end{tabular}
		\end{center}
		\caption{Signal discovery for one typical replication. Blue area shows the pixels correctly detected; yellow area shows the false positive pixels; and green area shows the false negative pixels.}
		\label{FIG:S08}
\end{figure}}

\begin{table}[t]
	\caption{Empirical coverage rates of the SCCs ($N=3682$).}
	\begin{center}
		\begin{tabular}{ccccccccccccccc} \hline \hline
			\multirow{2}{*}{$n$} &\multicolumn{3}{c}{$\alpha=0.10$} &\multicolumn{3}{c}{$\alpha=0.05$} &\multicolumn{3}{c}{$\alpha=0.01$}\\ \cline{2-10}
			&$\triangle_1$ &$\triangle_2$ &$\triangle_3$ &$\triangle_1$ &$\triangle_2$ &$\triangle_3$ &$\triangle_1$ &$\triangle_2$ &$\triangle_3$\\ \hline
			
			\multicolumn{10}{c}{$\mu(\bs{z})=20\left\{(z_1-0.5)^2+(z_2-0.5)^2\right\}$}\\ \hline
			\multirow{2}{*}{50} & 0.871 & 0.876 & 0.876 & 0.937 & 0.938 & 0.937 & 0.982 & 0.984 & 0.984 \\
			& (0.643) & (0.644) & (0.644) & (0.731) & (0.732) & (0.733) & (0.902) & (0.903) & (0.903) \\
			\multirow{2}{*}{100} & 0.885 & 0.881 & 0.882 & 0.939 & 0.942 & 0.941 & 0.979 & 0.979 & 0.979 \\
			& (0.460) & (0.458) & (0.458) & (0.522) & (0.521) & (0.521) & (0.643) & (0.641) & (0.642)  \\
			\multirow{2}{*}{200} & 0.901 & 0.902 & 0.883 & 0.949 & 0.949 & 0.941 & 0.987 & 0.988 & 0.987  \\
			& (0.330) & (0.331) & (0.326) & (0.374) & (0.375) & (0.370) & (0.460) & (0.461) & (0.457)  \\ \hline
			
			\multicolumn{10}{c}{$\mu(\bs{z})=5\exp\left[-15\left\{(z_1-0.5)^2+(z_2-0.5)^2\right\}\right]+0.5$}\\ \hline
			\multirow{2}{*}{50} & 0.868 & 0.871 & 0.871 & 0.934 & 0.934 & 0.934 & 0.982 & 0.984 & 0.982 \\
			& (0.643) & (0.644) & (0.644) & (0.731) & (0.732) & (0.733) & (0.902) & (0.903) & (0.903) \\
			\multirow{2}{*}{100} & 0.896 & 0.893 & 0.880 & 0.945 & 0.944 & 0.938 & 0.980 & 0.981 & 0.979 \\
			& (0.465) & (0.464) & (0.458) & (0.527) & (0.526) & (0.521) & (0.648) & (0.647) & (0.642)  \\
			\multirow{2}{*}{200} & 0.901 & 0.899 & 0.898 & 0.947 & 0.947 & 0.949 & 0.987 & 0.988 & 0.988  \\
			& (0.330) & (0.331) & (0.331) & (0.374) & (0.375) & (0.376) & (0.460) & (0.461) & (0.462)  \\ \hline
			
			\multicolumn{10}{c}{$\mu(\bs{z})=3.2(-z_1^3+z_2^3)+2.4$}\\ \hline	
			\multirow{2}{*}{50} & 0.860 & 0.870 & 0.869 & 0.927 & 0.931 & 0.929 & 0.985 & 0.987 & 0.987 \\
			& (0.628) & (0.628) & (0.629) & (0.716) & (0.716) & (0.718) & (0.887) & (0.887) & (0.889) \\
			\multirow{2}{*}{100} & 0.892 & 0.894 & 0.895 & 0.942 & 0.947 & 0.947 & 0.982 & 0.983 & 0.983  \\
			& (0.451) & (0.451) & (0.452) & (0.514) & (0.514) & (0.515) & (0.635) & (0.635) & (0.635)  \\
			\multirow{2}{*}{200} & 0.899 & 0.902 & 0.898 & 0.942 & 0.947 & 0.949 & 0.988 & 0.988 & 0.989  \\
			& (0.320) & (0.320) & (0.320) & (0.364) & (0.365) & (0.365) & (0.451) & (0.451) & (0.451)  \\ \hline
			
			\multicolumn{10}{c}{$\mu(\bs{z})=-10[\sin\{5\pi (z_{1}+0.22)\}-\sin\{5\pi (z_{2}-0.18)\}]+2.8$}\\ \hline
			\multirow{2}{*}{50} & 0.885 & 0.892 & 0.867 & 0.940 & 0.943 & 0.928 & 0.983 & 0.985 & 0.982  \\
			& (0.703) & (0.717) & (0.719) & (0.790) & (0.804) & (0.807) & (0.959) & (0.973) & (0.977) \\
			\multirow{2}{*}{100} & 0.894 & 0.890 & 0.883 & 0.943 & 0.946 & 0.934 & 0.979 & 0.981 & 0.977 \\
			& (0.500) & (0.509) & (0.516) & (0.562) & (0.571) & (0.578) & (0.681) & (0.691) & (0.699)  \\
			\multirow{2}{*}{200} & 0.899 & 0.899 & 0.892 & 0.946 & 0.947 & 0.946 & 0.988 & 0.988 & 0.987 \\
			& (0.354) & (0.361) & (0.368) & (0.398) & (0.405) & (0.412) & (0.483) & (0.490) & (0.497)  \\ 
			\hline \hline
			
		\end{tabular}
	\end{center}
	\label{TAB:S02}
\end{table}

\begin{table}[t]
	\caption{Type I error and empirical power of two-sample test}
	\begin{center}
		\begin{tabular}{p{0.80cm}p{0.85cm}p{0.85cm}p{0.85cm}p{0.85cm}p{0.85cm}p{0.85cm}p{0.85cm}p{0.85cm}p{0.85cm}p{0.85cm}}
			\hline \hline
			\multirow{2}{*}{$n$}
			& \multirow{2}{*}{$\triangle$}
			& \multicolumn{9}{c}{$\delta$}\\
			\cline{3-11}
			& & 0.00  & 0.10 &  0.20 &  0.30  & 0.40  & 0.50  & 0.60  & 0.70 & 0.80\\
			\hline
			\multicolumn{11}{c}{ $\alpha=0.10$}\\
			\hline
			\multirow{3}{*}{50} & 49 & 0.110 & 0.204 & 0.374 & 0.620 & 0.842 & 0.974 & 1.000 & 1.000 & 1.000 \\
			& 80 & 0.102 & 0.194 & 0.366 & 0.612 & 0.844 & 0.972 & 1.000 & 1.000 & 1.000  \\
			& 144 & 0.101 & 0.192 & 0.367 & 0.616 & 0.844 & 0.976 & 1.000 & 1.000 & 1.000  \\
			\multirow{3}{*}{100} & 49 & 0.108 & 0.249 & 0.560 & 0.886 & 0.999 & 1.000 & 1.000 & 1.000 & 1.000 \\
			& 80 & 0.106 & 0.242 & 0.549 & 0.884 & 1.000 & 1.000 & 1.000 & 1.000 & 1.000 \\
			& 144 & 0.103 & 0.247 & 0.559 & 0.886 & 0.999 & 1.000 & 1.000 & 1.000 & 1.000 \\
			\multirow{3}{*}{200} & 49 & 0.087 & 0.334 & 0.848 & 1.000 & 1.000 & 1.000 & 1.000 & 1.000 & 1.000 \\
			& 80 & 0.085 & 0.319 & 0.836 & 1.000 & 1.000 & 1.000 & 1.000 & 1.000 & 1.000 \\
			& 144 & 0.082 & 0.325 & 0.844 & 1.000 & 1.000 & 1.000 & 1.000 & 1.000 & 1.000  \\\hline
			\multicolumn{11}{c}{ $\alpha=0.05$}\\
			\hline
			\multirow{3}{*}{50} & 49 & 0.053 & 0.110 & 0.250 & 0.474 & 0.700 & 0.900 & 0.992 & 1.000 & 1.000 \\
			& 80 & 0.049 & 0.101 & 0.244 & 0.467 & 0.692 & 0.894 & 0.988 & 1.000 & 1.000  \\
			& 144 & 0.051 & 0.107 & 0.252 & 0.472 & 0.699 & 0.899 & 0.989 & 1.000 & 1.000 \\
			\multirow{3}{*}{100} & 49 & 0.058 & 0.153 & 0.414 & 0.779 & 0.973 & 1.000 & 1.000 & 1.000 & 1.000  \\
			& 80 &  0.056 & 0.150 & 0.405 & 0.766 & 0.966 & 1.000 & 1.000 & 1.000 & 1.000\\
			& 144 & 0.056 & 0.151 & 0.415 & 0.770 & 0.969 & 1.000 & 1.000 & 1.000 & 1.000  \\
			\multirow{3}{*}{200} & 49 & 0.037 & 0.217 & 0.697 & 0.992 & 1.000 & 1.000 & 1.000 & 1.000 & 1.000 \\
			& 80 & 0.037 & 0.211 & 0.685 & 0.992 & 1.000 & 1.000 & 1.000 & 1.000 & 1.000  \\
			& 144 & 0.035 & 0.217 & 0.696 & 0.992 & 1.000 & 1.000 & 1.000 & 1.000 & 1.000 \\\hline
			
			\multicolumn{11}{c}{ $\alpha=0.01$}\\
			\hline
			\multirow{3}{*}{50} & 49 & 0.014 & 0.026 & 0.088 & 0.241 & 0.462 & 0.692 & 0.882 & 0.982 & 1.000 \\
			& 80 & 0.012 & 0.025 & 0.087 & 0.228 & 0.453 & 0.677 & 0.875 & 0.977 & 1.000 \\
			& 144 & 0.010 & 0.027 & 0.089 & 0.235 & 0.463 & 0.690 & 0.882 & 0.983 & 1.000 \\
			\multirow{3}{*}{100} & 49 & 0.013 & 0.032 & 0.181 & 0.509 & 0.825 & 0.976 & 1.000 & 1.000 & 1.000 \\
			& 80 & 0.012 & 0.032 & 0.172 & 0.486 & 0.817 & 0.978 & 0.999 & 1.000 & 1.000 \\
			& 144 & 0.012 & 0.032 & 0.186 & 0.509 & 0.828 & 0.979 & 0.999 & 1.000 & 1.000 \\
			\multirow{3}{*}{200} & 49 & 0.009 & 0.071 & 0.417 & 0.890 & 0.999 & 1.000 & 1.000 & 1.000 & 1.000 \\
			& 80 & 0.009 & 0.065 & 0.402 & 0.884 & 0.998 & 1.000 & 1.000 & 1.000 & 1.000 \\
			& 144 &0.009 & 0.068 & 0.420 & 0.884 & 0.998 & 1.000 & 1.000 & 1.000 & 1.000 \\\hline \hline
		\end{tabular}
	\end{center}
	\label{TAB:S01}
\end{table}

\begin{table}[t]
	\caption{FPRs, FNRs and FDRs for different methods.}
	\begin{center}
		\begin{tabular}{ccccccc} \hline \hline
			\multirow{2}{*}{$n$} &\multirow{2}{*}{Criterion} &\multicolumn{5}{c}{Method}\\ \cline{3-7}
			& &Bonferroni &Cluster (0.10) &Cluster (0.05) &Cluster (0.01) &SCC\\ \hline
			\multirow{3}{*}{100} &FPR &0.0000 &0.0472 &0.0233 &0.0067 &0.0090\\
			&FNR &0.3158 &0.1288 &0.1567 &0.2071 &0.1868\\
			&FDR &0.0000 &0.0876 &0.0449 &0.0142 &0.0169\\ \hline
			
			\multirow{3}{*}{200} &FPR &0.0000 &0.0534 &0.0260 &0.0044 &0.0043\\
			&FNR &0.2497 &0.0836 &0.1051 &0.1485 &0.1377\\
			&FDR &0.0000 &0.0893 &0.0478 &0.0081 &0.0062\\ \hline\hline
		\end{tabular}
	\end{center}
	\label{TAB:S00}
\end{table}

\bibliographystyle{asa}
\bibliography{biomreferences}

\end{document}